\documentclass[11pt]{article}
\usepackage[in]{fullpage}
\usepackage[titletoc]{appendix}

\usepackage{amssymb}
\usepackage{amsfonts}
\usepackage{amsmath,mathtools,amsthm}
\usepackage{graphicx}
\usepackage{enumerate}
\usepackage{array,multirow}
\usepackage{color}

\newcommand{\field}[1]{\mathbb{#1}}
\newcommand{\C}{\mathcal{C}}
\newcommand{\R}{\field{R}}
\newcommand{\Z}{\ensuremath{\field{Z}}}
\newcommand{\E}{\ensuremath{\field{E}}}

\newcommand{\N}{\field{N}}

\newcommand{\eat}[1]{}

\newcounter{ccc}

\newcounter{rc}
\newcommand{\brc}{\setcounter{rc}{0}}

\newcommand{\irc}{\addtocounter{rc}{1} \therc.}

\newcommand{\stream}{\ensuremath{\mathcal{S}}}
\newcommand{\event}[1]{\ensuremath{\mathcal{{#1}}}}

\providecommand{\abs}[1]{\lvert#1\rvert}

\providecommand{\card}[1]{\bigl\lvert#1\bigr\rvert}

\providecommand{\prob}[1]{\ensuremath{\text{\sf Pr}\left[ #1
\right]}} 

\newcommand{\probb}[1]{\ensuremath{\text{\sf Pr}\bigl[ #1
\bigr]}}
\newcommand{\probB}[1]{\ensuremath{\text{\sf Pr}\Bigl[ #1
\Bigr]}}

\newcommand{\expect}[1]{\ensuremath{\E\left[
#1 \right]}}

\newcommand{\expectb}[1]{\ensuremath{\E\bigl[ {#1} \bigr]}}

\newcommand{\expectbb}[1]{\ensuremath{\E\biggl[ {#1} \biggr]}}

\renewcommand{\exp}[1]{\ensuremath{\textrm{ exp}\left\{ {#1} \right\}}}

\newcommand{\probfi}[1]{\ensuremath{\text{\emph{P}} \left[ #1
\right]}}

\newcommand{\probsub}[2]{\ensuremath{{\sf Pr}_{#1}\left[ #2 \right]}}
\newcommand{\probsubb}[2]{\ensuremath{{\sf Pr}_{#1}\bigl[ #2 \bigr]}}

\providecommand{\variance}[1]{{\sf Var}\left[ {#1} \right]}
\newcommand{\varianceb}[1]{\ensuremath{\textsf{ Var}\bigl[ {#1} \bigr]}}

\providecommand{\covariance}[2]{{\sf Cov}\left( {#1},{#2} \right)}

\providecommand{\expectsub}[2]{\ensuremath{\mathbb{E}_{{#1}}\left[ {#2} \right]}}
\providecommand{\varsub}[2]{\textsf{Var}_{{#1}}\left[ {#2} \right]}

\providecommand{\expectsubb}[2]{\mathbb{E}_{{#1}}\bigl[ {#2} \bigr]}
\providecommand{\expectsubbb}[2]{\mathbb{E}_{{#1}}\biggl[ {#2} \biggr]}
\providecommand{\expectsubB}[2]{\mathbb{E}_{{#1}}\Bigl[ {#2} \Bigr]}

\newcommand{\smallU}{\textsc{small-u}}

\newcommand{\tpest}{\text{\sc tpest}}
\newcommand{\est}{\textsf{TPEst}}

\newcommand{\countsketch}{\textsf{CountSketch}}

\newtheorem{theorem}{Theorem}
\newtheorem{lemma}[theorem]{Lemma}
\newtheorem{corollary}[theorem]{Corollary}
\newtheorem{fact}[theorem]{Fact}

\newcommand{\Hss}{\textsf{Hss}}
\newcommand{\rank}{\mathrm{rank}}

\newcommand{\distr}{\ensuremath{\sim}}

\newcommand{\hh}{\ensuremath{\textsf{HH}}} 

\newcommand{\Card}[1]{\biggl\lvert {#1} \biggr \rvert}
\newcommand{\nocollision}{\textsc{nocoll}}

\newcommand{\nocoll}{\textsc{nocollision}}

\newcommand{\topk}{\text{\sc Topk}}
\newcommand{\hattopk}{\ensuremath{\widehat{\textsc{Topk}}}}
\newcommand{\epsbar}{\ensuremath{\bar{\epsilon}}}
\newcommand{\xibar}{\ensuremath{\bar{\xi}}}

\newenvironment{pfsktch}[1][Proof Sketch]{\begin{trivlist} \item[\hskip \labelsep {\bfseries#1}]}{\end{trivlist}}

\newcommand{\lmargin}{\text{lmargin}}
\newcommand{\rmargin}{\text{rmargin}}
\newcommand{\midreg}{\text{mid}}

\newcommand{\goodest}{\textsc{goodest}}
\newcommand{\accuest}{\textsc{accuest}}
\newcommand{\goodftwo}{\textsc{goodf}\ensuremath{_2}}

\newcommand{\smallhh}{\textsc{smallhh}}

\newcommand{\ftwores}[1]{\ensuremath{F_2^{\text{res}}\left({#1}\right)}}

\newcommand{\1}{\textbf{1}}

\newcommand{\level}{\mathsf{level}}

\newcommand{\G}{\ensuremath{\mathcal{G}}}

\newcommand{\bvtheta}{\ensuremath{\bar{\vartheta}}}

\newcommand{\iw}{\textsf{IW}}

\newcommand{\smallH}{\textsc{small-h}}
\newcommand{\lastlevel}{\textsc{goodfinallevel}}

\newcommand{\ako}{\textsf{AKO}}

\newcommand{\tp}{\textsc{tp}}
\newcommand{\avtp}{\textsc{avgtp}}

\newcommand{\sgn}{\textrm{sgn}}

\newcommand{\ghss}{\textsf{Geometric-Hss}}
\newcommand{\smallres}{\textsc{smallres}}

\newcommand{\for}{\textbf{for }}
\newcommand{\algoendfor}{\textbf{endfor}}

\newcommand{\algoto}{\textbf{to }}

\newcommand{\powf}[2]{\ensuremath{{#1}^{\underline{{#2}}}\;}}

\newtheorem*{relemma}{Restated Lemma}

\title{Taylor Polynomial Estimator  for  Estimating Frequency Moments
}

\author {
 Sumit Ganguly\\
 Indian Institute of Technology, Kanpur\\
 \texttt{sganguly@cse.iitk.ac.in}
}
\date{}
\begin{document}

\brc

\maketitle

\begin{abstract} We present a randomized  algorithm for estimating the $p$th moment
$F_p$ of the frequency vector of  a data  stream in the general update (turnstile) model to within a multiplicative  factor of $1 \pm \epsilon$, for  $p
> 2$, with high constant confidence.  For $0 < \epsilon \le 1$, the  algorithm   uses  space $O( n^{1-2/p} \epsilon^{-2} +
n^{1-2/p} \epsilon^{-4/p} \log (n))$ words. This improves over the current bound of $O(n^{1-2/p} \epsilon^{-2-4/p} \log (n))$ words by Andoni et. al. in \cite{ako:arxiv10}. Our  space  upper bound matches the lower bound of Li
and Woodruff \cite{liwood:random13} for $\epsilon = (\log (n))^{-\Omega(1)}$ and the lower
bound of Andoni et. al. \cite{anpw:icalp13} for $\epsilon = \Omega(1)$.
\end{abstract}

\newpage
\tableofcontents

\newpage
\section{Introduction} \label{sec:intro}
The data stream model is relevant for online
applications over massive data, where an algorithm may use only  sub-linear
 memory  and a single pass over the data to summarize a large data-set that
appears as  a sequence of  incremental updates. Queries may be answered using only the data summary. A  data stream is
viewed as a sequence of $m$  records of the form $(i, v)$, where, $i
\in
 [n] = \{1,2, \ldots, n\}$ and $v  \in  \{-M, -M+1, \ldots, M-1, M\}$.
The record $(i,v)$ changes the $i$th coordinate $f_i$ of the
$n$-dimensional  \emph{frequency vector} $f$ to  $ f_i + v$.
The $p$th
moment of  the frequency vector $f$ is defined as $F_p = \sum_{i
\in [n]}
 \abs{f_i}^p$, for $p\ge 0$.
 The (randomized) $F_p$  estimation problem is:
Given $p$ and $\epsilon \in (0,1]$, design an algorithm that makes
one pass over the input stream and  returns  $\hat{F}_p$ such that
$\probb{\abs{\hat{F}_p - F_p} \le \epsilon F_p} \ge 0.6$ (where,
the constant  0.6 can be replaced by any other constant $>1/2$.)
In this paper, we consider estimating $F_p$ for the regime $p >
2$, called the  \emph{high moments} problem. The problem was posed and studied in the seminal work of Alon, Matias and Szegedy in \cite{ams:jcss98}.

\emph{Space lower bounds.} Since a deterministic estimation algorithm for $F_p$ requires $\Omega(n)$ bits  \cite{ams:jcss98},   research has focussed on randomized algorithms \cite{B-YJKS:stoc02,cks:ccc03,wood:soda04,jw:soda11,wz:stoc12,g:arxiv11b,liwood:random13,anpw:icalp13}.
Andoni et. al. in  \cite{anpw:icalp13} present a bound of $\Omega(n^{1-2/p} \log (n))$ \emph{words} assuming that the algorithm is a\emph{ linear sketch}.\eat{ and  $\epsilon = \Omega(1)$. }
 Li and Woodruff in \cite{liwood:random13} show a  lower bound of  $\Omega(n^{1-2/p}\epsilon^{-2} \log (n) )$ bits in the turnstile streaming model. For \emph{linear} sketch algorithms, the lower bound is the sum of the above two lower bounds,  namely,
$\Omega(n^{1-2/p}( \epsilon^{-2} + \log (n)))$ words.

\;\emph{Space upper bounds.} The
table in Figure~\ref{fig:table} chronologically lists  algorithms and their properties for
estimating $F_p$ for  $p > 2$ of data streams in the \emph{turnstile mode}. Algorithms for  \emph{insertion-only} streams are not directly comparable to algorithms for update streams---however, we note that the best algorithm for insertion-only streams is by Braverman et. al. in \cite{bksv:arxiv14} that uses $O(n^{1-2/p})$ bits, for $p \ge 3$ and $\epsilon = \Omega(1)$.

\emph{Contribution.} We show that for each fixed $p > 2$ and $0 < \epsilon \le 1$,
there is an  algorithm for estimating
$F_p$ in the general update streaming  model that uses space
$O(n^{1-2/p}(\epsilon^{-2} + \epsilon^{-4/p} \log (n) ) )$
words, with word size $O(\log (nmM))$ bits. It is the most space economical algorithm as a function of $n$ and $1/\epsilon$.
 The space bound of our algorithm
matches the  lower bound of $\Omega(n^{1-2/p} \epsilon^{-2})$
of Li and Woodruff in \cite{liwood:random13} for $\epsilon \le
(\log n)^{-p/(2(p-2))}$ and  the lower bound  $\Omega(n^{1-2/p}
\log (n))$ words of  Andoni et.al. in \cite{anpw:icalp13} for
linear sketches and $\epsilon = \Omega(1)$.

\begin{figure}[htbp]
\begin{center}
\begin{tabular}{|m{0.7in}|c|c|}\hline
\multirow{1}{*}{Algorithm } & Space in $O(\cdot)$ words  & \multirow{1}{*}{Update time $O(\cdot)$ } \\ \hline
\iw \cite{indwoo:stoc05} & $n^{1-2/p}\left(\epsilon^{-1} \log (n)\right)^{O(1)}$   &
$(\log^{O(1)} n)(\log (mM))$  \\ \hline
\multirow{1}{1.7in}{\Hss
\cite{bgks:soda06}} & $ n^{1-2/p} \epsilon^{-2-4/p} \log (n) \log^2 (nmM) $    & $\log(n)
\log (nmM)$  \\\hline
\textsf{MW} \cite{mw:soda10} & $n^{1-2/p} (\epsilon^{-1}\log (n))^{O(1)}$   &
$n^{1-2/p}(\epsilon^{-1}\log n)^{O(1)}$ \\ \hline
\ako \cite{ako:arxiv10} & $ n^{1-2/p} \epsilon^{-2-4/p} \log (n) $
  & $\log n$\\\hline
\textsf{BO-I} \cite{braost:arxiv10} & $n^{1-2/p} \epsilon^{-2-4/p} \log (n)   \log^{(c)}(n)$
& $\log n$   \\\hline
 \textsf{this paper} &
 $n^{1-2/p}\epsilon^{-2}+n^{1-2/p}\epsilon^{-4/p} \log(n)$  & $\log^2
 (n)$  \\ \hline
\end{tabular}
\end{center}
\caption{Space requirement of published  algorithms  for estimating  $F_p$, $p > 2$. Word-size is $O(\log ( nmM)) $ bits for algorithms  for update streams. $\log^{(c)}(n)$ denotes $c$ times  iterated logarithm for $c = O(1)$. }
\label{fig:table}
\end{figure}

\vspace*{-0.3cm}
\emph{Techniques and Overview.} We design  the \ghss~ algorithm for estimating $F_p$ that builds upon
the \Hss~technique presented in \cite{bgks:soda06,gl:hssfull}. It uses a layered data structure with $L+1 = O(\log n)$ levels numbered from $0$ to $L$  and
uses an $\ell_2$-heavy-hitter  structure  based on \countsketch~
\cite{ccf:icalp02} at each level to identify and estimate
$\abs{f_i}^p$ for each heavy-hitter.  The heavy-hitters structure at each level has the same number of $s = O(\log n)$ hash  tables with each hash  table having the number of buckets (height of table). The main new ideas are as follows. The height of any  \countsketch~table at level $l$ is
  $\alpha^l$ times the height of any of the  tables of the level 0 structure, where, $0<\alpha < 1$ is a constant. The geometric decrease  ensures that  the total space required  is  a constant times the  space used by the lowest level and avoids increasing space by a factor of $O(\log n)$ as in the \Hss~algorithm. 

In all previous works, an estimate for $\abs{f_i}^p$ for a \emph{sampled item} $i$  was obtained by retrieving an estimate $\hat{f}_i$ of $f_i$ from the heavy-hitter structure of an appropriately chosen level, and then computing $\abs{\hat{f}_i}^p$. In order for $\abs{\hat{f}_i}^p$ to lie within $ (1\pm \epsilon) \abs{f_i}^p$,  $\abs{\hat{f}_i - f_i} $ had to be constrained to be at most $O(\epsilon \abs{f_i}/p)$. By the lower bound results of \cite{pw:focs11}, the estimation error for \countsketch~is in general optimal and cannot be improved. We circumvent this problem by designing a more accurate estimator $\bvtheta(\lambda, k)$ for $\abs{f_i}^p$ directly. If  $\lambda$ is an estimate for $\abs{f_i}$ that is accurate to within a constant relative error, that is, $\lambda \in (1\pm O(1/p))\abs{f_i}$ and there are independent, identically distributed  and unbiased  estimates $X_1, X_2, \ldots, X_{\Theta(k)}$ of  $\abs{f_i}$ with  standard deviation $\sigma[X_j] \le O(\abs{f_i}/p)$, then, it is shown that (i) $\expect{\bvtheta(\lambda, k)} \in (1\pm O(1/p)^k)\abs{f_i}^p$, and (ii) $ \variance{\bvtheta(\lambda, k)} \le O(\abs{f_i}^{2p-2} \sigma^2[X_j])$.

The estimator $\bvtheta$ is designed using a \emph{Taylor polynomial estimator}.
Given an estimate $\lambda = \abs{\hat{f}_i}$ for $\abs{f_i}$ such that $\lambda \in (1\pm O(1/p)) \abs{f_i}$, the $k+1$ term \emph{Taylor polynomial  estimator} denotes
$\vartheta(\lambda, k) = \sum_{j=0}^{k} \binom{p}{j} \lambda^{p-j}( X_1- \lambda) (X_2-\lambda)  \ldots (X_j - \lambda)$, where, $X_1, \ldots, X_k$ are independent and identically distributed estimators  of $\abs{f_i}$.  Note that replacing the $X_j$'s    by $\abs{f_i}$ gives the expression $\sum_{j=0}^{k-1} \binom{p}{j} \lambda^{p-j}(\abs{f_i} - \lambda)^j$, which is the degree-$k$ term Taylor polynomial expansion of $\abs{f_i}^p$ around $ \lambda$ (i.e., $(\lambda + (\abs{f_i} - \lambda))^p$.  A new estimator $\bvtheta(\lambda, k, r)$ is defined as the  average of $r$ \emph{dependent} Taylor polynomial estimators $\vartheta$'s, where,  each of these $r$ $\vartheta$-estimators is
obtained from a certain $k$-subset of  random variables $X_1,
\ldots, X_s$, with $s= O(k)$, and each $k$-subset is drawn from an appropriate code and  has a controlled overlap with another $k$-subset from the code.  Note that now, only  a constant factor (i.e.,  within a factor of $1 \pm O(1/p)$ )  accuracy for the estimate $\lambda$ of $\abs{f_i}$ is needed, rather than an $O(\epsilon)$-accuracy needed earlier.

 Finally, we note that \Hss~algorithm \cite{gl:hssfull} used full independence of hash functions and then invoked Indyk's method \cite{indy:focs00}   of using Nisan's pseudo-random generator to fool space-bounded computations \cite{nisan:stoc90}. In our algorithm, we show that it suffices to  use only limited $d = O(\log n)$-wise independence of hash families, by changing the way the hash functions are composed.

\eat{\emph{Taylor polynomial estimator}, denoted as \tpest, views $\abs{f_i}^p$    for
estimating $\psi(\mu)$   where, (a) $X$ is a random variable and
$\mu = \expect{X}$, (b)  $\psi$ is   analytic in a certain range, and, (c) an
estimate $\lambda$ for $\mu$ is available. Let $\psi(\mu) = \sum_{j=0}^{k} \gamma_{j}(\lambda)(x-\lambda)^j$  be the degree $k$-term Taylor expansion of $\psi(\mu)$ around $\lambda$. This is converted into
 the estimator
$\vartheta(\psi,\lambda,k) = \sum_{j=0}^{k-1}
\gamma_j(\lambda)(X_1 - \lambda)\cdot (X_2 - \lambda)\cdot  \ldots
\cdot (X_j - \lambda)$, where $X_1, \ldots, X_{k}$ are independent
copies of $X$. For $\psi(t) = \abs{t}^p$, $ \lambda \in  \mu( 1 \pm O(1/p))$ and $ k \ge p+1$,  it is  shown that the estimator has low bias, that is,
$\expect{\vartheta(\abs{t}^p, \abs{\lambda}, k)} \in \abs{\mu}^p (1 \pm O(1/p^k))$, and has  variance  $O(p^2 \abs{\mu}^{2p-2} \variance{X_j})$. Using the Taylor polynomial estimator provides greater accuracy in the estimation of $\abs{f_i}^p$ given an estimate $\abs{\hat{f}_i}$, rather than
The estimator $\bvtheta$ is defined as
 the average of $r = O(k)$ dependent  $\vartheta$ estimators, where,  each of these $r$ $\vartheta$-estimators is
obtained from a certain $k$-subset of  random variables $X_1,
\ldots, X_s$, with $s= O(k)$, where each $k$-subset has a certain bounded overlap with another $k$-subset used. The construction is obtained using codewords from an appropriate code. This is used to show  that $\variance {\bvtheta} =
O(\variance{\vartheta}/k)$, while preserving the bias.}

\eat{

\emph{Comparative overview of analysis.} We now present a high level comparison of the   calculations for deriving the space bound in the  analysis of previous algorithms and our algorithm.
 First consider the \Hss~algorithm from ~\cite{gl:hssfull}. Here items are grouped into groups of decreasing frequency ranges, the $l$th group $G_l $ consists of items with $\abs{f_i} \in [T/2^{l/2}, T/2^{(l-1)/2})$ and the zeroth group has items with $\abs{f_i} \ge T$. Here $T$ is a threshold parameter and is a dependent on the space provided to the heavy-hitters structure.  Items in the group $G_l$ are sampled with frequency (close to) $1/2^l$. The algorithm ensures that for each sampled item $i$, $\abs{\hat{f}_i}^p \in (1\pm O(\epsilon))\abs{f_i}^p$, which is achieved by ensuring that $\abs{\hat{f}_i} \in (1\pm O(\epsilon)/p)\abs{f_i}$.  If $B$ is the number of buckets in each table of the \countsketch~structure at each level, then, the additive error $\abs{\hat{f}_i - f_i}$ at level $l$ can be shown to be $O(F_2/(2^l B)^{1/2})$. This would mean that
 the items in $G_l$ \emph{must} have frequencies $\abs{f_i} = \Theta((F_2/ (\epsilon^2 2^l B))^{1/2})$.
  The estimate $\hat{F}_p$ is  defined as the scaled sum  $\sum_{0\le l \le L} \sum_{i \in \text{sample}(G_l)} O(2^l) \abs{\hat{f}_i}^p$ of the sampled items.
 Since each estimate  is $\epsilon$-close, it can be routinely shown that $\expectb{\hat{F}_p} \in F_p \pm O(\epsilon F_p)$.
 \eat{  The variance of the estimate $\abs{\hat{f}_i}^p$ is  $O(2^{l} \abs{f_i}^{2p})$, if $i \in G_l$ ($l \ge 1)$. } $\varianceb{\hat{F}_p}$ can be shown to be  very close to  the sum of the individual  variances  of the estimates, that is,  $\sum_{l} \sum_{i \in G_l} O(\abs{f_i}^{2p} \cdot 2^{l})$. \eat{ Since, the expectation of $\hat{F}_p$  is $O(\epsilon)$-close to $F_p$, hence for $\epsilon$-closeness of $\hat{F}_p$ with $F_p$ with high constant confidence, it suffices to ensure that  $\varianceb{\hat{F}_p} \le O(\epsilon^2 F_p^2)$, by Chebychev's inequality.} Therefore,
\begin{align*}
O(\varianceb{\hat{F_p}}) & = \sum_{l, i \in G_l} 2^{l} \abs{f_i}^{2p}  = \sum_{l} \sum_{i \in G_l} \left( \frac{F_2}{\epsilon^2 2^l B} \right)^{p/2} 2^l  \abs{f_i}^p  \le  \left(\frac{ F_2 }{\epsilon^2 B  }  \right)^{p/2} \sum_{l, i \in G_l}   2^{l-lp/2} \abs{f_i}^p\\ &  \le  \frac{ F_2}{(\epsilon^2 B)^{p/2}} \cdot  F_p  \le \epsilon^2 F_p^2
\end{align*}
which holds if $B = n^{1-2/p} \epsilon^{-2-4/p}$, since, $F_2 \le F_p^{2/p} n^{1-2/p}$. An application of Chebychev's inequality now shows that $\hat{F}_p$ is $O(\epsilon)$-close to $F_p$ with high constant probability.   The space requirement is $\tilde{O}(B) = \tilde{O}(n^{1-2/p} \epsilon^{-2-4/p})$, where the additional log factors arise since, (a) the number of tables required at each level is $O(\log n)$ for high confidence, and (b) the number of levels is $O(\log n)$. An additional log factor was incurred in \cite{gl:hssfull} for an application of Nisan's pseudorandom generator to fool the space bounded algorithm \cite{nisan:stoc90,indy:focs00}.
\footnote{\label{footnote:ako} Andoni et. al. in \cite{ako:arxiv10} use a continuous distribution for sampling that basically achieves the following.  (a) Items with $\abs{f_i} > O(\epsilon (F_2/B)^{1/2})$ are always sampled, and (b) otherwise, $i$ is sampled with probability   $ p_i =  \left(F_2/(\epsilon^2 B f_i^2)\right)^{-p/2}  = \epsilon^p \abs{f_i}^p/(F_2/B)^{p/2}$.  Then,
$
O(\varianceb{\hat{F_p}})  = \sum_{i} (1/p_i) \abs{f_i}^{2p} = \sum_i  \frac{\epsilon^{-p} (F_2/B)^{p/2}}{\abs{f_i}^p}\abs{f_i}^{2p} \le \left( \frac{F_2}{ B} \right)^{p/2} \epsilon^{-p}  F_p   \le \epsilon^2 F_p^2
$.
which holds if $B = n^{1-2/p} \epsilon^{-2-4/p}$.
}

In the new algorithm, the frequency of a sampled item is estimated only  to within relative error of $1 \pm O(1/p)$ and not to within $1\pm \epsilon$. An estimate for $\abs{f_i}^p$ is obtained using  the Taylor polynomial estimator, denoted $\bvtheta_i$, that satisfies $\abs{\expect{ \bvtheta_i} -\abs{f_i}^p}  \le 2^{-\Omega(k)} \abs{f_i}^p$, where, $k = O(\log n)$ is the degree of the polynomial used. The estimate for $F_p$ is (approximately) $\hat{F}_p = \sum_{0 \le l \le L} \sum_{i \in \text{sample}(G_l)} \bvtheta_i 2^l$, since $2^{-l}$ is the probability with which $i \in G_l$ is sampled.  Hence, $\abs{\expectb{\hat{F}_p} - F_p} \le n^{-\Omega(1)} F_p$.
Secondly, an item in group $l$ (i.e., sampled with probability $2^{-l}$) has $\abs{f_i} = \Theta( (F_2/B)^{1/2})$--note that these are lower than the threshold for the \Hss~algorithm  by a factor of $1/\epsilon$.
 The variance of the estimate $\bvtheta_i 2^l$  is  (as before) $O(2^{l}\abs{f_i}^{2p})$ for items sampled into  levels $l=1$ and higher.
Considering the contribution to variance from the items in level $l \ge 1$ or higher and using the property that the variance of $\hat{F}_p$ is (almost) the sum of the variances of $\bvtheta_i$'s, we obtain in a similar manner that
\begin{multline*}
\sum_{l \in [L], i \in G_l} 2^{l} \abs{f_i}^{2p}  = \sum_{l\in [L]} \sum_{i \in G_l} \left( \frac{F_2}{2^l B} \right)^{p/2} 2^l  \abs{f_i}^p  \le  \left(\frac{ F_2 }{B  }  \right)^{p/2} 2^{l-lp/2} \sum_{l\in [L], i \in G_l} \abs{f_i}^p \\\le  \left(\frac{ F_2}{ B}\right)^{p/2}  F_p  \le \epsilon^2 F_p^2
\end{multline*}
which holds if $B \ge  n^{1-2/p} \epsilon^{-4/p}$, since, $F_2 \le F_p^{2/p} n^{1-2/p}$. Items in level 0 are treated separately, since they are sampled with probability 1. Here, it is shown that $\varianceb{\bvtheta_i} = \abs{f_i}^{2p-2}O(F_2)/(B\log (n))$, where, the $\log n$ factor in the denominator is a property  of the averaged Taylor polynomial estimator.  Thus, for $\abs{f_i} > O ((F_2/B)^{1/2})$, the variance of the contribution from items in level 0 to $\hat{F}_p$ is
$\sum_{i \in G_0}\frac{ \abs{f_i}^{2p-2} F_2}{B \log n} \le \frac{ F_{2p-2}F_2}{B \log n} \le \frac{ F_p^{2-2/p}F_p^{2/p} n^{1-2/p}}{B \log n} \le \epsilon^2 F_p^2
$, provided
 $B \ge n^{1-2/p} \epsilon^{-2}/\log (n)$. Setting $B$ to be the sum of the two values satisfies both constraints so that $B = (n^{1-2/p}(\epsilon^{-2}/\log n + \epsilon^{-4/p}))$. The total space required is $O(B\log n)$, since there are $O(\log n)$ tables in a \countsketch~structure at each level. However, due to the geometric decrease in the table heights, the total space is a constant factor times the space at level 0.

}

\subsubsection*{Notation} Let $\R$ denote the field of real numbers, $\N$ denote the set of natural numbers, that is, $\N = \{0,1,2, \ldots, \}$, $\Z$ denote the ring of integers, and $\Z^+$  and $\Z^{-}$ denote the set of positive integers and the set of negative integers respectively.

 For $a \in \R$ and $s \in \N$, define
 \begin{align*}
 \powf{a}{s} = \begin{cases} a \cdot (a-1)\cdot \cdots \cdot (a-s+1) & \text{ if } s \in \Z^+ \\
 1 & \text{ if } s = 0  \enspace .
 \end{cases}
 \end{align*}
 It follows that, (i) for $s_1, s_2 \in \N$,  $\powf{a}{s_1+s_2} = \powf{a}{s_1} \powf{(a-s_1)}{s_2} $, and (ii) for $a<0$, $\powf{a}{s} = (-1)^s \powf{(-a+s-1)}{s}$. The notation $\powf{a}{s}$ is taken from \cite{ConcreteMath:book}.

For $p \in \R$ and $k \in \N$, denote
\begin{align*}\binom{p}{k} = \begin{cases} \cfrac{\powf{p}{k}}{k!}  & \text{ if } p \in \R \text{ and } k \in \N \\
0 & \text{ if } p \in \R \text{ and } k \in \Z^{-} \enspace .
 \end{cases}
 \end{align*}
 We use the well-known following identities for binomial coefficients, namely, the \emph{absorption identity}: $\binom{p}{k} = \frac{p}{k} \binom{p-1}{k-1}$,  for integer $k \ne 0$, and,  the \emph{upper negation identity}: $\binom{p}{k} = (-1)^k \binom{k-p-1}{k}$, for integer $k$.

\subsection*{Review: Residual second moment and \countsketch~algorithm}

Let $f \in \Z^n$ and let  $\rank:[n] \rightarrow [n]$ be any permutation
that orders the indices of $f$   in non-decreasing order by their absolute
frequencies, that is, $\abs{f_{\rank(1)}} \ge \abs{f_{\rank(2)}} \ge \ldots \abs{f_{\rank(n)}}$. The $k$-residual second moment of $f$ is denoted by  $\ftwores{k} $ and is defined as $\ftwores{k} = \sum_{ i \in
[n], \rank(i) > k} f_i^2$.

We will use the \countsketch~algorithm by Charikar, Chen and Farach-Colton \cite{ccf:icalp02}, which is a classic algorithm  for identifying  $\ell_2$-based heavy-hitters and for estimating item frequencies in data streams.  The \countsketch$(C,s)$ structure consists of $s$ hash tables denoted $T_1, \ldots, T_s$, each having $C$ buckets. Each bucket stores an $\log (nmM)$ bit integer.  The $j$th hash table uses the hash function $h_j:[n] \rightarrow [C]$, for $j=1,2, \ldots, s$. The hash functions are chosen independently and randomly  from a pair-wise independent hash family mapping $[n] \rightarrow [C]$. A pair-wise independent Rademacher family $\{\xi_j(i)\}_{i \in [n]}$ is associated with each table index $j \in [s]$, that is $\xi_j(i) \in_R \{-1, 1\}$. The  Rademacher families for different $j$'s are independent. Corresponding to a stream update of the form $(i,v)$, all tables are updated as follows.

\begin{tabbing}
xxxx\=xxxx\=xxxx\=xxxx\=xxxx\=xxxx\=xxxx\= \kill
\for $j =1$ \algoto $s$ \textbf{do} \\
\>  $T_j[h_j(i)] = T_j[h_j(i)] + v \cdot \xi_j(i) $ \\
\algoendfor
\end{tabbing}
Given an index $i \in [n]$, the estimate $\hat{f}_i$ returned for $f_i$ is the median of the estimates obtained from each table, namely,
$$ \hat{f}_i = \text{median}_{j=1}^s T_j[h_j(i)] \cdot \xi_j(i) \enspace . $$
It is shown in \cite{ccf:icalp02} using an elegant argument that
\begin{align} \label{eq:cskbasic}
\left\lvert \hat{f}_i - f_i \right\rvert \le \left(\frac{ 8\ftwores{C/8}}{C} \right) ^{1/2} \enspace .
\end{align}

\section{Taylor polynomial estimator
}\label{sec:taylor} Let $X$ be a random variable with $ \expect{X}
= \mu$ and $\variance{X} = \sigma^2$. Singh in
\cite{singh:sankhya64} considered the following problem: Given a function $\psi:\R \rightarrow \R$, design  an unbiased estimator
$\theta$ for $\psi(\expect{X})$ (i.e., $ \expect{\theta} = \psi(\expect{X})$. His  solution for  an analytic   function $\psi$ was the following.
Let
$\psi(t) = \sum_{k \ge 0} \gamma_{k}(0) t^k$.  Let $\nu$ be a distribution over $\N$  with probability mass function
$p_{\nu}(n)$, for $n=0,1,2, \ldots, $.  Choose $n \distr \nu$ and define the estimator
 $$\theta = (p_{\nu}(n))^{-1}
\gamma_n(0)\cdot X_1 \cdot X_2 \ldots \cdot X_n $$ where the
$X_i$'s are \emph{independent copies of} $X$. The estimator satisfies
$$\expect{\theta} = \sum_{n \ge 0} (p_{\nu}(n) )^{-1}\cdot p_{\nu}(n) \cdot  \gamma_n(0) \expect{X_1} \expect{X_2} \ldots \expect{X_n} = \sum_{n \ge 0} \gamma_n(0) \mu^n = \psi(\mu) \enspace . $$
However, the variance can be
large;
  for the geometric distribution $\nu$ with    $p_{\nu}(n) = q(1-q)^{n}$, for $n \ge 0$ and
  $0<q\le 1$, it is shown in \cite{css:colt10} that
 $ \expect{\theta^2}= (1/q)\sum_{n \ge 0}  \gamma_n^2(0) ((\mu^2 + \sigma^2)/(1-q))^n
 $.

\subsection{Taylor Polynomial Estimator}
The Taylor polynomial estimator (abbreviated as \tp~estimator)  is derived
from the Taylor's series of $\psi(\mu) = \psi(\lambda + (\mu-\lambda))$ by expanding it around
$\lambda$, an estimate of $\mu$, and then  truncating it after the
first $k+1$ terms. Let   $X_1, \ldots, X_k$ be independent variables with the same expectation $\expect{X_j} = \mu = \expect{X}$ and whose variance is each bounded above by $\sigma^2$. Define
\[ \begin{array}{l}
\vartheta(\psi, \lambda, k, \{X_l\}_{l=1}^k) = \sum_{j=0}^k \gamma_j(\lambda)
(X_1-\lambda)(X_2-\lambda) \ldots (X_j-\lambda) \enspace .
\end{array}\]
where, $\gamma_j(t)$ is the function $\psi^{(j)}(t)/j!$, for $j =0,1,
\ldots$. Its expectation and variance properties are
 given below. Let
 $\eta^2 = \expect{(X_j - \lambda)^2} = \sigma^2 + (\mu-\lambda)^2 $, for $j=1, \ldots, k$.

\begin{lemma} \label{lem:vvtheta} Let $\{X_l\}_{l=1}^k$ be independent random variables with expectation $\mu$ and standard deviation at most $\sigma$. Let $\eta = (\sigma^2 + (\mu-\lambda)^2)^{1/2} $ and let   $\psi$ be analytic in the region $[\lambda, \mu]$.  Then the following hold.
 \begin{enumerate}
\item  For some $\lambda' \in (\mu, \lambda)$,  $\card{\expect{\vartheta(\psi, \lambda, k, \{X_l\}_{l=1}^k) } - \psi(\mu)}  \le
 \abs{\gamma_{k+1}(\lambda')}  \cdot \abs{\mu-\lambda}^{k+1}$.
 \item $\variance{\vartheta(\psi, \lambda, k, \{X_l\}_{l=1}^k) }   \le \Bigl(\sum_{j=1}^k
  \abs{\gamma_j(\lambda)} \eta^j\Bigr)^2 \enspace . $
 \end{enumerate}
\end{lemma}
Corollaries~\ref{cor:fpbias} and \ref{cor:fptpvar} apply the
Taylor polynomial estimator to  $\psi(t) = t^p$.

\begin{corollary}\label{cor:fpbias} Assume the  premises of Lemma~\ref{lem:vvtheta}. Further, let   $\psi(t) = t^p$,  $p \ge 2$,  $\mu > 0$,
$\abs{\lambda - \mu} \le\alpha \mu$, for some $0\le \alpha <
1/2$ and $k+1 > p$.
 Then,
 $$ \left\lvert\expect{\vartheta(x^p, \lambda , k, \{X_l\}_{l=1}^k}- \mu^p \right\rvert \le \left( \frac{\alpha}{1-\alpha} \right)^{(k+1)} \cdot \mu^p \cdot  \left( \frac{ p}{k+1}\right)^{\lfloor p \rfloor + 1} \enspace . $$
  In particular, for $p$ integral, $\expect{\vartheta(x^p, \lambda , k, \{X_l\}_{l=1}^k} = \mu^p$.
\end{corollary}

\begin{corollary}\label{cor:fptpvar}Assume the  premises of Lemma~\ref{lem:vvtheta} and Corollary~\ref{cor:fpbias}. 
Then $$\variance{\vartheta(x^p, \lambda ,k,
\{X_l \}_{l=1}^k}\le  (1.08)p^2
\mu^{2p-2}\eta^2 \enspace . $$
\end{corollary}

\subsection{Averaged Taylor  polynomial estimator} \eat{ The
averaged \tp~estimator, denoted by \avtp, keeps $s = \Theta(k)$
independent copies $X_1, X_2, \ldots, X_s$ of $X$ and returns  the
average of  the $s$ \tp~estimators, each obtained from specific
$k$-subsets of $\{X_1, \ldots, X_s\}$ that correspond to codewords
of a certain code, which are then randomly permuted. This is shown
to reduce the variance by a factor of $O(k)$ for $\psi(t) =
\abs{t}^p$.}
We use a version of the Gilbert-Varshamov theorem from \cite{dipw:soda10}.
\begin{theorem} [Gilbert-Varshamov] \label{thm:gv}For positive integers  $q \ge 2$ and $ k >
1$, and real value $0 <  \epsilon < 1-1/q$, there exists a set $\C
\subset \{0,1\}^{qk}$ of binary vectors with exactly $k$ ones such
that $\C$ has minimum Hamming distance $2\epsilon k$ and $\log
\abs{\C} > (1-H_q(\epsilon))k \log q$, where, $H_q$ is  the
$q$-ary entropy function $H_q(x) = -x \log_q \frac{x}{q-1} - (1-x)
\log_q  (1-x) $.
\end{theorem}

\begin{corollary} \label{cor:gv}
For $k \ge 1$, there exists a code $Y \subset \{0,1\}^{8k}$ such
that  $\abs{Y} \ge 2^{0.08k}$, each $y\in Y$ has exactly $k$ 1's, and
the minimum Hamming distance among distinct codewords in $Y$ is
$3k/2$.
\end{corollary}

Let $Y$ be a code as given by Corollary~\ref{cor:gv}. Each  $y \in Y$ is a boolean vector  $y = (y(1), y(2), \ldots, y(s))$ of dimension
 $s = 8k$ with exactly $k$ 1's.  It  can be equivalently
 viewed as a $k$-dimensional ordered sequence
 $y \equiv (y_1, y_2, \ldots, y_k)$ where $1\le y_1 < y_2< \ldots < y_k \le s$,  and $y_j$ is the index of  the $j$th occurrence of 1 in $y$.
 Let $\pi: [k] \rightarrow [k]$ be a
permutation and $y =  (y_1, \ldots, y_k)$ be an ordered sequence
of size $k$. Then,   $\pi(y)$ denotes the sequence of indices
$(y_{\pi(1)}, \ldots, y_{\pi(k)})$.

Let $X_1, X_2, \ldots, X_s$ be independent random variables with expectation $\mu$ and standard deviation at most $\sigma$.
We first
define the  Taylor polynomial estimator, denoted  \tp~estimator,  for $\psi(\mu)$, given  (i) an  estimate
$\lambda$ for $\mu$, (ii)  a codeword
$y \in Y$, and  (iii) a permutation $\pi:[k]
\rightarrow [k]$.
 The \tp~estimator corresponding  to $y \in Y$ and permutation $\pi$ is defined as
$$
\vartheta(\psi,\lambda,k,s,y,\pi, \{X_t\}_{
t=1}^s) = \sum_{v=0}^k \gamma_v(\lambda) \prod_{l=1}^v \left(X_{y_{\pi(l)}} -
\lambda\right) \enspace .
$$
 Let $\{\pi_y\}_{y \in Y}$
denote a set of $\abs{Y}$ randomly and independently chosen
permutations that map $[k] \rightarrow [k]$ that is  placed in (arbitrary) 1-1 correspondence with $Y$. The averaged Taylor polynomial estimator \avtp~
averages the  $\abs{Y}$ \tp~estimators corresponding to
 each codeword in $Y$, ordered by  the permutations $\{\pi_y\}_{y \in Y}$ respectively, as follows.
\begin{align} \label{eq:bvtheta}
 \bvtheta(\psi, \lambda, k,s, Y, \{\pi_y\}_{y\in Y} , \{X_l\}_{l=1}^s)
 =\frac{1}{\abs{Y}}\sum_{y \in Y} \vartheta(\psi,\lambda,k,s,y, \pi_y,
  \{X_l\}_{l =1}^s)
\end{align}
The Taylor polynomial estimator in \emph{RHS} of Eqn.~\eqref{eq:bvtheta} corresponding to each $y \in Y$ is referred to simply as $\vartheta_y$, when the other parameters are clearly understood from context. Note that for any $y\in Y$ and permutation $\pi_y$, $\expect{\vartheta_y}$ is the same. Therefore,
due to averaging,
 the \avtp~estimator has the same expectation as the expectation of each of the $\vartheta_y$'s.
\begin{lemma} \label{lem:vbvtheta}
 Let $p \ge 2, q=8$, $k \ge \max(1000, 40(\lfloor p \rfloor +2))$ and   $s = qk$. Let
 $Y  \subseteq \{0,1\}^s$  such that, (a) $\abs{Y}
\ge 2^{0.08k}$, (b) each  $y \in Y$ has exactly $k$ ones, and (c) the
minimum Hamming distance among distinct codewords in $Y$ is
$3k/2$. Let $\{X_1, \ldots, X_s\}$ be a family of independent  random variables, each having  expectation $\mu >0$
and variance  bounded above by $\sigma^2$.
Let $\lambda$ be an estimate for $\mu$ satisfying $
\abs{\lambda-\mu} \le \min(\mu,\lambda)/(25p)$ and let $\sigma < \min(\mu, \lambda)/(25p)$. Let $\eta=
((\lambda-\mu)^2 + \sigma^2)^{1/2} >0$.
 Let $\bvtheta $ denote $   \bvtheta(t^p, \lambda, k,s,
 Y,  \{\pi_y\}_{y\in Y}, \{X_l\}_{l=1}^s) $.
Then
 $$\variance{\bvtheta} \le \left( \frac{(0.288)p^2}{k} \right)\mu^{2p-2} \eta^2
\enspace . $$
 \end{lemma}

\section{Algorithm} \label{sec:algo}
 \label{sec:ghss}

The \ghss~algorithm uses  a level-wise structure  corresponding to levels $l =0, 1, \ldots, L $,
where, the values of $L$ and the other parameters are given in Figure~\ref{table:params}.
\begin{figure}[t]
\begin{center}
\begin{tabular}{|p{1.8in}|l|}\hline \hline
\emph{Description of Parameter}&  \emph{Parameter and its value} \\ [3mm] \hline
 Number of levels  &~~ $L =  \lceil
\log_{2\alpha} \frac{n}{C} \rceil$ \\ [2mm]\hline
Reduction factor  &~~ $\alpha =  1- (1-2/p) \nu, ~\nu = 0.01$ \\  \hline
\multirow{2}{1.7in}{Basic  space parameters} &~~
 $ B    = \left(\cfrac{425 (2\alpha)^{p/2} n^{1-2/p}\epsilon^{-2}}{\min(\epsilon^{4/p-2}, \log (n))}\right)$ \\ &\\
 &~~  $ C = (27p)^2 B $    \\[1mm] \hline
 \multirow{3}{1.7in}{Level-wise space parameters }
  &~~ $ B_l = 4 \alpha^l B,  ~~~l  =0,1, \ldots, L-1 $ \\
  &~~ $C_l =  4\alpha^lC, ~~~ l = 0,1,  \ldots,  L -1 $  \\
  &~~ $C_L = 16 (4 \alpha^L C)$,  \\ [3mm] \hline
Degree of independence of $g_1, \ldots, g_L$   &~~   $d = 50\lceil  \log n \rceil$ \\ [3mm]  \hline

Taylor Polynomial Estimator Parameters &~~ $k = 1000\lceil\log n \rceil$, $r = 16k, s = 8k$
  \\[2mm] \hline
Degree of independence of table hash functions &~~ $t =11$ \\ \hline
\end{tabular}
\end{center}
\caption{Parameters used by the \ghss~algorithm.}
\label{table:params}
\end{figure}

\subsubsection*{Level-wise structures}
Corresponding to each level $l =0,1, \ldots, L-1$,  a pair of structures $(\hh_l, \est_l)$ are kept, where, $\hh_l$ is  a
\countsketch$(16C_l,s)$ structure with  $s = O(\log n)$ hash tables each consisting
of $16C_l$ buckets. The $\est_l$ structure is used by the Taylor  polynomial estimator  at level $l$
and is a standard \countsketch$(16C_l,2s)$ structure with the following minor changes.
 \begin{enumerate} [(a)]
 \item The hash functions $h_{lr}$'s
used for the hash tables $T_{lr}$'s  are  6-wise independent.
\item  The Rademacher family
$\{\xi_{lr}(i)\}_{i \in [n]}$ is $4$-wise independent for each  table index $r \in [2s]$, and is independent across
the $r$'s, $r \in [2s]$.
\end{enumerate}
The hash tables $\{T_{lr}\}_{r \in [2s]}$ have  $16C_l$ buckets each and use the hash function $h_{lr}$, for $r \in [2s]$. Corresponding to the final level $L $, only an  \hh$_L$ structure  is kept which is a  \countsketch$(C^*_L, s)$ structure, where $C^*_L = 16 C_L$. The structure at  level $L$ uses  $O(1)$ times larger space for $\hh_L$  to facilitate the discovery of all
items  and their  frequencies mapping  to this level (with very high probability).

\subsubsection*{Hierarchical Sub-sampling}
The original stream $\stream$ is sub-sampled hierarchically to produce random sub-streams
for each of the levels $\stream_0  = \stream \supset \stream_1 \supset \stream_2 \supset
\cdots \stream_L$, where, $\stream_l$ is the sub-stream that maps to level $l$. The stream
$\stream_0$ is the  entire input stream. $\stream_1$ is obtained by sampling each   item $i$
appearing in $\stream_0$ with probability $1/2$; if $i$ is sampled, then all its records $(i,v)$ are
included in $\stream_1$, otherwise none of its records are included. In general,
$\stream_{l+1}$ is obtained by sampling items from $\stream_l$ with probability $1/2$, so
that
$
\prob{i \in \stream_{l+1} \mid i \in \stream_l} = 1/2
$. This is done by a sequence of  independently chosen  random hash functions
  $g_1, g_2, \ldots, g_{L}$  each mapping $[n] \rightarrow \{0,1\}$. Then,
 \begin{equation*}
 i \in \stream_l ~\text{ iff }  g_1(i) =1,   g_2(i) =1,   \ldots, g_{l}(i)=1, ~~~ l =1,2, \ldots, L \enspace .
 \end{equation*}
 If $i \in \stream_l$, then for each stream update of the form $(i,v)$, the update is propagated to the structures  $\hh_l$ and $\tpest_l$.

\subsubsection*{Group thresholds and Sampling into groups}
Let $\hat{F}_2$ be an estimate satisfying
 $F_2 \le \hat{F}_2 \le  (1 + 0.01/(2p))F_2$ with probability $1-n^{-25}$ and is computed using random bits that are independent of the ones used in the above structures.

Let $\epsbar = (B/C)^{1/2} =1 /(27p)$. The  level-wise  thresholds are defined  as follows.
\begin{gather} \label{eq:Tl}
 T_0 = \left( \frac{ \hat{F}_2}{B} \right)^{1/2},  ~T_l = \left( \frac{1}{2\alpha
 }\right)^{l/2} T_0, ~~~~l \in  [L-1], \eat{ ~~~ T_L = , }\text{ and }  \notag \\
 Q_l = T_l - \epsbar T_l, ~l \in \{0\} \cup  [L-1], ~~~Q_L = 1/2
 \enspace .
 \end{gather}
 Let $\hat{f}_{il}$ be the estimate for $f_i$ obtained
from  level $l$ using $\hh_l$. For $l \in \{0\} \cup [L-1]$, we say that $i$ is ``discovered'' at level $l$, or that $l_d(i) = l$,    if $l$ is the smallest level such that  $\abs{\hat{f}_{il}} \ge Q_l $. Define
$\hat{f}_{i} = \hat{f}_{i,l_d(i)}$.  $l_d(i)$ is set to $ L$ iff $i \in \stream_L$ and $i$ has not been discovered at any earlier level.

Items are  placed into sample groups, denoted by $\bar{G}_l$, for $l \in \{0\} \cup [L]$, as follows. An item is placed into the sampled group $\bar{G}_l$ if the following holds.
\begin{enumerate}
\item  If $i$ is discovered at level $l$  and $\abs{\hat{f}_{il}} \ge T_l$, then, $i$ is included in $\bar{G}_l$.
\item If
$i$ is discovered at level $l-1$ but  $\abs{\hat{f}_{i,l-1}} < T_{l-1}$ and  the flip of an  unbiased coin $K_i$ turns up \emph{heads}.
\end{enumerate}
An item $i$ is placed in $G_0$ if $\abs{\hat{f}_{i0}} \ge T_0$.
In other words, the sample groups are defined as follows.
\begin{align*}
\bar{G}_0 &= \{i : \abs{\hat{f}_i} \ge T_0\}, \\
\bar{G}_l &= \{ i : ( l_d (i)  = l \text{ and }  \abs{\hat{f}_{i}}  \ge T_l   ) \text{ or } (  l_d (i) = l-1\text{ and } \abs{\hat{f}_i} < T_{l-1}\text{ and }
K_i = 1) \}, ~l=1,2, \ldots, L-1, \\
\bar{G}_L &= \{ i: l_d(i) = L \text{ or } (l_d(i) = L-1\text{ and } \abs{\hat{f}_i } < T_{L-1}\text{ and } K_i = 1) \} \enspace .
\end{align*}
We refer to an item as being \emph{sampled} if it belongs to a sample group. From the construction above,  it follows that (1) only an item that is discovered may be sampled, and  (2) if  $i \in [n]$ is discovered at level $l$, then, $i$ may belong to sampled group $\bar{G}_l$ or to the sampled group $\bar{G}_{l+1}$, or to neither (and hence to no sampled group). That is,  there is a possibility that discovered items are not sampled (this happens when $ Q_l \le \hat{f}_{il} < T_l$ and $K_i = 0$ (tails)).

\subsubsection*{The {\sc nocollision}~event} Let  $\hattopk_l(C_l)$ be the set of  the top-$C_l$ elements in terms of  the   estimates
$\abs{\hat{f}_{il}}$ at level $l$. For $l \in\{0\} \cup  [L]$, $\nocollision_l$ is said to hold if for each $i
\in \hattopk_l(C_l)$, there exists a set $R_l(i) \subset [2s]$ of indices of hash tables of the structure  $\tpest_l$ such that $\abs{R_l(i)} \ge s$ and
that $i$ does not collide with any other item of $\hattopk_l(C_l)$ in the buckets
$h_{lq}(i)$, for $q \in R_l(i)$. More precisely,
\begin{multline} \label{eq:nc} \nocollision_l \equiv  \forall i \in \hattopk_l(C_l), \exists R_l(i) \subset [2s]
\left(\abs{R_l(i)} \ge s \text{ and }  \right. \\ \left.
\forall q \in R_l(i), \forall j \in \hattopk_l(C_l) \setminus \{i\}  ~~h_{lq}(i) \ne h_{lq}(j) \right)
\enspace .
\end{multline}
The    event $\nocollision$ is defined as  $$ \nocollision \equiv \wedge_{l=0}^L \nocollision_l \enspace . $$ The analysis shows \nocollision~to be a very  high probability event, however,
if  $\nocollision$ fails, then, the estimate for $F_p$  returned is 0.

\subsubsection*{The estimator $\hat{F}_p$} Assume that  the event \nocollision~holds, otherwise, $\hat{F}_p$ is set to 0.  For each item $i$ that is discovered at  level $l_d(i) < L$ and is sampled into sampled group at level  $l_s(i)$, the averaged Taylor polynomial estimator is used to
obtain an estimate of $\abs{f_i}^p$ using  the structure $\tpest_{l_d(i)}$ at level  $l_d(i)$ and scaled by factor of  $2^{l_s(i)}$ to compensate for sampling.
If $l_d(i) = l_s(i) =L$, then  the simpler estimator $\abs{\hat{f}_i}^p$ is used  instead and the resulting estimate is scaled by $2^L$.

The
parameter $\lambda $ used in the Taylor polynomial estimator for estimating $\abs{f_i}^p$ is set to $\abs{\hat{f}_{i}} = \abs{\hat{f}_{i,l_d(i)}}$. Let $l = l_d(i)$. By  $\nocollision$, let
$R_{l}(i) = \{t_1, t_2, \ldots, t_s\} \subset [2s]$. Let
 $X_{ijl}$  be the (standard) estimate  for $\abs{f_i}$ obtained from table $T_{lj}$,  that
 is,
 $$ X_{ijl} = T_{lj}[h_{lj}(i)] \cdot \xi_{lj}(i) \cdot \sgn(\hat{f}_{i}), ~~~\text{ for $j \in R_l(i)$. }
$$
 The estimator  $ \bvtheta_{i}$ is defined as
 $$ \bvtheta_i = \bvtheta(t^p, \abs{\hat{f}_{i}},k, s,Y,
 \{\pi_j\}_{j \in Y}, \{X_{ijl}\}_{j \in R_l(i)} \}) $$ where, $Y$ is a code satisfying Corollary~\ref{cor:gv} and $\{\pi_j\}_{j \in Y} \}$ is a family of independently and randomly chosen permutations from $[k] \rightarrow [k]$.  The parameters $k$  and $s$ are given in  Figure~\ref{table:params}. The estimator $\hat{F}_p$ for $F_p$
 is defined below.
\begin{align} \label{eq:hatFp}
\hat{F}_p =  \sum_{l = 0}^{L} \sum_{i \in \bar{G}_l, l_d(i) < L}  2^l \cdot  \bvtheta_{i} +  \sum_{i \in \bar{G}_L, l_d(i) = L} 2^L \cdot \abs{\hat{f}_i}^p \enspace  . \end{align}
\section{Analysis}
\label{sec:anal:ghss}
In this section, we analyze the \ghss~algorithm.

\subsection{The event $\G$}
 Let $\ftwores{k,l}$ denote  the (random) $k$-residual second
moment   of the frequency vector corresponding to $\stream_l$.
The analysis is conditioned on the  conjunction of a set of  events denoted by $\G$, as defined in Figure~\ref{fig:G}.
\begin{figure}[htbp]
\begin{center}
\begin{tabular}{ll>{$}l<{$}}
(1) &~\goodftwo &\equiv   F_2 \le  \hat{F}_2  \le \left(1+\cfrac{0.001}{2p}\right)F_2, \\
(2) &~\nocollision &~~~\text{ defined in ~\eqref{eq:nc}} \\
(3) &~\goodest  &\equiv   \forall l: 0 \le l \le L, ~\forall i \in [n], ~~ \abs{\hat{f}_{il} - f_i}\le  \left(
\cfrac{\ftwores{2C_l,l}}{C_l} \right)^{1/2}\\[2mm]
(4) & ~\smallres & \equiv  \forall l: 0 \le l \le L,~ \ftwores{2C_l,l} \le  \cfrac{1.5 \ftwores{ 
(2\alpha)^l C  }}{2^{l-1}} \\[2mm] (5) & ~\accuest &\equiv  \forall l: 0 \le l \le L, ~\forall i \in [n], ~ \abs{\hat{f}_{il} - f_i}\le   \left(
\cfrac{\ftwores{ (2\alpha)^l C }}{2(2\alpha)^l C}\right)^{1/2} \\ [2mm]
(6)& ~\lastlevel~ &\equiv \forall i \in \stream_L, \hat{f}_{iL} = f_{i}\\[2mm]
(7)&~\smallhh& \equiv   \forall l: 0 \le l \le L, \{i: \abs{\hat{f}_{il}}\ge Q_l\} \subset \overline{\topk}(C_l).
\end{tabular}
\end{center}
\caption{$\G$ is the conjunction of these 7 events}
\label{fig:G}
\end{figure}

The events comprising $\G$ are as follows.  \goodftwo~is the  event that $\hat{F}_2$ is an $1+O(1/p)$-factor approximation of $F_2$.  The event $\goodest$  states that for all $i \in [n]$ and levels $l \in \{0\} \cup [L]$, the frequency estimation errors incurred by the $\hh_l$ structure remains within the high-probability error bound for  the \countsketch~algorithm \cite{ccf:icalp02} given by Eqn.~\eqref{eq:cskbasic}. However, the bounds in \goodest~have to be  expressed in terms of $\ftwores{2C_l,l}$, which are themselves  random variables. The event \smallres~gives some control on this random variable by giving an upper bound on $\ftwores{2C_l,l}$ as $\frac{1.5\ftwores{ (2\alpha)^l C)}}{2^{l-1}}$. The event \accuest~holds if  the frequency estimation for an item $i$ at a certain level $l$  has an additive accuracy of $\frac{\ftwores{ (2\alpha)^l C)}}{(2\alpha)^l C}$. The bounds given by \accuest~are non-random functions of $l$.
An item $i$ is \emph{classified as a heavy-hitter at level} $l$ if $\hat{f}_{il} \ge Q_l$, that is, its estimate obtained from the $\hh_l$~structure exceeds the threshold $Q_l$. The event \smallhh~is said to hold if at each level, each heavy-hitter item at that level  is among those with the  top-$C_l$ absolute estimated  frequencies at that level.  The \nocoll~event is used only by the \tpest~family of structures at each level, and ensures that each heavy-hitter remains isolated from all the other heavy-hitters of that level in at least half ( $s$)  of the tables of the \tpest~structure at that level.

Lemma~\ref{lem:hss:G} shows that $\G$ holds except with inverse polynomial probability.
\begin{lemma} \label{lem:hss:G} For the choice of parameters in Figure~\ref{table:params},
$\G$ holds with probability $1- O(n^{-24})$.
\end{lemma}
\eat{
 }

\subsection{Grouping items by frequencies}
Items are divided into groups based upon frequency ranges, as follows.
\begin{align*}
 G_0 & = \{ i: \abs{f_i} \ge T_0\}\\
  G_l  &= \{ i: T_l \le \abs{f_i} < T_{l-1} \}, l=1,2, \ldots, L-1 \\
 G_L  &= \{i: 1\le \abs{f_i} < T_{L-1}\} \enspace .
 \end{align*}
Note that this grouping is for purposes of analysis, since the true frequencies are unknown to the algorithm. Since estimated frequencies may have errors, it is possible that the sampling algorithm samples an item $i$ into the sampled group $\bar{G}_l$, although, the item does not belong to the group $G_l$. It will be useful to understand the conditions under which such errors do not  occur, and the conditions under which such  errors may occur and their extent.

Each group is further  partitioned into subsets defined by frequency ranges, namely,   $\lmargin(G_l)$, $ \midreg(G_l)$ and $\rmargin(G_l)$.
 \begin{align*}
 \lmargin(G_l) &= \{i: T_l \le \abs{f_i}  <  T_l(1+  \epsbar) \}, ~~l =0, \ldots, L-1, \\
 \rmargin(G_l) &= \{ i :  T_{l-1}(1 - 2\epsbar) \le \abs{f_i} <  T_{l-1} \}, ~~l \in [L]\\
  \midreg(G_l) &= \{ i: T_l+ T_{l}\epsbar \le \abs{f_i}  < T_{l-1} - 2T_{l-1}\epsbar\}, ~~l \in [L-1], \\
 \midreg(G_0) &= \{i: \abs{f_i} \ge T_0(1+\epsbar)\} \\
 \midreg(G_L) & = \{ 1 \le \abs{f_i} < T_{L-1}(1-2\epsbar)\} \enspace
 .
 \end{align*}
$G_0$ and $G_L$ have no \rmargin$(G_0)$~and \lmargin$(G_L)$ defined, respectively.  These definitions are similar (though not identical) to the \Hss~algorithm \cite{gl:hssfull}. The ratio $\frac{T_{l-1}}{T_l}  = (2\alpha)^{1/2}$, for $l = 1, 2, \ldots, L-1$.  The last group $G_L$ has frequency range is $[1, T_{L-1})$ and the frequency ratio $T_{L-1}/1$ can be  large.

\subsection{Properties of the sampling scheme}
 In the remainder of this paper, we assume that $c > 23$ is a constant satisfying $ \prob{\neg \G}/\prob{\G} \le n^{-c}$.

 \subsubsection*{Basic Property}
 Lemma~\ref{lem:margin} presents the   basic property of the sampling scheme.

\begin{lemma} \label{lem:margin}
Let $i \in G_l$.
\begin{enumerate}
\item Let  $i \in \midreg(G_l)$. Then,  $$\card{2^l\prob{ i \in \bar{G}_l \mid \G} - 1} \le 2^ln^{-c} \enspace . $$ Further, conditional on $\G$,  (i)  $i \in \bar{G}_l $ iff $i \in \stream_l$,  and, (ii) $i$ may
not  belong to any $\bar{G}_{l'}$, for $l' \ne l$, that is, (i) $\prob{i \in \bar{G}_l \mid \G} = \prob{ i \in \stream_l \mid \G} =2^l \pm n^{-c} $, and,  (ii) $ \prob{ i \in \cup_{l' \ne l} \bar{G}_{l'} \mid \G} =0$.

\item Let  $i \in \lmargin(G_l)$. Then
$$
 \card{2^{l+1} \prob{i \in \bar{G}_{l+1} \mid \G} + 2^l\prob{i \in \bar{G}_l \mid \G} -
 1} \le 2^l n^{-c}  \enspace . $$
Further, conditional on $\G$,  $i$ may  belong to either $ \bar{G}_l$ or $\bar{G}_{l+1}$,
but not to any other sampled group, that is,
$\prob{ i \in \cup_{l' \not\in \{l,l+1\}} \bar{G}_{l'} \mid \G} = 0$.
\item If $i \in \rmargin(G_l)$, then
 $$
 \card{ 2^l \prob{ i \in \bar{G}_l \mid \G} + 2^{l-1}\prob{ i \in \bar{G}_{l-1} \mid \G} -
 1} \le O(2^l n^{-c}) \enspace . $$
 Further, conditional on $\G$,   $i$ can belong to  either $ \bar{G}_{l-1}$ or $
 \bar{G}_{l}$ and not to any other sampled group, that is,
 $\prob{i \in \cup_{l' \not\in \{l-1,l\}} \bar{G}_{l'} \mid \G}  = 0$.
 \end{enumerate}
\end{lemma}

Lemma~\ref{lem:margin} is essentially  true (with minor changes) for the   \Hss~method
 \cite{bgks:soda06,gl:hssfull}, although the  \Hss~analysis used full-independence of hash functions whereas here we work with limited independence. A straightforward corollary of Lemma~\ref{lem:margin} is the following.

\begin{corollary}\label{cor:margin} Let $i \in G_l$. Then,
\begin{align*}
\sum_{l'=0}^L 2^{l'} \prob{i \in \bar{G}_{l'} \mid \G}  = \sum_{ \substack{l' \in  \{0,1, \ldots, L\}  \cap \{l-1,l,l+1\}}} \prob{i \in \bar{G}_{l'} \mid \G} =  1 \pm 2^{l+1} n^{-c} \enspace .
\end{align*}
\end{corollary}

\subsubsection*{Approximate pair-wise independence property}

Lemma~\ref{lem:hsscond} essentially repeats the results of Lemma~\ref{lem:margin},
conditional upon the event that another item maps to a substream at some level $l$. This property is a step towards  proving an approximate pair-wise independence property in the following section.

\begin{lemma} \label{lem:hsscond} Let $i,j \in [n]$,  $i \ne j$ and $j \in  G_r$.
\begin{enumerate}
\item Let $j \in \midreg(G_r)$.  Then $$
 \left\lvert 2^r \prob{ j \in \bar{G}_r \mid i \in \stream_l, \G}  -1 \right\rvert  \le  2^r n^{-c} \enspace .
 $$ Further, for  any $r \ne r'$, $ \probb{j \in \bar{G}_{r'}  \mid i \in \stream_l, \G} = 0 \enspace . $

\item Let $j \in \lmargin(G_r)$. Then,
$$
  \left\lvert 2^{r+1} \prob{j \in \bar{G}_{r+1} \mid i \in \stream_l, \G} + 2^r\prob{j \in
 \bar{G}_r \mid i \in \stream_l,\G} - 1 \right\rvert
  \le  2^{r+1} n^{-c} \enspace .
$$
Further, for any $ r' \not\in \{r,r+1\}$, $ \probb{j \in \bar{G}_{r'}\mid i \in \stream_l, \G}  = 0$.
\item If $j \in \rmargin(G_r)$, then  $$
 \left \lvert 2^{r} \prob{ j \in \bar{G}_r \mid i \in \stream_{l}, \G} + 2^{r-1}\prob{ j \in
 \bar{G}_{r-1} \mid i \in \stream_{l},\G} - 1 \right\rvert \le   2^{r+1} n^{-c}  \enspace . $$
 Further, for any $ r' \not\in \{r-1,r\},  \prob{j \in \bar{G}_{r'}\mid i \in \stream_l, \G} = 0$.
 \end{enumerate}
\end{lemma}

\begin{corollary} \label{cor:hsscond}
Let $i,j  \in [n]$, $i \ne j$ and $j \in G_r$. Then,
$$ \left \lvert
  \sum_{r'=0}^L 2^{r'} \prob{j \in \bar{G}_{r'} \mid i \in \stream_l, \G} -1 \right\rvert \le  O(2^{r}n^{-c}) \enspace .$$
\end{corollary}

We can now prove an approximate pair-wise independence property.

\begin{lemma} \label{lem:hssj} For $i \in G_l$, $j \in G_m$ and $i,j$ distinct,
$$\left\lvert \sum_{r,r' =0}^L 2^{r+r'}\prob{ i \in \bar{G}_r, j \in \bar{G}_{r'}  \mid \G}  - 1 \right\rvert \le
O((2^l + 2^m)n^{-c}) \enspace .$$
\end{lemma}

\subsection{Application of  Taylor Polynomial Estimator}
Let $i \in \bar{G}_{l'}$ for some $l' \in \{0\} \cup [L-1]$. Then, $i$ has been discovered at a level $l_d(i) =l $ (say). The algorithm estimates $\abs{f_i}^p$ from the \tpest~structure at the discovery level $l$ using  the estimator $$\bvtheta_{i} = \bvtheta(\psi(t)=t^p, \allowbreak \abs{\hat{f}_{i}}, \allowbreak k,  \allowbreak s, \allowbreak  Y,\allowbreak
 \{\pi_y\}_{y\in Y}, \allowbreak \{X_{ijl}\}_{j \in R_l(i)} \})  \enspace . $$ By construction,  $\hat{f}_i$ is defined as $\hat{f}_{il}$ and for any $j \in R_l(i)$,  $\sigma_{ijl} = (\variance{X_{ijl}})^{1/2}$ and $\eta_{ijl} = (\sigma_{il}^2 + (\abs{f_i} - \abs{\hat{f}_{il}})^2$.
 We first show that the premises of Corollary~\ref{cor:fpbias} and Lemma~\ref{lem:vbvtheta} are satisfied so that we can use their implications.

\begin{lemma} \label{lem:hssavtp}
Assume the parameter values listed in Figure~\ref{table:params} and that  $\G$ holds. Suppose $l_d(i) = l$ for some $l \in  \{0\} \cup [L-1]$. Then the following properties hold.
\renewcommand*\theenumi{\alph{enumi}}
\renewcommand*\labelenumi{\theenumi)}
\begin{enumerate}
 \item $\abs{\hat{f}_{il} - f_i} \le \abs{f_i}/(26p)$,
 \item $\expect{X_{ijl} \mid l_d(i) = l, \abs{\hat{f}_{il}} > Q_l,  j \in R_l(i),\G} = \abs{f_i}$,
 \item  $\abs{f_i} \ge 15p\eta_{ijl_d(i)}$, for $j \in R_{l_d(i)}(i)$,
 \item $\eta_{ijl_d(i)}^2 \le 2.7 (\epsbar T_l)^2 $, for $j \in R_{l_d(i)}(i)$,
 \item $\abs{\hat{f}_{il} - f_i}  \le \abs{\hat{f}_i}/(26p)$,
\item  $ \abs{\hat{f}_i}/\eta_{ijl_d(i)} \ge 16p$, for $j \in R_{l_d(i)}(i)$, 
\item  if $l_d(i)=L$, then, $\hat{f_i} = f_i$ and $\eta_{iL} = 0$.
\end{enumerate}
\end{lemma}

\renewcommand*\labelenumi{\theenumi)}

\noindent
For $i,k \in \stream_l$,  $j \in [2s]$, let $u_{ikjl} =1$  iff $h_{lj}(i) = h_{lj}(k)$ and 0 otherwise.

\begin{lemma}\label{lem:fp:expect1} Assume the parameters in Figure~\ref{table:params} and let   $p \ge 2$. Suppose $i \in \bar{G}_l$, for some $l \in \{0\} \cup [L-1]$. Then,
$$\card{\expect{\bvtheta_i \mid \G }- \abs{f_i}^p } \le n^{-4000p}\abs{f_i}^p \enspace. $$ Further if $p$ is integral, then, $\expect{\bvtheta_i \mid \G} = \abs{f_i}^p$.
\end{lemma}

We denote by $\xibar$ the set of random bits defining the family of Rademacher random variables used by the \tpest~structures, that is, the set of random bits that defines the family $\{\xi_{lj}(i) \mid i \in  [n], j \in [2s], l \in \{0\} \cup [L]\}$.
Lemma~\ref{lem:bvthetacross} shows that the  event \nocollision~implies that the Taylor polynomial estimators are pair-wise uncorrelated.
\begin{lemma} \label{lem:bvthetacross} Suppose   $i \in \bar{G}_r$ and $i' \in \bar{G}_{r'}$. Then, $$\expectsubb{\bar{\xi}}{\bvtheta_i\bvtheta_{i'} \mid \hat{f}_{i}, \hat{f}_{i'}, \G} =
\expectsubb{\bar{\xi}}{\bvtheta_i\mid \hat{f}_{i},\G}\expectsubb{\bar{\xi}}{\bvtheta_{i'} \mid
\hat{f}_{i'},\G} \enspace . $$
\end{lemma}
\eat{\begin{pfsktch} Let $i$ be discovered at level $l$ and $i'$ at level $l'$.  The inference for $\bvtheta_i$ is made as the average of $\vartheta_i$'s by taking some subset of the table indices in $Q_l(i)$ and permuting them randomly (and similarly for $\bvtheta_{i'}$). It therefore suffices to show that the statement of the lemma holds with $\bvtheta_i$ replaced by $\vartheta_i$ and analogously for $i'$. Since, $\vartheta_i$ is a sum of products of $X_{ijl}$, for some  $j$'s $ \in Q_l(i)$, it suffices to show that for any $1 \le t, t' \le k$, $ \expectsubB{\xi}{ \prod_{w=1}^t X_{ij_wl} \prod_{w'=1}^{t'} X_{ij'_{w'}l'} \mid \G}  = \abs{f_i}^t\abs{f_{i'}}^{t'}$. Now if $l \ne l'$, then, $X_{ij_wl}$ and $X_{ij'_{w'} l'}$ are independent, since the families $\{\xi_{jl}\}$ and $\{\xi_{j'l'}\}$ are independent whenever $l \ne l'$. So suppose that $l = l'$. The inference from distinct table indices are independent, since the $\xi_{jl}$ family is independent of $\xi_{j'l}$ for $j \ne j'$.  The expectation of the product form can be split into products of expectation  over $ \xi_{jl}$, for distinct table indices  $j$. Let $J_t = \{j_1, \ldots, j_t\}$ and $J'_{t'} = \{ j'_1, \ldots, j'_t\}$. Thus,
\begin{align*}
\expectsubB{\xi_{l}}{ \prod\limits_{w=1}^t X_{ij_wl} \prod_{w'=1}^{t'} X_{ij'_{w'}l'} \mid \G} = \prod\limits_{j \in J_t \setminus J'_t} \expectsubB{\xi_{jl}}{X_{ijl} \mid \hat{f}_i,\hat{f}_j, \G} \prod\limits_{j'\in J'_{t'} \setminus J_t}\expectsubB{\xi_{jl}}{X_{i'j'l}\mid  \hat{f}_i,\hat{f}_j, \G} \prod\limits_{j \in J_t \cap J'_t} \expectsubB{\xi_{jl}}{X_{ijl}X_{i'jl} \mid \hat{f}_i,\hat{f}_j, \G}
\end{align*}
The expectation of a single $X_{ijl}$ or $X_{i'j'l}$ is $\abs{f_i}$ and $\abs{f_j}$ respectively, and is  therefore easy to calculate.
Let $j \in J_t \cap J'_{t'}$ and consider $ \expectsubb{\xi}{X_{ijl}X_{i'jl}}$. Now $i$ and $i'$ are both in $\hattopk_l(C_l)$ conditional on $\G$ (this is the event \smallhh) and $\sgn(\hat{f}_i) = \sgn(f_i)$ (and likewise for $i'$). Since, $\nocollision$ holds, $i$ and $i'$ map to different buckets under $h_{jl}$, that is, $h_{jl}(i) \ne h_{jl}(j)$. Hence,
 \[
\begin{array} {l}\expectsubb{\xi_{jl}}{X_{ijl}X_{i'jl} \mid \hat{f}_i, \hat{f}_j, \G}  \\
= \expectsubB{\xi_{jl}}{\Bigl(\abs{f_i} + \sum\limits_{k\ne i} f_k \xi_{jl}(k) \xi_{jl}(i) \sgn(\hat{f}_i) u_{ijkl} \Bigr)
\Bigl(\abs{f_{i'}} + \sum\limits_{k'\ne {i'}} f_{k'} \xi_{jl}(k') \xi_{jl}({i'}) \sgn(\hat{f}_{i'})u_{i'jk'l}\Bigr)
\mid \hat{f}_i, \hat{f}_j, \G}\\
 = \abs{f_i}\abs{f_{i'}} + \sum\limits_{k \ne i, k' \ne i'} f_k f_{k'} \sgn(\hat{f}_i) \sgn(\hat{f_{i'}}) \expect{ \xi_{jl}(k) \xi_{jl}(i) \xi_{jl}(k') \xi_{jl}({i'})u_{ijkl} u_{i'jk'l'}\mid  \hat{f}_i, \hat{f}_j, \G} = \abs{f_i}\abs{f_{i'}}
\end{array} \]
where, each term of the summation is 0 on expectation, since, $k \ne i, k' \ne i'$, and since, $i \ne i'$ and $i$ and $i'$ map to different buckets, therefore, $u_{ijkl}u_{i'jk'l} = 0$ whenever $k = k'$.  Thus, all four items $i,i',k$ and $k'$ are distinct, and by 4-wise independence of the Rademacher sketch family, the expectation of each term in the summation is 0.  Hence, the expectation of the product form is the product of the individual expectations.

\end{pfsktch}
}

\subsection{Expectation and Variance of $\hat{F}_p$ Estimator.} For uniformity of notation, let  $\bvtheta_i$ denote $\abs{\hat{f}_{i}}$ when $l_d(i) = L$ and otherwise,  let its meaning be unchanged.\eat{  $\bvtheta_i$ mean the usual application of the averaged Taylor polynomial estimator, provided, $l_d(i) < L$. Otherwise, let $\bvtheta_i$  denote $\abs{\hat{f}_{iL}}^p$, which, with very high probability, equals $\abs{f_i}^p$.} Let $z_{il}$ be an indicator variable
that is 1 if $i \in \bar{G}_l $  and 0 otherwise. Since an item may be sampled into at most one group, $\sum_{l\in [L]} z_{il} \in \{0,1\}$. Using the extended definition of $\bvtheta_i$ mentioned above, we can write $\hat{F}_p$ as,
\begin{align} \label{eq:expFp0}
\hat{F}_p & = \sum_{l = 0}^{L} \sum_{i \in \bar{G}_l} 2^l  \bvtheta_{i}  \notag \\
& = \sum_{i \in [n]}  \sum_{l=0}^{L}  z_{il} \cdot  2^l \cdot \bvtheta_i  \notag \\
 & = \sum_{i \in [n]} Y_i \end{align}
where,
\begin{align} \label{eq:Yi} Y_i = \sum_{l'=0}^{L-1}2^{l'} z_{il'} \bvtheta_i \enspace .
\end{align}
Lemma~\ref{lem:Y} shows that $\hat{F}_p$ is almost an unbiased estimator for $F_p$. This follows  from Lemma~\ref{lem:fp:expect1}.
 \begin{lemma} \label{lem:Y}
 $\expectb{\hat{F}_p\mid \G} =  F_p (1\pm O(n^{-c+1}))$.
 \end{lemma}

We will use the following facts that are easily proved (see Appendix).
\begin{align} \label{eq:fpfacts}
F_2 &  \le n^{1-2/p} F_p^{2/p}, &\text{ $p \ge 2$,}  \notag \\
 F_{2p-2} & \le F_p^{2-2/p},  &\text{$p \ge 2$.}
\end{align}

\begin{lemma} \label{lem:varYi} Let 
$B =  K n^{1-2/p} \epsilon^{-2}/\log (n)$ and $C = (27p)^2 B$. Then,
 \begin{align*}
 \variance{Y_i \mid \G} \le
  \begin{dcases} \cfrac{  \epsilon^2 \abs{f_i}^{2p-2}F_p^{2/p}}{(5)(10)^4 K} & \text{if } i \in \midreg(G_0) \\
 2^{l+1} (1.002) \abs{f_i}^{2p} & \text{ if }  i \in \lmargin(G_0) \cup_{l=1}^L G_l
 \end{dcases}
 \end{align*}
\end{lemma}

Lemma~\ref{lem:varcross} builds on the approximate pair-wise independence of the sampling scheme  (Lemma 
 ~\ref{lem:hssj})  and the pair-wise uncorrelated property of  the $\bvtheta_i$ estimators (Lemma~\ref{lem:bvthetacross}) to   show that the $\covariance{Y_i}{Y_j}$, for $i \ne j$ is very small.
\begin{lemma} \label{lem:varcross} Let $i \ne j$. Then,
$$ \card{\covariance{Y_i}{Y_j \mid \G} } \le   O(n^{-c+1})\abs{f_i}^p \abs{f_j}^p \enspace . $$
\end{lemma}

Lemma~\ref{lem:varFp} gives a bound on the variance of the $\hat{F}_p$ estimator.
\begin{lemma}\label{lem:varFp}
$$\varianceb{\hat{F}_p \mid \G} \le  \frac{\epsilon^2 F_p^2}{50} \enspace . $$
\end{lemma}

\subsection*{Putting things together}

Theorem~\ref{thm:fp} states the space bound for the algorithm and the update time.
\begin{theorem} \label{thm:fp} For each fixed  $p > 2$ and $0 < \epsilon \le 1$, there exists an algorithm in the general update data stream model that returns $\hat{F}_p$ satisfying $\card{ \hat{F}_p - F_p} < \epsilon F_p$ with probability $3/4$. The algorithm uses space $O(n^{1-2/p} \epsilon^{-2} + n^{1-2/p} \epsilon^{-4/p} \log (n))$ words of size $O(\log (nmM))$ bits. The time taken to process each stream update is $O(\log^2 n)$.
\end{theorem}

\subsection*{Acknowledgement} The author thanks Venugopal G. Reddy for correcting an error in the analysis.

\bibliographystyle{plain}


\begin{appendices}

\section{Proofs for the Taylor Polynomial estimator}

\begin{fact} \label{fact:1}
Let $k > p \ge 0$. Then,  $\card{ \binom{p}{k}} \le \left(\frac{p}{k} \right)^{\lfloor p \rfloor+1}$. In particular, if $p \in \Z^{+} $, then, $\binom{p}{k} = 0$.
\end{fact}
\begin{proof} The second statement is obvious, since for $k>p \ge 0$ and $p$ integral,  $\powf{p}{k} = 0$. Otherwise, for non-integral $p$, using the absorption identity $\lfloor p \rfloor +1$ times, gives
\[ \begin{array}{c}\binom{p}{k} = \left(\frac{\powf{p}{\lfloor p \rfloor +1}}{\powf{k}{\lfloor p \rfloor +1}} \right) \binom{p-\lfloor p \rfloor -1}{k-\lfloor p \rfloor -1} =
\left(\frac{\powf{p}{\lfloor p \rfloor +1}}{\powf{k}{\lfloor p \rfloor +1}} \right) (-1)^k \binom{k -p-1}{k-\lfloor p \rfloor -1}
\end{array}
\]
Now, for $0 \le j \le \lfloor p \rfloor$,  $\frac{p-j}{k-j} \le \frac{p}{k}$, since $p <k$. Therefore, $ \frac{\powf{p}{\lfloor p \rfloor +1}}{\powf{k}{\lfloor p \rfloor +1}} \le \left( \frac{p}{k} \right)^{\lfloor p \rfloor +1}$. Similarly,  $ \binom{k-p-1}{k-\lfloor p \rfloor -1} \le \left(\frac{k-p-1}{k-\lfloor p \rfloor -1} \right)^{k-\lfloor p \rfloor -1} < 1$. Taking absolute values,
$\card{ \binom{p}{k}} \le \left( \frac{p}{k} \right)^k$. \hfill
\end{proof}

\eat{
\begin{fact}\label{fact:kt}  Let $k \ge 1$, $1 \le t \le k-1$ and $1 \le v \le k-1$, then,
\begin{align} \label{eq:fact:kt}
 \frac{\powf{(k-t)}{v}}{\powf{k}{v}} \left(\frac{t}{k-v}\right)   \le
2\left(1-\frac{\powf{(k-t)}{v}}{\powf{k}{v}}\right) \enspace .
\end{align}
\end{fact}
\begin{proof}If $v >k-t$, then the \emph{LHS} is 0 and the \emph{RHS} $=2$, implying the statement of the lemma.

So now let $v \le k-t$. For a fixed value of $k$ and $t$, let  $\beta_{v}  = \frac{\powf{(k-t)}{v}}{\powf{k}{v}}$ and let $\alpha_v =  \frac{t}{k-v}$.  Then, Eqn.~\eqref{eq:fact:kt} may be equivalently written as
\begin{align} \label{eq:fact:gamma} & \beta_v\alpha_v\le 2 (1-\beta_v), \text{ or, } ~~ (2 + \alpha_v)\beta_v \le 2 \enspace .
\end{align}
Now $\beta_v = \frac{\powf{(k-t)}{v}}{\powf{k}{v}} \le \left( \frac{k-t}{k} \right)^v \le \exp{ -\frac{vt}{k}}$. Further, $2+\alpha_v = 2\left(1 + \frac{t}{2(k-v)} \right)\le  2\exp{\frac{t}{2(k-v)}}$. Therefore,
\begin{align} \label{eq:fact:gamma:1}
(2+\alpha_v)\beta_v \le  2 \exp{ \frac{t}{2(k-v)} -\frac{vt}{k} }= 2\exp{-\frac{t ( -k + 2v(k-v))}{2k(k-v)}} \enspace .
\end{align}
Now
$ -k+2v(k-v) = -k -2(v-k/2)^2 + k^2/2$
 and  is minimized when $\abs{v -k/2}$ is maximized. In the interval $v \in [1, k-1]$, $\abs{v-k/2}^2$ attains a maximum value of $(k/2-1)^2$ and therefore, $-k+2v(k-v) \ge -k  -2(k/2-1)^2 + k^2/2 = k-2 \ge 0$, for $k \ge 2$. Hence, the expression in the \emph{RHS} of Eqn.~\eqref{eq:fact:gamma:1} is at most $2$, which implies ~\eqref{eq:fact:gamma}. \hfill
\end{proof}

}

\begin{proof} [Proof of Lemma~\ref{lem:vvtheta}.] Fix $\psi, \lambda$ and $k$ and let
$\vartheta=\vartheta(\psi, \lambda, k,$ $X_1, \ldots, X_k) $. Using linearity of expectation
and independence of $X_i$'s we have,
\begin{align*}
\expect{\vartheta}  & = \expect{  \sum_{j=0}^k \gamma_j(\lambda) \prod_{v=1}^j (X_v -\lambda)
} = \sum_{j=0}^k \gamma_j(\lambda) \prod_{v=1}^j (\mu-\lambda)   = \psi(\lambda+\mu-\lambda) -
 \gamma_{k+1}(\lambda')(\mu-\lambda)^{k+1}
\end{align*}
for some $\lambda' \in (\mu,\lambda)$ by the Taylor series
expansion of $\psi(\mu) = \psi(\lambda + (\mu-\lambda))$ around
$\lambda$. The Taylor series expansion of $\psi(\mu)$ around
$\lambda$ exists since $\psi$ is analytic in the interval
$[\mu,\lambda]$.  Therefore,
 $$ \left\vert \expect{\vartheta} - \psi(\mu)\right\vert \le
 \abs{\gamma_{k+1}(\lambda')}\abs{\mu-\lambda}^{k+1}$$ proving part (i) of the lemma.

 \noindent
 For $j=0,1, \ldots, k$, let $$P_j = \prod_{l=1}^j (X_{l}-\lambda)$$ (which implies that  $P_0=1$).  Then,  $$\vartheta = \sum_{j=0}^k \gamma_j(\lambda) P_j \enspace . $$

 By the independence of the $X_l$'s, \begin{align*} \variance{P_j}  = \variance{ \prod_{l=1}^j (X_l-\lambda)} &= \prod_{l=1}^j \expect{(X_l-\lambda)^2} - \prod_{l=1}^j (\expectb{X_l- \allowbreak
 \lambda})^2 \\ & = \eta^{2j} - (\mu-\lambda)^{2j} \enspace .
 \end{align*}
 Further for $1\le j < j' \le k$,
 \begin{align*} \covariance{P_j}{P_{j'}} &= \covariance{\prod_{l=1}^j (X_l-\lambda)}{ \prod_{l=1}^{j'}(X_l-\lambda)}\\ & = \expect{ \prod_{l=1}^j (X_l-\lambda) \prod_{l=1}^{j'}(X_l-\lambda)} - (\mu-\lambda)^{j+j'} \\ & = \prod_{l=1}^j \expect{(X_l-\lambda)^2} \prod_{l=j+1}^{j'} \expect{X_l-\lambda} - (\mu-\lambda)^{j+j'} \\
  &= \eta^{2j}(\mu-\lambda)^{j'-j} - (\mu-\lambda)^{j+j'} \enspace .
 \end{align*}
 Thus we have,
 \begin{align} \label{eq:vvtheta:a}
 &\variance{\vartheta}  = \sum_{j=0}^k (\gamma_j(\lambda))^2\variance{P_j} + \sum_{j<j'}2\gamma_j(\lambda) \gamma_{j'} (\lambda) \covariance{P_j}{ P_{j'}} \notag \\
 & =  \sum_{j=0}^k (\gamma_j(\lambda))^2(\eta^{2j} - (\mu-\lambda)^{2j}) + \sum_{0 \le j < j' \le k}2\gamma_j(\lambda) \gamma_{j'} (\lambda)  (\eta^{2j}(\mu-\lambda)^{j'-j} - (\mu-\lambda)^{j+j'} )  \notag \\
 & = \sum_{j=1}^k (\gamma_j(\lambda))^2(\eta^{2j} - (\mu-\lambda)^{2j}) + \sum_{1 \le j < j' \le k} 2\gamma_j(\lambda) \gamma_{j'} (\lambda)  (\eta^{2j}(\mu-\lambda)^{j'-j} - (\mu-\lambda)^{j+j'} ) \notag \\
 & = \sum_{j=1}^k (\gamma_j(\lambda))^2(\eta^{2j} - (\mu-\lambda)^{2j}) + \sum_{1 \le j < j' \le k} 2\gamma_j(\lambda) \gamma_{j'} (\lambda)   \prod_{i \in Q^{vv'}} \eta^{2j}(\mu-\lambda(^{j'-j} \left(1 - \left(\frac{(\mu-\lambda)^2}{\eta^2} \right)^j \right)
 \end{align}
 Let $t_j  = (\mu-\lambda)^{j'-j}\eta^{2j} \bigl(1 - \allowbreak  \bigr(\frac{ (\mu-\lambda)^{2}}{\eta^{2}}\bigr)^{2j} \bigr)$. Since, $\eta^2 = \sigma^2 + (\mu-\lambda)^2$, we have, $\abs{t_j} \le \abs{\mu-\lambda}^{j'-j} \eta^{2j} \le \eta^{j+j'}$. Taking absolute values on both sides of  Eqn.~\eqref{eq:vvtheta:a}, we have,
 \begin{align*}
\variance{\vartheta}  & \le \sum_{j=1}^k \gamma_j^2(\lambda) \eta^{2j} + \sum_{1\le j < j' \le k} 2 \abs{\gamma_j(\lambda)} \abs{\gamma_{j'}(\lambda)}\eta^{j+j'}  \\
  & = \biggl(\sum_{j=1}^k \abs{\gamma_j(\lambda)}\eta^j\biggr)^2 \enspace . ~~~ ~~~
\end{align*}
\end{proof}

\begin{proof} [Proof of Corollary~\ref{cor:fpbias}.] $\lambda \ge \mu(1-\alpha) > 0$ since, $0 \le \alpha < 1$ and $\mu > 0$. Hence, $\psi(t) = t^p$ is analytic in the interval $[\mu, \lambda]$ (or, $[\lambda, \mu]$ depending on whether $\mu < \lambda$ or $\lambda < \mu$).

Let $\vartheta$ abbreviate $\vartheta(\psi(t) = t^p,  \lambda, k, \{X_l\}_{l=1}^k) $. Note that for the function $\psi(t) = t^p$, $\gamma_k(w) = \left . \frac{1}{k!}\left( \frac{d^k}{dt^k} t^p \right)\right\vert_{t = w}  = \binom{p}{k}w^{p-k}$.
Applying
  Lemma~\ref{lem:vvtheta},
  there  exists  $\lambda' \in (\lambda, \mu)$ such that,
\begin{align*} 
&\card{ \expect{\vartheta} - \mu^p} = \card{ \gamma_{k+1}(\lambda')} \abs{\mu-\lambda}^{k+1} =  \Card{\binom{p}{k+1}}\card{ \lambda'^{p-k-1}}
\abs{\mu-\lambda}^{k+1}\notag  \\
& \le
\left( \frac{ p}{k+1} \right)^{\lfloor p \rfloor + 1}  \mu^{p-k-1} \left(1 - \alpha
 \right)^{p-k-1}(\alpha \mu)^{k+1}, ~~\text{ since, $k+1 > p$ and by Fact~\ref{fact:1}}  \notag \\
 & = \left( \frac{ p}{k+1} \right)^{\lfloor p \rfloor + 1}  \left( \frac{ \alpha}{1-\alpha} \right)^{k+1} (1-\alpha)^p \mu^p \enspace .
\end{align*}
In particular, if $p$ is integral, then, $ \binom{p}{k+1} = 0$ and $\expect{\vartheta} = \mu^p$. \hfill
\end{proof}

\begin{proof} [Proof of Corollary ~\ref{cor:fptpvar}.]
For $\psi(t) = t^p$, $\gamma_v(\lambda) = \binom{p}{v} \lambda^{p-v}$. We also have from the assumptions that  $\eta^2 = (\mu-\lambda)^2 + \sigma^2 \le 2(\frac{\lambda}{25p})^2$, or, $\frac{\eta}{\lambda} \le \frac{\sqrt{2}}{25}$.

 By
Lemma~\ref{lem:vvtheta}, part (2),
\begin{align} \label{eq:vvtheta1}
\variance{\vartheta}  \le \left( \sum_{v=1}^k \Card{\binom{p}{v}}\lambda^{p-v}
\eta^v\right)^2  =
\lambda^{2p-2} \eta^2  \left(\sum_{v=1}^{k}
\Card{\binom{p}{v}}\biggl(\frac{\eta}{\lambda}\biggr)^{v-1}\right)^2
 \end{align}
The ratio of the $(v+1)$st term in the  summation in the \emph{RHS}
 to the $v$th term, for $1 \le v \le k-1$, is $$\biggl\lvert \frac{p-v}{v+1}\biggr\rvert \cdot
\frac{\eta}{\lambda} \le \frac{(p-1) \sqrt{2}}{2 (25p)} < \frac{1}{25\sqrt{2}}$$
 Substituting in Eqn.~\eqref{eq:vvtheta1} for $\variance{\vartheta}$ and using $\lambda\le \mu(1+\frac{1}{25p}) \le  e^{1/(25p)} \mu$, we have,
\begin{align*}
\variance{\vartheta} &  \le  \lambda^{2p-2} \eta^2 p^2 \left(\sum_{v=1}^{k} (25 \sqrt{2})^{-(v-1)}\right)^2     \le (1.08)p^2 \mu^{2p-2} \eta^2 \enspace . \hfill ~~~~
\end{align*}

\end{proof}

\section{Proofs for Averaged Taylor Polynomial Estimator}

\begin{proof}[Proof of Corollary~\ref{cor:gv}.]
Choosing $q=8$ and $\epsilon = 3/4$ in Theorem~\ref{thm:gv} gives
a code $Y \subset \{0,1\}^{8k}$ of binary vectors with exactly $k$
1's and minimum distance $3k/{2}$. So, $H_q(\epsilon) =
0.9722648\ldots$  and  hence, by Theorem~\ref{thm:gv}, $\log \abs{Y}
> (1-H_q(\epsilon))k \log 8$ or, $\abs{Y} > 2^{
3(1-H_q(\epsilon))k }> 2^{0.08k}$.
\end{proof}

Recall that $Y \subset \{0,1\}^{s}$ where,  $s=8k$,  is a code such that every $y \in Y$ has exactly $k$ 1's, and the minimum Hamming distance between any pair of codewords in $Y$ is at least $3k/2$. Equivalently, $y$ can be written as an ordered sequence $(y_1, y_2, \ldots, y_k)$ where, $1\le y_1 < y_2 < \ldots < y_k \le s$ are the coordinates of the position of 1's in the $s$-dimensional binary vector $y$.
For example, let  $s = 4$ and $k = 2$---then  the vector $(1,0,1,0)$ is written as the $2$-dimensional  ordered sequence $(1,3)$.
We will say that $u \in y$ if $u$ is one of the $y_i$'s in the ordered
sequence notation. This notation views the sequence $(1,3)$  above as a set $\{1,3\}$.

 Given codewords $y,y' \in Y$, $y \cap y'$ denotes  the set of indices that
 are 1 in  both $y$ and $y'$.
Let $\pi: [k] \rightarrow [k]$ be a permutation and $y =  (y_1, \ldots, y_k)$
be an ordered sequence of size $k$. Then, $\pi(y)$ denotes the sequence $(y_{\pi(1)}, y_{\pi(2)}, \ldots, y_{\pi(k)})$.    The prefix-segment of $\pi(y)$ consisting of its first $v$ entries   is  $ (y_{\pi(1)}, \ldots, y_{\pi(v)})$.
Let $y,y'$ be   ordered sequences    of length $k$ and  let
$\pi, \pi'$  be  permutations mapping $[k] \rightarrow [k]$.
   Let $Q^{vv'}_{yy'\pi\pi'}$ denote the set of  common  indices shared
   among the  first $v$ positions of $\pi(y)$ with the first $v'$ positions
   of  $\pi'(y')$, that is,
 $$ Q^{vv'}_{yy'\pi\pi'} =\{y_{\pi(1)}, y_{\pi(2)}, \ldots, y_{\pi(v)}\} \cap \{y'_{\pi'(1)}, y'_{\pi'(2)}, \ldots, y'_{\pi'(v')} \}  \enspace . $$
 Let  $q^{vv'}_{yy'\pi\pi'} $ denote  the  number of common  indices, that is,
$$
q^{vv'}_{yy'\pi\pi'} =\card{Q^{vv'}_{yy'\pi\pi'}} \enspace . $$
 Given distinct codewords $y,y' \in Y $ and  permutations $\pi$ and $\pi'$, $Q^{vv'}_{yy'\pi\pi'}$ is abbreviated as $Q^{vv'}$ and $q^{vv'}_{yy'\pi\pi'}$ as $q^{vv'}$.

In the remainder of this  section, we will assume that $Y$ is a code of $s=8k$-dimensional boolean vectors of size exponential in $k$, as given by Corollary~\ref{cor:gv}. The  function for the Taylor polynomial estimator will be $\psi(t) = t^p$. Let $\vartheta_y$ abbreviate the estimator $\vartheta_y \equiv \vartheta(\psi(t) = t^p, \lambda, k,s,y, \pi_y, \{X_l\}_{l=1}^s)$, where, $\lambda$ is some parameter.

\subsection{Covariance of $\vartheta_y, \vartheta_{y'}$}

\begin{lemma} \label{lem:bvthetabasic}
Let $q=8$, $k>1$ and   $s = qk$.  Let $Y$ be a code satisfying Corollary~\ref{cor:gv}. Let $\{X_1, \ldots, X_s\}$ be a family of independent  random variables, each having  expectation $\mu >0$
and variance  bounded above by $\sigma^2$.
Let $\lambda$ be an estimate for $\mu$ satisfying $\abs{\lambda-\mu} \le \min(\mu,\lambda)/(25p)$ and let $\sigma < \min(\mu, \lambda)/(25p)$. Let $\eta=
((\lambda-\mu)^2 + \sigma^2)^{1/2} >0$.
 Let $\bvtheta $ denote $   \bvtheta(t^p, \lambda, k,s, Y,  \{\pi_y\}_{y\in Y}, \{X_l\}_{l=1}^s) $ and let $\vartheta_y$ denote the estimator $\vartheta_y = \vartheta(t^p, \lambda, k,s,y, \pi_y, \{X_l\}_{l=1}^s)$. Then,
for $y,y' \in Y$ and $y \ne y'$,
 \begin{align*}\covariance{\vartheta_y}{\vartheta_{y'}}
 &= \begin{dcases}
  \sum_{v,v' = 1}^k \gamma_v(\lambda) \gamma_{v'}(\lambda)
(\mu-\lambda)^{v+v'} \expectsub{\pi_y,\pi_{y'}}{ \left(
\frac{\eta^2}{(\mu-\lambda)^2} \right) ^{q^{vv'}}  -  1 } & \text{ if $\mu \ne \lambda$,} \\
 = \sum_{v=1}^k
\gamma^2_{v}(\lambda) \eta^{2v} \probsub{\pi_y,\pi_{y'}}{q^{vv}_{yy'\pi_y\pi_{y'}} = v}& \text{ if $\mu = \lambda$.}
\end{dcases}
 \end{align*}
 \end{lemma}

\begin{proof} [Proof of Lemma~\ref{lem:bvthetabasic}] By definition,
$
\bvtheta = \frac{1}{\abs{Y}} \sum_{y\in Y} \vartheta_y$.
\eat{
and therefore,
\begin{align*}
\expect{ \bvtheta^2}  & = \frac{1}{\abs{Y}^2} \sum_{y \in Y
}\expect{\vartheta_y^2} \\ & + \frac{1}{\abs{Y}^2}\sum_{y,y' \in Y,
y\ne y'} \expect{ \vartheta_y\vartheta_{y'}}  \enspace .
\end{align*}
Since for $y,y' \in Y$,  we have $\expect{\vartheta_y} = \expect{\vartheta_{y'}} = \expect{\bvtheta}$, hence,
\begin{align*}
\left(\expect{\bvtheta}\right)^2 =  \frac{1}{\abs{Y}^2} \left(\sum_{y \in Y} \expect{ \vartheta_y}\right)^2 = \frac{1}{\abs{Y}^2}  \sum_{y \in Y} (\expect{\vartheta_y})^2 + \frac{1}{\abs{Y}^2}  \sum_{y,y' \in Y, y\ne y'} \expect{\vartheta_y}\expect{\vartheta_{y'}} \enspace .
\end{align*}
Subtracting, we obtain
\begin{align} \label{eq:avtp0}
\variance{\bvtheta} & = \expect{\bvtheta^2} - (\expect{\bvtheta})^2 \notag \\
 &= \frac{1}{\abs{Y}^2} \sum_{y \in Y} \variance{\vartheta_y} + \frac{1}{ \abs{Y}^2} \sum_{y,y' \in Y, y \ne y'}\left( \expect{ \vartheta_y\vartheta_{y'}} - \expect{\vartheta_y} \expect{\vartheta_{y'}}\right)
\end{align}
This proves part (i) of the Lemma.
}
Fix $y,y' \in Y$, with $y \ne y'$  and let  $\pi=\pi_y$ and $\pi'= \pi_{y'}$ abbreviate  the random
permutations corresponding to $y$ and $y'$. Let $q_{yy' \pi_y \pi_{y'}}^{vv'}  $ be denoted by
$q^{vv'}$. Now,
$$\expect{\vartheta_y} \expect{\vartheta_{y'}}
= \left(\sum_{v=0}^k \gamma_v(\lambda)(\mu- \lambda)^v\right)^2  = \sum_{v=0}^k \sum_{v'=0}^k \gamma_v(\lambda)\gamma_{v'}(\lambda) (\mu-\lambda)^{v+v'} \enspace . $$
Further, from the definition of $\vartheta_y$ and $\vartheta_{y'}$, and by linearity of expectation,
\begin{align*}
\expect{\vartheta_y \vartheta_{y'}} & = \expect{ \left(\sum_{v=0}^k \gamma_v(\lambda) \prod_{l=1}^v (X_{y_{\pi(l)}} - \lambda) \right)
\left( \sum_{v'=0}^k \gamma_{v'}(\lambda) \prod_{m=1}^{v'} (X_{y'_{\pi'(m)}} - \lambda) \right)} \\
& = \sum_{v,v'=0}^k \gamma_v(\lambda) \gamma_{v'}(\lambda)
\expect{ \prod_{l=1}^v (X_{y_{\pi(l)}} - \lambda) \prod_{m=1}^{v'} (X_{y'_{\pi'(m)}} -  \lambda)}
\end{align*}
Fix  $\pi,\pi'$. There are $q^{vv'} = q^{vv'}_{yy'\pi_y\pi_{y'}}$ indices
that are common among the first $v$ positions of $\pi_y(y)$ and the first $v'$ positions of $\pi_{y'}(y')$. This set of common  indices is given by $Q^{vv'} =  Q^{vv'}_{yy'\pi_y\pi_{y'}} = \{y_{\pi(1)}, \ldots, y_{\pi(v)}\} \cap \{y'_{\pi'(1)}, \ldots, y'_{\pi'(v)}\}$. Also, let $U^{vv'} = U^{vv'}_{yy'\pi_y\pi_{y'}}$ denote the union $\{y_{\pi(1)}, \ldots, y_{\pi(v)}\} \cup \{y'_{\pi'(1)}, \ldots, y'_{\pi'(v)}\}$.
Hence we have,
\begin{align*}\prod_{l=1}^v(X_{y_{\pi(l)}} - \lambda)
\prod_{m=1}^{v'} (X_{y'_{\pi'(m)}}- \lambda) = \prod_{i \in Q^{vv'}} (X_i - \lambda)^2  \prod_{i \in U^{vv'}\setminus Q^{vv'}} (X_i - \lambda) \enspace .
\end{align*}
Taking expectation,
\begin{align*}
&\expect{\prod_{l=1}^v(X_{y_{\pi(l)}} - \lambda)
\prod_{m=1}^{v'} (X_{y'_{\pi'(m)}}- \lambda) } \\
& = \expectsub{\pi_y, \pi_{y'}}{\expectsub{X_1, \ldots, X_s}{\prod_{l=1}^v(X_{y_{\pi(l)}} - \lambda)
\prod_{m=1}^{v'} (X_{y'_{\pi'(m)}}- \lambda) \mid \pi_y, \pi_{y'} }}\\
 &= \expectsub{\pi_y, \pi_{y'}}{\expectsub{X_1, \ldots, X_s}{\prod_{i \in Q^{vv'}} (X_i - \lambda)^2  \prod_{i \in U^{vv'} \setminus Q^{vv'}} (X_i - \lambda) \bigr\rvert \pi_y, \pi_{y'} }} \\
& = \expectsub{\pi_y, \pi_{y'}}{ \prod_{i \in Q^{vv'}}\expectb{(X_i - \lambda)^2} \prod_{i \in U^{vv'} \setminus Q^{vv'}} \expectb{X_i}{(X_i - \lambda)} \bigr\rvert \pi_y, \pi_{y'} } \notag \\
& = \expectsub{\pi_y,\pi_{y'}}{\eta^{2q^{vv'}}(\mu-\lambda)^{v+v'-2q^{vv'}}},
\end{align*}
by independence of the $X_i$'s for $i \in [s]$.

Therefore,
\begin{align} \label{eq:vbvt:10}
& \covariance{\vartheta_y}{\vartheta_{y'}} \notag \\
 &= \expect{\vartheta_y \vartheta_{y'}}  - \expectb{\vartheta_y}\expectb{\vartheta_{y'}} \notag \\
 &= \sum_{v,v' = 0}^k \gamma_v(\lambda)
\gamma_{v'}(\lambda)\left(\expect{\prod_{l=1}^v(X_{y_{\pi(l)}} - \lambda)
\prod_{m=1}^{v'} (X_{y'_{\pi'(m)}}- \lambda)} -    (\mu-\lambda)^{v+v'} \right) \notag \\
&=\sum_{v,v' = 0}^k \gamma_v(\lambda)
\gamma_{v'}(\lambda)  \left(\expectsub{\pi_y,\pi_{y'}}{\eta^{2q^{vv'}}
(\mu-\lambda)^{v+v'-2q^{vv'}}} -    (\mu-\lambda)^{v+v'} \right)\notag \\
&=\sum_{v,v' = 1}^k \gamma_v(\lambda)
\gamma_{v'}(\lambda)  \left(\expectsub{\pi_y,\pi_{y'}}{\eta^{2q^{vv'}}
(\mu-\lambda)^{v+v'-2q^{vv'}}} -    (\mu-\lambda)^{v+v'} \right)
\end{align}
where the last step follows by noting that
if  $v = 0$ or  $v' = 0$, then  $q^{vv'} = 0$ and so, $\eta^{2q^{vv'}}
(\mu-\lambda)^{v+v'-2q^{vv'}} = (\mu-\lambda)^{v+v'}$. Hence
the summation indices $v,v'$ in ~\eqref{eq:vbvt:10} may start from
1 instead of 0.

\emph{Case 1: $\mu=\lambda$.} If $v \ne v'$, then, $2q^{vv'} \le 2\min(v,v') < v + v'$, Hence, the term $(\mu-\lambda)^{v+v'-2q^{vv'}} = 0$. In this case, Eqn.~\eqref{eq:vbvt:10} becomes
\begin{align} \label{eq:vbvt:11}
 \expect{\vartheta_y \vartheta_{y'}}  -
 \expectb{\vartheta_y}\expectb{\vartheta_{y'}} & = \sum_{v=1}^{\abs{y \cap y'}}
\gamma^2_{v}(\lambda) \eta^{2v} \probsub{\pi_y,\pi'_{y'}}{q^{vv}_{yy'\pi\pi'} = v}
\end{align}

\emph{Case 2: $\mu \ne \lambda$.} Then, Eqn.~\eqref{eq:vbvt:10} can be written as

\begin{align} \label{eq:vbvt:1a}
 &\expect{\vartheta_y \vartheta_{y'}}   -\expectb{\vartheta_y}\expectb{\vartheta_{y'}} \notag \\
& = \sum_{v,v' = 1}^k \gamma_v(\lambda) \gamma_{v'}(\lambda)
(\mu-\lambda)^{v+v'} \left(\expectsubbb{\pi_y,\pi_{y'}}{ \biggl(
\frac{\eta^2}{(\mu-\lambda)^2} \biggr) ^{q^{vv'}}}  -  1 \right) \enspace
.  
\end{align}
This proves the Lemma.
\end{proof}

Let $Y$ be a code satisfying the properties of Corollary~\ref{cor:gv} and let $y,y' \in Y$ and distinct such that $t = \abs{y \cap y'}$. Let $\pi_y, \pi_{y'}$ denote randomly  and independently chosen permutations from $[k] \rightarrow [k]$.  Define
\begin{align}
 P_{yy'}& = \lambda^{2p} \sum_{v,v'=1}^k \binom{p}{v}  \binom{p}{v'} \left(\frac{\mu-\lambda}{\lambda}\right)^{v+v'} \sum_{r=1}^{t} \left( \frac{\eta^2}{(\mu-\lambda)^2} \right)^r \probsub{\pi_y,\pi_{y'}}{q^{vv'} = r} \label{eq:Srta}\\
 Q_{yy'}& = \lambda^{2p}\sum_{\substack{1\le v,v' \le k}}\binom{p}{v}\binom{p}{v'} \left(\frac{\mu-\lambda}{\lambda}\right)^{v+v'} \left(\probsub{\pi_y,\pi_{y'}}{q^{vv'} = 0}-1\right)  \label{eq:Srtb} \end{align}
 \begin{corollary} \label{lem:vbvcross} Assume the premises and notation of Lemma~\ref{lem:bvthetabasic} and let $\mu \ne \lambda$. For $y,y \in Y$ and $y \ne y'$ such that $t = \abs{y \cap y'}$,  let $\pi_y, \pi_{y'}$ denote randomly  and independently chosen permutations from $[k] \rightarrow [k]$. Then,
$$\covariance{\vartheta_y}{\vartheta_{y'}}   \le  P_{yy'}+ Q_{yy'} \enspace . $$
\end{corollary}

\begin{proof}
Since $\psi(x) = x^p$,
$\gamma_v(\lambda) = \binom{p}{v} \lambda^{p-v}$.  The Corollary follows by substituting this  into  Lemma~\ref{lem:bvthetabasic}.
\end{proof}

\subsection{Probability of overlap of prefixes of  $y$ and $y'$ after random ordering}

\begin{lemma} \label{lem:bvapower} Let  $Y$ be a code satisfying the properties of Corollary~\ref{cor:gv}. Let $\{\pi_y\}_{y \in Y}$ be a family of random and independently chosen permutations from $[k] \rightarrow [k]$. For distinct $y,y' \in Y$,
\begin{align}\label{eq:probqvvprime}
 \probsubb{\pi_y,\pi_{y'}} {q^{vv'} = r}
 = \frac{1}{\binom{k}{v}\binom{k}{v'}} \sum_{s=0}^{t-r} \binom{t}{r} \binom{t-r}{s} \binom{k-t}{v-(r+s)} \binom{k-(r+s)}{v'-r}
 \end{align}
\end{lemma}

\begin{proof}
Fix $y,y' \in Y$ and distinct and let $t = t(y,y') = \abs{ y \cap y'}$.
By notation, $\pi_y(y)[v]$ is  the $v$-sequence $\tau = (y_{\pi_y(1)}, \ldots, y_{\pi_y(v)})$ and  $\pi_{y'}(y')[v]$ is  the $v'$-sequence $\nu= (y'_{\pi_y^{\prime}(1)}, \ldots, y'_{\pi_y^{\prime}(v)})$. The permutations $\pi_y$ and $\pi_{y'}$ are each uniformly randomly and independently chosen from the space of all permutations $[k] \rightarrow [k]$ (i.e., $S_k$).

 The problem is to count the number of ways in which the  $v$ positions in $\tau$ and the $v'$ positions in $\nu$ can be filled, using the elements of $y$ and $y'$ under permutations $\pi_y$ and $\pi_{y'}$ such that $\tau \cap \nu$ has exactly $r$ elements. Since $\pi_y$ and $\pi_{y'}$ are uniformly  random and independent permutations, the sample space has size $\powf{k}{v} \cdot \powf{k}{v'} = \binom{k}{v} \binom{k}{v'} v! v'!$. There are $t$ elements in common among $y$ and $y'$ and we wish for $\tau$ and $\nu$ to have $r$ elements in common. Suppose $\tau$ has $r+s$ elements from the $t$ elements in common, where, $s$ ranges from $0$ to $\max(t-r,v-r)$. These are selected in $\binom{t}{r+s}$ ways. Having chosen these elements, we select $r$ elements in $\binom{r+s}{r}$ ways--these elements are included in $\nu$ as well. We have now filled $r+s$ positions of $\tau$ and $r$ positions of $s$. The remaining $v-(r+s)$ positions may be filled out of the $k-t$ elements of $y$ that are not common with $y'$. This is done in $\binom{k-t}{v-(r+s)}$ ways. There are $v'-r$ positions remaining to be filled in $\nu$. There are $k-t + (t-(r+s))$ elements to choose from, which can be done in $ \binom{k-(r+s)}{v'-r}$ ways. The $v$ elements chosen for $\tau$ and the $v'$ elements chosen for $\nu$ can be rearranged in $v!$ and $v'!$ ways. Thus,
 \begin{align}\label{eq:probqvvprime}
 \probsubb{\pi_y,\pi_{y'}} {q^{vv'} = r} & =
 \frac{v! v'!}{\binom{k}{v}\binom{k}{v'} v! v'!} \sum_{s=0}^{t-r} \binom{t}{r+s} \binom{r+s}{r} \binom{k-t}{v-(r+s)} \binom{k-(r+s)}{v'-r} \notag \\
 & = \frac{1}{\binom{k}{v}\binom{k}{v'}} \sum_{s=0}^{t-r} \binom{t}{r} \binom{t-r}{s} \binom{k-t}{v-(r+s)} \binom{k-(r+s)}{v'-r}
 \end{align}
 which proves the lemma. \hfill

 \eat{
We approximate Eqn.~\eqref{eq:probqvvprime} as follows.
\begin{align} \label{eq:approxprob}
 \probsub{\pi_y,\pi_{y'}}{q^{vv'} = r} \eat{ & = \binom{t}{r} \frac{ v_{(r)}}{k_{(r)}} \cdot \frac{ (k-t)_{(v-r)}}{(k-r)_{(v-r)}}  \cdot \frac{ v'_{(r)}}{k_{(r)}} \cdot \frac{ (k-t)_{(v'-r)}}{(k-r)_{(v'-r)}} \notag  \\}
 & \le \binom{t}{r} \left( \frac{v}{k} \right)^r \left( 1- \frac{t-r}{k-r} \right)^{v-r} \left( \frac{v'}{k} \right)^r \left( 1 - \frac{t-r}{k-r} \right)^{v'-r}
 \end{align}
for $\max(0,t-k + \max(v,v')) \le r \le \min(v,v',t)$.

}
\end{proof}

\subsection{Estimating $Q_{yy'}$}

\begin{lemma} \label{lem:Q} Assume the  premises and notation of Lemma~\ref{lem:bvthetabasic} and Corollary~\ref{lem:vbvcross}.  Let $p \ge  2$ and let $y,y' \in Y$  and distinct. If $\mu \ne\lambda$, then   $Q_{yy'} <0$.
\end{lemma}

\begin{proof} Fix $y,y' \in Y$ and distinct and let $Q$ denote  $Q_{yy'}$. Let $\alpha = \frac{\mu-\lambda}{\lambda} \le \frac{1}{25p}$.  Then, $$Q = -Q_1 + Q_2$$ where,
\begin{align} Q_1  & =  \sum_{1 \le v,v' \le k }\binom{p}{v} \binom{p}{v'} \lambda^{2p-v-v'} (\mu-\lambda)^{v+v'}  = \lambda^{2p} \left( \sum_{v=1}^k \binom{p}{v} \alpha^v\right)^2
\label{eq:Q1}  \\
  Q_2 & = \sum_{\substack{1\le v,v' \le k}}\binom{p}{v}\binom{p}{v'} \lambda^{2p-v-v'} (\mu-\lambda)^{v+v'} \probsub{\pi,\pi'}{q^{vv'} = 0} \enspace . \label{eq:Q2def}
\end{align}

Consider $\sum_{v=1}^k \binom{p}{v} \alpha^v$. The  absolute value of the ratio of the $v+1$st term to the $v$th term, for $v=1,2, \ldots, k-1$, is
$$\frac{\abs{p-v}}{v+1} \cdot \alpha \le \left(\frac{p}{2}\right) \cdot \frac{1}{25p} \le \frac{1}{50} \enspace . $$
Therefore, $$ \Card{\sum_{v=1}^k \binom{p}{v} \alpha^v - p\alpha} \le  (p\alpha)\sum_{v\ge 1} (50)^{-(v-1)} = \frac{ (p\alpha)}{49} \enspace. $$
Therefore,
\begin{equation} \label{eq:Q1a}
Q_1  = \lambda^{2p}p\alpha\left( 1 \pm \frac{1}{49} \right)^2  \in \lambda^{2p} p \alpha \left( 1 \pm \frac{1}{24} \right)
\end{equation}

Consider $Q_2$. Let $t =t(y,y') = \abs{y \cap y'}$.
\begin{align} \label{eq:Q2}
Q_2 & = \lambda^{2p}\sum_{v=1}^k \sum_{v'=1}^k \binom{p}{v}\binom{p}{v'} \alpha^{v+v'}\sum_{u=0}^{t}
\frac{\binom{t}{u} \binom{k-t}{v-u} \binom{k-u}{v'}}{\binom{k}{v}\binom{k}{v'}} \notag \\
& = \lambda^{2p} \sum_{u=0}^t \binom{t}{u} \sum_{v=1}^k \binom{p}{v} \frac{\binom{k-t}{v-u}}{\binom{k}{v}}\alpha^v \sum_{v'=1}^k \binom{p}{v'} \frac{\binom{k-u}{v'}}{\binom{k}{v'}} \alpha^{v'} \notag \\
& = \lambda^{2p} \sum_{u=0}^t \binom{t}{u}R_{ut}S_{ut}
\end{align}
where,
\begin{align*} R_{ut} & = \sum_{v=1}^k \binom{p}{v} \frac{\binom{k-t}{v-u}}{\binom{k}{v}}\alpha^v  = \sum_{v=\max(u,1)}^{k-t+u} \frac{\binom{p}{v} \binom{k-t}{v-u}}{\binom{k}{v}} \alpha^v,
~~~\text{ and } \\
S_{ut}  &  = \sum_{v'=1}^k \binom{p}{v'}\frac{ \binom{k-u}{v'}}{\binom{k}{v'}} \alpha^{v'} \enspace .
\end{align*}
Consider $R_{ut}$. The  absolute value of the ratio of the $(v+1)^{\text{st}}$  term in the summation $R_{ut}$ to the $v$th term   for $\max(u,1) \le v \le k-t+u-1$ is
\begin{align*}
\frac{ \abs{p-v}}{v+1} \cdot \left(\frac{ k-t-v+u}{v-u+1}\right) \cdot \left( \frac{ v+1}{k-v} \right)
\alpha \le \left( \frac{p}{2} \right) \left( \frac{ k-v -(t-u)}{k-v} \right)\cdot  \frac{1}{25p} \le  \frac{1}{50} \enspace .
\end{align*}
\emph{Case 1: $u\le 1$. } Then,
\begin{align*}
 R_{ut} \in  \frac{p\alpha}{k}\left( 1 \pm \frac{1}{49} \right) \enspace .
 \end{align*}
 \emph{Case 2: $ u \ge 2$.} Then,
 \begin{align*}
 R_{ut} \in \frac{ \binom{p}{u} \alpha^u }{\binom{k}{u}} \left( 1 \pm \frac{1}{49}\right) \enspace .
\end{align*}
In either case,
\begin{align} \label{eq:Rut}
R_{ut} \in \frac{ \binom{p}{\max(u,1)} \alpha^{\max(u,1)}}{\binom{k}{\max(u,1)}} \left( 1 \pm \frac{1}{49}\right) \enspace .
\end{align}

Now consider $S_{ut} = \sum_{v=1}^k \binom{p}{v}\frac{ \binom{k-u}{v}}{\binom{k}{v}} \alpha^{v} $. The  absolute value of the ratio of the $v+1$th term in the summation $S_{ut}$ to the $v$th term, $v=1,2, \ldots, k-u-1$ is
\begin{align*}
\frac{ \abs{p-v}}{v+1} \left( \frac{k-u-v}{v+1}\right) \left( \frac{v+1}{k-v} \right) \alpha \le \frac{p\alpha}{2} \le \frac{1}{50} \enspace .
\end{align*}
Therefore,
\begin{align} \label{eq:Rprime}
 S_{ut} \in \left( \frac{p (k-u) \alpha}{k} \right) \left(1 \pm \frac{1}{49} \right)
\end{align}
Substituting Eqns.~\eqref{eq:Rut} and ~\eqref{eq:Rprime} in Eqn.~\eqref{eq:Q2}, we have,
\begin{align} \label{eq:Q2b1}
Q_2 &= \lambda^{2p}\sum_{u=0}^t  \binom{t}{u} R_{ut} S_{ut}\notag \\
 & \in \left(1 \pm  \frac{1}{49}\right)\left(1\pm \frac{1}{49} \right)  \lambda^{2p} \sum_{u=0}^t \frac{\binom{t}{u} \binom{p}{{\max(u,1)}}  \alpha^{\max(u,1)} (p\alpha) (k-u)}{\binom{k}{{\max(u,1)}} k}
\end{align}
Consider the summation term in Eqn. ~\eqref{eq:Q2b1}.
\begin{align}
& \sum_{u=0}^t \frac{\binom{t}{u} \binom{p}{{\max(u,1)}}  \alpha^{\max(u,1)} (k-u)}{\binom{k}{{\max(u,1)}} }  = p\alpha + \sum_{u=1}^t \frac{ \binom{t}{u} \binom{p}{u} \alpha^u (k-u)}{\binom{k}{u} } \label{eq:Q2c1}
\end{align}
Consider the summation term in Eqn.~\eqref{eq:Q2c1}.
The ratio of the absolute value of the $u+1$st term to the $u$th term, for $ 1 \le u \le t-1$ is
\begin{align*}
\left( \frac{t-u}{u+1}\right) \left(\frac{\abs{p-u}}{u+1}\right)\left( \frac{ u+1}{k-u} \right) \left(\frac{k-u-1}{k-u} \right) \alpha \le \left( \frac{ (t-u)}{k-u} \right) \left( \frac{p}{2} \right)(1)(\alpha)\le \frac{1}{200}
\end{align*}
since, $t \le k/4$ from the property  of the code $Y$.

Therefore, from Eqn.~\eqref{eq:Q2c1},
\begin{align*}
\sum_{u=1}^t \frac{ \binom{t}{u} \binom{p}{u} \alpha^u (k-u)}{\binom{k}{u} }  & \in \frac{tp\alpha(k-1)}{k^2} \left ( 1 \pm \frac{1}{199} \right) \in \frac{p\alpha}{4} \left( 1 \pm \frac{1}{199} \right)
\end{align*}
since $ t \le k/4$.

Substituting in Eqn.~\eqref{eq:Q2b1}, we have,
\begin{align} \label{eq:Q2d1}
Q_2 &  \in \left(1 \pm  \frac{1}{49}\right)^2 \left( \frac{p\alpha}{k} \right) \lambda^{2p} \left( p\alpha + \frac{p\alpha }{4}\left(1 \pm \frac{1}{199}\right) \right) \le \left( \frac{(1.31) (p\alpha)^2}{k} \right) \lambda^{2p}
\end{align}

Using Eqns. ~\eqref{eq:Q1a} and ~\eqref{eq:Q2d1}, we have,
\begin{align*}
Q_1 - Q_2 & \ge  \lambda^{2p} (p\alpha)^2\left (1 - \frac{1}{24} \right)  -
\lambda^{2p}  \frac{ (1.31) (p\alpha)^2}{k} \\ &   > 0
\end{align*}
since, $k  \ge 3$.  Hence, $Q = -Q_1 + Q_2 < 0$.
\end{proof}

\subsection{Estimating $P_{yy'}$.}

\textbf{Notation.} Let $Y$ be a code satisfying Corollary ~\ref{cor:gv}. Let $y,y' \in Y$ and distinct and let $t = \abs{y \cap y'}$. Let $P$ denote $P_{yy'}$.
Let $\alpha = \frac{\mu-\lambda}{\lambda}$ and $ \beta =  \frac{\eta^2}{\lambda^2}$. Define
\begin{align*}
P_1 & =\lambda^{2p}  \sum_{u=1}^{t} \binom{t}{u} \sum_{r=1}^u  \binom{u}{r} \beta^{r} \left(\frac{ \card{\powf{p}{u}}\card{\powf{p}{r}}\abs{\alpha}^{u-r}}{\powf{k}{u}\powf{k}{r}}\right) \\
& \hspace*{0.4in} \cdot
\left((1-\abs{\alpha})^{p-u+p-r} + 2 (1-\abs{\alpha})^{p-u} (27)^{-(3/4)k} + (27)^{-(1.5k)} \right)\1_{r > p, p \text{non-integral}} \\
P_2 & = \lambda^{2p}  \sum_{u=1}^{t} \binom{t}{u} \sum_{r=1}^u  \binom{u}{r} \beta^{r} \left(\frac{ \card{\powf{p}{u}}\card{\powf{p}{r}}\abs{\alpha}^{u-r}}{\powf{k}{u}\powf{k}{r}}\right)
\cdot  \left(( 1- \abs{\alpha})^{p-u} +  (27)^{-(3/4)k}\right) \left( \frac{50}{49}\right)\1_{r\le p < u, p \text{ non-integral}} \\
P_3 & = \lambda^{2p}  \sum_{u=1}^{t} \binom{t}{u} \sum_{r=1}^u  \binom{u}{r} \beta^{r} \left(\frac{\card{\powf{p}{u}}\card{\powf{p}{r}}\abs{\alpha}^{u-r}}{\powf{k}{u}\powf{k}{r}}\right)\cdot  \left( \frac{50}{49} \right)^2 \1_{u \le p}
\end{align*}

\begin{lemma} \label{lem:P}
Assume the premises and notation of  Lemma~\ref{lem:bvthetabasic} and Corollary~\ref{lem:vbvcross}. Let  $y,y' \in Y$ and distinct and let $\pi = \pi_y$ and $ \pi' = \pi_{y'}$ be  random permutations from $[k] \rightarrow [k]$. Let $\alpha = \frac{\mu-\lambda}{\lambda}$ and $ \beta =  \frac{\eta^2}{\lambda^2} $. Then, $P_{yy'} =0$ if $p$ is integral, and otherwise,
$P_{yy'} \le P_1 + P_2+P_3$.
\end{lemma}
\begin{proof}
Let $P$ denote $P_{yy'}$. Then,
\begin{align}  \label{eq:P1b}
P   & = \lambda^{2p} \sum_{v,v'=1}^k \binom{p}{v} \binom{p}{v'} \alpha^{v+v'-2r}  \sum_{r=1}^{t} \beta^r \probsub{\pi,\pi'}{q^{vv'} = r} \notag \notag \\
 & = \lambda^{2p} \sum_{v,v'=1}^k \binom{p}{v} \binom{p}{v'} \sum_{r=1}^{t} \alpha^{v+v'-2r} \beta^r \cdot \cfrac{1}{\binom{k}{v} \binom{k}{v'}} \sum_{u=r}^t \binom{t}{u}\binom{u}{r} \binom{k-t}{v-u} \binom{k-u}{v'-r} \notag \notag \\
 & = \lambda^{2p}  \sum_{u=1}^{t} \binom{t}{u}  \sum_{r=1}^u  \binom{u}{r} \beta^{r}
 \left(\sum_{v=u}^{k} \binom{p}{v} \frac{\binom{k-t}{v-u}}{\binom{k}{v}} \cdot  \alpha^{v-r} \right) \left(\sum_{v'=r}^k  \binom{p}{v'} \frac{\binom{k-u}{v'-r}}{\binom{k}{v'}} \cdot \alpha^{v'-r}\right) \notag \notag \\
 & = \lambda^{2p}  \sum_{u=1}^{t} \binom{t}{u}  \sum_{r=1}^u  \binom{u}{r} \beta^{r} U_{ur} V_{ur}
\end{align}
where,
\begin{align}
U_{ur} & = \sum_{v=u}^{k}  \frac{\binom{p}{v}\binom{k-t}{v-u} \alpha^{v-r}}{\binom{k}{v}}  ~~ \text{ and }~~
V_{ur} = \sum_{v'=r}^k  \frac{\binom{p}{v'} \binom{k-u}{v'-r} \alpha^{v'-r}}{\binom{k}{v'}} \enspace . \label{eq:UVur}
\end{align}

We first obtain upper bounds on $U_{ur}$ and $V_{ur}$.
\begin{align} \label{eq:U1}
U_{ur} & =\sum_{v=u}^{k}  \frac{\binom{p}{v}\binom{k-t}{v-u} \alpha^{v-r}}{\binom{k}{v}}  \notag \\
 &= \sum_{v=u}^k \frac{ \frac{\powf{p}{u}}{\powf{v}{u}} \binom{p-u}{v-u}  \binom{k-t}{v-u} \alpha^{v-u+(u-r)}}{\frac{\powf{k}{u} }{\powf{v}{u}} \binom{k-u}{v-u}} \notag \\
  &= \frac{ \powf{p}{u} \alpha^{u-r} }{\powf{k}{u}}  \sum_{w=0}^{k-u} \binom{p-u}{w} \frac{\binom{k-t}{w}}{ \binom{k-u}{w}}\alpha^{w} \enspace .
\end{align}
by letting $w = v-u$.
\eat{Similarly, $$V_{ur} = \frac{ \powf{p}{r}  }{\powf{k}{r}}  \sum_{w=u-r}^{k-r} \binom{p-r}{w} \alpha^{w} \frac{\binom{k-u}{w}}{ \binom{k-r}{w}} \enspace . $$
}

\emph{Case U.1: $u > p$.} Note that if $p$ is integral then $U_{ur}=0$. Otherwise, $\sgn(\binom{p-u}{w}) =(-1)^w$. Using this and since $0 \le t \le u$, we have,
\begin{align} \label{eq:U:case1}
\left \lvert \sum_{w=0}^{k-u} \binom{p-u}{w} \alpha^{w} \frac{\binom{k-t}{w}}{ \binom{k-u}{w}}\right\rvert  & \le   \sum_{w=0}^{k-u} \binom{p-u}{w} (-1)^w \abs{\alpha}^w   = (1-\abs{\alpha})^{p-u} + \binom{p-u}{k-u+1} \gamma^{k-u+1}
\end{align}
for some $\gamma \in (-\abs{\alpha}, 0)$, by Taylor's series expansion of $(1-\abs{\alpha})^{p-u}$ around 0 up to $k-u$ terms.

Now, for $u > p$, $1 \le u \le t \le k/4$, we have,
\begin{align} \label{eq:U1a}
\left \lvert \binom{p-u}{k-u+1}\gamma^{k-u+1}\right\rvert &  \le  \binom{k-p}{k-u+1} \abs{\alpha}^{k-u+1} \le  \left( \frac{ (k-p)e \abs{\alpha}}{k-u+1} \right)^{k-u+1} \le (27)^{-(3/4)k} \enspace .
\end{align}
since, $1 \le u \le t \le k/4$ and $ \abs{\alpha} \le \frac{1}{25p} \le \frac{1}{50}$.

\emph{Case U.2: $u \le p$.} Consider
$\sum_{w=0}^{k-u} \binom{p-u}{w} \alpha^{w} \frac{\binom{k-t}{w}}{\binom{k-u}{w}}$. Let the $w$th term in the summation be $\tau_w$, for $0 \le w \le k-u-1$. Then, for $1 \le w \le k-u-1$,  $$\left \lvert \frac{ \tau_{w+1}}{\tau_w}\right\rvert= \left(\frac{ \abs{p-u-w}}{w+1}\right) \cdot \abs{\alpha}  \cdot  \left(\frac{k-t-w}{k-u-w}\right) \le \frac{1}{50} $$
since, (a) $1 \le u \le t$ and $k-u-w \ge 1$, and, (b) $\frac{\abs{(p-u)-w}}{w+1} \le \frac{p}{2}$.

Therefore,
\begin{align*}
\left \lvert \sum_{w=0}^{k-u} \binom{p-u}{w} \alpha^{w} \frac{\binom{k-t}{w}}{ \binom{k-u}{w}} - 1\right\rvert  \le \sum_{w \ge 1} (50)^{-w} = \frac{1}{49} \enspace .
\end{align*}

Combining Cases U.1 and U.2, we have,
\begin{align} \label{eq:Ub}
\abs{U_{ur}} & \le \left(\frac{\abs{\powf{p}{u}} \abs{\alpha}^{u-r}}{\powf{k}{u}}\right) \left[\left((1-\abs{\alpha})^{p-u} + (27)^{-(3/4)k}\right)\1_{u > p, p \text{ non-integral}} + \frac{50}{49}  \1_{u \le p} \right] \enspace .
\end{align}

\emph{Case $V$:} Proceeding similarly for evaluating $V_{ur}$, we have,

\begin{align*}
V_{ur} &= \sum_{v=r}^k  \frac{\binom{p}{v} \binom{k-u}{v-r} \alpha^{v-r}}{\binom{k}{v}}\\
& = \sum_{v=r}^k \cfrac{ \frac{\powf{p}{r}}{\powf{v}{r}} \binom{p-r}{v-r} \binom{k-u}{v-r} \alpha^{v-r}}{\frac{\powf{k}{r}}{\powf{v}{r}} \binom{k-r}{v-r}} \\
& = \frac{\powf{p}{r}}{\powf{k}{r}} \sum_{w=0}^{k-r} \cfrac{ \binom{p-r}{w} \binom{k-u}{w} \alpha^w }{\binom{k-r}{w}} \enspace .
\end{align*}

\emph{Case V.1: $r > p$}. We note that if $p$ is integral then $\powf{p}{r}=0$ and therefore  $V_{ur}=0$. Otherwise, $\sgn(\binom{p-r}{w}) = (-1)^w$. Thus,

\begin{align*}
 \left\lvert\sum_{w=0}^{k-r} \cfrac{ \binom{p-r}{w} \binom{k-u}{w} \alpha^w }{\binom{k-r}{w}} \right\rvert & \le
 \sum_{w=0}^{k-r} \cfrac{ \card{ \binom{p-r}{w}}  \abs{\alpha}^w \binom{k-u}{w}}{\binom{k-r}{w}} \\
& \le \sum_{w=0}^{k-r}  \card{ \binom{p-r}{w}}  \abs{\alpha}^w , \text{ since, $k \ge u \ge r \ge 1$,} \\
& = \sum_{w=0}^{k-r} \binom{p-r}{w} (-\abs{\alpha})^w, \text{for some $\gamma \in (-\abs{\alpha}, 0)$,}  \\
& = (1-\abs{\alpha})^{p-r} + \binom{p-r}{k-r+1} \gamma^{k-r+1} \\
& \le  (1-\abs{\alpha})^{p-r} + (27)^{-(3/4)k}
\end{align*}
following the same argument as in Eqn.~\eqref{eq:U1a},  and using $1 \le r \le t \le k/4$.
 Thus,
\begin{align*}
\card{V_{ur}}  \le \frac{ \left\lvert\powf{p}{r} \right\rvert}{\powf{k}{r}} \left((1-\abs{\alpha})^{p-r}+ (27)^{-(3/4)k}\right)
\end{align*}

\emph{Case V.2: $r \le   p$.} Consider the ratio of the absolute value of the $w+1$st term, denoted $\nu_{w+1}$ to the $w$th term $\nu_w$ of the summation $\sum_{w=0}^{k-r} \binom{p-r}{w} \alpha^w \frac{ \binom{k-u}{w}}{\binom{k-r}{w}}$. Then,
\begin{align*}
\left \lvert \frac{\nu_{w+1}}{\nu_w} \right \rvert = \left( \frac{ \abs{p-r-w}}{w+1} \right) \alpha \left( \frac{ k-u-w}{k-r-w} \right) \le \left( \frac{p}{2} \right) \alpha \le \frac{1}{50} \enspace .
\end{align*}
Therefore,
\begin{align*}
\sum_{w=0}^{k-r} \binom{p-r}{w} \alpha^w \frac{ \binom{k-u}{w}}{\binom{k-r}{w}} \in \left(1 \pm \frac{1}{49} \right) \enspace .
\end{align*}
and so, $$\abs{V_{ur}} \in \frac{\powf{p}{r}}{\powf{k}{r}} \left( 1 \pm \frac{1}{49} \right)  \enspace . $$

Combining Cases \emph{V.1} and \emph{V.2} gives
\begin{align} \label{eq:Vb}
\abs{V_{ur} } &  \le \frac{ \abs{\powf{p}{r}}}{\powf{k}{r}} \left( \left((1-\abs{\alpha})^{p-r}+ (27)^{-(3/4)k}\right) \1_{r>p, p \text{ non-integral}} +\frac{50}{49} \cdot \1_{r\le p} \right)
\end{align}

Substituting Eqn.~\eqref{eq:Ub} and ~\eqref{eq:Vb} in Eqn.~\eqref{eq:P1b}, we have,
\begin{align} \label{eq:Pa}
 P& =  \lambda^{2p}  \sum_{u=1}^{t} \binom{t}{u}  \sum_{r=1}^u  \binom{u}{r} \beta^{r} U_{u,r} V_{u,r} \notag \\
& \le \lambda^{2p}  \sum_{u=1}^{t} \binom{t}{u}  \sum_{r=1}^u  \binom{u}{r} \beta^{r} \card{U_{u,r}}\card{ V_{u,r}}
\end{align}
Now, since, $1 \le r \le u \le t \le k/4$, we have,
\begin{align*}
&\card{U_{ur}} \cdot  \card{V_{ur}} \\
& \le
 \frac{\card{\powf{p}{u}}\abs{\alpha}^{u-r}}{\powf{k}{u}}\left(\left(( 1- \abs{\alpha})^{p-u} + (27)^{-(3/4)k} \right)\1_{u > p, p \text{ non-integral}} + \frac{50}{49} \cdot \1_{u \le p} \right) \\
&~~ \cdot \left(\frac{ \card{\powf{p}{r}}}{\powf{k}{r}}\right) \left(\left(( 1- \abs{\alpha})^{p-r} + (27)^{-(3/4)k} \right)\1_{r > p, p \text{ non-integral}} + \frac{50}{49} \cdot \1_{r \le p} \right) \\
 & \le \left(\frac{ \card{\powf{p}{u}}\card{\powf{p}{r}}\abs{\alpha}^{u-r}}{\powf{k}{u}\powf{k}{r}}\right)\left(
\left( (1-\abs{\alpha})^{p-u+p-r} + 2 (1-\abs{\alpha})^{p-u} (27)^{-(3/4)k} + (27)^{-(1.5k)} \right)\1_{r > p, p \text{ non-integral}} \right.\\
&  \hspace*{1.0in} \left. \left. +  \left(( 1- \abs{\alpha})^{p-u} +  (27)^{-(3/4)k}\right) \left( \frac{50}{49}\right)\1_{r\le p < u, p \text{ non-integral}}\right)  +\left( \frac{50}{49} \right)^2 \1_{u \le p}\right)
\end{align*}

Therefore,
\begin{align*}
P & = \lambda^{2p}  \sum_{u=1}^{t} \binom{t}{u}  \sum_{r=1}^u  \binom{u}{r} \beta^{r} U_{ur} \cdot V_{ur} \\
& \le \lambda^{2p}  \sum_{u=1}^{t} \binom{t}{u} \sum_{r=1}^u  \binom{u}{r} \beta^{r} \left(\frac{ \card{\powf{p}{u}}\card{\powf{p}{r}}\abs{\alpha}^{u-r}}{\powf{k}{u}\powf{k}{r}}\right) \\
& \hspace*{0.4in} \cdot \left(
\left( (1-\abs{\alpha})^{p-u+p-r} + 2 (1-\abs{\alpha})^{p-u} (27)^{-(3/4)k} + (27)^{-(1.5k)} \right)\1_{r > p, p \text{ non-integral}} \right.\\
&  \hspace*{1.50in} \left. \left. +  \left(( 1- \abs{\alpha})^{p-u} +  (27)^{-(3/4)k}\right) \left( \frac{50}{49}\right)\1_{r\le p < u, p \text{ non-integral}}\right)  +\left( \frac{50}{49} \right)^2 \1_{u \le p}\right)\\
& = P_1 + P_2 + P_3
\end{align*}
\end{proof}

\eat{
where,
\begin{align*}
P_1 & =\lambda^{2p}  \sum_{u=1}^{t} \binom{t}{u} \sum_{r=1}^u  \binom{u}{r} \beta^{r} \left(\frac{ \card{\powf{p}{u}}\card{\powf{p}{r}}\abs{\alpha}^{u-r}}{\powf{k}{u}\powf{k}{r}}\right) \\
& \hspace*{0.4in} \cdot
\left((1-\abs{\alpha})^{p-u+p-r} + 2 (1-\abs{\alpha})^{p-u} (27)^{-(3/4)k} + (27)^{-(1.5k)} \right)\1_{r > p} \\
P_2 & = \lambda^{2p}  \sum_{u=1}^{t} \binom{t}{u} \sum_{r=1}^u  \binom{u}{r} \beta^{r} \left(\frac{ \card{\powf{p}{u}}\card{\powf{p}{r}}\abs{\alpha}^{u-r}}{\powf{k}{u}\powf{k}{r}}\right)
\cdot  \left(( 1- \abs{\alpha})^{p-u} +  (27)^{-(3/4)k}\right) \left( \frac{25}{24}\right)\1_{r\le p < u} \\
P_3 & = \lambda^{2p}  \sum_{u=1}^{t} \binom{t}{u} \sum_{r=1}^u  \binom{u}{r} \beta^{r} \left(\frac{\card{\powf{p}{u}}\card{\powf{p}{r}}\abs{\alpha}^{u-r}}{\powf{k}{u}\powf{k}{r}}\right)\cdot  \left( \frac{25}{24} \right)^2 \1_{u \le p}
\end{align*}
}

\subsubsection{Estimating $P_3$}

\begin{lemma} \label{lem:P3}
Assume the premises and notation of  Lemma~\ref{lem:bvthetabasic} and Corollary~\ref{lem:vbvcross}. Let  $y,y' \in Y$ and distinct and let $\pi = \pi_y$ and $ \pi' = \pi_{y'}$ be  random permutations from $[k] \rightarrow [k]$. Let $\alpha = \frac{\mu-\lambda}{\lambda}$ and $ \beta =  \frac{\eta^2}{\lambda^2} $.
Then,
\begin{align} \label{eq:P3b}
P_3 & \le \frac{0.275p^2 }{k}\lambda^{2p} \beta \enspace . 
\end{align}
\end{lemma}

\begin{proof}

Consider the sum $P_3$.
\begin{align} \label{eq:P3}
P_3 &= \lambda^{2p}  \sum_{u=1}^{t} \binom{t}{u} \sum_{r=1}^u  \binom{u}{r} \beta^{r} \left(\frac{\card{\powf{p}{u}}\card{\powf{p}{r}}\abs{\alpha}^{u-r}}{\powf{k}{u}\powf{k}{r}}\right)\cdot  \left( \frac{50}{49} \right)^2 \1_{u \le p} \notag  \\
& = \left( \frac{50}{49} \right)^2 \lambda^{2p}\sum_{u=1}^{\min(p,t)} \binom{t}{u} \left( \frac{ \powf{p}{u}}{\powf{k}{u}}\right) \abs{\alpha}^u \sum_{r=1}^u \binom{u}{r}  \left( \frac{\powf{p}{r}}{\powf{k}{r}}\right)\left( \frac{\beta}{\abs{\alpha}} \right)^r\notag \\
& \le \left( \frac{50}{49} \right)^2 \lambda^{2p}\sum_{u=1}^{\min(p,t)} \binom{t}{u} \left( \frac{ p}{k}\right)^u \abs{\alpha}^u \sum_{r=1}^u \binom{u}{r}  \left( \frac{p}{k}\right)^r\left( \frac{\beta}{\abs{\alpha}} \right)^r\notag \\
& = \left( \frac{50}{49} \right)^2 \lambda^{2p}\sum_{u=1}^{\min(p,t)} \binom{t}{u} \left( \frac{ p}{k}\right)^u \abs{\alpha}^u \left(\left(1+ \frac{p\beta}{k \abs{\alpha}} \right)^u -1 \right)\notag \\
& \le \left( \frac{50}{49} \right)^2 \lambda^{2p}\sum_{u=1}^{t} \binom{t}{u} \left( \frac{ p}{k}\right)^u \abs{\alpha}^u \left(\left(1+ \frac{p\beta}{k\abs{ \alpha}} \right)^u -1 \right)\notag \\
& = \left( \frac{50}{49} \right)^2 \lambda^{2p}\left(P_{31} - P_{32} \right)
\end{align}

where,
\begin{align*}
P_{31} & = \sum_{u=1}^{t} \binom{t}{u} \left( \frac{ p}{k}\right)^u \abs{\alpha}^u \left(1+ \frac{p\beta}{k \abs{\alpha}} \right)^u = \left( 1 + \frac{ p \abs{\alpha}}{k} \left(1 + \frac{p\beta}{k \abs{\alpha}} \right) \right)^t -1 \\
P_{32} & = \sum_{u=1}^{t} \binom{t}{u} \abs{\alpha}^u\left( \frac{ p}{k}\right)^u  = \left(1 + \frac{p\abs{\alpha}}{k} \right)^t - 1 \enspace .
\end{align*}
Let $a = \left( 1 + \frac{ p \abs{\alpha}}{k} \left(1 + \frac{p\beta}{k \abs{\alpha}} \right)\right)  \le \exp{ \frac{ p \abs{\alpha}}{k} \left(1 + \frac{p\beta}{k \abs{\alpha}}\right)}$ and $b = \left(1 + \frac{p\abs{\alpha}}{k} \right)$.
Therefore,
\begin{align*}
P_{31}-P_{32}   & = a^t -b^t \le (a-b)(t a^{t-1}) \\ &\le \left(\frac{p^2 \beta}{k^2}\right) (t) \exp{ (t-1) \frac{ p \abs{\alpha}}{k} \left(1 + \frac{p\beta}{k \abs{\alpha}} \right)} \\
& \le \frac{p^2 \beta}{4k} \exp{ \frac{p \abs{\alpha}}{4} + \frac{p^2 \beta}{4k }} \\
& \le  \frac{p^2 \beta}{4k} \exp{ \frac{1}{100} + \frac{ 1}{50k}} \\
& \le \frac{( 1.0102)p^2 \beta}{4k}
\end{align*}
Therefore,  subsituting in Eqn.~\eqref{eq:P3}, we have,
\begin{align*}
P_3 & \le \frac{0.275p^2 }{k}\lambda^{2p} \beta = \frac{0.275 p^2}{k} \lambda^{2p-2} \eta^2 \enspace .
\end{align*}
\end{proof}

\subsubsection{Estimating $P_2$}
We now consider $P_2$.
\begin{lemma} \label{lem:P2}
Assume the premises and notation of  Lemma~\ref{lem:bvthetabasic} and Corollary~\ref{lem:vbvcross}. Let  $y,y' \in Y$ and distinct and let $\pi = \pi_y$ and $ \pi' = \pi_{y'}$ be  random permutations from $[k] \rightarrow [k]$. Let $\alpha = \frac{\mu-\lambda}{\lambda}$ and $ \beta =  \frac{\eta^2}{\lambda^2} $.
\begin{align}
P_2 \le \frac{ p^2\lambda^{2p}  \beta}{(30)(40)k} \enspace .
\label{eq:P2f}
\end{align}
\end{lemma}
\begin{proof}
\begin{align} \label{eq:P2a}
P_2 & =  \lambda^{2p}  \sum_{u=1}^{t} \binom{t}{u} \sum_{r=1}^u  \binom{u}{r} \beta^{r} \left(\frac{ \card{\powf{p}{u}}\card{\powf{p}{r}}\abs{\alpha}^{u-r}}{\powf{k}{u}\powf{k}{r}}\right)
\cdot  \left(( 1- \abs{\alpha})^{p-u} +  (27)^{-(3/4)k}\right) \left( \frac{50}{49}\right)\1_{r\le p < u} \notag \\
& = \left( \frac{50}{49}\right)\lambda^{2p}  \sum_{u=\lfloor p \rfloor +1}^{t} \binom{t}{u} \sum_{r=1}^{\lfloor p \rfloor }  \binom{u}{r} \beta^{r} \left(\frac{ \card{\powf{p}{u}}\card{\powf{p}{r}}\abs{\alpha}^{u-r}}{\powf{k}{u}\powf{k}{r}}\right)
\cdot  \left(( 1- \abs{\alpha})^{p-u} +  (27)^{-(3/4)k}\right)
\end{align}
The first summation is empty if $ t < \lfloor p \rfloor +1$ in which case $P_2=0$. Also, $P_2=0$ if $p$ is integral, since $\powf{p}{u} = 0$, for $u \ge \lfloor p \rfloor +1$.  So we now assume that $t \ge \lfloor p \rfloor +1$ and $p$ is not integral. Further, $(1-\abs{\alpha})^{p-u} \ge 1$, for $u \ge \lfloor p \rfloor +1$ and $\abs{\alpha} \le 1/(50p)$. Hence, $(27)^{-(3/4)k} + (1-\abs{\alpha})^{p-u} \le (1-\abs{\alpha})^{p-u} (1 + (27)^{-(3/4)k}$. Using this simplification and also using the fact that $\powf{p}{r}/\powf{k}{r} \le (p/k)^r$, for $ 1 \le r \le \lfloor p \rfloor$,  Eqn.~\eqref{eq:P2a} can be written as follows.
\begin{align*}
P_2 & \le \left( \frac{50}{49}\right)\left( 1 +  (27)^{-(3/4)k}\right)  \lambda^{2p}  \sum_{u=\lfloor p \rfloor +1}^{t} \binom{t}{u}   \frac{\card{\powf{p}{u}} \abs{\alpha}^u}{\powf{k}{u}} \sum_{r=1}^{\lfloor p \rfloor }  \binom{u}{r} \left( \frac{\beta}{\abs{\alpha}}\right)^{r} \left(\frac{ \card{\powf{p}{r}}}{\powf{k}{r}}\right)
\cdot  ( 1- \abs{\alpha})^{p-u} \\
& \le (1.042)  ( 1- \abs{\alpha})^p   \lambda^{2p}  \sum_{u=\lfloor p \rfloor +1}^{t} \sum_{r=1}^{\lfloor p \rfloor }  \binom{t}{u} \binom{u}{r}  \frac{\card{\powf{p}{u}} \gamma^u}{\powf{k}{u}}    \left( \frac{\beta}{\abs{\alpha}}\right)^{r} \left(\frac{p}{k}\right)^r \\
\end{align*}
where,  $ \gamma = \frac{ \abs{\alpha}}{1-\abs{\alpha}}$.

Let
\begin{align*}
Q_2 & = \sum_{u=\lfloor p \rfloor +1}^{t}\sum_{r=1}^{\lfloor p \rfloor }  \binom{t}{u} \binom{u}{r}  \frac{\card{\powf{p}{u}} \gamma^u}{\powf{k}{u}}    \left( \frac{\beta}{\abs{\alpha}}\right)^{r} \left(\frac{p}{k}\right)^r \end{align*}
so that
\begin{align}  \label{eq:P2b}
P_2 \le (1.042) e^{-\abs{\alpha}p} \lambda^{2p}Q_2 \le (1.001) \lambda^{2p}Q_2
\enspace .
\end{align}

Then,
\begin{align}
Q_2 & = \sum_{u=\lfloor p \rfloor +1}^{t} \sum_{r=1}^{\lfloor p \rfloor }  \binom{t}{u} \binom{u}{r}  \frac{\card{\powf{p}{u}} \gamma^u}{\powf{k}{u}}    \left( \frac{\beta}{\abs{\alpha}}\right)^{r} \left(\frac{p}{k}\right)^r\notag \\
 & =  \sum_{u=\lfloor p \rfloor +1}^{t} \sum_{r=1}^{\lfloor p \rfloor }  \binom{t}{r} \binom{t-r}{u-r}  \frac{\card{\powf{p}{u}} \gamma^u}{\powf{k}{u}}    \left( \frac{\beta}{\abs{\alpha}}\right)^{r} \left(\frac{p}{k}\right)^r \notag \\
& = \sum_{r=1}^{\lfloor p \rfloor } \binom{t}{r} \left( \frac{ p \beta}{\abs{\alpha} k} \right)^r \sum_{u=\lfloor p \rfloor +1}^t  \binom{t-r}{u-r}  \frac{ \powf{p}{r} \card{ \powf{(p-r)}{u-r}}}{\powf{k}{r} \powf{(k-r)}{u-r}}\gamma^{r+ u-r}  \notag \\
& =  \sum_{r=1}^{\lfloor p \rfloor } \binom{t}{r} \left( \frac{ p \beta}{\abs{\alpha} k} \right)^r \left( \frac{\powf{p}{r}}{\powf{k}{r}}\right) \gamma^{r} \sum_{u-r = \lfloor p \rfloor +1-r}^{t-r}\binom{t-r}{u-r}
\frac{\card{\powf{(p-r)}{u-r}}}{ \powf{(k-r)}{u-r}} \cdot  \gamma^{u-r} \label{eq:Q2a}
 \end{align}
Consider the inner summation in Eqn.~\eqref{eq:Q2a}, namely,
\begin{align}
& \sum_{w = \lfloor p \rfloor +1-r}^{t-r}\binom{t-r}{w}
\frac{\card{\powf{(p-r)}{w}}}{ \powf{(k-r)}{w}} \cdot  \gamma^{w} \label{eq:Q2ia}
\end{align}
The ratio of $(w+1)$st term to the $w$th term, for $w= \lfloor p \rfloor +1 -r, \ldots, t-r-1$,  in the above summation is
\begin{align*}
\left( \frac{ t-r-w}{w+1}\right) \left( \frac{ \abs{p-r-w}}{k-r-w} \right) \gamma  = \left( \frac{ t-r-w}{k-r-w} \right) \left( \frac{w-(p-r)}{w+1}\right) \gamma \le \frac{t \gamma}{k} \le \frac{1}{(4)(50p-1)} \le \frac{ 1}{(4)(49)}
\end{align*}
since $t \le k/4$ and $\gamma = \frac{\abs{\alpha}}{1- \abs{\alpha}} \le \frac{1}{50p-1} \le \frac{1}{49}$.  Therefore, Eqn.~\eqref{eq:Q2ia} may be upper bounded as follows.
\begin{align*}
&\sum_{w = \lfloor p \rfloor +1-r}^{t-r}\binom{t-r}{w}
\frac{\card{\powf{(p-r)}{w}}}{ \powf{(k-r)}{w}} \cdot  \gamma^{w}\\
&\le \binom{t-r}{\lfloor p \rfloor +1 -r} \left( \frac{ \card{ \powf{(p-r)}{\lfloor p \rfloor +1 -r} }}{\powf{(k-r)}{\lfloor p \rfloor +1 -r}}\right) \gamma^{\lfloor p \rfloor +1 -r}  \left( 1 + \frac{1}{195} \right) \enspace .
\end{align*}
Substituting in Eqn.~\eqref{eq:Q2a}, we have,
\begin{align} \label{eq:Q2c}
Q_2 & \le   (1.0052) \sum_{r=1}^{\lfloor p \rfloor } \binom{t}{r} \left( \frac{ p \beta}{\abs{\alpha} k} \right)^r \left( \frac{\powf{p}{r}}{\powf{k}{r}}\right) \gamma^{r} \binom{t-r}{\lfloor p \rfloor +1 -r} \left( \frac{ \card{ \powf{(p-r)}{\lfloor p \rfloor +1 -r} }}{\powf{(k-r)}{\lfloor p \rfloor +1 -r}}\right) \gamma^{\lfloor p \rfloor +1 -r} \notag \\
& \le (1.0052) \sum_{r=1}^{\lfloor p \rfloor } \binom{t}{r} \binom{t-r}{\lfloor p \rfloor +1 -r} \left( \frac{ p^2\beta \gamma}{\abs{\alpha} k^2} \right)^r  \left( \frac{ \card{ \powf{(p-r)}{\lfloor p \rfloor +1 -r} }}{\powf{(k-r)}{\lfloor p \rfloor +1 -r}}\right) \gamma^{\lfloor p \rfloor +1 -r}\notag \\
& =  (1.0052) \sum_{r=1}^{\lfloor p \rfloor } \binom{t}{\lfloor p \rfloor +1} \binom{ \lfloor p \rfloor +1}{r} \left( \frac{ p^2\beta}{(1-\abs{\alpha}) k^2} \right)^r  \left( \frac{ \card{ \powf{(p-r)}{\lfloor p \rfloor +1 -r} }}{\powf{(k-r)}{\lfloor p \rfloor +1 -r}}\right) \gamma^{\lfloor p \rfloor +1 -r} \notag \\
& = (1.0052)S
\end{align}
Consider the summation above and let $t_r = \binom{ \lfloor p \rfloor +1}{r}  \left( \frac{ p^2\beta}{(1-\abs{\alpha}) k^2} \right)^r  \left( \frac{ \powf{(p-r)}{\lfloor p \rfloor +1 -r}}{\powf{(k-r)}{\lfloor p \rfloor +1 -r}}\right) \gamma^{\lfloor p \rfloor +1 -r}$
be the $r$th term. Let $r_m = \text{argmax}_{r=1}^{\lfloor  p \rfloor } t_r$, that is $t_{r_m}$ is the largest among the $t_r$'s. Then, clearly, $S = \sum_{r=1}^{\lfloor p \rfloor} t_r \le \lfloor p \rfloor t_{r_m}$. For $r_m = r \in \{1,2, \ldots, \lfloor p \rfloor\}$, we have,

\begin{align}
 S&  \le  \lfloor p \rfloor\binom{t}{\lfloor p \rfloor +1}  \binom{\lfloor p \rfloor +1}{r}  \left( \frac{ p^2\beta}{(1-\abs{\alpha}) k^2} \right)^r  \left( \frac{ \powf{(p-r)}{\lfloor p \rfloor +1 -r}}{\powf{(k-r)}{\lfloor p \rfloor +1 -r}}\right) \gamma^{\lfloor p \rfloor +1 -r} \notag  \\
&  = \left( \frac{\lfloor p \rfloor \gamma}{r!}\right) \left(\frac{ \powf{t}{\lfloor p \rfloor +1}}{ k^r \powf{(k-r)}{\lfloor p \rfloor +1 -r} }\right) \left( \frac{ p^2\beta}{(1-\abs{\alpha}) k} \right)^r \left( \frac{\powf{(p-r)}{\lfloor p \rfloor +1 -r}}{(\lfloor p \rfloor +1-r)!}\right) \gamma^{\lfloor p \rfloor  -r} \label{eq:S2c}
\end{align}
Now $$\frac{ p^2\beta}{(1-\abs{\alpha}) k} \le \frac{ p^2\beta}{(1- \frac{1}{50p}) k}
\le  \frac{(1.011)p^2 \beta }{k} $$
since $p \ge 2$.

Therefore, Eqn.~\eqref{eq:S2c} may be written as
\begin{align}
S & \le \left( \frac{\lfloor p \rfloor \gamma}{r!}\right) \left( \frac{t}{k} \right)^r \left( \frac{ t-r}{k-r} \right)^{\lfloor p \rfloor + 1  - r} \left( \frac{(1.011)p^2 \beta }{k} \right)^r \gamma^{\lfloor p \rfloor  -r} \notag \\
& \le \left( \frac{(1.011)\lfloor p \rfloor \gamma p^2 \beta }{k^r r!}\right)  (4)^{-(\lfloor p \rfloor +1)}, \text{ since, $\frac{t}{k} \le \frac{1}{4}$ and $\gamma \le \frac{1}{(50p-1)}$} \notag \\
& \le \frac{ p^2 \beta }{(30)(49) k} , \text{ since, $p \ge 2$ and $p^2 \beta \ll 1$.}  \label{eq:Q2d}
\end{align}
Substituting in Eqn.~\eqref{eq:Q2c}, we have that $Q_2 \le (1.0052)S$ and from Eqn.~\eqref{eq:P2b}, we have,
\begin{align*}
P_2 \le (1.001) \lambda^{2p} Q_2 \le \frac{ p^2\lambda^{2p}  \beta}{(30)(40)k}
\end{align*}
\end{proof}

\subsubsection{Estimating $P_1$}
 We now calculate $P_1$.
 \begin{lemma} \label{lem:P1}
Assume the premises and notation of  Lemma~\ref{lem:bvthetabasic} and Corollary~\ref{lem:vbvcross}. Let  $y,y' \in Y$ and distinct and let $\pi = \pi_y$ and $ \pi' = \pi_{y'}$ be  random permutations from $[k] \rightarrow [k]$. Let $\alpha = \frac{\mu-\lambda}{\lambda}$ and $ \beta =  \frac{\eta^2}{\lambda^2} $.
Then for $n \ge 2$,
\begin{align} \label{eq:P1}
 P_1 \le (0.3)\left(\frac{p^2 \beta}{k^a}\right)\left(\frac{1}{(2)(25)^{2} p}\right)^{(a-1)}
\end{align}
\end{lemma}

\begin{proof}
 \begin{align}
 P_1 & =\lambda^{2p}  \sum_{u=1}^{t} \binom{t}{u} \sum_{r=1}^u  \binom{u}{r} \beta^{r} \left(\frac{ \card{\powf{p}{u}}\card{\powf{p}{r}}\abs{\alpha}^{u-r}}{\powf{k}{u}\powf{k}{r}}
 \right) \notag \\
  & \hspace*{0.4in} \cdot
\left((1-\abs{\alpha})^{p-u+p-r} + 2 (1-\abs{\alpha})^{p-u} (27)^{-(3/4)k} + (27)^{-(1.5k)} \right)\1_{r > p}  \label{eq:P1a}
\end{align}

First, we  note that for $k \ge c\log n$ (where, $c = 100$ as per Table ~\ref{table:params}),  $(27)^{(-3/4)k} = n^{-(3.5)c} \le n^{-(3.5)c} (1-\abs{\alpha})^{p-u}$. Hence,
$\left((1-\abs{\alpha})^{p-u+p-r} + 2 (1-\abs{\alpha})^{p-u} (27)^{-(3/4)k} + (27)^{-(1.5k)}\right) = (1-\abs{\alpha})^{2p-u-r} (1 + O(n^{-(3.5c)})$. Therefore,
\begin{align}
P_1 &= \left( 1 + O(n^{-3.5c})\right) \lambda^{2p} \sum_{u=1}^{t} \binom{t}{u} \sum_{r=1}^u  \binom{u}{r} \beta^{r} \left(\frac{ \card{\powf{p}{u}}\card{\powf{p}{r}}\abs{\alpha}^{u-r}}{\powf{k}{u}\powf{k}{r}}
 \right) (1-\abs{\alpha})^{p-u+p-r} \1_{r>p} \notag \\
 & =  \left( 1 + O(n^{-3.5c})\right)  \lambda^{2p}
 \sum_{u=\lfloor p \rfloor +1}^t \sum_{r=\lfloor p \rfloor +1}^u \binom{t}{u} \binom{u}{r} \beta^{r} \left(\frac{ \card{\powf{p}{u}}\card{\powf{p}{r}}\abs{\alpha}^{u-r}}{\powf{k}{u}\powf{k}{r}}
 \right) (1-\abs{\alpha})^{2p-u-r}  \notag \\
& =\left( 1 + O(n^{-3.5c})\right)  \lambda^{2p} L \label{eq:P1b}
\end{align}
where,
 $$L = \sum_{u=\lfloor p \rfloor +1}^t \sum_{r=\lfloor p \rfloor +1}^u \binom{t}{u} \binom{u}{r} \beta^{r} \left(\frac{ \card{\powf{p}{u}}\card{\powf{p}{r}}\abs{\alpha}^{u-r}}{\powf{k}{u}\powf{k}{r}}
 \right) (1-\abs{\alpha})^{2p-u-r} \enspace . $$

 Let $a =\lfloor p \rfloor +1$ and let $v = u-a$ and $w = r-a$. Then,
 \begin{align}
L &=(1-\abs{\alpha})^{2p-2a} \beta^a  \powf{t}{a}\left( \frac{\powf{p}{a}}{\powf{k}{a}} \right)^2
  \sum_{v=0}^{t-a}\binom{t-a}{v}
   \sum_{w=0}^{v}
\binom{v}{w} \abs{\alpha}^{v-w} \notag \\
&\hspace*{1.0cm}\left( \frac{ \beta}{(1-\abs{\alpha})}\right)^{w}
  \left(  \frac{ \powf{(a+v-p-1)}{v}  \powf{(a+w-p-1)}{w}}{\powf{(k-a)}{v} \powf{(k-a)}{w}  \powf{(w+a)}{a}}\right) (1-\abs{\alpha})^{-v} \label{eq:La}
\end{align}

Now $\powf{a -p + w - 1}{w} \le \powf{w}{w} = w!$. Similarly, $\powf{a-p+v-1} \le v!$. Therefore, $$
\binom{t-a}{v} \binom{v}{w}  \powf{(a+v-p-1)}{v}  \powf{(a+w-p-1)}{w} \le \powf{(t-a)}{v} \powf{v}{w} \enspace . $$ Hence,
\begin{align}
L & \le (1-\abs{\alpha})^{2p-2a} \beta^a  \powf{t}{a}\left( \frac{\powf{p}{a}}{\powf{k}{a}} \right)^2
  \sum_{v=0}^{t-a}\left(\frac{\powf{(t-a)}{v}}{\powf{(k-a)}{v}} \right) (1-\abs{\alpha})^{-v} \notag \\
& \hspace*{1.0in}  \sum_{w=0}^{v}
\frac{\powf{v}{w}}{\powf{(k-a)}{w}  }\abs{\alpha}^{v-w} \left( \frac{ \beta}{(1-\abs{\alpha})}\right)^{w}
  \left(  \frac{1}{ \powf{(w+a)}{a}}\right)  \notag \\
  &  \le (1-\abs{\alpha})^{2p-2a} \beta^a  \powf{t}{a}\left( \frac{\powf{p}{a}}{\powf{k}{a}} \right)^2\left( 1 +
  \sum_{v=1}^{t-a} \sum_{w=0}^{v}
\frac{c^v   \powf{v}{w} \abs{\alpha}^{v-w} }{\powf{(k-a)}{w}  }\beta'^{w}
  \left(  \frac{1}{ \powf{(w+a)}{a}}\right)\right) \label{eq:Lb}
 \end{align}
where $c = \frac{t-a}{(k-a)(1-\abs{\alpha})}$ and $ \beta' = \left( \frac{ \beta}{(1-\abs{\alpha})}\right)$.  Let $l_{vw}$ denote the summand
 $$l_{vw} = \frac{c^v   \powf{v}{w} \abs{\alpha}^{v-w} }{\powf{(k-a)}{w}  }\beta'^{w}
  \left(  \frac{1}{ \powf{(w+a)}{a}}\right), ~~~1 \le v \le t-a, 0 \le w \le v \enspace .  $$
  The summation in Eqn.~\eqref{eq:Lb} may be written as
  $$J = \sum_{v=1}^{t-a} K_v, \text{ where, } K_v= \sum_{w=0}^v l_{vw}, ~~v = 1,2, \ldots, t-a \enspace . $$
  Therefore,

  Comparing $l_{vw}$ and $l_{v+1,w}$, we have,
 \begin{align*}
 l_{vw} & = \frac{c^v   \powf{v}{w} \abs{\alpha}^{v-w} \beta'^{w} }{\powf{(k-a)}{w} \powf{(w+a)}{a} }\\
  l_{v+1,w} & = \frac{c^{v+1} \powf{(v+1)}{w} \abs{\alpha}^{v+1-w}
 \beta'^w}{\powf{(k-a)}{w} \powf{(w+a)}{a} }
 \end{align*}
 Then,
 \begin{align} \label{eq:lratio}
 \frac{l_{v+1,w+1}}{l_{vw}} = \frac{ c (v+1) \abs{\alpha} \beta' (w+1)}{(k-a-w)(w+1+a)}, ~~~ 1 \le v \le t-a-1, 0 \le w \le v \enspace .
 \end{align}
 Since,
 $l_{v,0} = \frac{c^{v} \abs{\alpha}^{v}}{a!}$, therefore,
 $
 \frac{ l_{v+1,0}}{\sum_{w=0}^{v} l_{vw}} \le \frac{l_{v+1,0}}{l_{v0}}\le c\abs{\alpha}
$.
Therefore,  for $1 \le v \le t-a-1$,
\begin{align*}
\frac{K_{v+1}}{2K_v} = \frac{\sum_{w=0}^{v+1} l_{v+1,w}}{2\sum_{w=0}^v l_{vw}} \le \frac{l_{v+1,0}}{l_{v_0}}  + \max_{w=0}^v \left( \frac{ l_{v+1,w+1}}{l_{vw}} \right) \le 2c \abs{\alpha}, \text{ by Eqn.~\eqref{eq:lratio}.}
\end{align*}
or,
\begin{align*}
\frac{K_{v+1}}{K_v} \le \frac{\sum_{w=0}^{v+1} l_{v+1,w}}{\sum_{w=0}^v l_{vw}} \le  4 c \abs{\alpha} \le \frac{ 4 (t-a)}{(k-a)(1-\frac{1}{25p})25p} \le \frac{1}{25p-1} = \frac{1}{49} \enspace .
\end{align*}

\begin{align*}
L  &\le (1-\abs{\alpha})^{2p-2a} \beta^a  \powf{t}{a}\left( \frac{\powf{p}{a}}{\powf{k}{a}} \right)^2\left( \frac{1}{a!} +
  \sum_{v=1}^{t-a}K_v \right) \\
& \le   (1-\abs{\alpha})^{2p-2a} \beta^a  \powf{t}{a}\left( \frac{\powf{p}{a}}{\powf{k}{a}} \right)^2\left( \frac{1}{a!} + \frac{49 K_1}{48} \right)\\
& =(1-\abs{\alpha})^{2p-2a} \beta^a  \powf{t}{a}\left( \frac{\powf{p}{a}}{\powf{k}{a}} \right)^2 \left( \frac{1}{a!} + \frac{49}{48} \left( \frac{ c \abs{\alpha}}{a!} + \frac{c\beta' \abs{\alpha}}{(k-a) (a+1)!} \right)\right) \\
& \le (1-\abs{\alpha})^{2p-2a} \beta^a  \powf{t}{a}\left( \frac{\powf{p}{a}}{\powf{k}{a}} \right)^2\left( \frac{1}{a!} \right) (1.006) \\
& \le (1.006)(1-\abs{\alpha})^{-2}\left( \frac{\powf{t}{a}}{\powf{k}{a}} \right) \left( \frac{ \beta^a \powf{p}{a}}{\powf{k}{a}} \right) \left(\frac{ \powf{p}{a}}{a!} \right)   \notag \\
& \le (0.2625)\left(\frac{p^2 \beta}{k^a}\right)\left(\frac{1}{(2)(25)^{2} p}\right)^{(a-1)}
\end{align*}
since, (i) $c  = \frac{(t-a)}{(k-a) (1-\abs{\alpha})} \le \frac{t}{k(1-\frac{1}{50})} \le 0.256$,  (ii)  $\left( \frac{\powf{t}{a}}{\powf{k}{a}} \right) \le \left( \frac{t}{k} \right)^a \le \left( \frac{1}{4} \right)^a$, (iii) $a = \lfloor p \rfloor +1$ and therefore, $\powf{p}{a} \le a!$, (iv) $\beta p \le \frac{(2)p}{(25p)^2} \le \frac{2}{(25)^{2}p}$ and so, $\frac{ \beta^a \powf{p}{a}}{\powf{k}{a}} \le \left( \frac{ \beta p }{k} \right)^a =  (p^2\beta) \frac{ (p\beta)^{a-1}}{p k^a} \le  (p^2 \beta) \left(\frac{2}{(25)^{2} p}\right)^{(a-1)} \frac{1}{ k^a}$.

Substituting in Eqn.~\eqref{eq:P1b}, we have,
\begin{align*}
P_1 \le (0.3)\left(\frac{p^2 \beta}{k^a}\right)\left(\frac{1}{(2)(25)^{2} p}\right)^{(a-1)}
\end{align*}
\end{proof}

\subsection{Completing Variance calculation for Averaged Taylor Polynomial Estimator}

\begin{lemma} \label{lem:covarb} Assume the premises of Lemma~\ref{lem:bvthetabasic} and let $\mu=\lambda$. Let $y,y' \in Y$ be distinct. Then,
 $$\covariance{\vartheta_y}{ \vartheta_{y'}} \le  \left(\frac{ (0.261) p^2}{k} \right) \mu^{2p-2}\eta^2  \enspace . $$
 \end{lemma}

 \begin{proof}
 By Lemma~\ref{lem:bvthetabasic},
 \begin{align*}
 \covariance{\vartheta_y}{\vartheta_{y'}} & =  \sum_{v=1}^k
\gamma^2_{v}(\lambda) \eta^{2v} \probsub{\pi_y,\pi_{y'}}{q^{vv}_{yy'\pi_y\pi_{y'}} = v} \\
& = \sum_{v=1}^k \binom{p}{v}^2 \lambda^{2(p-v)} \eta^{2v} \left( \frac{ \binom{t}{v}}{\binom{k}{v}^2} \right)
\end{align*}
Taking the ratio of the $v+1$st term and the $v$th term  of the summation above, we obtain,
\begin{align*}
\left(\frac{ (p-v)^2}{(v+1)^2}\right) \left( \frac{\eta^2}{\lambda^2}\right) \left( \frac{ (t-v)}{v+1} \right) \left( \frac{ (v+1)}{k-v} \right)^2 \le  \left(\frac{ (p-1)}{2}\right) \left( \frac{2}{(25p)^2} \right) \le \frac{1}{2500}
\end{align*}
Therefore,
\begin{align*}
\covariance{\vartheta_y}{ \vartheta_{y'}} \le  p^2 \lambda^{2p-2}{\eta^2} \left( \frac{ t}{k^2} \right) \left( 1 + \frac{1}{2499} \right) \le \left(\frac{ (0.251) p^2}{k} \right) \lambda^{2p-2}\eta^2  \enspace .\
\end{align*}
since, $ \frac{t}{k} \le \frac{1}{4}$.
\end{proof}

\begin{lemma} \label{lem:covar}
Assume the premises of Lemma~\ref{lem:bvthetabasic}.  Let $y,y' \in Y$ be distinct. Then,
 $$\covariance{\vartheta_y}{ \vartheta_{y'}} \le  \frac{ 0.276 p^2 \lambda^{2p} \beta}{k} \enspace . $$
 \end{lemma}

 \begin{proof}

\emph{Case 1:} $\mu=\lambda$. By Lemma~\ref{lem:covarb},
$$ \covariance{\vartheta_y}{\vartheta_{y'}} \le  \left(\frac{ (0.251) p^2}{k} \right) \lambda^{2p-2} \eta^2 \enspace .
$$
\emph{Case 2:} $\mu \ne \lambda$.
Adding the expressions for $P_3, P_2$ and $P_1$ respectively from Lemmas~\ref{lem:P3} to ~\ref{lem:P1}, we obtain,
\begin{align} \label{eq:P1b}
P &  \le  \frac{p^2 \lambda^{2p} \beta}{k} \left( (0.275)  + \left(\frac{1}{1200} + \left(\frac{0.3}{k^{a-1}}\right)\left(\frac{1}{(2)(25)^{2} p}\right)^{(a-1)}\right)\1_{p \text{ non-integral}} \right)  \notag \\
& \le \frac{ 0.276 p^2 \lambda^{2p} \beta}{k}
\end{align}
Therefore,
\begin{align*}
 \covariance{\vartheta_y}{\vartheta_{y'}} &\le P_{yy'} + Q_{yy'}, ~~~~~\text{ by Corollary ~\ref {lem:vbvcross}} \\
 & \le \frac{ 0.276 p^2 \lambda^{2p} \beta}{k}, ~~~ \text{ by Eqn.~\eqref{eq:P1b} and Lemma~\ref{lem:Q}}
 \end{align*}
 Thus, in all cases, $$\covariance{\vartheta_y}{\vartheta_{y'}}  \le \frac{ 0.276 p^2 \lambda^{2p} \beta}{k} \enspace . $$
 \end{proof}

 \begin{lemma} \label{lem:vbvtz}
 Assume the premises of Lemma~\ref{lem:bvthetabasic} and let $k \ge 1000$ and $n \ge 2$. Then,
 \begin{align*}
 \variance{\bvtheta}  \le \left(\frac{(0.288)p^2}{k} \right)\mu^{2p-2} \eta^2 \enspace .
 \end{align*}

\end{lemma}
\begin{proof}
 \begin{align} \label{eq:vbvtza}
\variance{\bvtheta} &  = \frac{1}{\abs{Y}^2} \sum_{y \in Y} \variance{\vartheta_y} + \sum_{\substack{ y \ne y'\\ y,y' \in Y}} \covariance{\vartheta_y}{ \vartheta_{y'}} \notag \\
& \le  \left(\frac{1}{\abs{Y}^2}\right) \abs{Y} (1.08)p^2 \mu^{2p-2} \eta^2 + \left( \frac{ \abs{Y} (\abs{Y}-1)}{\abs{Y}^2} \right) \left(\frac{ (0.276) p^2 \lambda^{2p-2} \eta^2}{k}  \right) \notag \\
 & = \left(\frac{1}{2^{0.08k}} \right) (1.08)p^2 \mu^{2p-2} \eta^2 + (0.276)(e^{1/25})\left(\frac{p^2 \mu^{2p-2} \eta^2}{k} \right)\notag \\
 & \le \left(\frac{(0.288)p^2}{k} \right)\mu^{2p-2} \eta^2 ~~~~\text{ for $k \ge 1000$.}
\end{align}
The second step uses Corollary~\ref{cor:fptpvar} and ~\ref{lem:covar}.
 \end{proof}

\section{Proof that $\G$ holds with very high probability}

\subsection{Preliminaries  and Auxiliary Events}

\emph{The event \goodftwo.} Using standard algorithms for estimating $F_2$ such as \cite{ams:jcss98,tz:soda04}, one can obtain an estimate $\tilde{F}_2$ satisfying $\abs{\tilde{F}_2 - F_2}\le \frac{0.001}{8p}F_2$, with probability $1 - n^{-25}$ using space $O(\log^2 n)$ bits. Then, $\hat{F}_2 = \left(1 - \frac{0.001}{8p}\right)^{-1} \hat{F}_2$ satisfies $F_2 \le \hat{F}_2 \le \left(1 + \frac{0.001}{2p} \right) F_2$, which is the event \goodftwo.

The event \goodest~essentially states that the \countsketch~guarantees for accuracy of estimation holds for all items and at all levels.
\begin{lemma} \label{lem:goodest}
\goodest~holds with probability $1-n^{-23}$.
\end{lemma}
\begin{proof}
By guarantees of  \countsketch~ structure \cite{ccf:icalp02} using  tables with $16C_l$ buckets and $ s=8k = (8)(1000)(\log n)$ tables with independent hash functions, we have,
$
\abs{\hat{f}_{il} - f_i} \le  \left( \ftwores{C_l, l}/C_l \right)^{1/2}$ with probability $ 1-  n^{-25}$. Using union bound to add the  error probability over the levels $L = O(\log n)$ and $i \in [n]$, we obtain that \goodest~holds except with probability $n^{-25}(L)(n) \le n^{-23}$. \hfill
\end{proof}

The above events comprising $\G$ will be  shown to hold with probability $1-n^{-\Omega(1)}$.  In order to do so, we define a few auxiliary events.
\subsubsection*{Auxiliary Events}
For  $l \in \{0\} \cup  [L]$ and $q \ge 1$, define the random variable  $$ H_{lq} = \sum_{1 \le \rank(i) \le 2^l q } y_{il}  \text{
and } U_{lq} = \sum_{ \rank(i) >  2^l q} f_i^2 \cdot   y_{il}  $$
where,  for $i \in [n]$,  $y_{il} $ is an indicator variable that is 1  if $i \in \stream_l$ and is 0
otherwise. For $l \in \{0\}\cup [L]$, define two auxiliary events parameterized by a parameter $q$, as follows.
\begin{align*}
\smallH(l, q)  &\equiv H_{lq} \le 2q, ~\text{ and } \\
 \smallU(l, q) & \equiv U_{l, q} \le
\frac{1.5 \ftwores{2^{l-1}q }}{2^{l-1}} \enspace .
\end{align*}

\subsection{Proof that space parameter $C_l$ is  polynomial sized}

We will now show that $C_l = n^{\Omega(1)}$ for each $l \in \{0\} \cup [L]$. This would also imply that $B_l = C_l (27p)^{-2} = n^{\Omega(1)}$ for each$l \in \{0\} \cup [L]$.
\begin{lemma} \label{lem:CL} Assume the parameter values given in Figure~\ref{table:params}. Then  for $p > 2$,
$C_L \ge n^{\Omega(1)} $.
\end{lemma}

\begin{proof} Since $L = \lceil \log_{2\alpha} (n/C) \rceil $,
\begin{align} \label{eq:CL1}
C_L &= 4 \alpha^L C \ge  (4\alpha) \alpha^{\log_{2\alpha} (n/C)} C = \frac{ (4\alpha) (2\alpha)^{\log _{2\alpha} (n/C)} C }{2^{\log_{2\alpha}(n/C)}} \notag \\ &= \frac{4\alpha n}{(2^{\log_2(n/C)})^{1/(\log_2 (2\alpha))}} = \frac{4 \alpha n}{(n/C)^{1/\log_2(2\alpha)}} \enspace .
\end{align}
Let $\alpha = 1- \gamma$. Then,  $$\log_2(2\alpha) = 1 +\log_2 (\alpha) = 1 + \frac{\ln(\alpha)}{\ln 2} \ge  1 - \frac{2\gamma}{\ln ( 2)}$$
 since, $ \gamma < 1/2$. Hence, $$\frac{1}{\log_2 (2\alpha)} = \frac{1}{\left( 1 - 2\gamma/\ln ( 2)\right)} \le  1 + \frac{4\gamma}{\ln 2} \enspace . $$

 Let $C = Kn^{1-2/p}$.  Substituting in ~\eqref{eq:CL1},
\begin{align} \label{eq:CL2}
C_L  & \ge \frac{ 4 \alpha n}{ (n/C)^{1/\log_2(2\alpha)}} \notag  \\
& \ge \frac{ 4 \alpha n}{(n/C)^{1 + 4\gamma/\ln (2)}} \notag \\ & =  4 \alpha C (n/C)^{-4\gamma/\ln (2)} \notag \\
&= 4 \alpha K n^{1-2/p} \cdot (K^{-1} n^{2/p})^{-4\gamma/\ln (2)}  \notag \\
& =  4 \alpha  K \cdot K' \cdot n^{1-2/p -(2/p)(4\gamma/\ln (2))}
\end{align}
where,  $K' = K^{4\gamma/\ln(2)}$.

\noindent
Since, $\alpha = 1 - (1-2/p)\nu$, $\gamma =1-\alpha =  (1-2/p) \nu$. The exponent of $n$ in ~\eqref{eq:CL2} is
\begin{align*}1-2/p -(2/p)(4\gamma/\ln (2) &=  1-2/p-  (2/p)\left(1-2/p\right)\left( \frac{4\nu}{\ln (2)}\right) \\ & = \left(1-2/p\right)\left( 1 - (2/p)\left(\frac{4\nu}{\ln 2}\right)  \right)
\end{align*}
which is a  positive constant for all $p > 2$ and $\nu < (\ln 2)/4$.  Thus, $C_L= n^{\Omega(1)}$.
\end{proof}

\emph{\bf Remark.} This is the only place where the fact $p > 2$ is explicitly used. If $p=2$, then, $C_L$ would be $\Theta(\epsilon^{-2})$, and $L$ would be $\log_2 (n\epsilon^{2}) + O(1)$. The analysis would work, although the space bound would increase by a factor of $O(\log (n\epsilon^2))$.

\subsection{Application of Chernoff-Hoeffding bounds for Limited Independence}

We will use the following version of Chernoff-Hoeffding bounds for limited independence, specifically,  Theorem 2.5 (II a) from \cite{sss:soda93}.
\begin{theorem} [\cite{sss:soda93}] \label{thm:sss93}
Let $X_1, X_2, \ldots, X_n$ be  $d $-wise independent random variables with support in
$[0,1]$.
Let $X = X_1 + \ldots + X_n$, with $\expect{X} = \mu$.
Then, for $\delta \ge 1$  and $d \le \lceil \delta \mu e^{-1/3} \rceil$,
$\prob{ \abs{X-\mu} \ge \delta \mu} \le e^{-\lfloor d/2 \rfloor}$. ~~~~
\end{theorem}

The following lemma is shown whose proof is given later  in this section.
\begin{lemma} \label{lem:fres:Ulk} Suppose $d \le \lfloor qe^{-1/3} \rfloor$. Then,  for $l
\in\{0\} \cup  [L]$  the following hold, 
 \begin{enumerate}
 \item $\prob{\smallH(l,q)} \ge 1-e^{-\lfloor d/2 \rfloor }$, and, 
 \item  either $U_{lq}=0$ or
$\prob{\smallU(l,q) } \ge 1- e^{-\lfloor d/2 \rfloor}$.
\end{enumerate}
\end{lemma}
Lemma~\ref{lem:CL} shows that $C_L = n^{\Omega(1)}$. This implies that $B_L = \epsbar^2 n^{\Omega(1)} = n^{\Omega(1)}$ since $\epsbar = 1/(27p)$.  Therefore, $C_l > B_l \ge B_L = n^{\Omega(1)}$ for all $l \in \{0\} \cup [L]$. Hence we can  use Lemma~\ref{lem:fres:Ulk} and the union bound over $l \in \{0\} \cup [L]$ to show that  the following events hold  with probability $1-Le^{-\lfloor d/2\rfloor} = 1- Le^{-\Omega(\log n)} \ge 1-n^{-24}$, for suitable choice of the constant.
\begin{align*}
(a)& ~\wedge_{l \in \{0\}\cup [L]} \smallH(l,C_l), ~~~(b)~~~\wedge_{l \in \{0\}\cup [L]} \smallH(H, C_l/2), ~~\text{ and} \\
 (c)& ~\wedge_{l \in \{0\}\cup [L]}\smallH(l,\lceil \alpha^l B_l/(1-2\epsbar)^2 \rceil
\end{align*}

We now prove Lemma~\ref{lem:fres:Ulk}.
\begin{proof} [Proof of Lemma~\ref{lem:fres:Ulk}.]  For any fixed $l$, $y_{il}$ is an indicator variable that is 1 iff $g_1(i) = g_2(i) = \ldots = g_l(i) =1$. Since the $g_l$'s are drawn independently from $d$-wise independent hash family,  the $y_{il}$'s are $d$-wise independent.

By definition, $H_{lq} = \sum_{1 \le rank(i) \le 2^lq} y_{il}$ is the number of items with rank $2^lq$ or less that have hashed to level $l$.  Since, $\prob{y_{il}} = 1/2^l$, we have,  $\expect{H_{lq}} = 2^{l}q \cdot \frac{1}{2^l} = q $. Therefore,
$$ \prob{ H_{lq} > 2q} \le \prob{ \abs{H_{lq} - q} > q}  \le  e^{-\lfloor d/2 \rfloor}$$
by using  Theorem~\ref{thm:sss93} and assuming $d \le qe^{-1/3}$.

We now prove the bound on $U_{lq}$. By definition, $U_{lq} = \sum_{ \rank(i) > 2^lq} f_i^2 y_{il}$. Taking expectation,  $\expect{U_{lq}} = \sum_{\rank(i) > 2^lq} f_i^2 /2^l = \ftwores{2^l q}/2^l$.

Since,  $\abs{f_{\rank(2^l q)}} \le \abs{f_{\rank(j)}}$ for each  $j \in \{ 2^{l-1}q+1,
\ldots, 2^lq\}$, it follows that
$f^2_{\rank(2^l q)}    \le  \ftwores{ 2^{l-1}q }/(2^{l-1} q)$.

\emph{Case 1: Suppose $\ftwores{ 2^{l-1}q } > 0$.}  Define a scaled down variable $U'_{lq}$ as follows.
\begin{align*} U'_{lq} = \sum_{\rank(i) > 2^l q} \frac{f_i^2}{ \ftwores{ 2^{l-1}q
}/(2^{l-1} q )} \cdot y_{il} &  =  \frac{(2^{l-1} q) U_{lq}}{\ftwores{2^{l-1}q}} \enspace
.
\end{align*}
By the above argument, the multiplier $f_i^2/ (\ftwores{ 2^{l-1}q
}/(2^{l-1} q )) \le 1$.  Since $y_{il}$ are indicator variables,
$U'_{lq}$ is  the sum of $d$-wise independent variables \eat{$\{y_{il}\}_{i\in [n]}$} with support in the interval $[0,1]$.

Taking expectation,
$$
\expectb{U'_{lq}} = \frac{(2^{l-1} q)  \expect{U_{lq}}  }{\ftwores{2^{l-1}q}}=
\frac{(2^{l-1} q) }{\ftwores{2^{l-1}q}} \cdot \frac{\ftwores{2^l q}}{2^l}  \le
\frac{q}{2}\enspace .
$$
 By Theorem~\ref{thm:sss93}, we obtain,
\begin{align*}
\probb{U' _{lq}>  \expectb{U'_{lq}} + q} \le   \probb{ \card{U'_{lq} - \expectb{U'_{lq}}} >
q} & \le
e^{-\lfloor d/2\rfloor} \enspace .
\end{align*}
provided, $d \le   \lceil qe^{-1/3} \rceil $, which is assumed.

\noindent
The event $U'_{lq} > \expectb{U'_{lq}} + q$  may be equivalently written  (by rescaling) as $
 U_{lq} > \expectb{U_{lq}} + \frac{q \ftwores{2^{l-1}q}} {2^{l-1} q}$, which is the same as   $U_{lq} > \frac{\ftwores{2^l q}}{2^l} + \frac{\ftwores{2^{l-1}q}}{2^{l-1}}$. This in turn is implied by the event $ U_{lq} > \frac{1.5 \ftwores{2^{l-1}q}}{2^{l-1}}$.

\noindent
Therefore, $$
\probB{ U_{lq}  >\frac{ 1.5\ftwores{2^{l-1}q}}{ 2^{l-1}} }\le  \prob{U'_{lq} >
\expectb{U'_{lq}} +q} \le e^{-\lfloor d/2\rfloor}$$

\emph{Case 2:  $\ftwores{2^{l-1} q} = 0$.}  Then,  $U_{lq} = 0$. \hfill
\end{proof}

\begin{lemma} \label{lem:smallTU}
$\forall l \in \{0\} \cup [L], \smallH(l,C_l) , \smallH(l, \lceil B_l/(1-2\epsbar)^2 \rceil)$ and $ \allowbreak \smallU(l, C_l)$ hold simultaneously with probability $1-O(n^{-25})$.
\end{lemma}
\begin{proof}
From Lemma~\ref{lem:fres:Ulk}, $\smallH(l,C_l)$ and $\smallU(l,C_l)$ each holds with probability\\
 $e^{- \text{min}(\lfloor d/2 \rfloor\allowbreak, C_le^{-1/3}/2)}$. Similarly, $\smallH(l,\lceil B_l/(1-2\epsbar)^2 \rceil)$ holds with probability\\ $e^{-\min\left(\lfloor d/2 \rfloor,
 \left(B_l/(1-2\epsbar)^2\right) e^{-1/3}/2\right)}$.

From Lemma~\ref{lem:CL}, we have, $C_L \ge n^{\Omega(1)}$, and hence, $C_l \ge C_L \ge n^{\Omega(1)}$ for each $l\in \{0\} \cup [L]$.  Hence, $ d =O(\log n) =  o(C_L) = O(C_l)$ for each $l$. The failure probability is therefore $e^{-d/2}$, since $\epsbar = 1/(27p)$, $B_l = \epsbar^2 C_l = n^{\Omega(1)}$ and therefore,  $d = o(B_l)$, for each $l$.

Taking union bounds over the $O(\log n)$ values of $l$,   the three events hold simultaneously  except with probability $ (L+1) (3)e^{-d/2} \le (L+1)(3)e^{-50 (\log n)/2} = o(n^{-24})$.

\end{proof}

\subsection{ Proof that  {\sc smallres}, {\sc accuest}, {\sc goodl}, {\sc smallhh} hold with very high probability}

\begin{lemma} \label{lem:smallres1} Let $L = \lceil \log_{2\alpha} (n/C) \rceil $ and the hash functions $g_1, g_2, \ldots, g_L$ are drawn from $d$-wise independent family with $d = O(\log n)$ and even. Suppose $\smallH(l,C_l)$ and $\smallU(l,C_l)$ holds for each $l \in \{0\} \cup [L]$.
Then, \smallres~holds.
\end{lemma}
\begin{proof}
We first show that $\smallres_l \equiv \ftwores{2C_l,l} \le 1.5 \ftwores{ (2\alpha)^l C}/2^{l-1}$~is implied by $\smallH \allowbreak (l, C_l)$ and $\smallU\allowbreak(l,\allowbreak C_l)$.

 If $\smallH(l,C_l)$ holds, then, $H_{l,C_l} \le 2C_l$, that is,  $\sum_{1\le \rank(i) \le 2^{l}C_l} y_{il} \le 2C_l$. Hence,
\begin{align*}
\ftwores{2C_l,l} \le \sum_{\rank(i) > 2^l C_l} f_i^2y_{il} = U_{l, C_l} \le  \frac{1.5 \ftwores{2^{l-1}C_l}}{2^{l-1}}
\end{align*}where the last  inequality follows since  $\smallU(l, C_l)$ holds.

\noindent
Further, $2^{l-1}C_l = 2^{l-1} (4 \alpha^l C)  \ge 2(2\alpha)^{l} C$, since, $0 < \alpha < 1$. Thus, $$\ftwores{2C_l,l} \le\frac{ (1.5)\ftwores{2(2\alpha)^{l} C_l}}{2^{l-1}} \enspace .$$

\noindent
Hence $\smallres_l$ holds, for each $l \in \{0\} \cup [L]$, or equivalently,  $\smallres$ holds.
\end{proof}

\begin{lemma} \label{lem:accuest1}
$\goodest~ \wedge ~\smallres$ imply $\accuest$.
\end{lemma}
\begin{proof} Fix $i \in [n]$ and $l \in \{0\} \cup [L]$. By construction,
   $C_l = 4\alpha^l C$. Thus,
\begin{align*}
\abs{\hat{f}_{il} - f_i}^2\le  \frac{\ftwores{C_l, l}}{C_l}\le \frac{ 1.5 \ftwores{
2(2\alpha)^l C  }}{2^{l-1} (4\alpha^l C)}  \le \frac{ \ftwores{
(2\alpha)^l C  }}{ 2(2 \alpha)^l C}
\end{align*}
where the first step follows from \goodest~ and the second step follows from \smallres.  \hfill
\end{proof}

We now show that  the $\hh_L$ structure discovers all  items and their exact frequencies  that map to level $L$ (with high probability).
\begin{lemma} \label{lem:levelL} For $L = \lceil \log_{2\alpha} \frac{n}{C} \rceil$ and assuming $\smallH(L,C_L)$  and $\goodest_L$ holds, the frequencies of all the  items  in $\stream_L$ are discovered without error  using $\hh_L$. That is, $\smallH(L,C_L)  \wedge \goodest_L$ implies $\lastlevel$.
\end{lemma}
\begin{proof}

Let $L = \lceil \log_{2\alpha} (n/C) \rceil$. Then, $$2^L (C_L/2) =  2^L ( 4\alpha^L C /2 )= 2(2\alpha)^L  C \ge 2 (n/C) C = 2n \enspace . $$
By definition, $H_{L, C_L/2} = \sum_{1 \le \rank(i) \le 2^L (C_L/2)} y_{il}$ counts the number of items that map to level $L$ with ranks in $1,2, \ldots, 2^L (C_L/2) $. But $2^L(C_L/2) > n$. Hence, $H_{L,C_L/2}$ is the  number of items that map to level $l$. Since, $\smallH(L,C_L/2)$ holds, $H_{L,C_L/2} \le C_L$. Hence, $\ftwores{C_L,L}= 0$.  By $\goodest_L$,   $\abs{\hat{f}_{iL} - f_i} \le \left( \ftwores{ C_L,L}/C_L\right)^{1/2} = 0 \enspace .$ Thus  if $i \in \stream_L$ then $\hat{f}_{iL} = f_i$.
 \hfill
\end{proof}

\emph{Remark \irc} Lemma~\ref{lem:levelL} can be proved as an implication of the event \smallH$(l,C_L)$ by using an $\ell_2/\ell_1$-compressed sensing recovery procedure as in \cite{crt:ieeetit06a,donoho:tit06}.

\emph{Remark \irc} In the turnstile streaming model assumed, we say that $i$ appears in the stream iff $\abs{f_i} \ge 1$. By Lemma~\ref{lem:levelL}, the frequencies of all items are discovered exactly. Hence items with non-zero frequencies, that is, those with $\abs{f_i} \ge 1$ would satisfy $\abs{\hat{f}_{iL}}  = \abs{f_i} >  1/2 = Q_L$ and thus would qualify the criterion of being  discovered at level $L$.   All other items would satisfy $\abs{\hat{f}_{iL}} = 0$ and will not be discovered at level $L$.

At each level $l$, the algorithm finds the top-$C_l$ items by absolute values of estimated frequencies. A heavy-hitter at a level $l$ is however defined as an item whose estimated frequency crosses the threshold $Q_l$.  The event $\smallhh_l$ states that the heavy-hitters at a level $l$ are always among the top-$C_l$ items by absolute estimated frequencies.

\begin{lemma} \label{lem:Hl} Suppose  $\smallH(l, \lceil  B_l/(1-2\epsbar)^2 \rceil)$ holds for each $l \in \{0\} \cup [L-1]$ and suppose $\accuest$ holds. Then, \smallhh~holds.
\end{lemma}
\begin{proof} Let  $H'_l $ denote the set of items that are discovered as heavy-hitters at level $l$, that is, $H'_l = \{i \in \stream_l \mid \abs{\hat{f}_i } \ge Q_l \}$, where, $Q_l =  T_l(1-\epsbar)\} \enspace . $ By \accuest~and since $\epsbar = (B/C)^{1/2}$, we obtain $$\abs{\hat{f}_{il} - f_i}
\le \left(\frac{\ftwores{ (2\alpha)^l C }}{2(2\alpha)^l C}\right)^{1/2} \le \frac{\epsbar}{\sqrt{2}} \left(\frac{F_2}{(2\alpha)^l B}\right)^{1/2} \enspace . $$

\noindent
 Suppose  $i \in H'_l$. Then, $$ \abs{f_i} \ge Q_l - \frac{\epsbar}{\sqrt{2}} \left(\frac{F_2}{(2\alpha)^l B}\right)^{1/2} \ge T_l(1-\epsbar) - T_l (\epsbar/\sqrt{2})\ge  T_l\left(1-2\epsbar\right) \enspace . $$
 since, $T_l = (\hat{F}_2/((2\alpha)^l B))^{1/2} \ge (F_2/((2\alpha)^l B))^{1/2}$.

 \noindent
 Therefore,
\begin{align*} \rank(i) \le \frac{F_2}{\abs{f_i}^2 }\le \frac{F_2}{( T_l (1-2\epsbar))^2} = \frac{F_2  (2\alpha)^l B}{\hat{F}_2 (1-2\epsbar)^2} \le  \frac{2^l  B_l}{(1-2\epsbar)^2}
 \end{align*}
 Hence $H'_l \subset H_{lq}$, where we let $q =B_l/(1-2\epsbar)^2$.

\noindent
Since $\smallH(l, q)$ ~holds, $H_{lq} \le 2q$. Further, since, $H'_l \subset H_{lq}$,  therefore, $\abs{H'_l} \le 2q = 2 B_l/(1-2\epsbar)^2 \le C_l$, since, by choice of parameters, $\epsbar = (B_l/C_l)^{1/2} = 1/(27p)$ and $p \ge 1$.

\noindent
By construction, $H'_l$ is the set of items whose estimated frequencies are at least $Q_l$. Hence, $$H'_l = \hattopk(\abs{H'_l}) \subset \hattopk(C_l) \enspace .$$
\end{proof}

\subsection{Proof that \text{\sc nocollision}~holds with very high probability}

\begin{lemma} \label{lem:nc} If $t \ge 6$ and $s = \Theta(\log n)$, then, $\nocollision$
holds with probability at least $1 - n^{-150}$.
\end{lemma}
\begin{proof}
Assume full independence of hash functions. For $i \in \hattopk_l(C_l)$ and $l \in [2s]$, let
$w_{ijl} =1 $ if $i$ collides with some other item in $\hattopk_l(C_l)$ in the $j$th table of the
\tpest~structure at level $l$. Since, each table at level $l \in \{0\} \cup [L-1]$ has $16C_l$ buckets, therefore,
 $$q = \prob{w_{ijl}=1} =
1-\left(1-\frac{1}{16C_l}\right)^{C_l-1} \le 1/16 \enspace . $$ Let $W_{il} = \sum_{j=1}^{2s} (1-w_{ijl})$ be the number of tables where $i$
 does not collide with any other item of $\hattopk_l(C_l)$. Then, $\expectb{W_{il}} \ge (1-q)
 (2s) \ge (15/8)s$. By Chernoff's bounds,
 \begin{align*}
 \prob{W_{il} \ge s} & \ge 1- \exp{-(15/8)s(7/15)^2/2} \ge 1- e^{-0.2s} \\ &  = 1 - e^{-(0.2)(8)(100)\log (n)}  = 1- n^{-160}
 \end{align*}
 since, $s = 8k = 8 (100 \log (n))$.

 \noindent By
 union bound,  $$\prob{\forall i \in \hattopk_l(C_l) \left(W_{il} \ge  s\right)} \ge 1 -
 C_l e^{-0.2s} \ge 1 - n^{-150}\enspace. $$

 Assuming $t$-wise independence of the hash family from which the $h_{lj}$'s are drawn, denote  $q'_t= \probsubb{t}{w_{ijl}=1}$, where the subscript $t$ denotes $t$-wise independence. Let $u_{ikjl} = 1$  if $i$ and $k$  collide under hash function $h_{lj}$ for the $j$th hash table in the structure $\tpest_l$. Let $S_{li} = \overline{\topk}(C_l) \setminus \{i\}$. Then, by inclusion-exclusion,
 \begin{align} \label{eq:nocoll:q}1-q & =\probb{ w_{ijl}=0}  = 1- \probb{w_{ijl}=1}
 = 1 - \prob{\bigvee_{k\in S_{li}} (u_{ikjl} = 1)} \notag \\ &  = 1 - \sum_{r=1}^{\abs{S_{li}}} (-1)^{r-1} \sum_{\substack{\{k_1, k_2, \ldots, k_r\} \subset S_{li}}} \probb{u_{ik_1jl}=1, u_{ik_2jl} = 1, \ldots, u_{ik_rjl} =1} \\
 1- q'_t &= \probsub{t}{ w_{ijl}=0} \notag \\ & =  1 - \sum_{r=1}^{\abs{S_{li}}} (-1)^{r-1} \sum_{\substack{\{k_1, k_2, \ldots, k_r\} \subset S_{li}}} \probsubb{t}{u_{ik_1jl}=1, u_{ik_2jl} = 1, \ldots,  u_{ik_rjl} =1} \label{eq:nocoll:qprime}
  \end{align}
  Further, the sum of the tail starting from position $t+1$ to $\abs{S_{li}}$ is, in absolute value, dominated by the $t$th term. Therefore, from ~\eqref{eq:nocoll:q}, we have,
    \begin{multline}
    \card{q-\sum_{r=1}^{t-1} (-1)^{r-1} \sum_{\substack{\{k_1, k_2, \ldots, k_r\} \subset S_{li}}} \probb{u_{ik_1jl}=1, u_{ik_2jl} = 1, \ldots,  u_{ik_rjl} =1}} \\ \le
    \sum_{\{k_1, k_2, \ldots, k_t\} \subset S_{li}} \probb{u_{ik_1jl}=1, u_{ik_2jl} = 1, \ldots,  u_{ik_rjl} =1} \label{eq:nocoll:qa}
    \end{multline}
  Similarly from ~\eqref{eq:nocoll:qprime}, we have,
  \begin{multline}
  \card{q'-\sum_{r=1}^{t-1} (-1)^{r-1} \sum_{\{k_1,k_2, \ldots ,k_r\} \subset S_{li}} \probsubb{t}{u_{ik_1jl}=1, u_{ik_2jl} = 1, \ldots,  u_{ik_rjl} =1}} \\ \le
    \sum_{\substack{\{k_1, k_2, \ldots, k_t\} \subset S_{li}}}  \probsubb{t}{u_{ik_1jl}=1, u_{ik_2jl} = 1, \ldots,  u_{ik_rjl} =1} \label{eq:nocoll:qprimea}
    \end{multline}
By $t$-wise independence, the probability terms in the above expression are identical for $r=1, \ldots, t$, that is, for any  $1  \le k_1 < k_2 < \ldots < k_r \le n$ and $2 \le r \le t$. \begin{align*}&  \probsubb{t}{u_{ik_1jl}=1, u_{ik_2jl} = 1, \ldots,  u_{ik_rjl} =1}\\ & = \probb{u_{ik_1jl}=1, u_{ik_2jl} = 1, \ldots,  u_{ik_rjl} =1}\end{align*}

Therefore, by triangle inequality,
\begin{align}\label{eq:nocoll:qqprime}
\abs{ q - q'} \le 2\sum_{j_1<j_2< \ldots < j_t} \probsub{t}{u_{ik_1jl}=1, u_{ik_2jl} = 1, u_{ik_tjl} =1}
\end{align}
Since there are $16C_l$ buckets in the \tpest~structure at level $l$, we have,  $ \probb{u_{ik_rjl} \allowbreak =1} = 1/(16C_l)$.  Substituting in ~\eqref{eq:nocoll:qqprime},
 \begin{align*}\card{q-q'_t} &\le 2\sum_{j_1<j_2< \ldots < j_t} \probsubb{t}{u_{ik_1jl}=1, u_{ik_2jl} = 1, u_{ik_tjl} =1}\\ & = 2\binom{\abs{S_{li}}}{t}\left(\frac{1}{16C_l}\right)^{t}   \le  2 \binom{(C_1-1)}{t} (16 C_l )^{-t}
   \le2 \left( \frac{ C_le}{16 C_l t} \right)^t \le 2 \left( \frac{e}{16 t} \right)^t
 \end{align*}
since, $\abs{S_{li}} = C_l-1$.
For $t \ge 6$, $\abs{q-q'_t} \le 2(32)^{-6} \le 2^{-29}$.

The above Chernoff's bound argument may be repeated using probability of success  $1-q'_t \ge 1-q-2^{-29}$,
 instead of $1-q$. Hence, $\nocollision(H)$ holds except with
 probability  $n^{-150}$ by calculations similar to the previous one.
  \end{proof}
\subsection{Proof that $\G$ holds with very high probability}
\begin{relemma}[Restatement of Lemma~\ref{lem:hss:G}]
$\prob{\G} \ge 1-O(n^{-24})$.
\end{relemma}
\begin{proof} By adding the failure probabilities of all the events comprising $\G$ using
 Lemmas~\ref{lem:goodest} through ~\ref{lem:smallTU}, the statement of the lemma follows.
\end{proof}

\subsection{Technical fact}

The following fact gives a bound on the difference between  the unconditional probability of an event $E$ and its probability  conditioned on an event $F$. It essentially shows that if $\prob{E} = 1/n^{O(1)}$, its probability is not significantly altered if it is conditioned by a very high probability event $F$, that is, $ \prob{F} = 1 - n^{-\Omega(1)}$.
 \begin{fact} \label{fact:cndhpev} Let $E$ and $F$ be a pair of events such that $\prob{F} >0$.  Then,  $ \card{\prob{E \mid F} - \prob{E}} \le 1-\prob{F}$.
 \end{fact}
\eat{Fact~\ref{fact:cndhpev} is used to closely estimate the unknown probability of an event conditional on a very high probability event such as $\G$  from the (known) probability of the (unconditioned) event. For example, $\prob{i \in \stream_l, j \in \stream_r \mid \G} = \prob{i \in \stream_l, j \in \stream_r} \pm O(n^{-11}) = 2^{-l-r} \pm O(n^{-11})$.}
\begin{proof} [Proof of Fact~\ref{fact:cndhpev}.] If $\prob{F}=1$, then $\prob{E, F} = \prob{E \cup F} - \prob{E} - \prob{F} = 1-\prob{E} - 1 = \prob{E}$, and hence the statement holds. Otherwise,
\begin{align*}
 \prob{E} =
\prob{E \mid F} \prob{F} + \prob{E \mid \neg F} \prob{\neg F}
\end{align*}
Subtracting $ \prob{E\mid F}$ from both sides yields,
\begin{align*}
\prob{E} - \prob{E \mid F}  & = \prob{E \mid F}(\prob{F} -1)+ \prob{E\mid \neg F} \prob{\neg F} \\ &= \left(- \prob{E \mid F} + \prob{E \mid \neg F}\right) \prob{\neg F}
\end{align*}
Taking absolute values and noting that $\abs{ -\prob{E \mid F} + \prob{E\mid \neg F}} \le 1$, we have,
$\card{\prob{E} - \prob{E \mid F}} \le \prob{\neg F}$.
\hfill

\end{proof}
The fact is used by letting $F = \G$. Then, for any event $E$, $\card{\prob{E \mid \G} -
\prob{E}}\le \prob{\neg G} = O(n^{-24})$, by Lemma~\ref{lem:hss:G}.

\section{Basic Sampling Properties of \ghss~Algorithm}

\eat{We reiterate the definitions of the groups, sampled groups and the partition of the sampled groups in to left margin, mid-region and right-margin respectively.

 Items are divided into groups according to their frequencies, as follows.
\begin{align*}
 G_0 &= \{ i: \abs{f_i} \ge T_0\}, \\G_l  &= \{ i: T_l \le \abs{f_i} < T_{l-1} \}, ~~l =1,2, \ldots, L-1, \text{ and } \\
 G_L &  = \{i:  \abs{f_i} < T_{L-1}\} \enspace .
 \end{align*}
Each group $G_l$ is  partitioned as per frequency ranges into  $\lmargin(G_l), \midreg(G_l)$ and $\rmargin(G_l)$.
 \begin{align*}
 \lmargin(G_l) &= \{i: T_l \le \abs{f_i}  <  T_l +  T_l\epsbar \}, \\
 \rmargin(G_l) &= \{ i :  T_{l-1} - 2T_{l-1}\epsbar \le \abs{f_i} <  T_{l-1} \}, \\
 \midreg(G_l) & = \{ i: T_l+ T_{l}\epsbar \le \abs{f_i}  < T_{l-1} - 2T_{l-1}\epsbar\} \enspace
 .
 \end{align*}
 }

\emph{Preliminaries.} The following lemma argues that  the frequency ranges defining  $\lmargin, \midreg$ and $\rmargin$ are non-empty intervals.
 \begin{lemma} \label{lem:nereg}
 For $p \ge 2$ and for each $l \in \{0\} \cup [L]$, the frequency ranges that define $\lmargin(G_l), \midreg\allowbreak (G_l)$ and $\rmargin(G_l)$ are non-empty intervals.
 \end{lemma}
 \begin{proof} The statement of the lemma is obviously true from the definitions for $\lmargin(G_l)$ and $\rmargin(G_l)$.

 For $\midreg(G_l)$, the interval range is $[T_l(1+\epsbar), T_{l-1}(1-2\epsbar))$. This range is non-empty iff  $T_{l-1} (1-2\epsbar) > T_l(1+\epsbar)$, or, $T_{l-1}/T_l > (1+\epsbar)/(1-2\epsbar) $, or, $(2\alpha)^{1/2} > (1+ 1/(27p))/(1-2/(27p))$, which is true for $\alpha = 1-2(0.01)/p \ge 0.99$.
\end{proof}

Our analysis is conditioned on $\G$. Assuming $\G$ holds, the event \accuest~holds, and therefore, the frequency estimation error by the $\hh_l$~structure is bounded as follows.
\begin{align} \label{eq:esterr}
\abs{\hat{f}_{il} - f_i} \le \left( \frac{\ftwores{ (2\alpha)^l C }}{(2\alpha)^l C}\right)^{1/2} \le \left( \cfrac{ \hat{F}_2}{(2\alpha)^l C}\right)^{1/2} = \epsbar T_l \enspace .
\end{align}

We first prove a property about the relation between the level at which an item is discovered and the group $G_l$ to which an item belongs.  This property is then used to a relation between the probabilities  with which an item may belong to different sampled groups.

\subsection{Properties concerning levels at which an item is discovered}
\begin{lemma} \label{lem:discovery} 
The following properties hold conditional on $\G$.
\begin{enumerate}
\item Suppose $i \in \lmargin(G_l)$ for some $0 \le l \le L-1$. Then, (a) $ \prob{l_d(i) \le l-1 \mid \G} =0$, and (b) the event $\{l_d(i)=l,\G\} \equiv \{i \in \stream_l, \G\}$.
\item
 Suppose $i \in  \midreg(G_l)$  for some $0 \le l \le L$.  Then, (a) $\prob{l_d(i) \le l-1\mid \G} = 0$, (b) the event $\{l_d(i)=l,\G\} \equiv \{i \in \stream_l, \G\}$, and, (c) $\prob{\hat{f}_{il} \ge T_l \mid i \in \stream_l,\G} = 1$.
\item Suppose $i \in  \rmargin(G_l)$ for some $2\le l \le L$. Then,  (a)  $\prob{l_d(i) \le l-2 \mid \G}=0$, (b) $\{i \in \stream_l, \G\}$ implies $\{\abs{\hat{f}_{il}} \ge T_l \}$ , and  (c)
$\prob{l_d(i)=l \mid l_d(i) \ne l-1, \G} = \prob{ i \in \stream_l \mid l_d(i) \ne l-1, \G}$.
\end{enumerate}
\end{lemma}

\begin{proof} [Proof of Lemma~\ref{lem:discovery}.]
Since, $\accuest$ holds as a sub-event  of $\G$, we have, $\abs{\hat{f}_{il} - f_i} \le \epsbar T_l$, by Eqn. ~\eqref{eq:esterr}. Also,
$Q_l = T_l(1-\epsbar)$. All statements below are conditional on $\G$.

~\emph{Case: $i \in \lmargin(G_l) \cup \midreg(G_l)$, $l \ge 1$. } Then,  $T_l + \epsbar T_l \le \abs{f_i} <T_{l-1} - 2\epsbar T_{l-1}$.
Therefore for $ r\le l-1$,
$$ \abs{\hat{f}_{ir}} \le \abs{f_i} + \epsbar T_r < T_{l-1} - 2\epsbar T_{l-1} + \epsbar T_r \le T_{r}- 2\epsbar T_r + \epsbar T_r = T_r - \epsbar T_r = Q_r \enspace .
$$
Hence, $\prob{l_d(i) \le l-1\mid \G} =0$.

Further, if $i \in \stream_l$ and $i \in \lmargin(G_l) \cup \midreg(G_l)$,  then,
$
\abs{\hat{f}_{il}} \ge \abs{f_i} - \epsbar T_l \ge T_l - \epsbar T_{l} =  Q_{l}
$
and so $i$ is discovered at level $l$, if $i$ has not been discovered at an earlier level. However, part(a) states that $i$ cannot be discovered at levels $l-1$ or less. Hence $i$ is discovered at level $l$.  Thus, conditional upon $\G$, if $i \in \stream_l$, then, $l_d(i)=l$. Conversely, if $i \not\in \stream_l$, then $l_d(i) \ne l$. Hence, the events $\{i \in \stream_l\}$ and $\{l_d(i)=l\}$ are equivalent, conditional on $\G$. This proves  parts 1(b) and 2(b).

\emph{Case: $i \in \midreg(G_l)$.} If $i \in \stream_l$,
then, $\abs{\hat{f}_{il}} \ge \abs{f_i} - \epsbar T_l  \ge T_l +\epsbar T_l - \epsbar T_l = T_l$. This proves part 2(c).

\emph{Case: $i \in \rmargin(G_l)$.}  Then, $\abs{f_i} < T_{l-1}$. Let $r \le l-2$. Then,
$$\abs{\hat{f}_{ir}} \le \abs{f_i} + \epsbar T_r < T_{l-1} + \epsbar T_r < T_r - \epsbar T_r = Q_r$$ where, the last inequality $T_{l-1} + \epsbar T_r < T_r - \epsbar T_r$ follows since, it is equivalent to $\frac{T_{l-1}}{ T_{l-2}} <(1-2\epsbar)$, which holds since, $\frac{T_{l-1}}{ T_{l-2}} = \frac{1}{\sqrt{2\alpha}} \le (0.72) $ and $(1-2\epsbar) = 1 - \frac{2}{27p} \ge 0.96$.
Hence $\prob{l_d(i) \le l-2 \mid G} = 0$.

 We are given that  $i \in \rmargin(G_l)$. Suppose that $i \in \stream_l$. Then,  \begin{align} \label{eq:disc:rmb}\abs{\hat{f}_{il}} \ge T_{l-1} - 2 \epsbar T_{l-1} -\epsbar T_l  & = T_l( 2\alpha)^{1/2} -2 (2\alpha)^{1/2} \epsbar  T_l -\epsbar T_l  \notag \\
&\ge T_l\left( 1.40 (1-(2)(0.04) )- (0.04)\right)  = 1.248T_l > T_l
\end{align}
Hence, $\prob{\abs{\hat{f}_{il}} \ge T_l \mid i \in \stream_l, G} = 1 $, and
therefore, by Eqn.~\eqref{eq:disc:rmb}
$
\prob{l_d(i) \in \{l-1,l\} \mid i \in \stream_l,G} = 1$.  This proves part 2(b).

Since, $l_d(i) > l$ implies $i \in \stream_l$, we have,
\begin{align*} \prob{l_d(i) > l \mid \G} & = \prob{l_d(i) > l, i \in \stream_l\mid \G} \\
& = \prob{l_d(i) > l \mid i \in \stream_l, \G}\cdot  \prob{ i \in \stream_l\mid \G} \\
& \le  \left( 1- \prob{l_d(i) \in \{l-1,l\} \mid i \in \stream_l, \G} \right) \cdot \prob{ i \in \stream_l \mid \G} \\
& = 0 \enspace .
\end{align*}
Hence,
\begin{align} \label{eq:disc:rm2}
\prob{l_d(i) \ne l-1\mid  \G} &  = \prob{l_d(i) \le l-2 \mid \G} + \prob{l_d(i)  = l \mid \G} + \prob{l_d(i) >l \mid  \G}  \notag\\
& = 0 + \prob{l_d(i) = l \mid \G} + 0 \enspace .
\end{align}
It follows that,
\begin{align*}
\prob{l_d(i) = l \mid  l_d(i) \ne l-1,  \G} & = \frac{\prob{l_d(i) =l \mid  \G}}{\prob{l_d(i) \ne l-1\mid  \G}} = 1
\end{align*}
by Eqn.~\eqref{eq:disc:rm2}.
\end{proof}

\subsection{Probability of items belonging to sampled groups}

\begin{relemma}[Re-statement of Lemma~\ref{lem:margin}.]
Let $i \in G_l$.
\begin{enumerate}
\item Suppose  $i \in \midreg(G_l)$. Then, (a) the event $\{i \in \bar{G}_l,\G\} \equiv \{i \in \stream_l, \G\}$, (b)
 $2^l\prob{ i \in \bar{G}_l \mid \G} = 1 \pm 2^ln^{-c}$,  and, (c) $\prob{i \in \cup_{l' \ne l} \bar{G}_{l'} \mid \G} = 0$.

\item Suppose $i \in \lmargin(G_l)$. Then, (a) $\prob{i \in \cup_{l' \ne \{l,l+1\}} \bar{G}_{l'}} = 0$, (b) the event $\{i \in \bar{G}_l \cup \bar{G}_{l+1}, \G\} \equiv \{i \in \stream_l,\G\}$, and   (c)
$
 2^{l+1} \prob{i \in \bar{G}_{l+1} \mid \G} + 2^l\prob{i \in \bar{G}_l \mid \G}=
 1 \pm  2^l n^{-c} $.
\item Suppose  $i \in \rmargin(G_l)$. Then, (a) $\prob{i \in \cup_{l' \ne \{l-1,l\}} \bar{G}_{l'}} =0$, (b) the events $\{i \in \bar{G}_{l-1} \cup \bar{G}_{l}\} \subset \{i \in \stream_l\}$, (c) $\{i\in \stream_l, l_d(i) \ne l-1\} \subset \{i \in \bar{G}_l\}$ ,  and, (d)
 $
  2^l \prob{ i \in \bar{G}_l \mid \G} + 2^{l-1}\prob{ i \in \bar{G}_{l-1} \mid \G} = 1 \pm  O(2^l n^{-c}) \enspace . $
 \end{enumerate}
\end{relemma}

\begin{proof}[Proof of Lemma~\ref{lem:margin}.]
Assume $\G$ holds for the arguments in this proof.   Suppose $i \in \stream_l$. Then  $\abs{\hat{f}_{il}- f_i} \le T_l \epsbar$.

\emph{Case: $i \in \midreg(G_l)$.} \emph{Part 1 (b).} Since $i \in \midreg(G_l)$, $\abs{f_i} \ge T_l + \epsbar T_l$. Conditional on $\G$, $\accuest$ holds, and therefore,
$$ \abs{\hat{f}_{il}} \ge \abs{f_i} - \epsbar T_l \ge T_l + \epsbar T_l - \epsbar T_l = T_l \enspace . $$
Therefore,
\begin{equation} \label{eq:margin:mid3}\{i \in \stream_l, \G\} \subset \{\abs{\hat{f}_{il}} \ge T_l, \G\}
\end{equation}
Then,
\begin{align*}
\{i \in \bar{G}_l, \G\} &  \equiv \{l_d(i) =  l, \abs{\hat{f}_{il}} \ge T_l, \G \} \notag\\
 & \equiv \{i \in \stream_l, \abs{\hat{f}_{il}} \ge T_l, \G\},  ~ \text{ since, $\{l_d(i) =l,\G\} \equiv \{i \in \stream_l,\G\}$, Lemma~\ref{lem:discovery}, (2b)}\\
 & \equiv \{ i \in \stream_l, \G\}, ~~~~~~~~~~~~~~ \text{ by Eqn.~\eqref{eq:margin:mid3}.}\\
 \end{align*}
 This proves part 1 (b).

 \emph{Part 1 (a).}
\begin{align} \label{eq:margin:mid1} \prob{i \in \bar{G}_l \mid \G} & = \prob{l_d(i) = l, \abs{\hat{f}_{il}} \ge T_l \mid \G} + \prob{ l_d(i)=l-1, Q_l \le \abs{\hat{f}_{i,l-1}} < T_l, K_i=1 \mid \G}
\end{align}
Denote by $\event{E}_1$ the  event $ l_d(i)=l-1, Q_l \le \abs{\hat{f}_{i,l-1}} < T_l$ and by $\event{E}_2 $ the event $ Q_l\le \abs{\hat{f}_{i,l-1}} < T_l$. Then,
\begin{align*}
& \prob{ \E_1,K_i=1 \mid \G} = \prob{K_i =1 \mid \E_1,\G} \cdot  \prob{\E_2 \mid l_d(i)=l-1, \G} \cdot \prob{l_d(i)=l-1 \mid \G}  = 0
\end{align*}
since, $\prob{l_d(i)=l-1 \mid \G} = 0$, by Lemma~\ref{lem:discovery}, part (2a). Substituting in Eqn.~\eqref{eq:margin:mid1}, we have,
\begin{align*}
\prob{i \in \bar{G}_l \mid \G} & = \prob{l_d(i) = l, \abs{\hat{f}_{il}} \ge T_l \mid \G} \\
 & = \prob{ i \in \stream_l, \abs{\hat{f}_{il}} \ge T_l\mid \G}, \text{ since, $\{l_d(i)=l,\G\} \equiv \{i \in \stream_l,\G\}$, Lemma~\ref{lem:discovery}, (2b)} \notag \\
&=  \prob{i \in \stream_l \mid \G}, ~~~~~~~~\text{ by  part 1 (a)} \notag \\ 
 & = 2^{-l} \pm n^{-c}, ~~~~~~~~~~~~~\text{ by Fact~\ref{fact:cndhpev}.} 
\end{align*}
Multiplying by $2^l$ and  transposing, we have $
2^l\textsf{Pr}\bigl[ i \in \bar{G}_l \mid \G\bigr]  \in 1 \pm n^{-c}\cdot 2^l $, as claimed in part 1(a).

\emph{Part 1(c).}  We have by \accuest~that for any $0 \le r \le l-1$,
\begin{align*}
\abs{\hat{f}_{i,r}} < T_{l-1} - 2\epsbar T_{l-1} + \epsbar T_r \le T_{l-1} \left( 1 -\epsbar \right) = Q_{l-1}
\end{align*}
Hence, $i$ cannot be in $\bar{G}_{r}$ for any $r \le l-1$.
We have by part (1a) that $\{i \in \bar{G}_l,\G\} \equiv \{i \in \stream_l, \G\}$.

Let $i \in \bar{G}_r$ for some $r \ge l+1$. Since, for $i $ to belong to $ \bar{G}_r$, $i $ must be in $\stream_{r-1} $ and hence by the sub-sampling procedure, $i \in \stream_{l}$. By part 1(a), $\{i \in \bar{G}_l,\G\} \equiv \{i \in \stream_l, \G\}$, and therefore, $i \in \bar{G}_l$. Hence, $i\not\in \bar{G}_r$, for any $r \ge l+1$. Thus,
\begin{align*}
\prob{i \in \cup_{r\ne l} \bar{G}_r \mid \G} = 0 \enspace .
\end{align*}

\emph{Case: $i \in \lmargin(G_l)$.} From Lemma~\ref{lem:discovery}, $l_d(i) \nless l$ and $l_d(i) = l $ iff $i \in \stream_l$. Since $l_d(i) \nless l$, $i \not\in \bar{G}_r$, for any $r < l$. Consider $r > l+1$. If $i \in \bar{G}_r$, then, $l_d(i) \ge r-1 \ge l+1$. Since, $i \in \stream_{l_d(i)}$, and $l_d(i) \ge l+1$, it follows that $i \in \stream_l$, by the sub-sampling procedure.  However, by Lemma~\ref{lem:discovery}, part (1b), $ \{l_d(i)=l, \G\} \equiv \{i \in \stream_l, \G\}$. Hence, in this case, $l_d(i)=l$, contradicting the implication that $l_d(i) \ge l+1$. Thus,

\begin{align*}
\prob{ i \in \cup_{l' \not\in \{l,l+1\}} \bar{G}_{l'}} = 0
\end{align*}
proving part 2 (a).

Suppose $i \in \stream_l$.  Then, $l_d(i) = l$ and $\hat{f}_{i} = \hat{f}_{il}$.
By construction,  
\begin{equation} \label{eq:hss:margin:lm1:pil}
 \prob{i \in \bar{G}_l \mid i \in \stream_l, \G}  = \prob{\abs{\hat{f}_{il} } \ge T_l \mid i \in \stream_l,\G} = p_{il} ~~~(\text{say.})
\end{equation}
Further,
\begin{align} \label{eq:conda}
\prob{i \in \bar{G}_{l+1} \mid i \in \stream_l, \G}
 &= \prob{Q_l \le \abs{\hat{f}_{il} } <  T_l, K_i = 1 \mid i \in \stream_l, \G} \notag  \\
 & ~~+ \prob{\abs{\hat{f}_{il} } < Q_l, i \in \stream_{l+1}, \abs{\hat{f}_{i,l+1}} \ge T_{l+1} \mid i \in \stream_l, \G}
\end{align}
However, conditional on $\G$ and $i \in \stream_l$, by Lemma~\ref{lem:discovery}, $\abs{\hat{f}_{il}} \ge Q_l$. Hence, the second probability in the \emph{RHS} of Eqn.~\eqref{eq:conda} is 0.
Therefore,
\begin{align}  \label{eq:hss:margin:lm2}
 &\prob{ i \in \bar{G}_{l+1} \mid i \in \stream_{l},\G} \notag \\
  &= \prob{ Q_l \le \abs{\hat{f}_{il}} < T_l,  K_i = 1 \mid i \in \stream_l,\G} \notag \\ & =
\prob{K_i = 1 \mid Q_l \le \abs{\hat{f}_{il}} < T_l, i \in \stream_l, \G} \cdot \prob{ Q_l \le  \abs{\hat{f}_{il}} < T_l \mid i \in \stream_l, \G} \notag \\ & = (1/2)\left( 1 - p_{il} \right)
\end{align}
since,  (a)  $K_i$ is independent of all other random bits, and, (b) $\probb{ Q_l \le  \abs{\hat{f}_{il}} < T_l \mid \allowbreak  i \in \stream_l, \G} + \prob{\abs{\hat{f}_{il} } \ge T_l \mid i \in \stream_l,\G} = \prob{\abs{\hat{f}_{il}} \ge Q_l \mid i \in \stream_l, \G} = 1$.

Eliminating $p_{il}$ using ~\eqref{eq:hss:margin:lm1:pil} and ~\eqref{eq:hss:margin:lm2},  we have,
 \begin{align} \label{eq:hss:margin:lm2a}
 2 \prob{i \in \bar{G}_{l+1} \mid i \in \stream_l,\G} + \prob{i \in \bar{G}_l \mid i \in
 \stream_l,\G} = 1 \enspace .
 \end{align}

 Multiplying Eqn.~\eqref{eq:hss:margin:lm2a} by $ \prob{i \in \stream_l\mid \G}$, we have,
\begin{equation} \label{eq:hss:margin:lm2d}
2 \prob{i \in \bar{G}_{l+1}, i \in \stream_l\mid \G}
 + \prob{i \in \bar{G}_l, i \in \stream_l \mid \G} = \prob{i \in \stream_l\mid \G} \enspace .
 \end{equation}
By Lemma~\ref{lem:discovery}, if $i \in \lmargin(G_l)$, then, $l_d(i) \nless l$ and $l_d(i) = l$ (or, $\abs{\hat{f}_{il}} \ge Q_l$) iff $i \in \stream_l$. By construction therefore, $(i \in \bar{G}_l $ or $ i \in \bar{G}_{l+1})$ iff $i \in \stream_l$.  This proves part 2(b).

Thus, $i \in \bar{G}_{l+1}$ implies $i \in \stream_l$ and $i \in \bar{G}_l$ also implies that $i \in \bar{G}_l$. Hence, Eqn.~\eqref{eq:hss:margin:lm2d} can be written as
\begin{equation}\label{eq:hss:margin:lm2e}
2 \prob{i \in \bar{G}_{l+1}\mid \G}
 + \prob{i \in \bar{G}_l\mid \G} = \prob{i \in \stream_l\mid \G} = 2^{-l} \pm n^{-c}
 \end{equation}
using Fact~\eqref{fact:cndhpev}.
 Multiplying by $2^l$ gives part 2(c)  of the lemma.

\emph{Case: $i \in \rmargin(G_l)$.} Assume that $\G$ holds.  By Lemma~\ref{lem:discovery},  $l_d(i) \in \{l-1, l\}$ but $l_d(i) \nless l-1$ and $l_d(i) \ngeq l+1$. Since, $l_d(i) \nless l-1$, it follows that $i \not\in \bar{G}_r$ for any $ r < l-1$.
If  $i \in \stream_l$, we have, \begin{align*}\abs{\hat{f}_{il}} \ge \abs{f_{i}} - \epsbar T_l \ge T_{l-1} - 2\epsbar T_{l-1} - \epsbar T_l & = T_l\left( (2\alpha)^{1/2} - \epsbar(2(2\alpha)^{1/2} +1)\right)  \ge (1.3)T_l  > T_l \end{align*}
by the choice of parameters $\alpha$ and $\epsbar = 1/(27p)$. Hence, if $i \not\in \bar{G}_{l-1}$ and $i \in \stream_l$, then, $i \in \bar{G}_l$. In other words,
$$\prob{i \in \bar{G}_l \mid i \not\in \bar{G}_{l-1}, i \in \stream_l, \G} = 1 \enspace . $$
If $i \in \bar{G}_r$ for some $r \ge l+1$, then, $ i \in \stream_l$ and this implies that $i \in \bar{G}_l$, which is a contradiction. Hence,
$$ \prob{ i \in \cup_{r\not\in \{l-1,l\}}\bar{G}_r\mid \G} = 0 \enspace . $$

  By construction, we have,
\begin{align} \label{eq:hss:margin:rm1a}
\prob{ i \in \bar{G}_{l-1} \mid i \in \stream_{l-1}, \G}  &= \prob{ \abs{\hat{f}_{i,l-1}} \ge
T_{l-1} \mid i \in \stream_{l-1}, \G} = p_{i,l-1} \text{~(say)}\\
 \probb{ i \in \bar{G}_{l} \mid i \in \stream_{l-1},\G} & = \prob{  Q_{l-1} \le \abs{\hat{f}_{i,l-1}}< T_{l-1} \text{ and } K_i = 1 \mid i \in \stream_{l-1},\G} \notag \\
 & + \prob{ \abs{\hat{f}_{i,l-1}} < Q_{l-1}, i \in \stream_l, \abs{\hat{f}_{il}} \ge T_l \mid i \in \stream_{l-1}, \G}\notag \\
 & = A+B\label{eq:hss:margin:rm1}
 \end{align}
where, we let $A$ and $B$  denote the probability expressions in the first and second  terms in the \emph{RHS} respectively of Eqn.~\eqref{eq:hss:margin:rm1}. Then,
\begin{align} \label{eq:hss:margin:rmA1}
A
& = \prob{  Q_{l-1} \le \abs{\hat{f}_{i,l-1}}< T_{l-1},  K_i = 1 \mid i \in \stream_{l-1},\G}  \notag \\
&= \prob{K_i = 1 \mid Q_{l-1} \le \abs{\hat{f}_{i,l-1}}< T_{l-1}, i \in \stream_{l-1},\G}\cdot \prob{ Q_{l-1} \le \abs{\hat{f}_{i,l-1}}< T_{l-1} \mid i \in \stream_{l-1},\G}  \notag \\
& = (1/2) \prob{ Q_{l-1} \le \abs{\hat{f}_{i,l-1}}< T_{l-1} \mid i \in \stream_{l-1},\G}
\end{align}

Therefore, for $i \in \rmargin(G_l)$,  $i$ could possibly be a member of   $\bar{G}_{l-1}$ which can happen only if $i \in \stream_{l-1}$. However,  if  $i \not\in \bar{G}_{l-1}$ and  $i \in \stream_{l-1}$, then $i$ can possibly be a member of $\bar{G}_l$. This can happen in two ways, either  (i) $Q_{l-1} \le \abs{\hat{f}_{i,l-1}} < T_{l-1}$ and the coin toss $K_i =1$, or,  (ii) $Q_{l-1} > \abs{\hat{f}_{i,l-1}}$ and $i \in \stream_l$ and $\abs{\hat{f}_{il}} \ge T_l$.  In the latter case, if $ i \in \stream_l$, then, $\abs{\hat{f}_{il}} $ is at least $T_l$ with probability 1,  conditional on $\G$. This follows from Lemma~\ref{lem:discovery}, part (2).    In particular, $i \not\in \bar{G}_{l'}$ for any $l' \not\in \{l-1,l\}$.

Hence,
\begin{align*}
B  & = \prob{ \abs{\hat{f}_{i,l-1}} < Q_{l-1}, i \in \stream_l, \abs{\hat{f}_{il}} \ge T_l \mid i \in \stream_{l-1}, \G} \\
& = \prob{ \abs{\hat{f}_{i,l-1}} < Q_{l-1}, i \in \stream_l \mid i \in \stream_{l-1}, \G} \\
 & = \prob{ \abs{\hat{f}_{i,l-1}} < Q_{l-1} \mid i \in \stream_l, \G} \cdot \prob{ i \in \stream_l \mid  i \in \stream_{l-1},\G} \\
 & = \prob{ \abs{\hat{f}_{i,l-1}} < Q_{l-1} \mid i \in \stream_l, \G} \cdot \left( 1/2 \pm n^{-c} \right)
\end{align*}

Note that $\probb{\abs{\hat{f}_{i,l-1}} < Q_{l-1} \mid i \in \stream_l} =
\probb{\abs{\hat{f}_{i,l-1}} < Q_{l-1} \mid i \in \stream_{l-1}}$ for the following reason.
$\abs{\hat{f}_{i,l-1}}$ is a function of the frequencies of the items that conflict with $i$ in
the set of hash buckets to which $i$ maps in the $\hh_{l-1}$ structure. By construction of the
hash function, whether  $i$ maps to the next level $l$ depends on whether $g_l(i) = 1$, which is independent of the hash functions $g_1, g_2, \ldots, g_{l-1}$.
Hence,  $$\probb{\abs{\hat{f}_{i,l-1}} < Q_{l-1} \mid i \in \stream_l} =
\probb{\abs{\hat{f}_{i,l-1}} < Q_{l-1} \mid i \in \stream_{l-1}}$$ Using Fact~\eqref{fact:cndhpev}, we have,
$$\prob{\abs{\hat{f}_{i,l-1}} < Q_{l-1} \mid i \in \stream_l,\G}  =
\prob{\abs{\hat{f}_{i,l-1}} < Q_{l-1} \mid i \in \stream_{l-1},\G} \pm
n^{-c} \enspace .$$

Thus Eqn.~\eqref{eq:hss:margin:rm1}  may be written as
\begin{align} \label{eq:hss:margin:rm2}
 &\probb{ i \in \bar{G}_{l} \mid i \in \stream_{l-1},\G}   = A+B \notag \\
&=  (1/2)\probb{ Q_{l-1} \le \abs{\hat{f}_{i,l-1}}< T_{l-1} \mid i \in
\stream_{l-1},\G} 
  + (1/2)\probb{ \abs{\hat{f}_{i,l-1}} < Q_{l-1} \mid i \in
\stream_{l-1},\G} \pm O(n^{-c}) \notag \\
  &= (1/2)  \probb{ \abs{\hat{f}_{i,l-1}} < T_{l-1} \mid i \in \stream_{l-1},\G}  \pm
  O(n^{-c}) \notag \\
&= \frac{1-p_{i,l-1}}{2} \pm O(n^{-c}) \enspace .
\end{align}
 From Eqns.~\eqref{eq:hss:margin:rm1a} and~\eqref{eq:hss:margin:rm2} we obtain,
 \begin{align} \label{eq:hss:rmargin:rm3}
 & 2 \prob{ i \in \bar{G}_l \mid i \in \stream_{l-1},\G} + \prob{ i \in \bar{G}_{l-1} \mid i \in \stream_{l-1},\G}   = 1 \pm O(n^{-c}) \enspace .
 \end{align}
Multiplying Eqn.~\eqref{eq:hss:rmargin:rm3} by $\prob{i \in  \stream_{l-1} \mid \G}$,  we have,
\begin{align} \label{eq:hss:rmargin:rm4}
  &2 \prob{ i \in \bar{G}_l, i \in \stream_{l-1} \mid \G} + \prob{ i \in \bar{G}_{l-1}, i \in \stream_{l-1}\mid \G}  = \prob{i \in \stream_{l-1} \mid\G}\left( 1 \pm O(n^{-c}) \right) \enspace .
 \end{align}
 From the  discussion after Eqn.~\eqref{eq:hss:margin:rmA1}, it follows that $i $ may belong to $ \bar{G}_{l-1} \cup \bar{G}_l$, and in either case, this is possible only if $i \in \stream_{l-1}$. This proves part 3 (b).

 Thus, $i \in \bar{G}_l$ or $i \in \bar{G}_{l-1}$ implies that $i \in \stream_{l-1}$.
 Hence, Eqn.~\eqref{eq:hss:rmargin:rm4} is equivalent to
 \begin{gather*}
 2 \prob{ i \in \bar{G}_l \mid \G} + \prob{ i \in \bar{G}_{l-1}\mid \G} = ( 2^{-(l-1)} \pm n^{-c})\left( 1 \pm O(n^{-c}) \right) \\ = 2^{-(l-1)} \pm O(n^{-c})
 \end{gather*}
 Multiplying by $2^{l-1}$ gives  statement 3 (c) of the lemma.
   \hfill

\end{proof}
\section{Approximate  pair-wise independence of the sampling}
In this section, we prove an approximate pair-wise independence property of the sampling
technique.

\begin{lemma} \label{lem:hss:ij}
Let $i\ne j$. Then, $\prob{i \in \stream_l \mid j \in \stream_r, \G} = 2^{-l} \pm n^{-c}$.
\end{lemma}
\begin{proof}
By pair-wise independence of the hash functions $\{g_l\}$ mapping items to levels, we have
$\prob{i \in \stream_l \mid j \in \stream_r} = \prob{i \in \stream_l} = 2^{-l}$. By Fact~\ref{fact:cndhpev},
$ \prob{i \in \stream_l \mid j \in \stream_r,\G} = 2^{-l} \pm n^{-c}$. \hfill
\end{proof}

\subsection{Sampling probability of items conditional on another item mapping to a level}

\eat{
\begin{lemma} \label{lem:hss:dwise}
Let $d$ be the degree of independence of the hash families from which $g_i$'s are drawn. Let $i_1, i_2, \ldots, i_d \in [n]$ and  distinct.  Then, for any $l_1, l_2, \ldots, l_d \in \{0\} \cup [L]$,  $
 \prob{i_d \in \stream_{l_d} \mid \left(\forall r \in [d-1], i_r \in \stream_{l_r}\right), \G}  = 2^{-l_d}\pm   n^{-c}$.
\end{lemma}
The proof is a direct extension of Lemma~\ref{lem:hss:ij}.
}

\begin{relemma}[Restatement of Lemma~\ref{lem:hsscond}.] 
Let $i,j \in [n]$,  $i \ne j$ and $j \in  G_r$. Then,
$$
  \sum_{r'=0}^L 2^{r'} \prob{j \in \bar{G}_{r'} \mid i \in \stream_l, \G}  =1 \pm O(2^{r}\cdot n^{-c}) \enspace .$$ In particular, the following hold.
\begin{enumerate}
\item Suppose  $j \in \midreg(G_r)$. Then,   $$
  2^r \prob{ j \in \bar{G}_r \mid i \in \stream_l, \G}  = 1 \pm  2^r n^{-c} \enspace .
$$ Further,    for any $r \ne r'$, $ \probb{j \in \bar{G}_{r'}\mid i \in \stream_l, \G} = 0$.

\item If $j \in \lmargin(G_r)$, then,
$$
  2^{r+1} \prob{j \in \bar{G}_{r+1} \mid i \in \stream_l, \G} + 2^r\prob{j \in
 \bar{G}_r \mid i \in \stream_l,\G} = 1 \pm   2^{r+1} n^{-c} \enspace .
$$
Further, for any $ r' \not\in \{r,r+1\},  \prob{j \in \bar{G}_{r'}\mid i \in \stream_l, \G} = 0$.
\item If $j \in \rmargin(G_r)$, then  $$
 2^{r} \prob{ j \in \bar{G}_r \mid i \in \stream_{l}, \G} + 2^{r-1}\prob{ j \in
 \bar{G}_{r-1} \mid i \in \stream_{l},\G}  = 1 \pm  2^{r+1} n^{-c} \enspace . $$
 Further, for any $ r' \not\in \{r-1,r\},  \prob{j \in \bar{G}_{r'}\mid i \in \stream_l, \G} = 0$.
 \end{enumerate}

\end{relemma}

\begin{proof}[Proof of Lemma~\ref{lem:hsscond}.]
The proof proceeds  identically as in the proof of Lemma~\ref{lem:margin}, except that all probabilities are, in addition to being conditional on $\G$, also conditional on $i \in \stream_l$.

\emph{Case 1: $j \in \midreg(G_r)$.} Conditional on $\G$,  as argued in the proof of Lemma~\ref{lem:discovery}, part 1 (b),
$j \in \bar{G}_{r}$ iff $j \in \stream_r$.  Therefore,
\begin{align} \label{eq:hsscondmid}
\prob{j \in \bar{G}_{r} \mid i \in \stream_l, \G} & = \prob{j \in \stream_r \mid i \in
\stream_l, \G}  \in  2^{-r}  \pm n^{-c}
\end{align}
where, the last step follows from Lemma~\ref{lem:hss:ij}.

\emph{Case 2: $j \in \lmargin(G_r)$.} Let  $$p'_{jr}= \prob{ \abs{\hat{f}_{ir}} \ge T_r \mid i \in
\stream_l, j \in \stream_r,\G} \enspace . $$Then,
\begin{align} \label{eq:hsscondlm}
\textsf{Pr}\bigl[ j \in \bar{G}_r \mid i \in \stream_l,\G\bigr] &  = \prob{ \abs{\hat{f}_{ir}} \ge T_r, j \in \stream_r \mid i \in \stream_l,\G} \notag \\
  & = \prob{ \abs{\hat{f}_{ir}} \ge T_r \mid i \in \stream_l, j \in \stream_r,\G} \prob{j \in
  \stream_r \mid i \in \stream_l,\G} \notag \\
  & =  p'_{jr} \cdot (2^{-r} \pm n^{-c}), ~~~\text{ by Lemma~\ref{lem:hss:ij}. }
\end{align}
Further,
\begin{align} \label{eq:hsscondlm1}
\probb{ j \in \bar{G}_{r+1} \mid i \in \stream_l,\G}  &= \prob{Q_r \le \abs{\hat{f}_{ir}} <
T_r, j \in \stream_r,  K_i = 1
\mid i \in \stream_l, \G} \notag \\
& ~~+ \prob{ \abs{\hat{f}_{ir}} < Q_r, i \in \stream_{r+1}, \abs{\hat{f}_{i,r+1}} \ge T_{r+1} \mid i \in \stream_l, \G}
\end{align}
Conditional on $\G$, $\abs{\hat{f}_{ir}} \ge \abs{f_i} -\epsbar T_r  \ge T_r - \epsbar T_r = Q_r$, since $j \in \lmargin(G_r)$. Hence, $\prob{\abs{\hat{f}_{ir}} < Q_r\mid \G} = 0$. Further, since the coin toss $K_i = 1$ is independent of other random bits, Eqn.~\eqref{eq:hsscondlm1} becomes
\begin{align} \label{eq:condlm2}
\prob{ j \in \bar{G}_{r+1} \mid i \in \stream_l,\G}
& = (1/2) \prob{ Q_r \le \abs{\hat{f}_{ir}}< T_r, j \in \stream_r \mid i \in \stream_l, \G} \notag \\
& = (1/2) \prob{Q_r \le \abs{\hat{f}_{ir}}< T_r \mid i \in \stream_l, j \in \stream_r, \G} \prob{j \in \stream_r \mid i \in \stream_l, \G}  \notag \\
& = (1/2) (1-p'_{jr})(2^{-r} \pm n^{-c})
\end{align}
Multiplying  Eqn.~\eqref{eq:condlm2} by $2^{r+1}$, multiplying  Eqn.~\eqref{eq:hsscondlm1} by $2^r$ and adding, we have,
\begin{align*}
2^{r+1} \prob{ j \in \bar{G}_{r+1} \mid i \in \stream_l, \G} + 2^r \prob{j \in \bar{G}_r \mid i \in \stream_l, \G} = 1 \pm O(2^r n^{-c})
\end{align*}
which proves statement (2) of the lemma.

\emph{Case 3: $j \in \rmargin(G_r)$.} 
Then,
\begin{align} \label{eq:hsscondrm1}
\prob{j \in \bar{G}_{r-1} \mid i \in \stream_l,\G}  & = \prob{\abs{\hat{f}_{j,r-1}} \ge T_{r-1},  j \in \stream_{r-1},
 \mid i \in \stream_l,\G} \notag \\
& = \prob{\abs{ \hat{f}_{j,r-1} } \ge T_{r-1} \mid i \in \stream_l, j \in \stream_{r-1},\G} \cdot
\prob{ j \in \stream_{r-1} \mid i \in \stream_{l},\G} \notag \\
& = \prob{\abs{ \hat{f}_{j,r-1} } \ge T_{r-1} \mid i \in \stream_l, j \in \stream_{r-1},\G}  
(2^{-(r-1)} \pm n^{-c})
\end{align}

Also,
\begin{align} \label{eq:hsscondrma}
 \prob{ j \in \bar{G}_r \mid i \in \stream_l,\G} &= \prob{ j \in \stream_{r-1}, Q_{r-1} \le \abs{\hat{f}_{j,r-1}} < T_{r-1},  K_i = 1 \mid i \in
\stream_l,\G} \notag  \\
& + \prob{  \abs{ \hat{f}_{j,r-1}} < Q_r, j \in \stream_r, \abs{\hat{f}_{j,r}} \ge T_r\mid i \in \stream_l,\G}
\end{align}
For $j \in \rmargin(G_r)$ and conditional on $\G$, by following the argument of Lemma~\ref{lem:discovery},  it follows that  if $j \in \stream_r$ then, $\abs{\hat{f}_{jr}} \ge T_r$, viz.,  $\abs{\hat{f}_{jr}} \ge \abs{f_{jr}} - \epsbar T_r \ge T_{r-1} -2\epsbar T_{r-1} - \epsbar T_r > T_r$.
Therefore,
\begin{align} \label{eq:condrm2b}
&\prob{  \abs{ \hat{f}_{j,r-1}} < Q_r, j \in \stream_r, \abs{\hat{f}_{j,r}} \ge T_r\mid i \in \stream_l,\G} \notag \\
&= \prob{\abs{ \hat{f}_{j,r-1}} < Q_r, j \in \stream_r\mid i \in \stream_l,\G} \notag \\
& = \prob{ \abs{ \hat{f}_{j,r-1}} < Q_r\mid i \in \stream_l, j \in \stream_r, \G} \prob{j \in \stream_r \mid i \in \stream_l, \G} \notag \\
& = \prob{\abs{ \hat{f}_{j,r-1}} < Q_r\mid i \in \stream_l, j \in \stream_r, \G}(2^{-r} \pm n^{-c})
\end{align}

The estimate $\hat{f}_{j,r-1}$ is obtained at level
$r-1$, and this is independent of whether $j$ (or any other subset of items) is a member of
$\stream_r$. The latter is a consequence of the level-wise product of \emph{independent} hash values, namely,
$j \in \stream_r$ iff $j \in \stream_{r-1}$ and $g_r(j) = 1$. Therefore,
\begin{align} \label{eq:condrm4}
&\prob{\abs{ \hat{f}_{j,r-1}} < Q_r\mid i \in \stream_l, j \in \stream_r} \notag \\
& = \prob{\abs{\hat{f}_{j,r-1}} < Q_r\mid i \in \stream_l, j \in \stream_{r-1}, g_r(j)= 1} \notag \\
& = \frac{\prob{\abs{\hat{f}_{j,r-1}} < Q_r, g_r(j) = 1\mid i \in \stream_l, j \in \stream_{r-1}}}{\prob{g_r(j)=1 \mid i \in \stream_l, j \in \stream_{r-1}}} \notag \\
& = \left(\frac{\prob{g_r(j) =1 \mid \abs{\hat{f}_{j,r-1}} < Q_r,  i \in \stream_l, j \in \stream_{r-1}}}{\prob{g_r(j)=1 \mid i \in \stream_l, j \in \stream_{r-1}}}\right)
 \cdot \left(\prob{\abs{\hat{f}_{j,r-1}} < Q_r\mid i \in \stream_l, j \in \stream_{r-1}}\right)
\end{align}
Consider the numerator term of the fraction above: \\ $\prob{g_r(j) =1 \mid \abs{\hat{f}_{j,r-1}} < Q_r,  i \in \stream_l, j \in \stream_{r-1}}$. The event $\abs{\hat{f}_{j,r-1}} < Q_r$ depends only on the set of elements that have mapped to $\stream_{r-1}$, and is independent of whether $g_r(j) = 1$. Similarly, $ j \in \stream_{r-1}$ is independent of whether $g_r(j) =1$. Thus, the numerator term equals $ \prob{g_r(j)=1 \mid i \in \stream_l}$ and the denominator term also equals the same, for the same reasons. Hence, Eqn.~\eqref{eq:condrm4} becomes
\begin{align} \label{eq:condrm5}
\prob{\abs{ \hat{f}_{j,r-1}} < Q_r\mid i \in \stream_l, j \in \stream_r}
= \prob{\abs{\hat{f}_{j,r-1}} < Q_r\mid i \in \stream_l, j \in \stream_{r-1}}
\end{align}
Now, conditioning with respect
to $\G$, we have,
\begin{align} \label{eq:condrm6} \prob{ \hat{f}_{j,r-1} > Q_r  \mid j \in \stream_r, i \in
\stream_l,\G}  \in \prob{ \hat{f}_{i,r-1} > Q_r \mid j \in \stream_{r-1}, i \in \stream_l,\G}
\pm  n^{-c}\enspace .
\end{align}
Substituting Eqn.~\eqref{eq:condrm6} in  Eqn.~\eqref{eq:condrm2b}, we have,
\begin{align} \label{eq:condrm2b1}
&\prob{  \abs{ \hat{f}_{j,r-1}} < Q_r, j \in \stream_r, \abs{\hat{f}_{j,r}} \ge T_r\mid i \in \stream_l,\G} \notag \\
&
= \left( \prob{ \hat{f}_{j,r-1} < Q_r \mid j \in \stream_{r-1}, i \in \stream_l,\G} \right) (2^{-r} \pm n^{-c}) \pm 2^{-r} n^{-c}
\end{align}

Consider the first probability term in the \emph{RHS} of Eqn.~\eqref{eq:hsscondrma}.
\begin{align} \label{eq:condrm2c}
&\prob{ j \in \stream_{r-1}, Q_{r-1} \le \abs{\hat{f}_{j,r-1}} < T_{r-1},  K_i = 1 \mid i \in
\stream_l,\G} \notag \\
& = (1/2)\prob{ Q_{r-1} \le \abs{\hat{f}_{j,r-1}} < T_{r-1} \mid i \in \stream_l, j \in \stream_{r-1}, \G}  \prob{ j \in \stream_{r-1} \mid i \in \stream_l, \G} \notag \\
& = \prob{ Q_{r-1} \le \abs{\hat{f}_{j,r-1}} < T_{r-1} \mid i \in \stream_l, j \in \stream_{r-1}, \G} (1/2)(2^{-(r-1)} \pm n^{-c})
\end{align}
Substituting Eqns. ~\eqref{eq:condrm2b1} and ~\eqref{eq:condrm2c} in Eqn.~\eqref{eq:hsscondrma}, we have,
\begin{align} \label{eq:condrm3}
&\prob{j \in \bar{G}_r\mid i \in \stream_l, \G} \notag \\
 &~~= \prob{ Q_{r-1} \le \abs{\hat{f}_{j,r-1}} < T_{r-1} \mid i \in \stream_l, j \in \stream_{r-1}, \G} 2^{-r} \pm O(n^{-c}) \notag \\
& ~~~~+ \left( \prob{ \hat{f}_{i,r-1} < Q_r \mid j \in \stream_{r-1}, i \in \stream_l,\G} \right) (2^{-r} \pm n^{-c}) \pm 2^{-r} n^{-c}
\end{align}
Multiplying   Eqn.~\eqref{eq:hsscondrm1} by $2^{r-1}$ and Eqn.~\eqref{eq:condrm3} by $2^r$ and adding, we obtain
\begin{align*}
& 2^{r-1} \prob{j \in \bar{G}_{r-1} \mid i \in \stream_l,\G} + 2^r \prob{j \in \bar{G}_r\mid i \in \stream_l, \G} \\
& =  \prob{\abs{ \hat{f}_{j,r-1} } \ge T_{l-1} \mid i \in \stream_l, j \in \stream_{r-1},\G}  \pm 2^{r-1} n^{-c} \\
& ~~+ \prob{ Q_{r-1} \le \abs{\hat{f}_{j,r-1}} < T_{r-1} \mid i \in \stream_l, j \in \stream_{r-1}, \G} \pm O(2^r n^{-c})) \notag \\
& ~~~~+  \prob{ \hat{f}_{i,r-1} < Q_r \mid j \in \stream_{r-1}, i \in \stream_l,\G}   \pm O(2^r n^{-c})  \\
& =1 \pm O(2^r n^{-c}) \enspace .
\end{align*}
This proves statement (3) of the Lemma. \hfill
\end{proof}

\subsection{Sampling probability  of an item conditional on another item being sampled}

\begin{relemma} [Lemma~\ref{lem:hssj}.] Suppose $i \in G_l$, $j \in G_m$  and $ j \ne i$.
Then,\\
$$\sum_{r,r'=0}^L   2^{r+r'}\prob{ i \in \bar{G}_r, j \in \bar{G}_{r'} \mid \G} = 1  \pm
O((2^{l}+2^m) n^{-c}) \enspace . $$
\end{relemma}
\begin{proof}[Proof of Lemma~\ref{lem:hssj}.] Assume $\G$ holds for all the arguments in the proof.
\emph{Case 1:
$i \in \midreg(G_l)$.}
Then,
\begin{align*}
&\prob{i \in \bar{G}_r, j \in \bar{G}_{r'} \mid \G}  =\prob{ i \in \bar{G}_r \mid j \in \bar{G}_{r'}, \G} \cdot \prob{ j \in \bar{G}_{r'} \mid \G}
\end{align*}
Conditional on $\G$, $ i \in \bar{G}_r$ iff $r=l$ and $i \in \stream_l$. That is, for $r \ne l$, $\prob{i \in \bar{G}_r \mid j \in \bar{G}_{r'} ,  \G} \allowbreak = 0$.  Therefore,
\begin{align*}
 &\prob{ i \in \bar{G}_l \mid j \in \bar{G}_{r'}, \G} \cdot \prob{ j \in \bar{G}_{r'} \mid \G} \\
  &= \prob{i \in \stream_l \mid j \in \bar{G}_r,\G} \cdot \prob{ j \in \bar{G}_{r'} \mid \G} , \text{ by Lemma~\ref{lem:discovery}, part 2(b)} \\
 & = \prob{ j \in \bar{G}_{r'} \mid i \in \stream_l, \G} \cdot \prob{ i \in \stream_l \mid \G}~~, \text{ by Bayes' rule} \\
 & = \prob{j \in \bar{G}_{r'} \mid i \in \stream_l, \G} \cdot (2^{-l} \pm n^{-c})
\end{align*}
Multiplying by $2^l$, we have,
\begin{gather}
2^l \prob{i \in \bar{G}_r, j \in \bar{G}_{r'} \mid \G} = \prob{j \in \bar{G}_{r'}\mid i \in \stream_l, \G} (1 \pm 2^l n^{-c})  \label{eq:hss:ijmida}
\end{gather}
By Lemma~\ref{lem:hsscond}, we have, $\sum_{r'=0}^{L} \prob{j \in \bar{G}_{r'}\mid i \in \stream_l, \G} = 1 \pm 2^{m+1} n^{-c}$. Therefore, multiplying both sides of Eqn.~\eqref{eq:hss:ijmida} by $2^{r'}$ and summing over $r'$, we have,
\begin{gather}  \label{eq:hss:ijmidb}
\sum_{r'=0}^L 2^{l+r'} \prob{i \in \bar{G}_l, j \in \bar{G}_{r'} \mid \G}  = (1 \pm 2^{m+1} n^{-c} ) (1 \pm 2^l n^{-c}) = (1\pm  O(2^m + 2^l) n^{-c})
\end{gather}
Since $ \prob{i \in \bar{G}_r, j \in \bar{G}_{r'} \mid \G} = 0$ for $r \ne l$, we can equivalently write  Eqn. ~\eqref{eq:hss:ijmidb} as
\begin{align*}
\sum_{r,r'=0}^L 2^{r+r'}\prob{i \in \bar{G}_r, j \in \bar{G}_{r'} \mid \G} = (1\pm  O(2^m + 2^l) n^{-c})
\end{align*}

\emph{Case 2: $i \in \lmargin(G_l)$.} Then, $i$ may belong to either $ \bar{G}_l \cup \bar{G}_{l+1}$ and to no other sampled group  and $i \in \bar{G_l} \cup \bar{G}_{l+1}$ iff $ i \in  \stream_l$, by Lemma~\ref{lem:margin} parts 2(a) and 2(b) respectively.
\begin{align} \label{eq:hss:ijlma}
& \prob{i \in \bar{G}_{l}, j \in \bar{G}_{r'} \mid \G} \notag \\
& = \prob{ i \in \stream_l, \abs{\hat{f}_{il}} \ge T_l , j \in \bar{G}_{r'}\mid  \G} \prob{ j \in \bar{G}_{r'},\G} \notag \\
& = \prob{ \abs{\hat{f}_{il}} \ge T_l \mid j \in \bar{G}_{r'}, i \in \stream_l, \G} \prob{ j \in \bar{G}_{r'} \mid i \in \stream_l, \G} \prob{i \in \stream_l \mid \G} \notag \\
& = \prob{\abs{\hat{f}_{il}} \ge T_l \mid j \in \bar{G}_{r'}, i \in \stream_l, \G} \prob{ j \in \bar{G}_{r'} \mid i \in \stream_l, \G} (2^{-l} \pm n^{-c})
\end{align}
Let $$ p_{il} = \prob{\abs{\hat{f}_{il}} \ge T_l \mid j \in \bar{G}_{r'}, i \in \stream_l, \G} \enspace .$$
Multiplying both sides of Eqn.~\eqref{eq:hss:ijlma} by $2^l$, we obtain
\begin{align} \label{eq:hss:ijlmaa}
 2^l \prob{i \in \bar{G}_{l}, j \in \bar{G}_{r'} \mid \G} & = p_{il} \cdot \prob{ j \in \bar{G}_{r'} \mid i \in \stream_l, \G} (1\pm 2^l n^{-c})
\end{align}

We now consider the case when $i \in \bar{G}_{l+1}$.  By construction, $i \in \bar{G}_{l+1}$ in two ways, either (i) $i \in \stream_l$, $Q_l \le \abs{\hat{f}_{il}} < T_l$ and $K_i=1$, or, (ii) $i \in \stream_l$, $\abs{\hat{f}_{il}} < Q_l$ and $i \in \stream_l$ and $\abs{\hat{f}_{i,l+1}} \ge T_{l+1}$. Possibility (ii) cannot hold since, by Lemma~\ref{lem:discovery} (1b), $i\in\stream_l$ iff $l_d(i) =l$, which by definition is that $\abs{\hat{f}_{il}} \ge Q_l$. These calculations are conditioned on $\G$ and therefore hold conditioned on $j \in \bar{G}_{r'}$ as well.
Hence,
\begin{align} 
&  \prob{i \in \bar{G}_{l+1}, j \in \bar{G}_{r'} \mid \G}  \notag \\
& = \prob{i \in \stream_{l}, (Q_l \le  \abs{\hat{f}_{il}} < T_l ), K_i = 1, j \in \bar{G}_{r'} \mid  \G}    \notag \\
 \label{eq:hss:ijlmc}
 & = (1/2)\prob{ Q_l \le \abs{\hat{f}_{il}} < T_l \mid i \in \stream_l,  j \in \bar{G}_{r'},\G} \prob{ j \in \bar{G}_{r'}\mid i \in \stream_l,\G}  \prob{i \in \stream_l \mid  \G}   \notag \\
& = (1/2)(1-p_{il})\prob{ j \in \bar{G}_{r'}\mid i \in \stream_l,\G}  (2^{-l} \pm n^{-c})
 \end{align}
 or,  by multiplying both sides of Eqn.~\eqref{eq:hss:ijlmc},
 \begin{align} \label{eq:hss:ijlmdd}
 2^{l+1} \prob{i \in \bar{G}_{l+1}, j \in \bar{G}_{r'} \mid \G}  = (1-p_{il}) \cdot \prob{ j \in \bar{G}_{r'} \mid i \in \stream_l, \G} (1 \pm 2^{l+1} n^{-c})
 \end{align}
 Adding Eqns.~\eqref{eq:hss:ijlmaa} and ~\eqref{eq:hss:ijlmdd}, we have,
 \begin{align} \label{eq:hss:ijlme}
 &2^{l+1} \prob{i \in \bar{G}_{l+1}, j \in \bar{G}_{r'} \mid \G} + 2^l \prob{i \in \bar{G}_{l}, j \in \bar{G}_{r'} \mid \G} = \prob{ j \in \bar{G}_{r'} \mid i \in \stream_l, \G} (1 \pm 2^{l+2}n^{-c}) \enspace .
 \end{align}
 By Lemma~\ref{lem:hsscond}, $\sum_{r'=0}^L  2^{r'} \prob{j \in \bar{G}_{r'} \mid i \in \stream_l, \G} = 1 \pm  O(2^m n^{-c})$. Therefore, multiplying Eqn.~\eqref{eq:hss:ijlme} by $2^{r'}$ and summing over $r'$, we have,
 \begin{align*}
 &\sum_{r'=0}^L \left(2^{l+1} \prob{i \in \bar{G}_{l+1}, j \in \bar{G}_{r'} \mid \G} + 2^l \prob{i \in \bar{G}_{l}, j \in \bar{G}_{r'} \mid \G}\right)\\
&  = \sum_{r'=0}^L 2^{r'} \prob{j \in \bar{G}_{r'} \mid i \in \stream_l, \G} (1 \pm 2^{l+2} n^{-c})\\
 &  = (1 \pm O(2^m n^{-c}))(1 \pm O(2^l n^{-c})) \\
 & = 1 \pm O((2^{l} + 2^m) n^{-c})
 \end{align*}
Since, $\prob{i \in \bar{G}_r, j \in \bar{G}_{r'} \mid \G} = 0$ for any $r \not\in \{l,l+1\}$, we can rewrite the above equation as
\begin{align*}
& \sum_{r,r'=0}^L 2^{r+r'} \prob{i \in \bar{G}_{r}, j \in \bar{G}_{r'} \mid \G}
= (1 \pm O(2^m + 2^l)n^{-c})
\end{align*}

\emph{Case 3: $i \in \rmargin(G_l)$.} If $j \in \lmargin(G_m)$ or $j \in \midreg(G_m)$, then, we can interchange the roles of $i$ and $j$ and the lemma is proved. Hence, we may now assume that $j \in \rmargin(G_m)$. Let $m \le l$ without loss of generality.

By Lemma~\ref{lem:margin}, part (3), $i \in \bar{G}_{l-1} \cup \bar{G}_l$ and this implies that $i \in \stream_l$. Also, $i \not\in \cup_{l' \not\in \{l-1, l\}} \bar{G}_{l'}$ (with prob. 1). Let
$p_{i,l-1, j,r'} = \prob{ \abs{\hat{f}_{i,l-1}} \ge T_{l-1} \mid j \in \bar{G}_{r'}, i \in \stream_{l-1}}$. Then,
\begin{align} \label{eq:hss:ijrma}
\prob{ i \in \bar{G}_{l-1}, j \in \bar{G}_{r'} \mid \G} &= \prob{ \abs{\hat{f}_{i,l-1}} \ge T_{l-1}, i \in \stream_{l-1}, j \in \bar{G}_{r'} \mid \G} \notag \\
& = p_{i,l-1} \cdot \prob{ j \in \bar{G}_{r'} \mid i \in \stream_{l-1}, \G} \prob{ i \in \stream_{l-1} \mid \G} \notag \\
& = p_{i,l-1} \cdot \prob{ j \in \bar{G}_{r'} \mid i \in \stream_{l-1}, \G} (2^{-(l-1)} \pm n^{-c}) \enspace .
\end{align}
\eat{
& = \prob{i \in \bar{G}_{l-1} \mid j \in \bar{G}_{r'}, \G} \prob{j \in \bar{G}_{r'},\G} \notag \\
& = \prob{ i \in \stream_{l-1}, \abs{\hat{f}_{i,l-1}} \ge T_{l-1} \mid j \in \bar{G}_{r'},\G}  \prob{j \in \bar{G}_{r'},\G}   \notag \\
& = \prob{\abs{\hat{f}_{i,l-1}}  \ge T_{l-1} \mid i \in \stream_{l-1}, j \in \bar{G}_{r'},\G} \notag \\
& \hspace{1.0cm} \cdot \prob{ i\in \stream_{l-1} \mid j \in \bar{G}_{r'}, \G} \cdot \prob{j \in \bar{G}_{r'},\G} \notag \\
& = \prob{ \abs{\hat{f}_{i,l-1}}  \ge T_{l-1} \mid i \in \stream_{l-1}, j \in \bar{G}_{r'},\G}\notag \\
& \hspace{1.0cm} \cdot \prob{ j \in \bar{G}_{r'} \mid i \in \stream_{l-1},\G}\cdot \prob{ i \in \stream_{l-1}\mid \G} \notag \\
& = \prob{ \abs{\hat{f}_{i,l-1}}  \ge T_{l-1} \mid i \in \stream_{l-1}, j \in \bar{G}_{r'},\G} \notag \\
& \hspace{1.0cm}\cdot \prob{ j \in \bar{G}_{r'} \mid i \in \stream_{l-1},\G} (2^{-(l-1)} \pm n^{-c})
\end{align}
}
Let $q_{i,l-1,j,r'} = \prob{Q_{l-1} \le \abs{\hat{f}_{i,l-1}} < T_{l-1}\mid i \in \stream_{l-1}, j \in \bar{G}_{r'}}$. By Lemma~\ref{lem:margin} part 3 (b), $\{i\in \stream_l, l_d(i) \ne l-1\} \subset \{i \in \bar{G}_l\}$.
Then,
\begin{align} \label{eq:hss:ijrmb}
&\prob{ i \in \bar{G}_{l}, j \in \bar{G}_{r'} \mid \G}  \notag \\
&= \prob{ Q_{l-1} \le \abs{\hat{f}_{i,l-1}} < T_{l-1}, K_i = 1,  i \in  \stream_{l-1}, j \in \bar{G}_{r'} \mid \G} + \prob{ \abs{\hat{f}_{i,l-1}} < Q_{l-1}, i \in \stream_l, j \in \bar{G}_{r'} \mid \G} \notag \\
& = (1/2)q_{i,l-1,j,r'} \prob{j\in \bar{G}_{r'} \mid i \in \stream_{l-1}, \G} \prob{i \in \stream_{l-1}}  +  \prob{ \abs{\hat{f}_{i,l-1}} < Q_{l-1}, i \in \stream_l, j \in \bar{G}_{r'} \mid \G}  \notag  \\
& = q_{i,l-1,j,r'} \cdot  \prob{j\in \bar{G}_{r'} \mid i \in \stream_{l-1}, \G} (2^{-l} \pm O(n^{-c}) \notag \\
 & \hspace*{1.0in} +  \prob{ \abs{\hat{f}_{i,l-1}} < Q_{l-1}\mid i \in \stream_l, j \in \bar{G}_{r'}, \G}  \prob{j \in \bar{G}_{r'} \mid i \in \stream_l, \G} \prob{i \in \stream_l\mid \G} \notag \\
\end{align}
Consider the following term derived from the  second term in the above sum.
\begin{align}  \label{eq:ijrmb1}
\prob{ \abs{\hat{f}_{i,l-1}} < Q_{l-1}\mid i \in \stream_l, j \in \bar{G}_{r'},\G}  & = \prob{ \abs{\hat{f}_{i,l-1}} < Q_{l-1}\mid g_l(i) =1, i \in \stream_{l-1}, j \in \bar{G}_{r'}} \notag  \\
& = \frac{ \prob{\abs{\hat{ f}_{i,l-1}} < Q_{l-1}, g_l(i)=1 \mid i \in \stream_{l-1}, j \in \bar{G}_{r'}}}{\prob{g_l(i)=1 \mid i \in \stream_{l-1}, j \in \bar{G}_{r'}}}
\end{align}
\eat{
\begin{align}
& ~~~+  \prob{ \abs{\hat{f}_{i,l-1}}< Q_{l-1} \mid i \in \stream_{l}, j \in \bar{G}_{r'},\G} \prob{ j \in \bar{G}_{r'} \mid i \in \stream_l, \G} \prob{i \in \stream_l\mid \G} \\
& = (1/2) q_{i,l-1,j,r'}\prob{j\in \bar{G}_{r'} \mid i \in \stream_{l-1}, \G}  (2^{-(l-1)} \pm n^{-c}) + (1-q_{i,l-1,j,r'} - p_{i,l-1,j,r'})
& \hspace{1.0cm} \cdot \prob{ j \in \bar{G}_{r'} \mid i \in \stream_{l-1},\G}\cdot  (2^{-(l-1)} \pm n^{-c})  \notag \\
  &~~~~ +\prob{\abs{\hat{f}_{i,l-1}} < Q_{l-1},  \abs{\hat{f}_{il}} \ge T_l \mid i \in \stream_l, j \in \bar{G}_{r'},\G} \notag \\
& \hspace{1.0cm}\cdot \prob{ j \in \bar{G}_{r'} \mid i \in \stream_{l},\G}\cdot  (2^{-l} \pm n^{-c}) \notag \\
& = \prob{ Q_{l-1} \le \abs{\hat{f}_{i,l-1}} < T_{l-1}\mid i \in  \stream_{l-1}, j \in \bar{G}_{r'},\G}\notag \\
& \hspace{1.0cm} \cdot \prob{ j \in \bar{G}_{r'} \mid i \in \stream_{l-1},\G} \cdot  (2^{-l} \pm O(n^{-c}) ) \notag \\
  & ~~~+\prob{ \abs{\hat{f}_{i,l-1}} < Q_{l-1}, \abs{\hat{f}_{il}} \ge T_l \mid i \in \stream_l, j \in \bar{G}_{r'},\G}\notag \\
& \hspace{1.0cm} \cdot \prob{ j \in \bar{G}_{r'} \mid i \in \stream_{l},\G}\cdot  (2^{-l} \pm n^{-c})
\end{align}
Conditional on $\G$, since $i \in \rmargin(G_l)$, $ \abs{\hat{f}_{il}} \ge T_l $ holds surely. Therefore, Eqn.~\eqref{eq:hss:ijrmb} can be written as
\begin{align} \label{eq:hss:ijrmc}
&\prob{ i \in \bar{G}_{l}, j \in \bar{G}_{r'} \mid \G} \notag \\
&= \prob{ Q_{l-1} \le \abs{\hat{f}_{i,l-1}} < T_{l-1}\mid i \in  \stream_{l-1}, j \in \bar{G}_{r'},\G} \notag \\
& \hspace{1.0cm} \cdot \prob{ j \in \bar{G}_{r'} \mid i \in \stream_{l-1},\G}\cdot  (2^{-l} \pm O(n^{-c}) ) \notag \\
&+\prob{ \abs{\hat{f}_{i,l-1}} < Q_{l-1} \mid i \in \stream_l, j \in \bar{G}_{r'},\G} \notag \\
& \hspace{1.0cm} \cdot  \prob{ j \in \bar{G}_{r'} \mid i \in \stream_{l},\G}\cdot  (2^{-l} \pm n^{-c})
\end{align}
Further, since the event $i \in \stream_l \equiv g_l(i) = 1 \text{ and } i \in \stream_{l-1}$,
\begin{align}\label{eq:hssj:rm:conda} & \prob{\abs{\hat{f}_{i,l-1}} < Q_{l-1} \mid i \in \stream_l, j \in \bar{G}_{r'},\G}\notag \\
& \hspace{1.0cm} = \frac{ \prob{ \abs{\hat{f}_{i,l-1}} < Q_{l-1}, g_l(i)=1 \mid i \in \stream_{l-1}, j \in \bar{G}_{r'},\G}}{\prob{ g_l(i)=1\mid i \in \stream_{l-1}, j \in \bar{G}_{r'},\G}}
  \end{align}
  }

\begin{align*}
 &\prob{\abs{\hat{f}_{i,l-1}} < Q_{l-1}, g_l(i)=1 \mid i \in \stream_{l-1}, j \in \bar{G}_{r'},\G } \\
 & = \prob{ g_l(i)=1 \mid i \in \stream_{l-1}, \abs{\hat{f}_{i,l-1}} < Q_{l-1},  j \in \bar{G}_{r'}, \G} \prob{\abs{\hat{f}_{i,l-1}} < Q_{l-1} \mid i \in \stream_{l-1}, j \in \bar{G}_{r'}, \G}
 \end{align*}
The event $g_l(i) = 1$ is independent of the value of $\hat{f}_{i,l-1}$, since they depend on the values of $g_{l'}(k)$'s for $k \in [n]\setminus\{i\}$ and $1 \le l' <l$. Now conditional on $\G$ and given that $j \in \rmargin(G_m)$, for $m \le l$, the event $j \in \bar{G}_{r'}$ has zero probability unless $r' \in \{m-1, m\}$. 

\emph{Case 3.1.} $r'=m-1$. In this case, $j \in \bar{G}_{m-1}$. Since, $j \in \rmargin(G_m)$, the event $j \in \bar{G}_{m-1}$ depends only on the value of $\hat{f}_{j,m-1}$. Since, $m \le l$, the random bit defining $g_l$ is independent of the values of the random bits that determine $\hat{f}_{j,m-1}$.  Therefore,
\begin{align*}
\prob{ g_l(i)=1 \mid i \in \stream_{l-1}, \abs{\hat{f}_{i,l-1}} < Q_{l-1},  j \in \bar{G}_{r'}, \G} = \prob{g_l(i)=1 \mid \G} \enspace .
\end{align*}
Arguing similarly, $ \prob{g_l(i) = 1 \mid i \in \stream_l, j \in \bar{G}_{r'}, \G} = \prob{g_l(i) =1 \mid \G}$.

Therefore, it follows from Eqn. ~\eqref{eq:ijrmb1} that
\begin{align}  \label{eq:ijrmb2}
\prob{ \abs{\hat{f}_{i,l-1}} < Q_{l-1}\mid i \in \stream_l, j \in \bar{G}_{r'},\G}  & =  \prob{\abs{\hat{ f}_{i,l-1}} < Q_{l-1},  \mid i \in \stream_{l-1}, j \in \bar{G}_{r'}}
\end{align}

\emph{Case 3.2.} Suppose $r' = m$. Since $j \in \rmargin(G_m)$, therefore,  $j \in \bar{G}_m$ is equivalent to the event $ \hat{f}_{j,m-1} < Q_{m-1}$ and $ j \in \stream_{m}$. If $m < l$, then the event $g_l(i) = 1$ is independent of the values of $\hat{f}_{j,m-1}$ and the event $j \in \stream_m$. Hence the same conclusion as Eqn.~\eqref{eq:ijrmb2} holds when $r'=m$ and $m < l$.

Now suppose $r'=m$ and $m =l$. Then, we have,
\begin{align*}
& \prob{ g_l(i)=1 \mid i \in \stream_{l-1}, \abs{\hat{f}_{i,l-1}} < Q_{l-1},  j \in \bar{G}_{r'}, \G}\\
& = \prob{ g_l(i)=1 \mid j \in \bar{G}_{r'}, \G} \\
& = \prob{g_l(i) =1 \mid \abs{\hat{f}_{j,l-1}} < Q_{l-1}, g_l(j)=1, j \in \stream_{l-1}, \G} \\
& = \prob{g_l(i) = 1 \mid  g_l(j)=1, \G} \\
& = \prob{g_l(i) =1\mid \G} \enspace .
\end{align*}
Hence, Eqn.~\eqref{eq:ijrmb2} continues to hold in this case as well. Thus in all cases, Eqn~\eqref{eq:ijrmb2} holds. 

Substituting this into  Eqn.~\eqref{eq:hss:ijrmb}, we have, 
\begin{align} \label{eq:hss:ijrmb2}
&\prob{ i \in \bar{G}_{l}, j \in \bar{G}_{r'} \mid \G}  \notag \\
& = q_{i,l-1,j,r'} \cdot  \prob{j\in \bar{G}_{r'} \mid i \in \stream_{l-1}, \G} (2^{-l} \pm O(n^{-c}) \notag \\
 & \hspace*{1.0in} +  \prob{ \abs{\hat{f}_{i,l-1}} < Q_{l-1}\mid i \in \stream_{l-1}, j \in \bar{G}_{r'}, \G}  \prob{j \in \bar{G}_{r'} \mid i \in \stream_l, \G} \prob{i \in \stream_l\mid \G} \notag \\
& = \left(q_{i,l-1,j,r'}  + 1- (p_{i,l-1,j,r'} -q_{i,l-1,j,r'}) \right) \prob{j\in \bar{G}_{r'} \mid i \in \stream_{l-1}, \G} (2^{-l} \pm O(n^{-c}) \notag \\
& = (1-p_{i,l-1,j,r'})\prob{j\in \bar{G}_{r'} \mid i \in \stream_{l-1}, \G} (2^{-l} \pm O(n^{-c})
\end{align}

Multiplying Eqn.~\eqref{eq:hss:ijrma} by $2^{l-1}$ and Eqn.~\eqref{eq:hss:ijrmb2} by $2^l$, we have for $r' \in \{m-1,m\}$ that
\begin{gather} \label{eq:hss:ijrm5}
2^{l-1} \prob{ i \in \bar{G}_{l-1}, j \in \bar{G}_{r'} \mid \G} + 2^l\prob{ i \in \bar{G}_{l}, j \in \bar{G}_{r'} \mid \G}  = \prob{j \in \bar{G}_{r'} \mid i \in \stream_{l-1}, \G} (1 \pm O(2^l n^{-c}))
\end{gather}
The \emph{LHS} of ~\eqref{eq:hss:ijrm5} can be equivalently written as $\sum_{r=0}^L 2^r \prob{i \in \bar{G}_r, j \in \bar{G}_{r'}\mid \G}$, since, for $r \not\in \{l-1,l\}$, $\prob{i \in \bar{G}_r \mid \G} = 0$.  Therefore,
\begin{align} \label{eq:ijrm6}
\sum_{r=0}^L 2^r \prob{i \in \bar{G}_r, j \in \bar{G}_{r'}\mid \G} = \prob{j \in \bar{G}_{r'} \mid i \in \stream_{l-1}, \G} (1 \pm O(2^l n^{-c}))
\end{align}
By Lemma~\ref{lem:hsscond}, we have,
\begin{align*} \sum_{r'=0}^L  \prob{j \in \bar{G}_{r'} \mid i \in \stream_{l-1}, \G} & = \sum_{r'=m-1}^m \prob{j \in \bar{G}_{r'} \mid i \in \stream_{l-1}, \G} = 1\pm 2^m O(n^{-c})
\end{align*}
Combining with Eqn.~\eqref{eq:ijrm6}, we have,
\begin{multline*}
\sum_{r'=0}^L \sum_{r=0}^L  2^r \prob{i \in \bar{G}_r, j \in \bar{G}_{r'}\mid \G} =
\sum_{r'=0}^L \prob{j \in \bar{G}_{r'} \mid i \in \stream_{l-1}, \G} (1 \pm O(2^l n^{-c})) \\ = (1 \pm O(2^m n^{-c}))(1 \pm 2^l n^{-c})) = 1 \pm O( (2^l + 2^m) n^{-c}) \enspace .  ~~~~
\end{multline*}
\end{proof}

\eat{

We now consider $ \prob{j \in \bar{G}_{r'} \mid i \in \stream_l, \G}$ and $ \prob{j \in \bar{G}_{r'} \mid i \in \stream_{l-1},\G}$. Since, $j \in \rmargin(G_m)$, hence conditional on $\G$, $j $ can be either in $\bar{G}_{m-1}$ or in $\bar{G}_{m}$, but not in any other sampled group.
\begin{align} \label{eq:hss:ijrme}
&\prob{j \in \bar{G}_m \mid i \in \stream_l, \G}  \notag \\
& = \prob{ Q_{m-1} \le \abs{\hat{f}_{j,m-1}}  < T_{m-1},  K_j = 1, j \in \stream_{m-1} \mid i \in \stream_l, \G} \notag \\
& ~~+ \prob{ j \in  \stream_m, \abs{\hat{f}}_{j,m-1} < Q_{m-1}, \abs{\hat{f}_{j,m}} \ge T_m \mid i \in \stream_l, \G} \notag \\
& = (1/2)\prob{  Q_{m-1} \le \abs{\hat{f}_{j,m-1}}  < T_{m-1} \mid j \in
\stream_{m-1}, i \in \stream_l, \G}\notag \\
 & \hspace*{1.0in}\cdot \prob{ j \in \stream_{m-1} \mid i \in \stream_l,\G} \notag \\
& ~~+ \prob{ \abs{\hat{f}_{j,m-1}} < Q_{m-1} \mid j \in \stream_m, i \in \stream_l, \G} \prob{j \in \stream_m \mid i \in \stream_l,\G}
\end{align}
By pair-wise independence of the functions $\{g_l\}$ for each $l = 1,2, \ldots, L$, we have that
$ \prob{j \in \stream_m \mid i \in \stream_l} = 2^{-m}$ and $\prob{j \in \stream_{m-1} \mid i \in \stream_l} = 2^{-(m-1)}$. Conditioning with respect to $\G$ changes the probability by $O(n^{-c})$ each. Thus, Eqn. ~\eqref{eq:hss:ijrme} can be written as
\begin{align} \label{eq:hss:ijrmf}
&\prob{j \in \bar{G}_m \mid i \in \stream_l, \G} \notag \\
& = \prob{  Q_{m-1} \le \abs{\hat{f}_{j,m-1}}  < T_{m-1} \mid j \in
\stream_{m-1}, i \in \stream_l, \G} (2^{-m} \pm O(n^{-c})) \notag \\
& ~+ \prob{ \abs{\hat{f}_{j,m-1}} < Q_{m-1} \mid j \in \stream_m, i \in \stream_l, \G} (2^{-m} \pm O(n^{-c}))
\end{align}
The probability $\prob{ \abs{\hat{f}_{j,m-1}} < Q_{m-1} \mid j \in \stream_m, i \in \stream_l, \G}$ equals $\prob{ \abs{\hat{f}_{j,m-1}} < Q_{m-1} \mid j \in \stream_{m-1}, i \in \stream_l, \G}$,
since the event $\abs{\hat{f}_{j,m-1}} < Q_{m-1}$ depends only on the random bits used by hash functions $g_1, \ldots, g_{m-1}$ and $\{h_{m-1,r}\}_{r\in [2s]}$. Thus, Eqn.~\eqref{eq:hss:ijrmf} simplifies to
\begin{align}  \label{eq:hss:ijrmg}
&\prob{j \in \bar{G}_m \mid i \in \stream_l, \G}
& = \prob{ \abs{\hat{f}_{j,m-1}} < T_{m-1} \mid j \in \stream_{m-1}, i \in \stream_l, \G} (2^{-m} \pm O(n^{-c}))
\end{align}
In Eqn.~\eqref{eq:hss:ijrmg}, the value of $\hat{f}_{j,m-1}$ depends only on  $\stream_{m-1}$, which contains both $j$ and $i$, the latter because it is given that $i \in \stream_l$ and $l \ge m$ and so $i$ appears in each of $\stream_{l-1}, \stream_{l-2}, \ldots, \stream_0$.

By starting with the probability $\prob{j \in \bar{G}_m \mid i \in \stream_{l-1}, \G}$, this would simplify analogously to
\begin{align} \label{eq:hss:ijrmh}
&\prob{j \in \bar{G}_m \mid i \in \stream_{l-1}, \G} \notag \\
& = \prob{ \abs{\hat{f}_{j,m-1}} < T_{m-1} \mid j \in \stream_{m-1}, i \in \stream_{l-1}, \G} (2^{-m} \pm O(n^{-c}))
\end{align}
Since $l \ge m$,
\begin{align*}
\prob{ \abs{\hat{f}_{j,m-1}} < T_{m-1} \mid j \in \stream_{m-1}, i \in \stream_l, \G} = \prob{ \abs{\hat{f}_{j,m-1}} < T_{m-1} \mid j \in \stream_{m-1}, i \in \stream_{l-1}, \G} \enspace .
\end{align*}
Substituting the above in Eqn.~\eqref{eq:hss:ijrmd}, we have,
\begin{align} \label{eq:hss:ijrmi}
&\prob{ i \in \bar{G}_{l}, j \in \bar{G}_{r'} \mid \G}  \notag \\
&= \prob{ Q_{l-1} \le \abs{\hat{f}_{i,l-1}} < T_{l-1}\mid i \in  \stream_{l-1}, j \in \bar{G}_{r'},\G}  \cdot \prob{ j \in \bar{G}_{r'} \mid i \in \stream_{l-1},\G} \notag \\
& ~~~\cdot  (2^{-l} \pm O(n^{-c}) )  \notag \\
& +\prob{ \abs{\hat{f}_{i,l-1}} < Q_l \mid i \in \stream_{l-1}, j \in \bar{G}_{r'},\G}  \cdot  \prob{ j \in \bar{G}_{r'} \mid i \in \stream_{l-1},\G}\cdot (1 \pm O(n^{-c}))  (2^{-l} \pm n^{-c}) \notag \\
& = \prob{ \abs{\hat{f}_{i,l-1}} < T_{l-1} \mid i \in  \stream_{l-1}, j \in \bar{G}_{r'},\G}  \cdot \prob{ j \in \bar{G}_{r'} \mid i \in \stream_{l-1},\G}\notag \\ & ~~~~\cdot (2^{-l} \pm O(n^{-c}))
\end{align}
Multiplying Eqn.~\eqref{eq:hss:ijrmi} by $2^l$ and multiplying Eqn.~\eqref{eq:hss:ijrma} by $2^{l-1}$ and then adding them, we have,
\begin{align}
& 2^{l-1} \cdot \prob{ i \in \bar{G}_{l-1}, j \in \bar{G}_{r'} \mid \G}  + 2^l \cdot \prob{ i \in \bar{G}_{l}, j \in \bar{G}_{r'} \mid \G} \notag \\
& = \prob{ \abs{\hat{f}_{i,l-1}}  \ge T_{l-1} \mid i \in \stream_{l-1}, j \in \bar{G}_{r'},\G} \cdot \prob{ j \in \bar{G}_{r'} \mid i \in \stream_{l-1},\G}\notag \\
& ~~~~~~~\cdot  (1 \pm 2^{l-1} n^{-c}) \notag  \\
& ~~+ \prob{ \abs{\hat{f}_{i,l-1}} < T_{l-1} \mid i \in  \stream_{l-1}, j \in \bar{G}_{r'},\G}  \cdot \prob{ j \in \bar{G}_{r'} \mid i \in \stream_{l-1},\G}\notag \\ & ~~~~\cdot (1 \pm O(2^l n^{-c})) \notag \\
& = \prob{ j \in \bar{G}_{r'} \mid i \in \stream_{l-1},\G}\cdot (1 \pm O(2^l n^{-c}))
\end{align}
By Lemma~\ref{lem:hsscond}, we have,
$\sum_{r'=0}^L  2^{r'}\prob{ j \in \bar{G}_{r'} \mid i \in \stream_{l-1}, \G} =1 \pm O(2^m n^{-c})$.
Combining, we have,
\begin{align*}
\sum_{r,r'=0}^L 2^{r+r'} \prob{i \in \bar{G}_r, j \in \bar{G}_{r'} \mid \G} = (1 \pm O(2^l + 2^m)n^{-c})
\end{align*}
\hfill
\end{proof}
}

\section{Application of Taylor polynomial estimator}

Throughout the remainder of this section, let $Y$ denote a code given by Corollary~\ref{cor:gv}.

\subsection{Preliminaries}
\emph{Notation.} We first partition the random seeds used by the algorithm by their functionality. For strings $s$ and $t$, let $s \oplus t$ denote the string that is the concatenation of $s$ and $t$. 

Let $\bar{g}_l$  denote the random bit string representing  the seed used to generate the hash function $g_l$, for $l \in \{0\} \cup [L]$, and let $\bar{g}$ denote the concatenation of the seed strings  $\bar{g}_1 \oplus \bar{g}_2, \ldots \oplus \bar{g}_L$.  For $l \in \{0\} \cup [L]$ and $j \in [s]$, let $\bar{h}_{HH,l,j}$ denote the random bit string used to generate the hash function corresponding to the $j$th hash table in the $\hh_l$ structure; let $\bar{h}_{HH,l}$ denote the concatenation of the random bitstrings  $\oplus_{j \in [s]} \bar{h}_{HH,l}$ and $\bar{h}_{HH}$  denote the concatenation of the random bitstrings $\oplus_{l \in \{0,1,\ldots, L\}} \bar{h}_{HH,l}$. For $l \in \{0\} \cup [L]$ and $j \in [2s]$, let $\bar{h}_{lj}$ denote the random bit string used to generate the hash function $h_{lj}$ in the \tpest$_l$ structure. Let $\bar{h}_l$ denote the random bit string $\oplus_{j \in [2s]} \bar{h}_{lj}$ and let $\bar{h}$ denote the concatenation $\bar{h} = \oplus_{l \in \{0,1, \ldots, L\}}$.  Let $\bar{\xi}_{HH,l,j}$ denote the random bit string used to generate the Rademacher family used by the $j$th table of the $\hh_l$ structure, for $l \in \{0,1, \ldots, L\}$ and $j \in [s]$.  Let $\bar{\xi}_{HH,l} = \oplus_{j \in [s]} \bar{\xi}_{HH,l,j}$ and let $\bar{\xi}_{HH} = \oplus_{l \in \{0,1,\ldots, L\}} \bar{\xi}_{HH,l}$.
 Let $\bar{\xi}_{lj}$ denote the random seed that generates the Rademacher variables $\{\xi_{lj}(k)\}_{k \in [n]}$  used by the $j$th table in $\tpest_l$ structure, for $j \in [2s]$; let $\bar{\xi}_l = \oplus_{j \in [2s]} \bar{\xi}_{lj}$ and let $\bar{\xi} = \oplus_{l \in \{0,1, \ldots, L\}} \bar{\xi}_{l}$. Let $\bar{\zeta}$ denote the random bit string used to estimate $F_2$.

The full random seed string used to update and maintain the \ghss~structure  is  $\bar{\zeta} \oplus \bar{g} \oplus \bar{h}_{HH} \oplus \bar{\xi}_{HH} \oplus \bar{h} \oplus \bar{\xi}$. In addition, during estimation, an $n$-dimensional random bit vector $K$ is also used.

Note that the events in $\G$ are dependent only on $\bar{\zeta} \oplus \bar{g} \oplus \bar{h}_{HH} \oplus \bar{\xi}_{HH}$.  This is further explained in the table below.

\begin{center}
\begin{tabular}{l|>{$}c<{$}} 
Event &   \parbox[c]{3in}{Random bit string that determines the event} \\[5pt]\hline
\goodftwo & \bar{\zeta} \\ & \\
\nocollision  &  \bar{h} \\& \\
\goodest &  \bar{h}_{HH} \\& \\
\smallres & \bar{g} \\ & \\
\accuest & \bar{g} \oplus \bar{h}_{HH}  \\ & \\
\lastlevel &  \bar{g} \oplus \bar{h}_{HH} \\ & \\
\smallhh & \bar{g} \oplus \bar{h}_{HH}\\ 
\end{tabular}
\end{center}

\subsection{Basic properties of the application of Taylor polynomial estimator: Proof of Lemma~\ref{lem:hssavtp}-Part I}
For items $i, k \in [n]$ with $k \ne i$, hash table index $j \in [s]$ and $l \in [L] \cup \{0\}$, define the indicator variable  $u_{ikjl}$ to be 1 if  $h_{lj}(i) = h_{lj}(k)$.

\begin{proof}[Proof of Lemma~\ref{lem:hssavtp}, parts (a), (b) and  (e).] Suppose $\G$ holds.
The last  statement of the lemma follows from \lastlevel, which is a sub-event of $\G$.

Let $l=l_d(i)  \in \{0\} \cup [L-1]$.  By \accuest,
 $\abs{\hat{f}_{il} - f_i} \le \epsbar T_l$.  Since $i$ is discovered at level $l$,  $\abs{\hat{f}_{il}} \ge Q_l = T_l -\epsbar T_l$. So,
 $\abs{f_i} \ge \abs{ \hat{f}_{il}} - \epsbar T_l \ge Q_l -
  \epsbar T_l = T_l - 2\epsbar T_l$ and therefore,
$$\frac{\abs{\hat{f}_i - f_i}}{\abs{f_i}}
  \le \frac{\epsbar T_l}{(1-2\epsbar) T_l} \le \frac{1/(27p)}{(1-2/(27p))} < \frac{1}{26p} $$
  since, $\epsbar = (B/C)^{1/2} = 1/(27p)$ and $p \ge 2$.
  Therefore,
  $$\frac{\abs{\hat{f}_i - f_i}}{\abs{\hat{f}_i}} \le \frac{\epsbar T_l}{ (1-\epsbar)T_l} < \frac{1}{26p} \enspace . $$
  This proves parts (a) and (e) of the lemma.

Let $j \in R_l(i)$ and $l_d(i) = l$.  For $k \in [n]$, let $y_{lk}$ be an indicator variable that is 1 if $k \in \stream_l$ and is 0 otherwise. Then, $$X_{ijl} = \sum_{k \in [n]} f_k \cdot y_{lk} \cdot  \xi_{lj}(k) \cdot u_{ikjl} \cdot \xi_{lj}(i) \cdot \sgn(\hat{f}_i) $$
Since it is given that $l_d(i) = l$, it follows that
$$ X_{ijl} = f_i \cdot \sgn(\hat{f}_i)  + \sum_{k \in [n], k\ne i}  f_k \cdot y_{lk} \cdot  \xi_{lj}(k) \cdot u_{ikjl} \cdot \xi_{lj}(i) \cdot \sgn(\hat{f}_i) \enspace . $$

We now take  expectations.  Note that the events in $\G$ are independent of the Rademacher family random bits $\xi_{lj}(k)$.  Also, the event $u_{ikjl} = 1$ depends only on $\bar{g} \oplus \bar{h}_l$ and the event $l_d(i) =l$ depends only on $\bar{g} \oplus \bar{h}_{\hh}$.  Therefore,
 \begin{align} \label{eq:hssavtp:1}
&\expectsub{\bar{\xi}_{lj}}{X_{ijl} \mid l_d(i) = l, j \in R_l(i), \G} \notag \\
 & =  f_i \cdot \sgn(\hat{f}_i) + \sum_{k \in [n] \setminus \{i\}} f_k\expectsub{\bar{\xi}_{lj}}{ \xi_{lj}(k) \cdot  \xi_{lj}(i)\cdot  y_{lk} \cdot  u_{ikjl} \mid l_d(i)=l, j \in R_l(i),\G} \notag  \\
& = f_i \cdot \sgn(\hat{f}_i) + \sum_{k \in \stream_l} f_k \cdot \expectsub{\bar{\xi}_{lj}}{ \xi_{lj}(k) \xi_{lj} (i) \mid y_{lk} =1,  u_{ikjl} = 1, j \in R_l(i), \G} \notag  \\
& \hspace*{1.8in}\cdot \prob{u_{ikjl} =1, y_{lk}=1 \mid  j \in R_l(i), \G}  \notag \\
& = f_i \cdot \sgn(\hat{f}_i) + 0
\end{align}
since, $ \xi_{lj}(k)$ and $\xi_{lj}(i)$ depend only on $\bar{\xi}_{lj}$ and is independent of the conditioning events. The expectation is zero by 
pair-wise independence and zero-expectation  of the family $\{\xi_{lj}(s)\}_{s \in [n]}$.

Hence, Eqn.~\eqref{eq:hssavtp:1} becomes
\begin{align} \label{eq:hssavtp:2}
&\expectsub{\bar{\xi}_{lj}}{X_{ijl} \mid l_d(i) = l, j \in R_l(i), \G} = f_i  \cdot \sgn(\hat{f}_i) = f_i \cdot \sgn(f_i) = \abs{f_i}
\end{align}
because,
since, $l_d(i) = l$, $\abs{\hat{f}_{il}} \ge (1-\epsbar) T_l$ and therefore, $ \sgn( \hat{f}_i  \pm \epsbar T_l) = \sgn(\hat{f}_i)$, since, $\epsbar = 1/(27p) < 1/2$.
Since $\G$ holds, by \accuest~we have,   $$\sgn(\hat{f}_i) \sgn(f_i) = \sgn(\hat{f}_i) \sgn(\hat{f}_i \pm \epsbar T_l) = \sgn(\hat{f}_i) \sgn(\hat{f}_i) = 1 $$
and therefore $\sgn(\hat{f}_i)  = \sgn(f_i)$.  Hence Eqn.~\eqref{eq:hssavtp:2}  holds. \hfill

\end{proof}

\subsection{Expectation of $\bvtheta_i$}
\begin{proof} [Proof of Lemma~\ref{lem:fp:expect1}.]
By Lemma~\ref{lem:hssavtp}, we have, $$\expectsub{\bar{\xi}_{lj}}{ X_{ijl} \mid l_d(i) = l, j  \allowbreak \in  \allowbreak R_l(i), \allowbreak  \G} = \abs{f_i}$$  and therefore,
\begin{gather*}\expectsub{\bar{\xi}_l}{X_{ijl} \mid  l_d(i) = l, j \in R_l(i), \G} = \expectsubb{\bar{\xi}_l \setminus \bar{\xi}_{lj}}{\expectsubb{\bar{\xi}_{lj}} {X_{ijl} \mid  l_d(i) = l, j \in \allowbreak  R_l(i), \allowbreak \G}} = \abs{f_i} \enspace .
\end{gather*}
By Lemma~\ref{lem:hssavtp} part (a), if $i$ is discovered at level $l$, then, $\abs{ \hat{f}_{il} -f_i} \le \frac{\abs{f_i}}{26p}$.

Since \nocoll~holds as a sub-event of $\G$, $\abs{R_l(i)} \ge s$.  Let $ \{j_1, j_2, \ldots, j_s\}$ be any $s$-subset of $R_l(i)$  such that $1 \le j_1  < j_2 < \ldots < j_s \le 2s$ and $y \in Y$ be a code with $\pi_y : [k] \rightarrow [k]$ being a random permutation.  Let $y = (y_1, y_2, \ldots, y_k)$ be the  $k$-dimensional increasing sequence $1 \le y_1 < y_2 < \ldots < y_k \le  2s$ representing the $k$ non-zero positions in the $s$-dimensional bit vector $y$.
 Then, $\vartheta_{iyl} = \sum_{v=0}^k \binom{p}{v} \abs{\hat{f}_i}^{p-v} \prod_{r=1}^v (X_{i,j_{y_{\pi(r)}},l} - \abs{\hat{f}_i})$.  Therefore, each $j_{y_{\pi(r)}} \in R_l(i)$, for $1 \le r \le k$.
 \begin{align*}
 &\expectsub{\bar{\xi}_l}{\vartheta_{iyl} \mid l_d(i) =l, \G}\\
 &= \sum_{v=0}^k \binom{p}{v} \abs{\hat{f}_i}^{p-v} \prod_{r=1}^v \left(\expectsub{\bar{\xi}_{l}}{X_{i,j_{y_{\pi(r)}},l} \mid l_d(i) = l, \G, j_{y_{\pi(r)}} \in R_l(i)} - \abs{\hat{f}_i} \right) \\
 & = \sum_{v=0}^k \binom{p}{v} \abs{\hat{f}_i}^{p-v} \prod_{r=1}^v \left( \abs{f_i} - \abs{\hat{f}_i}\right)
 \end{align*}
 which by Corollary ~\ref{cor:fpbias} is bounded above as follows.
 \begin{align*}
 \card{\expectsub{\bar{\xi}_l}{\vartheta_{iyl} \mid l_d(i) =l, \G} - \abs{f_i}^p}& \le \left(\frac{\alpha}{1-\alpha}\right)^{k+1} \left( \frac{ \abs{f_i}}{k+1} \right)^{p} \\
 & \le \left(\frac{ 1/(26p)}{1- 1/(26p)}\right)^{k+1} \abs{f_i}^p \\
& \le (25p)^{-k-1}\abs{f_i}^p \le  n^{-4000p} \abs{f_i}^p
\end{align*}
since $k \ge 1000 \log (n)$.

Since,  $\expect{\bvtheta_{i}} = \expect{\vartheta_{iyl}}$ for each $y \in Y$ and random permutation $\pi_y$, the  lemma follows. Additionally, if $p$ is integral then, $\expect{\bvtheta_i} = \expect{ \vartheta_{il}} = \abs{f_i}^p$. \hfill

\end{proof}


\subsubsection{Probability that two items collide conditional on the event \text{\sc nocollision}}
We first prove a lemma that bounds the probability that two distinct items collide under a hash function $h_{lj}$ conditional on $j$ being in $ R_l(i)$. 
\begin{lemma} \label{lem:probcoll} Let $l_d(i) =l$, $k \in \stream_l$  and $k \not\in \hattopk(C_l)$ and $i \ne k$. If the degree of independence of the hash family from which the hash functions $h_{lj}$ are drawn is at least 11, then,
\begin{align*}
1. & ~
\prob{u_{ikjl} =1 \mid l_d(i)=l, j \in R_l(i), k \in \stream_l, k \not\in \hattopk(C_l)}  \\
 & ~~\in \left( 1- \frac{1}{16C_l} \right)^{C_l - 0.5 \mp 0.5} \pm  2 \binom{C_l}{t-1} \left( \frac{1}{16 C_l} \right)^{t-1} ,\text{ and, }\\
2. & ~\prob{u_{ikjl} =1 \mid l_d(i)=l, j \in R_l(i), k \in \stream_l, k \not\in \hattopk(C_l), \G} \\
 & ~~~\in  \left( 1- \frac{1}{16C_l} \right)^{C_l - 1} \pm  2 \binom{C_l}{t-1} \left( \frac{1}{16 C_l} \right)^{t-1} \pm O(n^{-c}) \enspace .
\end{align*}
\end{lemma}
 \begin{proof}
Since $u_{ikjl}=1$ is equivalent to $h_{lj}(i) = h_{lj}(k)$, we have,
\begin{align} \label{eq:bayes:1}& \probsub{t}{u_{ikjl} = 1 \mid l_d(i)=l, j \in R_l(i), k \in \stream_l, k \not\in \hattopk(C_l) } \notag  \\
& =
\left(\frac{\probsub{t}{ j \in R_l(i) \mid u_{ikjl}=1, l_d(i) =l,  k \in \stream_l, k \not\in \hattopk(C_l)} }{\probsub{t}{j \in R_l(i) \mid  l_d(i) =l, k \in \stream_l , k \not\in\hattopk(C_l)}  }\right)  \notag \\
& \hspace{1.0cm}  \cdot \probsub{t}{u_{ikjl}=1 \mid l_d(i)=l, k \in \stream_l, k \not\in \hattopk(C_l)} \enspace .
\end{align}
First,
\begin{gather} \label{eq:bayes:2}
\probsub{t}{u_{ikjl}=1 \mid l_d(i)=l, k \in \stream_l, k \not\in \hattopk(C_l)}  \notag \\ =
\probsub{t}{u_{ikjl}=1} = \frac{1}{16C_l} \pm \left( \frac{e}{16t} \right)^{t}
\end{gather}
since, the event $u_{ikjl}=1$ depends solely  on $\bar{h}_{lj}$ and is independent of the events $k \in \stream_l$ and $k \not\in \hattopk{C_l}$.

\noindent
Secondly,
\begin{align} \label{eq:nc:3}
&\textsf{Pr}_t \left[j \in R_l(i) \mid u_{ikjl}=1, l_d(i) =l, k \in \stream_l, k \not\in \hattopk(C_l) \right] \notag \\
&= \textsf{Pr}_t \left[ \forall i' \in \hattopk(C_l)\setminus\{ i \} ( h_{lj}(i') \ne h_{lj}(i)) \mid u_{ikjl}=1,  l_d(i) =l, k \in \stream_l, k \not\in \hattopk(C_l)\right] \notag \\
& = \frac{\textsf{Pr}_t \left[\left(\forall i' \in \hattopk(C_l)\setminus\{ i \} ( h_{lj}(i') \ne h_{lj}(i))\right)\text{ and } u_{ikjl}=1\mid l_d(i)=l, k \in \stream_l, k \not\in \hattopk(C_l)\right]}{\textsf{Pr}_t \left[u_{ikjl}=1 \mid l_d(i)=l, k \in \stream_l, k \not\in \hattopk(C_l)\right]}  \notag \\
& =\frac{\textsf{Pr}_t \left[\left(\forall i' \in \hattopk(C_l)\setminus\{ i \} ( h_{lj}(i') \ne h_{lj}(i))\right)\text{ and } u_{ikjl}=1 \mid l_d(i)=l, k \in \stream_l, k \not\in \hattopk(C_l)\right]}{\textsf{Pr}_t \left[u_{ikjl}=1\right]}  \enspace .
\end{align}
$u_{ikjl}=1$ is a function solely of $\bar{h}_{lj}$ and  it is independent of the events $l_d(i)=l$ and $k \not\in \hattopk(C_l)$. Hence, the denominator term in Eqn.~\eqref{eq:nc:3} is simply $\probsub{t}{u_{ikjl}=1}$.

Consider the numerator of Eqn.~\eqref{eq:nc:3}. Let  $A = \hattopk(C_l)$, $\abs{A} = k$. Then,
\begin{align} \label{eq:nc:4}
&\prob{\left(\forall i' \in \hattopk(C_l)\setminus\{ i \} ( h_{lj}(i') \ne h_{lj}(i))\right)\text{ and } u_{ikjl}=1\mid l_d(i)=l, k \in \stream_l, k \not\in \hattopk(C_l)}\notag \\
& =
\sum_{A \subset [n], \abs{A} = C_l} \textsf{Pr}_t \left[\left(\forall i' \in A \setminus\{ i \} ( h_{lj}(i') \ne h_{lj}(i))\right), u_{ikjl}=1 
\mid l_d(i)=l, k \in \stream_l, k \not\in A, \hattopk(C_l)=A \right] \notag \\ &\hspace*{1.0in} \cdot  \probsub{\bar{g} \oplus \bar{h}_{HH}}{\hattopk(C_l) = A \mid l_d(i)=l, k \in \stream_l, k \not\in A} \notag \\
& = \sum_{\substack{A \subset [n] \\\abs{A} = C_l}} \textsf{Pr}_t \left[\left(\forall i' \in A \setminus\{ i \} ( h_{lj}(i') \ne h_{lj}(i))\right), u_{ikjl}=1\right]
\prob{\hattopk(C_l) = A \mid l_d(i)=l, k \in \stream_l, k \not\in A}
\end{align}
since for a fixed $A$, the event $\left\{\forall i' \in A \setminus\{ i \} ( h_{lj}(i') \ne h_{lj}(i))\text{ and } u_{ikjl}=1\right\}$ is independent of the events $l_d(i)=l, k \in \stream_l$ and $ k \not\in A$.

We now estimate the probability $\probsub{t}{\left(\forall i' \in A \setminus\{ i \} ( h_{lj}(i') \ne h_{lj}(i))\right),  u_{ikjl}=1}$.
The event $\bigl\{\forall i' \in A \setminus\{ i \} ( h_{lj}(i') \ne h_{lj}(i)), u_{ikjl}=1\bigr\}$ is equivalent to $ \left(\neg \bigvee_{i' \in A \setminus \{i\}}  (u_{ii'jl} = 1) \right) \wedge (u_{ikjl} = 1)$. Therefore, by inclusion-exclusion, we have,
\begin{align*}
&\probsub{t}{\left(\neg \bigvee_{i' \in A \setminus \{i\}}  (u_{ii'jl} = 1) \right) \wedge (u_{ikjl} = 1)} \\
& = \probsub{t}{ \neg \bigvee_{i' \in A \setminus \{i\}}  (u_{ii'jl} = 1)  \mid u_{ikjl}=1}\probsub{t}{u_{ikjl}=1}\notag \\
& = \left( 1- \probsub{t}{ \bigvee_{i' \in A \setminus \{i\}}  (u_{ii'jl} = 1)  \mid u_{ikjl}=1}\right)\probsub{t}{u_{ikjl}=1}
\end{align*}
Following the inclusion-exclusion arguments as in Lemma~\ref{lem:nc} and using the notation that  that $P[\cdot]$ denotes the probability measure assuming full-independence of the same  hash family, we have,
\begin{align*}
& \left \lvert \left( 1- \probsub{t}{ \bigvee_{i' \in A \setminus \{i\}}  (u_{ii'jl} = 1)  \mid u_{ikjl}=1}\right)  - \left( 1 - P\left[\bigvee_{i' \in A \setminus \{i\}}  (u_{ii'jl} = 1)  \mid u_{ikjl}=1 \right] \right) \right \rvert \\
& \le 2 \sum_{\{i_1, i_2, \ldots, i_{t-1}\} \subset A \setminus \{i\}} P\left[ \bigwedge_{r=1}^{t-1} u_{i i_r jl}=1 \mid u_{ikjl}=1\right] \\
&\le 2 \binom{C_l}{t-1} \left( \frac{1}{16 C_l} \right)^{t-1}
\end{align*}
Therefore,
\begin{align*}
&\textsf{Pr}_t \left[\left(\forall i' \in A \setminus\{ i \} ( h_{lj}(i') \ne h_{lj}(i))\right), u_{ikjl}=1\right]\\
 &= \left( 1- \probsub{t}{ \bigvee_{i' \in A \setminus \{i\}}  (u_{ii'jl} = 1)  \mid u_{ikjl}=1}\right)\prob{u_{ikjl}=1}\\
& = \left(\left( 1 - P\left[\bigvee_{i' \in A \setminus \{i\}}  (u_{ii'jl} = 1)  \mid u_{ikjl}=1 \right] \right) \pm  2 \binom{C_l}{t-1} \left( \frac{1}{16 C_l} \right)^{t-1}\right)  \prob{u_{ikjl}=1}\\
& = \left(\left( 1- \frac{1}{16C_l} \right)^{C_l - \1_{i \not\in A}} \pm  2 \binom{C_l}{t-1} \left( \frac{1}{16 C_l} \right)^{t-1}\right) \prob{u_{ikjl}=1}
\end{align*}
Now, $ C_l - \1_{i \not\in A} \in C_l - 0.5  \mp 0.5$.

Substituting in Eqn.~\eqref{eq:nc:4}, we have,
\begin{align} \label{eq:nc:5}
&\textsf{Pr}_t \left[\left(\forall i' \in \hattopk(C_l)\setminus\{ i \} ( h_{lj}(i') \ne h_{lj}(i))\right)\text{ and } u_{ikjl}=1 \mid l_d(i)=l, k \in \stream_l, k \not\in \hattopk(C_l)\right] \notag \\
& = \sum_{A \subset [n], \abs{A} = k} \probsub{t}{\left(\forall i' \in A \setminus\{ i \} ( h_{lj}(i') \ne h_{lj}(i))\right)\text{ and } u_{ikjl}=1} \notag \notag \\
& \hspace*{1.0in } \cdot \probsub{\bar{g} \oplus \bar{h}_{HH}}{\hattopk(C_l) = A \mid l_d(i)=l, k \in \stream_l, k \not\in A} \notag \\
& \in  \left( \left( 1- \frac{1}{16C_l} \right)^{C_l - 0.5 \mp 0.5} \pm  2 \binom{C_l}{t-1} \left( \frac{1}{16 C_l} \right)^{t-1} \right) \prob{u_{ikjl}=1} \notag \\
& \hspace*{1.0in} \cdot \sum_{A \subset [n], \abs{A} = k} \probsub{\bar{g} \oplus \bar{h}_{HH}}{\hattopk(C_l) = A \mid l_d(i)=l, k \in \stream_l, k \not\in A} \notag \\
& = \left( \left( 1- \frac{1}{16C_l} \right)^{C_l - 0.5 \mp 0.5} \pm  2 \binom{C_l}{t-1} \left( \frac{1}{16 C_l} \right)^{t-1} \right) \prob{u_{ikjl}=1}
\end{align}
Substituting in Eqn~\eqref{eq:nc:3}, we have,
\begin{align} \label{eq:nc:6}
&\probsub{t}{j \in R_l(i) \mid u_{ikjl}=1, l_d(i) =l, k \in \stream_l, k \not\in \hattopk(C_l)}  \notag \\
& =\left(\frac{1}{\probsub{t}{u_{ikjl}=1}}\right)   \textsf{Pr}_t \left[\left(\forall i' \in \hattopk(C_l)\setminus\{ i \} ( h_{lj}(i') \ne h_{lj}(i))\right)\text{ and } u_{ikjl}=1\right. \notag \\ &\hspace*{1.7in} \left. \mid l_d(i)=l, k \in \stream_l, k \not\in \hattopk(C_l)\right]\notag \\
& =  \left( 1- \frac{1}{16C_l} \right)^{C_l - 0.5 \mp 0.5} \pm  2 \binom{C_l}{t-1} \left( \frac{1}{16 C_l} \right)^{t-1}
\end{align}

In a similar manner, we can show that
\begin{align} \label{eq:nc:7}
&\probsub{t}{j \in R_l(i) \mid  l_d(i) =l, k \in \stream_l , k \not\in\hattopk(C_l)} \notag \\
&= \left( 1- \frac{1}{16C_l} \right)^{C_l - 0.5 \mp 0.5} \pm  2 \binom{C_l}{t} \left( \frac{1}{16 C_l} \right)^{t}
\end{align}
Substituting Eqns.~\eqref{eq:nc:6}, ~\eqref{eq:nc:7} and ~\eqref{eq:bayes:2} in Eqn. ~\eqref{eq:bayes:1}, we have,
\begin{align} \label{eq:nc:8}
&\probsub{t}{u_{ikjl} = 1 \mid l_d(i)=l, j \in R_l(i), k \in \stream_l, k \not\in \hattopk(C_l) } \notag  \notag \\
& =
\left(\frac{\probsub{t}{ j \in R_l(i) \mid u_{ikjl}=1, l_d(i) =l,  k \in \stream_l, k \not\in \hattopk(C_l)} }{\probsub{t}{j \in R_l(i) \mid  l_d(i) =l, k \in \stream_l , k \not\in\hattopk(C_l)}  }\right) \notag \\
& \hspace*{2.0cm} \cdot \probsub{t}{u_{ikjl}=1 \mid l_d(i)=l, k \in \stream_l, k \not\in \hattopk(C_l)} \notag \\
& = \left(\cfrac{ \left( 1- \frac{1}{16C_l} \right)^{C_l - 0.5 \mp 0.5} \pm  2 \binom{C_l}{t-1} \left( \frac{1}{16 C_l} \right)^{t-1}}{\left( 1- \frac{1}{16C_l} \right)^{C_l - 0.5 \pm 0.5} \mp  2 \binom{C_l}{t} \left( \frac{1}{16 C_l} \right)^{t}} \right)\cdot  \left( \frac{1}{16C_l} \pm \left( \frac{e}{16t} \right)^{t}\right)
\end{align}
For $t = 11$, the above ratio is bounded by $\left( \frac{ 1 \pm 10^{-16}}{16 C_l}\right)$.

 Conditioning with respect to $\G$, by Fact~\ref{fact:cndhpev}, the above probability may
 change by $ n^{-c}$. Also, conditioned on $\G$, we have that $l_d(i)=l$ implies that $i \in \hattopk(C_l)$. Hence,
\begin{align*}
 &\probsub{t}{j \in R_l(i) \mid u_{ikjl}=1, l_d(i) =l, k \in \stream_l, k \not\in \hattopk(C_l),\G}  \\
& =  \left( 1- \frac{1}{16C_l} \right)^{C_l - 1} \pm  2 \binom{C_l}{t-1} \left( \frac{1}{16 C_l} \right)^{t-1}  \pm n^{-c}\enspace .
\end{align*}
Proceeding similarly as in Eqn.~\eqref{eq:nc:8}, we have,
\begin{align*}
&\probsub{t}{u_{ikjl} = 1 \mid l_d(i)=l, j \in R_l(i), k \in \stream_l, k \not\in \hattopk(C_l),\G} \notag  \notag \\
& =
\left(\frac{\probsub{t}{ j \in R_l(i) \mid u_{ikjl}=1, l_d(i) =l,  k \in \stream_l, k \not\in \hattopk(C_l),\G} }{\probsub{t}{j \in R_l(i) \mid  l_d(i) =l, k \in \stream_l , k \not\in\hattopk(C_l),\G}  }\right) \\
& \hspace*{1.0in} \cdot \probsub{t}{u_{ikjl}=1 \mid l_d(i)=l, k \in \stream_l, k \not\in \hattopk(C_l),\G}  \notag \\
& = \left( \cfrac{ \left( 1- \frac{1}{16C_l} \right)^{C_l - 0.5 \mp 0.5} \pm  2 \binom{C_l}{t-1} \left( \frac{1}{16 C_l} \right)^{t-1} \pm n^{-c}}{\left( 1- \frac{1}{16C_l} \right)^{C_l - 0.5 \pm 0.5} \mp  2 \binom{C_l}{t} \left( \frac{1}{16 C_l} \right)^{t} \mp n^{-c} }\right) \cdot  \left( \frac{1}{16C_l} \pm \left( \frac{e}{16t} \right)^{t} \pm n^{-c}\right)
\end{align*}
For $t = 11$, the above ratio is bounded by $\left( \frac{ 1 \pm 10^{-16}}{16 C_l}\right)$. \hfill

 \end{proof}

\subsection{Basic properties of the application of Taylor polynomial estimator: Proof of Lemma~\ref{lem:hssavtp}-Part II}
We now complete the proofs of the remaining parts of Lemma~\ref{lem:hssavtp}.

\begin{proof} [Proof of Lemma~\ref{lem:hssavtp}, parts (c), (d) and (f).]
Recall that $y_{lk}$ is an indicator variable that is 1 iff $k \in \stream_l$. Given that $i \in \stream_l$ the random variable $X_{ijl}$ is defined as
$$X_{ijl} =  (f_i  + \sum_{k \ne i } f_k \cdot u_{ikjl} \cdot \xi_{lj}(k) \cdot \xi_{lj}(i) \cdot y_{lk}) \sgn(\hat{f}_i) \enspace . $$
As shown in the proof of Lemma~\ref{lem:fp:expect1}, $\expectsub{\bar{\xi}_{lj}}{X_{ijl} \mid j \in R_l(i), l_d(i)= l, \G} = \abs{f_i}$.
Further,
\begin{align*}
&\expect{X_{ijl}^2 \mid j \in R_l(i), l_d(i) = l, \G} \\
&= \expectsub{\bar{h}_{HH,l} \oplus \bar{h}_{lj} \oplus \bar{g}}{\expectsub{\bar{\xi}_{lj}}{X_{ijl}^2 \mid j \in R_l(i), l_d(i) = l,  \G}} \\
& = f_i^2 + \expectsub{\bar{h}_{HH,l} \oplus \bar{h}_{lj} \oplus \bar{g}}{\sum_{k \in[n] \setminus\{i\} } f_k^2\cdot  u_{ijkl} \cdot y_{lk}\mid j \in R_l(i), l_d(i) = l, \G}
\end{align*}
since the expectation with respect to the Rademacher family of \tpest~ structure is independent of the random bits used to define $\G$ and $R_l(i)$.

Therefore,
\begin{align} \label{eq:hssavtp2:a}
\sigma^2_{ijl}  & =\varsub{\bar{\xi}_{lj} \oplus \bar{h}_{HH,l} \oplus \bar{h}_{lj} \oplus \bar{g}}{X_{ijl} \mid j \in R_l(i), l_d(i) = l,\G} \notag\\
& = \expectsub{\bar{\xi}_{lj} \oplus \bar{h}_{HH,l} \oplus \bar{h}_{lj} \oplus \bar{g}}{X_{ijl}^2 \mid j \in R_l(i), l_d(i) = l,\G} \notag \\
 & \hspace*{1.0cm} - \left( \expectsub{\bar{\xi}_{lj} \oplus \bar{h}_{HH,l} \oplus \bar{h}_{lj} \oplus \bar{g}}{X_{ijl} \mid j \in R_l(i), l_d(i) = l,\G}\right)^2\notag\\
& = f_i^2  + \expectsub{\bar{g} \oplus \bar{h}_{HH,l} \oplus \bar{h}_{lj}}{\sum_{k \in [n] \setminus \{i\} } f_k^2 \cdot  u_{ikjl} \cdot y_{lk}\mid j \in R_l(i), l_d(i) =l,  \G} - \abs{f_i}^2 \notag\\
&  = \sum_{k \in [n] \setminus \{i\}} f_{k}^2 \cdot \probsub{\bar{g} \oplus \bar{h}_{HH,l} \oplus \bar{h}_{lj}}{u_{ikjl} =1 \mid j \in R_l(i), l_d(i) =l, k \in \stream_l,\G}\notag\\
& \hspace*{1.0in}
\cdot \probsub{\bar{g} \oplus \bar{h}_{HH,l} \oplus \bar{h}_{lj}} {y_{lk}=1 \mid j \in R_l(i), l_d(i)=l, \G} \enspace .
\end{align}
Now,
\begin{align} \label{eq:hssavtp2:aa}
&\prob{u_{ikjl} =1 \mid j \in R_l(i), l_d(i) =l, k \in \stream_l,\G}\\
 &= \prob{u_{ikjl}=1, k \not\in \hattopk(C_l)\mid j \in R_l(i), k \in \stream_l,  l_d(i)=l, \G} \notag \\
& \hspace*{1.0in} + \prob{u_{ikjl}=1, k \in  \hattopk(C_l)\mid j \in R_l(i), l_d(i)=k, \G} \notag \\
& = \prob{u_{ikjl} =1 \mid j \in R_l(i), k \in \stream_l, k \not\in \hattopk(C_l), l_d(i) =l,\G}\notag \\
 & \hspace*{1.0in} \cdot \prob{k \not\in \hattopk(C_l) \mid j \in R_l(i), k \in \stream_l,  l_d(i) = l, \G} +0 \notag \\
& \le \left( \frac{1 + 10^{-16}}{(16C_l)} \right) \cdot \prob{k \not\in \hattopk(C_l) \mid j \in R_l(i), k \in \stream_l,  l_d(i) = l, \G}
 \end{align}
by (Lemma~\ref{lem:probcoll}, with $t=11$.

 Substituting in ~\eqref{eq:hssavtp2:a}, we have that
 \begin{align} \label{eq:hssavtp2:b}
 \sigma^2_{ijl}   &  \le  \sum_{k \in [n] \setminus \{i\}} f_{k}^2\cdot \left( \frac{1 + 10^{-16}}{(16C_l)} \right) \notag \\ & \hspace*{1.0cm}     \cdot \prob{k \not\in \hattopk(C_l), k \in \stream_l \mid j \in R_l(i),  l_d(i) = l, \G}    \notag \\
 & \le   \left( \frac{1 + 10^{-16}}{(16C_l)} \right) \sum_{\substack{k \in [n] \setminus \{i\}, k \in \stream_l, k \not\in \hattopk(C_l)}} f_k^2 \cdot 1 \notag \\
 & = \left( \frac{1 + 10^{-16}}{(16C_l)} \right) \ftwores{\hattopk(C_l),l}
 \end{align}
It can be shown     that, conditional on \goodest,  $$\ftwores{\hattopk(C_l),l} \le 9 \ftwores{C_l,l}$$ (this is
explicitly proved in \cite{g:isaac12}; variants appear in  earlier works for e.g.,
\cite{gks:fsttcs05,cm:sirocco06,jst:pods11}). Since, $\smallres$ holds as a sub-event of $\G$,
$\ftwores{C_l,l} \le   1.5\ftwores{  (2\alpha)^l C}/\allowbreak 2^{l-1} $. Therefore, Eqn.~\eqref{eq:hssavtp2:b} may be written as  follows.
\begin{align*}
\sigma^2_{ijl}   &  \le   \left( \frac{1 + 10^{-16}}{(16C_l)} \right) \ftwores{\hattopk(C_l),l} \notag \\
 &\le \frac{9 (1 + 10^{-16})
\ftwores{C_l,l}}{ 16 C_l}  ~~~\text{ since,  ($\ftwores{\hattopk_l(C_l),l} \le 9 \ftwores{C_l,l}$)  \cite{g:isaac12,jst:pods11}}\\
&  \le  \frac{  9 (1 + 10^{-16}))(1.5)\ftwores{  (2\alpha)^l C}}{C_l (16) 2^{l-1} } ~~~~~~~~~~~~ \text{ (\G~ implies \smallres.)}\\
&  \le \frac{  9 (1 + 10^{-16})(1.5)\hat{F}_2}{8  (2\alpha)^l  C}\\
 &\le (17/10)(\epsbar T_l)^2 \enspace .
\end{align*}
This proves part (d) of Lemma~\ref{lem:hssavtp}.

Hence,
$$\eta^2_{ijl} =  \abs{ \hat{f}_{il} - f_i}^2 + \sigma^2_{ijl} \le  (\epsbar T_l)^2  + (17/10)(\epsbar
T_l)^2\le 2.7(\epsbar T_l)^2 \enspace . $$
Since, $i$ is discovered at level $l$, $\abs{\hat{f}_{il}} \ge Q_l = T_l(1-\epsbar)$ and therefore, $\abs{f_i} \ge Q_l - \epsbar T_l = T_l(1-2\epsbar)$.

 Hence,
$\frac{\abs{f_i}}{\eta_{ijl}} \ge    \frac{T_l (1 - 2\epsbar)}{(\sqrt{2.7}) \epsbar T_l }  \ge 15p$. Further, $\frac{\abs{\hat{f}_i}}{\eta_{ijl}} \ge \frac{T_l(1-\epsbar)}{\sqrt{2.7} \epsbar T_l} \ge 16p$. This proves parts (c) and (f).
\hfill
\end{proof}

\subsection{Taylor polynomial estimators are uncorrelated with respect to $\bar{\xi}$}

\begin{proof} [Proof of Lemma~\ref{lem:bvthetacross}.] The expectations in this proof are only  with respect to $\bar{\xi}$.

Consider $\expectsub{\xibar}{\bvtheta_{i'}\bvtheta_i}$. $\bvtheta_i$ and $\bvtheta_{i'}$ each use the
\est~structure at levels $l_d(i)$ and $l_d(i')$ respectively. If $l_d(i) \ne l_d(i')$, then the
estimations are made from different structures and use independent random bits and
therefore,
$$\expectsub{\xibar}{ \bvtheta_i \bvtheta_{i'} \mid \hat{f}_i, \hat{f}_{i'} \G } = \expectsub{\xibar}{ \bvtheta_i
\mid  \hat{f}_i,\G } \expectsub{\xibar}{\bvtheta_{i'} \mid \hat{f}_i, \G}\enspace . $$

Now suppose that $l_d(i) = l_d(i') = l$ (say). Then,  $ \abs{\hat{f}_{il} } \ge Q_l$ and
$\abs{\hat{f}_{i'l}} \ge Q_{l}$. Since $\smallhh$ holds as a sub-event of $\G$, $\{i,i'\}\subset \{k: \abs{\hat{f}_{kl}} \ge Q_l\} \subset \hattopk(l,C_l)$.
Therefore,  by  $\nocollision_l$, the estimates  $\{X_{ijl}\}_{j \in R_l(i)}$ and $\{
X_{i'jl} \}_{j \in R_l(i')}$ are such that if $j \in R_l(i) \cap R_l(i')$, then, $h_{lj}(i) \ne h_{lj}(i')$.
Let $q_1, q_2, \ldots, q_s$ be some permutation of the table indices in $R_l(i)$. Likewise let $q'_1, q'_2,
\ldots, q'_s$ be a permutation of the table indices in $R_l(j)$.  Then,
\begin{align}\label{eq:varcross1}
&\expectsub{\xibar}{\vartheta_i\vartheta_{i'} \mid \hat{f}_i, \hat{f}_{i'},\G} \notag \\
&= \E_{\xibar}\left[
\left(\sum_{v=0}^k \gamma_v(\abs{\hat{f}_i}) \prod_{w=1}^v (X_{i,q_w, l} -
\abs{\hat{f}_i})\right) \left( \sum_{v'=0}^k \gamma_{v'}(\abs{\hat{f}_{i'}}) \prod_{w'=1}^{v'}(X_{i',
q'_w, l} -\abs{ \hat{f}_{i'}}) \right) \right. \notag \\
& \hspace*{1.0in} \left. \biggl\vert \hat{f}_i,\hat{f}_{i'},\G\right] \notag \\
& = \sum_{v,v'=0}^{k} \gamma_v(\abs{\hat{f}_i}) \gamma_{v'}(\abs{\hat{f}_{i'}}) \expectsub{\xibar}{
\prod_{w=1}^v (X_{i,q_w, l} - \abs{\hat{f}_i})\prod_{w'=1}^{v'}(X_{i', q'_{w'}, l} - \abs{\hat{f}_{i'}}) \mid \hat{f}_i, \hat{f}_{i'},\G}
\end{align}

Consider
$
\expectsub{\xibar}{
\prod_{w=1}^v (X_{i,q_w, l} - \abs{\hat{f}_i})\prod_{w'=1}^{v'}(X_{i', q'_{w'}, l} - \abs{\hat{f}_{i'}}) \mid \hat{f}_i, \hat{f}_{i'},\G}$. For some $1\le w' \le v'$, if $q'_{w'} \not\in \{q_1, q_2, \ldots, q_w\}$, then, the random variable $ X_{i',q'_{w'},l} - \abs{\hat{f}_{i'}}$ uses only the random bits of $\xi_{lq'_{w'}}$ and is independent of the random bits $\{\xi_{l,q_w} \mid 1\le w \le v\}$  used by any of the $X_{i,q_w,l}$, for $1 \le w \le v$. An analogous situation holds for any $1 \le w \le v$ such that $q_w \not\in \{q'_1, \ldots, q'_{v'}\}$. Clearly, for distinct tables, $j, j'$,  $\expectsub{\xibar}{X_{i,j,l}X_{i',j',l}}$  is the product of the
individual expectations, by independence of  the seeds of the
Rademacher families $\{\xi_{lj}(k)\}$ and $\{\xi_{lj'}(k)\}$.  Therefore,
\begin{align*}
&\expectsub{\xibar}{
\prod_{w=1}^v (X_{i,q_w, l} - \abs{\hat{f}_i})\prod_{w'=1}^{v'}(X_{i', q'_{w'}, l} - \abs{\hat{f}_{i'}}) \mid \hat{f}_i, \hat{f}_{i'},\G} \\
& = \prod_{w: q_w\not\in \{q'_1, \ldots, q'_{v'}\}} \expectsub{\xi_{l,q_w}}{(X_{i,q_w, l} - \abs{\hat{f}_i})\mid \hat{f}_i, \G} \\ & ~~~~\cdot \prod_{w': q'_{w'} \not\in \{q_1, \ldots, q_v\}} \expectsub{\xi_{l,q'_{w'}}}{(X_{i',q'_{w'}, l} - \abs{\hat{f}_{i'}}) \mid \hat{f}_{i'},\G} \\
& ~~~~\cdot \prod_{j \in \{q_1, \ldots, q_v\} \cap \{q'_1, \ldots, q'_{w'}\}}
\expectsub{\xi_{lj}}{(X_{ijl}- \abs{\hat{f}_i}) (X_{i'jl} - \abs{\hat{f}_{i'}})\mid \hat{f}_i, \hat{f}_{i'},\G}
\end{align*}

\eat{

Therefore,
~\eqref{eq:varcross1} becomes
\begin{align} \label{eq:varcross11}
& \expectsub{\xibar}{\vartheta_i\vartheta_{i'}\mid \hat{f}_{il}, \hat{f}_{i'l}, \G }  \notag \\
& = \sum_{v,v'=0}^{k} \gamma_v(\hat{f}_{il}) \gamma_{v'}(\hat{f}_{i'l})
\left(\prod_{\substack{ w \in [v] \\ q_w \not\in \{q'_1, \ldots, q'_{v'}\} }} \expectsub{\xibar}{X_{i,q_w, l} -
\abs{\hat{f}_{il}} \mid \hat{f}_{il},  \hat{f}_{i'l},\G}\right) \left(
\prod_{\substack{ w' \in [v'] \\ q'_{w'}\not\in \{q_1, \ldots, q_{v}\} }}
\expectsub{\xibar}{X_{i',q'_{w'}, l} - \abs{\hat{f}_{i'l}} \mid \hat{f}_{il},  \hat{f}_{i'l},\G} \right)\notag \\
& ~~~\times \prod_{  j \in \{ q_1, \ldots, q_v\} \cap \{q'_1, \ldots, q'_{v'}\}}
\expectsub{\xibar}{ \bigl(X_{i,j, l} -\abs{ \hat{f}_{il}}\bigr)\bigl( X_{i',j, l} - \abs{\hat{f}_{i'l}}\bigr)
\mid \hat{f}_{il},  \hat{f}_{i'l},\G}
\end{align}
}
We analyze $\expectsub{\xi_{lj}}{ X_{ijl}X_{i'jl} \mid  \hat{f}_{il}, \hat{f}_{i'l}\G}$.
\begin{align} \label{eq:varcross2}
&\expectsub{\xi_{lj}}{ X_{ijl} X_{i'jl} \mid \hat{f}_{il}, \hat{f}_{i'l}, \G } =
 \sgn(f_i ) \sgn(f_j) \notag \\
 & ~~\cdot \E_{\xi_{lj}}\biggl[ \Bigl( f_i + \xi_{lj}(i)\sum_{k
\ne i} f_k \cdot \xi_{lj}(k)  \cdot u_{ikjl} \Bigr) \cdot  \Bigl( f_{i'} +
 \xi_{lj}(i')\sum_{k' \ne i'} f_{k'} \cdot  \xi_{lj}(k') \cdot  u_{i'k'jl}\Bigr) \notag \\
 & \hspace*{1.0in} \biggl\vert \hat{f}_{il}, \hat{f}_{i'l}, \G \biggr]
 \end{align}
 Suppose we use linearity of expectation to expand the product and take the expectation of the individual terms.
 The expectation of the terms of the form $\expectsub{\xi_{lj}}{\xi_{lj}(i) \xi_{lj}(k) u_{ikjl}} = 0$ since $i \ne k$ and the random variable $u_{ikjl}$ is independent of $\xi_{lj}$.  Similarly, $\expectsub{\xi_{lj}}{\xi_{lj}(i') \xi_{lj}(k') u_{i'k'jl}} = 0$. We also obtain a set of terms of the form
 $\expectsub{\xi_{lj}}{\xi_{lj}(i) \cdot \xi_{lj}(i')  \cdot \xi_{lj}(k) \cdot \xi_{lj}(k') \cdot u_{ikjl} \cdot u_{i'k'jl}}$. Since, $j \in R_l(i) \cap R_l(i')$, $h_{lj}(i) \ne h_{lj}(i')$. Now  $u_{ikjl} \cdot u_{i'k'jl} =1$ only  if $h_{lj}(i) = h_{lj}(k)$ and $h_{lj}(i') = h_{lj}(k')$. We conclude that $\{i,i',k,k'\}$ are all distinct, and by 4-wise independence of the $\{\xi_{lj}(u)\}_{1\le u \le n}$ family, $\expectsub{\xi_{lj}}{\xi_{lj}(i) \cdot \xi_{lj}(i')  \cdot \xi_{lj}(k) \cdot \xi_{lj}(k') \cdot u_{ikjl} \cdot u_{i'k'jl}} =0$. Therefore, Eqn. ~\eqref{eq:varcross2} becomes
 \begin{align*}
 \expectsub{\xi_{lj}}{ X_{ijl} X_{i'jl} \mid \hat{f}_{il}, \hat{f}_{i'l}, \G}  = \abs{f_i}\abs{f_{i'}} = \expectsub{\xi_{lj}}{X_{ijl} \mid \hat{f}_{il},\G} \expectsub{\xi_{lj}}{X_{ijl} \mid \hat{f}_{i'l}, \G} \enspace .
 \end{align*}
It follows that
\begin{align*}
& \expectsub{\xi_{lj}}{ (X_{ijl} - \abs{\hat{f}_{il}} ) ( X_{i'jl} - \abs{\hat{f}_{i'l}})  \mid \hat{f}_{il}, \hat{f}_{i'l}, \G}  =  (\abs{f_i} - \abs{\hat{f}_{il}})) (\abs{f_{i'}} - \abs{\hat{f}_{i'l}}) \\
&= \expectsub{\xi_{lj}}{X_{ijl} - \abs{\hat{f}_{il}} \mid \hat{f}_{il},\G}\expectsub{\xi_{lj}}{X_{i'jl} - \abs{\hat{f}_{i'l}}\mid \hat{f}_{il},\G} \enspace .
\end{align*}
 For $l_d(i) = l_d(i') = l$, ~\eqref{eq:varcross1} simplifies to
\begin{align*}
\expectsub{\xi_l}{\vartheta_i\vartheta_{i'} \mid \hat{f}_{il}, \hat{f}_{i'l},\G} & = \expectsub{\xi_l}{\vartheta_i \mid \hat{f}_{il},\G} \expectsub{\xi_l}{\vartheta_{i'} \mid \hat{f}_{i'l},\G} \enspace .
\end{align*}
Thus, $\vartheta_i$ and $\vartheta_{i'}$ are uncorrelated in all cases.

Since, $\bvtheta_i$ is the average of the Taylor polynomial estimators $\vartheta_i$ for randomly chosen permutations, the variables $\bvtheta_i $ and $\bvtheta_{i'}$ are also uncorrelated in all cases, whether $l \ne l'$ or  $l = l'$, that is,
\begin{align} \label{eq:varcross3}
\expectsub{\xibar}{\bvtheta_i \bvtheta_{i'} \mid \hat{f}_{il}, \hat{f}_{i'l'},\G}
 = \expectsub{\xi_l}{\bvtheta_i \mid \hat{f}_{il},\G} \expectsub{\xi_{l'}}{\bvtheta_{i'} \mid \hat{f}_{i'l'},\G}
\end{align}
\end{proof}

\eat{
\begin{lemma} \label{lem:hsslimindep}
Let $i_1, i_2, \ldots, i_d \in [n]$ be distinct such that $i_j \in G_{l_j}$, for $1\le j \le d$. Then,
\begin{align*}
\sum_{0\le r_1, r_2, \ldots, r_d \le L} 2^{r_1 + r_2 + \ldots + r_d} \prob{ i_j \in \bar{G}_{r_j}, \text{ for } 1\le j \le d \mid \G} \in 1 \pm O(\sum_{j=1}^d 2^{l_j} n^{-c}) \enspace .
\end{align*}
\end{lemma}

\begin{proof}
The proof if by induction on $d$. The base case is when $d=1$. In this case the statement of the lemma becomes, for $i \in G_l$,  $ \sum_{0 \le r \le L} 2^{r} \prob{i \in \bar{G}_r\mid \G} \in 1 \pm O(2^l n^{-c})$, which  is the statement of Lemma~\ref{lem:margin}.

\emph{Induction Case:} We will assume by induction that

\begin{align} \label{eq:induction}
\sum_{0 \le r_2, \ldots, r_d \le L} 2^{\sum_{j=2}^d r_j} \prob{i_j \in \bar{G}_{r_j}, 2\le j \le d \mid \G} \in 1 \pm O(\sum_{j=2}^d 2^{l_j} n^{-c}) \enspace .
\end{align}
We consider three cases, namely, $i \in \midreg(G_l), i \in \lmargin(G_l)$ and $i \in \rmargin(G_l)$.

\emph{Case 1: $i \in \midreg(G_l)$.} Then, conditional on $\G$, by Lemma~\ref{lem:margin}, $i \in \bar{G}_l$ iff $i \in \stream_l$.
\begin{align*}
\prob{i \in \bar{G}_l}
\end{align*}

\end{proof}
}
\section{Expectation and Variance of $p$th moment estimator }
\label{sec:variance}
In this section, we  analyze the expectation and variance of the estimator $\hat{F}_p$.

\subsection{Expectation of the $\hat{F}_p$ estimator}

\begin{proof} [Proof of Lemma~\ref{lem:Y}.] Define $\level:[n] \rightarrow \{0,1,2\ldots, L+1\}$ to be the function that maps each item $i \in [n]$ to the index of the group it belongs to, that is,
\begin{align*}
\level(i) = \begin{cases}
 l  & \text{ if $i \in G_l$}  \\
 L+1 & \text{  if $f_i = 0$. }
 \end{cases}
 \end{align*}
Then, by definition of the $Y_i$'s,
\begin{align*}
 \expectb{ \hat{F}_p \mid \G} & = \expectbb{ \sum_{i\in [n]} Y_i \mid \G} \\
  &= \sum_{l=0}^L \sum_{i \in G_l}  \sum_{l'=0}^{L}2^{l'} \expect{ z_{il'} \bvtheta_i \mid \G}\\
 &= \sum_{l=0}^L \sum_{i \in G_l} \sum_{l'=0}^{L}2^{l'} \expect{\bvtheta_i \mid i \in \bar{G}_{l'}, \G} \prob{i \in \bar{G}_{l'} \mid \G}\\
& = \sum_{l=0}^L \sum_{i \in G_l} \sum_{l'=0}^{L}2^{l'}  \abs{f_i}^p (1 \pm n^{-4000p}) \prob{i \in \bar{G}_{l'} \mid \G}, ~~~\text{ by Lemma~\ref{lem:fp:expect1} }
 \\
& = \sum_{l=0}^L \sum_{i \in G_l} \abs{f_i}^p (1 \pm n^{-4000p}) \sum_{l'=0}^L 2^{l'} \prob{i \in \bar{G}_{l'} \mid \G} \\
&= \sum_{l=0}^L \sum_{i \in G_l} \abs{f_i}^p (1 \pm n^{-4000p}) (1 \pm O(2^{\level(i)} n^{-c})), ~~~\text{ by  Lemma~\ref{lem:margin} } \\
&= F_p (1 \pm 2^{L+1} n^{-c}))  \enspace .
\end{align*}
Let $C = K' n^{1-2/p}$ where $K' = \cfrac{(27p)^2 \epsilon^{-2}}{\min (\epsilon^{4/p-2}, \log n)}$, as given in  Figure~\ref{table:params}. 

Since $\alpha = 1 - (1-2/p)(0.01) > 0.99$, $$L = \lceil \log_{2\alpha} (n/C) \rceil \le 1 + \log_{1.98}(n/C) \le 1 + (1.02) \log_2(n/C) \le  1 + (1.02)\log_2 (n^{2/p}/K') \enspace . $$ Hence, $$2^L \le 2\left(\frac{n^{(2/p)}}{K'} \right)^{1.02}$$ and so $O(2^{L+1} n^{-c})  = O(n^{-(c-2)})$
 proving the lemma. \hfill
 \end{proof}

 \subsection{Variance of $Y_i$}
 
 In this section, we calculate $\variance{Y_i}$. For sake of completeness we first present proofs of some identities stated in Eqn.~\eqref{eq:fpfacts}.
 
 \begin{fact} \label{fact:fp1}
 For any $p \ge q$, $F_q \le n^{1-q/p}F_p^{q/p}$. In particular, $F_2 \le n^{1-2/p} F_p^{2/p}$ for any $p \ge 2$.
 \end{fact}
 
 \begin{proof}
 Let $X$ be a random variable that takes the value $\abs{f_i}^q$ with probability $1/n$, for $i \in [n]$.  Then,
 $$\expect{X} = \frac{F_q}{n} \enspace . $$
 
 \noindent By Jensen's inequality, for any function $f$ that is convex over the support of $X$, $\expect{f(X)} \ge f(\expect{X})$. Choose $f(t) = t^{p/q}$. Since $p \ge q$ and the support of $X$ is $\R^{\ge 0}$, $f(t) $ is convex in this range. Therefore, $\expect{f(X)} = \frac{F_p}{n}$.  By Jensen's inequality applied to $f$, we have,
 $$ \left( \frac{F_q}{n} \right)^{p/q} \le \frac{F_p}{n}, ~~~
 \text{ or, }~~~  F_q \le n^{1-q/p} F_p^{q/p} \enspace . $$
 \end{proof}

 In the following proofs, we will use the notion that the  \emph{sample group of an item is  consistent with the  frequency of the item} to mean that if $i \in G_l$ and $i$ is sampled into $\bar{G}_r$, then, $l$ and $r$ are related as given by Lemma~\ref{lem:margin}, conditional on  \G.  (For e.g., if $i \in \lmargin(G_l)$, then, $r \in \{l,l+1\}$, if $i \in \midreg(G_l)$, then, $r =l$, and if $i \in \rmargin(G_l)$, then, $r \in \{l-1,l\}$).

\begin{proof} [Proof of Lemma~\ref{lem:varYi}.] For this proof, assume that $\G$ holds.

\emph{Case 1:} $i \in \midreg(G_0)$. Then  $i \in \bar{G}_0$ with probability 1 and $l_d(i) = 0$. Therefore,  $$Y_i = \sum_{l=0}^{L-1} 2^l \cdot z_{il} \cdot  \bvtheta_{il} = \bvtheta_{i0}$$ since, $z_{i0} = 1$ and $z_{il} = 0$ for $l > 0$. Let $\bvtheta_{i}$ denote $\bvtheta_{i0}$. Therefore, $
\variance{Y_i \mid \G} = \variance{\bvtheta_{i} \mid \G} $.

From Figure~\ref{table:params}, we have,
$C = (27p)^2 B  \ge  (27p)^2 K \epsilon^{-2} n^{1-2/p}/ \log (n)$. Since the estimator $\bvtheta_i $ uses the \tpest~structure at level 0,
by Lemma~\ref{lem:hssavtp} (part (b)), we have, $\mu  = \expect{X_{ij0} \mid \G} = \abs{f_i}$  and by part (iv) of the same lemma,  $\eta_{ij0}^2 \le (2.7) \hat{F}_2/C $, for each $j \in R_0(i)$. Therefore, by Lemma~\ref{lem:vbvtheta},
\begin{align} \label{eq:var:midG0a}
\variance{\bvtheta_{i} \mid i \in \bar{G}_0, \G} & \le \left(\frac{(0.288)p^2}{k}\right)\abs{f_i}^{2p-2}\eta_{ij0}^2 \notag \\
&  \le
\left(\frac{ (0.288) p^2 \abs{f_i}^{2p-2}}{  (1000)(\log n)}\right)  \left(\frac{2.7 \hat{F}_2}{C}\right)  \notag \\
& \le \left(\frac{ (0.288)p^2 \abs{f_i}^{2p-2}}{(1000) (\log n)}\right) \left( \frac{ (2.7)(1.0005) F_2}{(27) p^2  K \epsilon^{-2} n^{1-2/p}/ (\log (n))} \right)\notag \\
 &  \le \frac{(0.3) \epsilon^2 \abs{f_i}^{2p-2} F_p^{2/p}}{(10)^4 K} \enspace .
\end{align}
where, the last step uses the fact that $F_2 \le F_p^{2/p} n^{1-2/p}$, for $p > 2$ from ~\eqref{eq:fpfacts}, and that $\hat{F}_2 \le (1+ 0.001/(2p))F_2$.

\emph{Case 2: $i \in \lmargin(G_0) \cup_{r=1}^L G_r.$}
If $i \in G_l$, then, $l_d(i) \in \{l,l-1\}$ and if $i \in \bar{G}_r$ then $l-1 \le r \le l+1$. By Lemma~\ref{lem:hssavtp},  $\eta_{ij l_d(i)} \le \abs{f_i}/(15p)$ for $j \in R_l(i)$.  From Lemma~\ref{lem:vbvtheta}, we have,  $$
\variance{\bvtheta_i \mid i \in \bar{G}_r, \G} = \left(\frac{(0.288)p^2 }{k} \right) \abs{f_i}^{2p-2} \eta^2_{ijl_d(i)} \le \frac{\abs{f_i}^{2p}}{(750)k} \enspace . $$ Hence,
\begin{align} \label{eq:varYi}
\variance{ Y_i \mid \G} 
&  = \variance{ \sum_{r=0}^{L} 2^r \bvtheta_i z_{ir}  \mid \G} \notag \\
&= \sum_{r=0}^L 2^{2r} \variance{ \bvtheta_i z_{ir}} + \sum_{\substack{0 \le r,r' \le L \\
r \ne r'}} 2^{r+r'}\covariance{\bvtheta_i z_{ir}}{\bvtheta_i z_{ir'}}\notag \\
& = \sum_{r=0}^L 2^{2r} \variance{ \bvtheta_i z_{ir}}
\end{align}
The last step follows since  $z_{ir} \cdot z_{ir'} = 0$ whenever $r \ne r'$, since $i$ may lie in only one sampled group. 

Simplifying ~\eqref{eq:varYi}, we have,
\begin{align} \label{eq:varYi1}
\variance{\bvtheta_i z_{ir}} &\le   \expect{\bvtheta_i^2 z_{ir}}  = \expect{ \bvtheta_i^2 \mid z_{ir}=1} \prob{ z_{ir} =1 }  
\end{align}
Assuming that $r$ is a level that is consistent with $i$ (otherwise $\prob{z_{ir} = 1\mid \G} = 0$), we have, by Lemma~\ref{lem:hssavtp}  that $\expect{\bvtheta_i \mid \G} \in \abs{f_i}^p(1 \pm \delta)$ where,  $ \eta_{i,j, l_d(i)} \le \abs{f_i}/(15p)$, for $j \in R_l(i)$. Using Lemma~\ref{lem:vbvtheta}, we obtain,
\begin{align}\label{eq:varYi2}
\expect{\bvtheta_i^2\mid z_{ir} = 1, \G}  & = \variance{\bvtheta_i \mid z_{ir} = 1, \G} + \left(\expect{ \bvtheta_i \mid z_{ir} =1, \G}  \right)^2  \notag \\
 &  \le \left(\frac{ (0.288) p^2 }{k}\right)( \abs{f_i}^{2p-2}) \left( \frac{\abs{f_i}^{2p}}{(15p)^2} \right) + \abs{f_i}^{2p} (1 + \delta) \notag \\
 & \le  \frac{\abs{f_i}^{2p}}{(750)k} +  \abs{f_i}^{2p} (1 + \delta)\notag \\
  &  \le  \abs{f_i}^{2p} (1.001)
\end{align}
where, $\delta \le  n^{-2500p}$.

Substituting ~\eqref{eq:varYi2} and  ~\eqref{eq:varYi1} into ~\eqref{eq:varYi}, we have,
\begin{align} \label{eq:varYi3a}
\variance{ Y_i \mid \G}  & \le  \sum_{r=0}^L 2^{2r} \expect{\bvtheta_i^2 z_{ir}\mid \G}  \notag \\
& \le
\abs{f_i}^{2p} (1.001) \sum_{r=0}^L 2^{2r} \prob{i \in \bar{G}_r \mid \G}  \notag  \\
 & \le 2^{l+1} (1.001) \abs{f_i}^{2p} \sum_{r=0}^L 2^r \prob{i \in \bar{G}_r \mid \G} \notag \\
  &\le (1.001) 2^{l+1} \abs{f_i}^{2p} (1+ \delta) \notag \\
 & \le (1.002) 2^{l+1} \abs{f_i}^{2p}
\end{align}
Step 2 uses ~\eqref{eq:varYi2}. Step 3 uses Lemma~\ref{lem:margin} to argue that if $i \in G_l$, then, $\prob{i \in \bar{G}_r \allowbreak  \mid \allowbreak \G}\allowbreak  = 0$ for all $ r > l+1$. Hence, the summation from $ r=0$ to $L$ is equivalent to $r$ ranging over  $ l-1, l$ and $l+1$. So the term $2^{2r} \le 2^{l+1} 2^r$.  The last step again uses Lemma~\ref{lem:margin} to note that $\sum_{r=0}^L 2^r \prob{i \in \bar{G}_r \mid \G} = 1 \pm O(2^ln^{-c})$.
 \hfill
\end{proof}

\subsection{Covariance of $Y_i$ and $Y_j$}

\begin{proof}[Proof of Lemma~\ref{lem:varcross}.] Let $i \ne j$, $i \in G_l$ and $j \in G_m$.
\begin{align}  \label{eq:YiYj}
& \covariance{Y_i}{Y_j \mid \G} = \expect{Y_iY_j \mid \G} - \expect{Y_i \mid \G} \expect{Y_j\mid \G} \notag \\
&  =
\expect{ \sum_{r=0}^L 2^{r} z_{ir} \bvtheta_i \sum_{r'=0}^L 2^{r'} z_{jr'} \bvtheta_{j}\mid \G} -\expect{\sum_{r=0}^L 2^{r} z_{ir} \bvtheta_i  \mid \G }\expect{\sum_{r'=0}^L 2^{r'} z_{jr'} \bvtheta_{j}\mid \G} \notag \\
& =  \sum_{0 \le r,r' \le L} 2^{r+r'}\expect{ \bvtheta_{i}\bvtheta_{j} \mid z_{ir}=1,  z_{jr'}=1,\G  } \prob{z_{ir} = 1, z_{jr'} = 1 \mid \G}  \notag \\
 & ~~~~~~-  2^{r+r'}\expect{\bvtheta_i \mid z_{ir} = 1 \mid  \G}  \expect{\bvtheta_{j}  \mid z_{jr'} = 1,\G} \prob{z_{ir} =1\mid \G} \prob{z_{jr'} = 1 \mid \G}\notag \\
 & = \sum_{0 \le r,r' \le L}   2^{r+r'} \sum_{ \hat{f}_i,\hat{f}_j
 } \E \Bigl[ \bvtheta_{i}\bvtheta_{j} \mid \hat{f}_i, \hat{f}_{j},  z_{ir}=1,  z_{jr'}=1, \G  \Bigr] \notag \\
&  \hspace*{1.3in}\cdot  \prob{\hat{f}_i,\hat{f}_{j} \mid z_{ir}=1,z_{jr' =1},\G} \cdot \prob{z_{ir} = 1, z_{jr'} = 1\mid \G}   \notag \\
 & ~~~~~~~~~~~-  2^{r+r'} \left(\sum_{\hat{f}_i
 }
 \expect{\bvtheta_i \mid \hat{f}_i, z_{ir} = 1,\G} \prob{\hat{f}_i \mid z_{ir} = 1,\G}\prob{z_{ir} =1\mid \G} \right) \notag \\
 &\left .  ~~~~~~~~~~~~ \left( \sum_{\hat{f}_{j}}
  \expect{\bvtheta_{j}  \mid \hat{f}_{j},  z_{jr'} = 1,\G}
  \prob{\hat{f}_{j} \mid z_{jr'=1},\G}  \prob{z_{jr'} = 1\mid \G} \right)\right] \enspace .
 \end{align}

By Lemma~\ref{lem:bvthetacross},
\begin{align} \label{eq:bvtija}
&\expect{ \bvtheta_{i}\bvtheta_{j} \mid \hat{f}_i, \hat{f}_{j},  z_{ir}=1,  z_{jr'}=1,\G  } \notag \\
&= \expect{\bvtheta_i \mid \hat{f}_i, \hat{f}_{j},  z_{ir}=1,  z_{jr'}=1,\G } \cdot   \expect{\bvtheta_{j} \mid \hat{f}_i, \hat{f}_{j},  z_{ir}=1,  z_{jr'}=1,\G  }
\enspace .
\end{align}

By Lemma~\ref{lem:fp:expect1}, for any value of $\hat{f}_i $ satisfying $\G$ and $i \in \bar{G}_r$ such that $r$ is consistent with $\abs{f_i}$,  we have,
$$\expectb{\bvtheta_i \mid \hat{f}_i, z_{ir}=1, E',\G} = \abs{f_i}^p(1 \pm \delta)  $$
 where, $E'$ is any subset (including the empty subset) of the events $\{\hat{f}_j \wedge z_{jr'} = 1\}$ and $\delta= O(n^{-2500p})$.

Substituting in ~\eqref{eq:bvtija} and for $r,r'$  consistent with $\abs{f_i}$ and $\abs{f_j}$ respectively, we have,
$$ \expect{ \bvtheta_{i}\bvtheta_{j} \mid \hat{f}_i, \hat{f}_{j},  z_{ir}=1,  z_{jr'}=1,\G  } = \abs{f_i}^p \abs{f_j}^p (1\pm O(\delta))$$
In a similar manner, it follows that
 $$\expect{\bvtheta_i \mid \hat{f}_i, z_{ir} =1,\G} \expect{\bvtheta_{j}\mid \hat{f}_{j}, z_{jr} = 1,\G} = \abs{f_i}^p\abs{f_{j}}^p (1 \pm O(\delta)) $$
Substituting these into ~\eqref{eq:YiYj}, we have,
 \begin{align} \label{eq:YiYj1}
 &\expect{Y_iY_j\mid \G} - \expect{Y_i \mid \G} \expect{Y_j \mid \G}  \notag \\
 & =\sum_{\substack {0 \le r,r' \le L \\ \text{ $r,r'$ consistent} \\ \text{ with  $i,j$ resp.} }} \biggl[2^{r+r'} \abs{f_i}^p \abs{f_j}^p (1 \pm O(\delta))
  \sum_{\hat{f}_i, \hat{f}_j}\prob{\hat{f}_i,\hat{f}_{j} \mid z_{ir}=1,z_{jr'} =1,\G}
 \cdot  \prob{z_{ir} = 1, z_{jr'} = 1\mid \G}  \notag \\
 & \hspace*{1.0in}-  2^{r+r'} \Bigl( \abs{f_i}^p \abs{f_{j}}^p (1\pm O(\delta)) \Bigr)\Bigl( \sum_{\hat{f}_i} 
  \prob{\hat{f}_i \mid z_{ir} = 1}\prob{z_{ir} =1,\G}\Bigr) \notag \\
  & \hspace*{1.5in} \cdot \Bigl( \sum_{\hat{f}_{j}}  
  \prob{\hat{f}_{j} \mid z_{jr'=1},\G}  \prob{z_{jr'} = 1\mid \G} \Bigr) \biggr]\notag \\
 & = \sum_{\substack {0 \le r,r' \le L \\ \text{ $r,r'$ consistent} \\ \text{ with  $i,j$ resp.} }}  \left[2^{r+r'} \abs{f_i}^p \abs{f_j}^p (1  \pm O(\delta))
\prob{z_{ir} = 1, z_{jr'} = 1\mid \G} \right.\notag \\
& \left.  ~~~~~~~~~~~~~-  2^{r+r'}  \abs{f_i}^p \abs{f_{j}}^p(1 \pm O(\delta)) \prob{z_{ir} =1\mid \G}\prob{z_{jr'} = 1 \mid \G} \right]
\end{align}
since,  each of the summations, namely,  (a) $\sum_{\hat{f}_i, \hat{f}_{j}} \prob{\hat{f}_i,\hat{f}_{j} \mid z_{ir}=1,z_{jr'} =1,\G} $, (b)
$\sum_{\hat{f}_i} \probb{\hat{f}_i \mid z_{ir} = 1,\G} $ and (c)   $\sum_{\hat{f}_{j}}  \prob{\hat{f}_{j} \mid z_{jr'=1},\G} $ are 1 respectively.

Further,
\begin{align*}
\sum_{\substack {0 \le r,r' \le L \\ \text{ $r,r'$ consistent} \\ \text{ with  $i,j$ resp.} }}  2^{r+r'}  \prob{z_{ir} = 1, z_{jr'} = 1\mid \G}  &  = \sum_{\substack {0 \le r,r' \le L  }}  2^{r+r'}  \prob{z_{ir} = 1, z_{jr'} = 1\mid \G} \\
& = 1 \pm O(2^l + 2^m)n^{-c}, \text{ by Lemma~\ref{lem:hss:ij}}
\end{align*}
since,  if  levels $r$ and  $r'$ are not consistent respectively with $\abs{f_i}$ and $\abs{f_j}$ respectively then $ \probb{z_{ir} =1, z_{jr'} = 1\mid \G} = 0$. The same applies to summations over $r$ of $2^r \prob{z_{ir}=1 \mid \G}$, etc..

By Lemma~\ref{lem:margin} (part 4),  $$\sum_{r=0}^L  2^r   \prob{z_{ir} = 1 \mid \G} = \sum_{\substack{\text{ $r$ consistent} \\ \text{ with $i$}}}  2^r  \cdot \prob{z_{ir} = 1 \mid \G}  = 1 \pm 2^{l} n^{-c} \enspace . $$
Similarly, $ \sum_{\text{ $r'$ consistent with $j$}} 2^{r'} \cdot \prob{j \in \bar{G}_{r'}} \in 1 \pm 2^m n^{-c}$. Combining and taking absolute values of both sides in ~\eqref{eq:YiYj1}  and replacing equality by $\le $,  we have,
\begin{align*} 
&\card{\covariance{Y_i}{Y_j \mid \G} } \\
& \le  \abs{f_i}^p \abs{f_j}^p \left( (1 \pm O(\delta)) (1\pm O(2^l + 2^m)n^{-c})
   - (1\pm O(\delta))(1 \pm O(2^l n^{-c}))(1\pm O(2^m n^{-c}))\right) \\
 & = \abs{f_i}^p \abs{f_j}^p O(\delta + (2^l + 2^m)n^{-c}) \\ 
 &= \abs{f_i}^p \abs{f_j}^p  \cdot  O(n^{-c+1}) ~~~
 \end{align*}
\eat{
Combining,
\begin{align} \label{eq:YiYj3}  & \sum_{0 \le r,r' \le L}  2^{r+r'} \left(\prob{z_{ir} = 1, z_{jr'} = 1\mid \G} - \prob{z_{ir} =1 \mid \G}\prob{z_{jr'} = 1 \mid \G}\right)    \notag \\
& = \sum_{0 \le r,r' \le L}  2^{r+r'} \cdot  \prob{z_{ir} = 1, z_{jr'} = 1} - \sum_{r} 2^r \cdot  \prob{z_{ir} =1} \sum_{r'} 2^{r'}\cdot \prob{z_{jr'}  = 1} \notag \\
& \in O(n^{-c} 2^{l+m})  = n^{-c+2}\enspace .
\end{align}
Substituting in ~\eqref{eq:YiYj2}, we have,
\begin{align*}
\card{\expect{Y_iY_{j} \mid \G} - \expect{Y_iY_j}\mid \G} \le \abs{f_i}^p \abs{f_j}^p \left( n^{-c+2} + O(n^{-10})\right)\enspace . ~~~~~~~~~~~~~~~~~~~~~~~~~~~~~~~~
\end{align*}
}
\end{proof}

\subsection{Variance of $\hat{F}_p$ estimator}
\begin{proof} [Proof of Lemma~\ref{lem:varFp}.] Let $K = 425 $, so that $B = Kn^{1-2/p} \epsilon^{-2}/\min(\log (n),\epsilon^{4/p-2}) \ge Kn^{1-2/p} \epsilon^{-4/p}$.
We have,
\begin{align}\label{eq:varFpa}
\varianceb{\hat{F}_p} & = \variance{\sum_{i \in [n]} Y_i\mid \G} \notag \\
 & \le     \sum_{i \in [n]} \variance{Y_i \mid \G} + \sum_{i \ne  j} \card{ \covariance{Y_i }{Y_j\mid \G}} \notag \\
& \le \sum_{i \in \midreg(G_0)} \variance{Y_i \mid \G}+  \sum_{i \in [n], i \not\in  \midreg(G_0)} \variance{Y_i \mid \G} +  F_p^2 \cdot O(n^{-c+1}) \notag \\
& \le \sum_{i \in \midreg(G_0)} \cfrac{  (0.3)\epsilon^2 \abs{f_i}^{2p-2}F_p^{2/p}}{(10)^4K} +
\left(\sum_{l=0}^{L} \sum_{i \in G_l, i \not\in \midreg(G_0)} 2^{l+1} (1.002) \abs{f_i}^{2p}\right)  + O(n^{-c+2})F_p^2
\end{align}
Step 3 follows from  Lemma~\ref{lem:varcross}, since,  $$\sum_{i \ne j}  \card{\covariance{Y_i}{Y_j \mid \G}} \le \sum_{i \ne j} O(n^{-c+1})\abs{f_i}^p \abs{f_j}^p  \le O(n^{-c+1} F_p^2) \enspace . $$ Step 4 uses Lemma~\ref{lem:varYi}.
 Since, $\hat{F}_2 \le F_2(1+0.01/(2p))$, $\hat{F}_2^{p/2} \le (1.01) F_2$.  Also, $F_2 \le F_p^{2/p} n^{1-2/p}$. Therefore,
\begin{equation} \label{eq:F2byBp} \left(\hat{F}_2/{B} \right)^{p/2}\le  \left(\frac{ (1.01)
F_2}{K n^{1-2/p} \epsilon^{-4/p}} \right)^{p/2}
\le  (1.01/K)^{p/2} \epsilon^2 F_p \enspace .
\end{equation}

For any set $S \subset [n]$ and $q \ge 0$,  let $F_q(S)$ denote $\sum_{i \in S} \abs{f_i}^q$.

Let $i \in \lmargin(G_0) \cup_{l=1}^{L} G_l$. 
By definitions of the parameters,
\begin{align*}
     \abs{f_i}  & \le T_{l-1} \le \left( \cfrac{F_2\left(1+ \frac{0.01}{2p}\right)}{(2\alpha)^{l-1}B} \right)^{1/2} & \text{ if $i \in G_l$ and $l \ge 1$, }  \\
\abs{f_i} & \le T_0(1+\epsbar) \le \left( \cfrac{F_2\left(1+ \frac{0.01}{2p}\right)}{B}\right)^{1/2} \left(1 + \frac{1}{27p}\right) & \text{ if $i \in \lmargin(G_0)$ \enspace . }
\end{align*}

We consider the first summation term of Eqn.~\eqref{eq:varFpa}, that is,
\begin{multline} \label{eq:varfp1}
 \sum_{i \in \midreg(G_0)} \cfrac{  (0.3)\epsilon^2 \abs{f_i}^{2p-2}F_p^{2/p}}{(10)^4K}  = \left(\frac{(0.3)\epsilon^2 F_p^{2/p}}{((10)^4)(425)} \right)F_{2p-2}(\midreg(G_0))  \\  \le \left(\frac{(0.3)\epsilon^2 F_p^{2/p}}{((10)^4)(425)}\right) F_p^{2-2/p} (\midreg(G_0))  \le \frac{(0.3) \epsilon^2 F_p^2}{((10)^4)(425)}
\end{multline}

We now consider the second summation term of Eqn. ~\eqref{eq:varFpa}, that is,
\begin{align} \label{eq:varFpaa}
 &\sum_{l=0}^{L} \sum_{i \in G_l, i \not\in \midreg(G_0)} 2^{l+1} (1.002) \abs{f_i}^{2p} \notag \\
 &  \le \sum_{i \in \lmargin(G_0)} (2)(1.002)(T_0(1+\epsbar))^p \abs{f_i}^p  + \sum_{l=1}^L \sum_{i \in G_l} 2^{l+1} T_{l-1}^p \abs{f_i}^p
 \end{align}
We will consider the two summations in Eqn.~\eqref{eq:varFpaa} separately.
\begin{align} \label{eq:varFpaaa}
& \sum_{i \in \lmargin(G_0)}  (2)(1.002)(T_0(1+\epsbar))^p \abs{f_i}^p  \notag \\
& \le  (2)(1.002)\left( \frac{F_2}{B}\right)^{p/2}(1.01)e^{1/27} F_p(\lmargin(G_0)) \notag \\
& \le  (2.11) \left( \frac{ n^{1-2/p} F_p^{2/p}}{ (425) n^{1-2/p} \epsilon^{-4/p}} \right)^{p/2} F_p(\lmargin(G_0)) \notag  \\
& = \left(\frac{1}{200} \right) \epsilon^2F_p \cdot F_p(\lmargin(G_0))
 \end{align}
 We now consider the second summation of Eqn.~\eqref{eq:varFpaa}.
\begin{align} \label{eq:varFpab}
& \sum_{l=1}^L \sum_{i \in G_l} 2^{l+1} T_{l-1}^p \abs{f_i}^p  \notag \\
& \le (2)(1.01) \sum_{l=1}^L 2^l\left( \frac{F_2}{(2\alpha)^{l-1} B} \right)^{p/2}  F_p(G_l)\notag  \\
& = (4)(1.01) \sum_{l=1}^L  2^l \left( \frac{ F_p^{2/p} n^{1-2/p}}{(425)(2\alpha)^{l-1} n^{1-2/p} \epsilon^{-4/p} } \right)^{p/2} F_p(G_l) \notag \\
& = \left(\frac{(4.04)}{(425)^{p/2}} \right) \epsilon^2 F_p \sum_{l=1}^L  2^l (2\alpha)^{-(l-1)(p/2)} F_p(G_l)
\end{align}
Further,
\begin{align} \label{eq:varFpc}
2^{l}
(2\alpha)^{(l-1)(-p/2)} &  = (2\alpha)^{p/2} 2^{l} (2\alpha)^{-lp/2} =  (2\alpha)^{p/2} 2^{l \left(1 - (p/2)\log_2 (2\alpha)\right)}  \enspace .
\end{align}
Let $ \gamma = 1-\alpha = (1-2/p)\nu$, where, $\nu = 0.01$. Therefore,
\begin{align} \label{eq:logalpha}\log_{2} (2\alpha)  &= 1 + \log_2 (\alpha) = 1 + \frac{\ln (\alpha)}{\ln 2}
= 1 + \frac{ \ln (1 - \gamma)}{\ln 2} \notag \\
 &\ge 1 - \frac{2\gamma}{\ln 2} = 1 - \frac{2(1-2/p)\nu}{\ln 2} \ge  1- (1-2/p)(3\nu)
\end{align}
Using eqn.~\eqref{eq:logalpha}, we can simplify the term $1-(p/2)\log_2 (2\alpha) $ as
\begin{align*}
1-(p/2)\log_2 (2\alpha) &  \le 1 - (p/2)\left( 1 - (1-2/p)(3\nu) \right)
= -(p/2-1)(1 - 3\nu)  = -(p/2-1)(0.97) < 0
\end{align*}
and is  a constant.

Substituting this into Eqn.~\eqref{eq:varFpc} and then into ~\eqref{eq:varFpab}, we have,
\begin{align} \label{eq:varFpd}
& \sum_{l=1}^L \sum_{i \in G_l} 2^{l+1} T_{l-1}^p \abs{f_i}^p  \notag \\
& \le \left(\frac{(4.04) }{(425)^{p/2}}\right) \epsilon^2 F_p \sum_{l=1}^L  2^l (2\alpha)^{-(l-1)(p/2)} F_p(G_l) \notag \\
& \le (4.04)\left(\frac{2\alpha }{425}\right)^{p/2} \epsilon^2 F_p \sum_{l=1}^L \left(2^{- (p/2-1)(1-3\nu)}\right)^l F_p(G_l) \notag \\
& \le \left( \frac{\epsilon^2}{53}\right) F_p  \sum_{l=1}^L F_p(G_l) \notag \\
& \le \left( \frac{\epsilon^2}{53}\right) F_p \cdot F_p( \cup_{l=1}^L G_l) \enspace .
\end{align}

Adding Eqns.~\eqref{eq:varFpaaa} and ~\eqref{eq:varFpd}, Eqn.~\eqref{eq:varFpaa} becomes
\begin{align} \label{eq:varFpe}
& \sum_{l=0}^{L} \sum_{i \in G_l, i \not\in \midreg(G_0)} 2^{l+1} (1.002) \abs{f_i}^{2p}  \notag \\
&  \le \left(\frac{1}{200} \right) \epsilon^2F_p \cdot F_p(\lmargin(G_0))  + \left( \frac{\epsilon^2}{53}\right) F_p \cdot F_p( \cup_{l=1}^L G_l) \notag \\
& \le \frac{\epsilon^2 F_p^2}{53}
\end{align}
Substituting  Eqn.~\eqref{eq:varfp1} and Eqn.~\eqref{eq:varFpe} in Eqn.~\eqref{eq:varFpa}, we have, for $n $ sufficiently large, that
\begin{align*}
\variance{\hat{F}_p} \le \frac{\epsilon^2 F_p^2}{50} ~~~~~
\end{align*}
\end{proof}

\subsection{Putting things together}

\begin{proof} [ Proof of Theorem ~\ref{thm:fp}.] Consider the \ghss~algorithm  using the parameters of Figure~\ref{table:params}.  By Lemma~\ref{lem:hss:G}, $\G$ holds except with probability $n^{-c}$, where, $c > 23$.
From Lemma~\ref{lem:varFp}, $\varianceb{\hat{F}_p} \le \epsilon^2 F_p^2/50$.  Using Chebychev's inequality,
 \begin{equation} \label{eq:hatfp}\prob{  \card{\hat{F}_p - \expectb{\hat{F}_p}} \le (\epsilon/2) F_p  \mid \G} \ge 1- \frac{\variance{F_p}}{((\epsilon/2) F_p)^2}  = 1-\frac{4}{50}  \enspace .
 \end{equation}
By Lemma~\ref{lem:Y}, $  \card{\expectb{\hat{F}_p\mid \G} - F_p} \le F_p (2^{L+1} n^{-c})$. 
Combining with Eqn.~\eqref{eq:hatfp}, by triangle inequality, we have,
$$ \prob{ \card{\hat{F}_p - F_p} \le \left((\epsilon/2) + 2^{L+1}n^{-c})\right)F_p \mid \G} \ge 1 - \frac{4}{50} $$
which implies that 
$$ \prob{ \card{\hat{F}_p - F_p} \le \epsilon F_p \mid \G} \le \frac{46}{50} \enspace . $$
since, $2^L \ll n$. 

\noindent
Since $\prob{\G} \ge 1-n^{-c}$,
unconditioning w.r.t. $\G$, we have,
$$ \prob{ \card{\hat{F}_p - F_p} \le \epsilon F_p } \ge \frac{ 46}{50} \left( 1 - O(n^{-c}) \right) \ge 0.9 \enspace . $$

The space required  at level 0 is $C_0s = Cs$, at level $l$ it is $C_l s$ and at level $L$ it is $16 C_L s$. Here, $s = 8k = 8(1000) \log (n) = O(\log n)$.  Further $C_l = 4(\alpha)^l C$. Thus, the total space is of the order of
\begin{align*}
&   \sum_{l=0}^L C_l s= (\log (n)) \sum_{l=0}^L  \alpha^l C
\le  \frac{C \log (n) }{1-\alpha} = \frac{C \log (n)}{(1-2/p)\nu}=  O\left(\frac{ n^{1-2/p} \log (n) \epsilon^{-2}}{\min(\log (n), \epsilon^{4/p-2})}\right)\enspace  .
\end{align*}
The last expression for space may also be written as $ O\bigl(  n^{1-2/p} \epsilon^{-2} + n^{1-2/p} \epsilon^{-4/p} \allowbreak  \log (n) \bigr)$.

The time taken to process each stream update consists of applying  the $L$ hash functions $g_1, \ldots, g_L$ to an item $i$. Each hash function is $O(\log n)$-wise independent and requires time $O(\log n)$ to evaluate it at a point. The time to evaluate $L = \log_{2\alpha} (n/C)$ functions is $O(\log^2 n)$. Additionally, for each level, the hash values for $i$ have to be computed for each of the $s$ hash functions of the $\hh_l$ and $\tpest_l$ structures. These hash functions are $O(1)$-wise independent, and they can collectively be computed in $O(Ls) = O(\log^2 n)$ time.  This proves the statement of the theorem. \hfill

\end{proof}

\eat{

 Continuing with ~\eqref{eq:varFpa}, except for  the small term $n^{-c+2}F_p^2$ which we add later,
\begin{align} \label{eq:varFpb}
\varianceb{\hat{F}_p} & \le  \frac{ \epsilon^2 F_p^{2/p} F_{2p-2}(\midreg(G_0)) }{(8000) K} + \sum_{i \in \lmargin(G_0)} (2)(1.01)(T_0(1+\epsbar))^p \abs{f_i}^p  + \sum_{l=1}^L \sum_{i \in G_l} 2^{l+1} T_{l-1}^p \abs{f_i}^p \notag \\
& \le \frac{ \epsilon^2 F_p^{2/p} F_p^{2-2/p}}{(8000)K} + (2)(1.01)e^{1/25}\Bigl( \frac{\hat{F}_2}{B} \Bigr)^{p/2}F_p(\lmargin(G_0)) + (2)(1.01) \sum_{l =1}^L  2^{l} \left( \frac{ \hat{F}_2}{ (2\alpha)^{l-1} B} \right)^{p/2} F_p(G_l) \notag \\
& = \frac{ \epsilon^2 F_p^2}{(8000)K} + \frac{2.12\epsilon^2 F_p}{K} \left( F_p(\lmargin(G_0)  + \sum_{l=1}^L 2^{l}
(2\alpha)^{(l-1)(-p/2)} F_p(G_l)\right)
\end{align}

Now,
\begin{align*} 
1-\frac{p}{2}\left(1+\frac{\ln (\alpha)}{\ln(2)}\right) \le 1 - \frac{p}{2}\left (1 -\frac{2\nu(1-2/p)}{\ln 2} \right)
= 1 -\frac{p}{2} + \frac{2\nu}{\ln 2} (p/2-1) = (p/2-1)\left(-1 + \frac{2\nu}{\ln 2}) \right) = \beta \text{ (say)}
\end{align*}
where $\beta$ is a negative constant, since, $p > 2$ and $\nu = 0.01$. Substituting in ~\eqref{eq:varFpc}, we obtain that  $2^{l+1}
(2\alpha)^{(l-1)(-p/2)} \le 2 (2\alpha)^{p/2} 2^{l\beta} \le 2 (2\alpha)^{p/2}$. Substituting this into ~\eqref{eq:varFpb}, we have,
\begin{align} \label{eq:varFpe}
\varianceb{\hat{F}_p} &  \le  \frac{ \epsilon^2 F_p^2}{(8000) K} + \frac{2.12 (2\alpha)^{p/2}\epsilon^2 F_p}{K} \left( F_p(\lmargin(G_0)  + \sum_{l=1}^{L}  F_p(G_l)\right)  +n^{-c+2} F_p^2 \notag  \\
 &= \frac{ \epsilon^2 F_p^2}{(8000) K} + \frac{2.12 (2\alpha)^{p/2}\epsilon^2 F_p^2}{K} + n^{-c+2} F_p^2  \le \frac{\epsilon^2F_p^2}{200} \enspace .
\end{align}
by using the value of $K = 425 (2\alpha)^{p/2}$. }

 \eat{
 \begin{lemma} \label{lem:Y}
 $\expectb{\hat{F}_p} =  F_p (1\pm O(n^{-c+1})$.
 \end{lemma}
 \begin{proof} Suppose $i \in G_l$. Then,  by  linearity of expectation,
\begin{align*}
\expect{Y_i \mid \G} & = \sum_{l'=0}^{L}2^{l'} \expect{ z_{il'} \bvtheta_i \mid \G}= \sum_{l'=0}^{L}2^{l'} \expect{\bvtheta_i \mid i \in \bar{G}_{l'}, \G} \prob{i \in \bar{G}_{l'} \mid \G} = \sum_{l'=0}^{L}2^{l'}  \abs{f_i}^p (1 \pm n^{-80}) \prob{i \in \bar{G}_{l'} \mid \G}\\ &   = \abs{f_i}^p (1 \pm n^{-80}) (1 \pm O(2^l n^{-c})) = \abs{f_i}^p (1 \pm O(2^ln^{-c}))
\end{align*}
Step 3 follows from  Lemma~\ref{lem:fp:expect1} and step 4 follows from Lemma~\ref{lem:margin}.
Hence, \\
$$\expectb{ \hat{F}_p} = \expectbb{ \sum_{i\in [n]} Y_i} = \sum_{l=0}^L \sum_{i \in G_l}  \expect{Y_i} =  \sum_{l=0}^L \sum_{i \in G_l} \abs{f_i}^p(1 \pm O(2^{l(i)}  n^{-c})) = F_p(1\pm O(2^L n^{-c}))  \enspace . $$
Let $ \gamma = 1-\alpha= (1-2/p)\nu$, $\nu = 0.01$.  Let $C = K n^{1-2/p}$, for some value of $K$ and let $K' = K^{-1(1+2\gamma/\ln2 )} $.  We have,
\begin{multline*}2^L  = 2^{\lceil \log_{2\alpha} (n/C) \rceil } \le 2\cdot  2^{\log_2 (n/C)/\log_{2}(2\alpha)} = 2\cdot (n/C)^{1/(1+ \ln \alpha/\ln 2)} \le  2 (n/C)^{1 + 2\gamma/\ln 2} \\ \le 2 (K'n^{(2/p)(1 + (1-2/p)\nu/\ln2)} = 2K' n^{1- (1-2/p)(1-2\nu/p\ln (2))} = n^{1-a}
\end{multline*}for $p > 2$ and $\nu < \ln 2$, where $a > 0$ and $a = \Theta(1)$ . Hence, $O(2^L n^{-c}) = O(n^{-(c-1)})$, proving the lemma. \hfill
 \end{proof}
 }

\eat{
In the following proofs, we will use the notion that the  \emph{sample group of an item is  consistent with the  frequency of the item} to mean that if $i \in G_l$ and $i$ is sampled into $\bar{G}_r$, then, $l$ and $r$ are related as given by Lemma~\ref{lem:margin}, conditional on  \G.  (For e.g., if $i \in \lmargin(G_l)$, then, $r \in \{l,l+1\}$, if $i \in \midreg(G_l)$, then, $r =l$, and if $i \in \rmargin(G_l)$, then, $r \in \{l-1,l\}$).

We will use the following facts.
\begin{equation} \label{eq:fpfacts}
F_2 \le n^{1-2/p} F_p^{2/p} \text{ and }  F_{2p-2} \le F_p^{2-2/p}, ~~~\text{ for $p \ge 2$.}
\end{equation}

\begin{lemma} \label{lem:varYi} Let $K = (425)(2\alpha)^{p/2}$.
 \begin{align*}
 \variance{Y_i \mid \G} \le \begin{cases} \cfrac{  \epsilon^2 \abs{f_i}^{2p-2}F_p^{2/p}}{(25)^3K}& \text{    if $i \in \midreg(G_0)$} \\
 2^{l+1} (1.002) \abs{f_i}^{2p} & \text{ $(i \in G_l$ for some $l \ge 1)$ or $(i \in \lmargin(G_0)$. }\
 \end{cases}
 \end{align*}
\end{lemma}

\begin{proof} For this proof, assume that $\G$ holds. \\
\emph{Case 1:} $i \in \midreg(G_0)$. Hence,  $i \in \bar{G}_0$ with probability 1 and $l_d(i) = 0$. Therefore,  $Y_i = \sum_{l=0}^{L-1} 2^l z_{il} \bvtheta_{il} = \bvtheta_{i}$, since, $z_{i0} = 1$ and $z_{il} = 0$ for $l > 0$. Therefore, $
\variance{Y_i \mid \G} = \variance{\bvtheta_{i} \mid \G} $. From Figure~\ref{table:params}, we have,
$C = (25p)^2 B = (25p)^2 K \epsilon^{-2} n^{1-2/p}/\min(\log (n), \epsilon^{4/p}) \ge (25p)^2 K \epsilon^{-2} n^{1-2/p}/ \log (n)$.
 From the \tpest~structure at level 0  and using Lemma~\ref{lem:vbvtheta}, we have, $\mu  = \expect{X_{ij0} \mid \G} = \abs{f_i}$ and from Lemma~\ref{lem:fp:expect1}, $\eta_{i0}^2 \le 3 \ftwores{C}/{C}$. Therefore,
\begin{align} \label{eq:var:midG0a}
\variance{\bvtheta_{i} \mid \G} & \le \frac{3}{(5k)} p^2 \abs{f_i}^{2p-2}\eta_{i0}^2 \le
\frac{ (3) p^2 \abs{f_i}^{2p-2}}{5 (16 \log (n))} \cdot \frac{ (3) F_2}{C}  \le \eat{\frac{ 3p^2 \abs{f_i}^{2p-2}F_2}{5 (16 \log n) (25p)^2  K n^{1-2/p} \epsilon^{-2}/\log(n)}}
 \frac{  \epsilon^2 \abs{f_i}^{2p-2}F_2}{(25)^3K n^{1-2/p}}~ \le \frac{\epsilon^2 \abs{f_i}^{2p-2} F_p^{2/p}}{(25)^3 K}
\end{align}
where, the last step uses the fact that $F_2 \le F_p^{2/p} n^{1-2/p}$, for $p > 2$. \\
\noindent
\emph{Case 2: $i \in \lmargin(G_0) \cup_{r=1}^L G_r.$}
If $i \in G_l$, then, $l_d(i) \le l $ and if $i \in \bar{G}_r$ then $l-1 \le r \le l+1$. By Lemma~\ref{lem:hssavtp},  $\eta_{i, l_d(i)} \le \abs{f_i}/(13p)$, \[ \begin{array}{l}
\variance{\bvtheta_i \mid i \in \bar{G}_r, \G} = 3p^2 \abs{f_i}^{2p-2} \eta_{i,l_d(i)}^2/(5k) \le \abs{f_i}^{2p}/(20k)
\enspace . \end{array} \]
Then
\begin{align} \label{eq:varYi}
&\variance{ Y_i \mid \G} = \variance{ \sum_{r=0}^{L} 2^r \bvtheta_i z_{ir}  \mid \G}  = \expect{ \Bigl( \sum_{r=0}^L 2^r \bvtheta_i z_{ir}   \Bigr)^2 \mid \G} -
\left( \expect{ \sum_r 2^r \bvtheta_i z_{ir}   \mid \G} \right)^2  \notag \\
& =  \sum_{r=0}^{L} 2^{2r} \expect{\bvtheta_i^2 z_{ir}\mid \G} + \sum_{r\ne r', 0 \le r,r' \le L}
2^{r+r'}\expect{\bvtheta_i^2 z_{ir} z_{ir'} \mid \G} \notag \\
& ~~~~ -  \sum_{r=0}^{L}  2^{2r} \left( \expect{ \bvtheta_i z_{ir} \mid \G} \right)^2  -
  \sum_{r\ne r', 0 \le r,r' \le L } 2^{r+r'} \expect{ \bvtheta_i z_{ir} \mid \G} \expect{ \bvtheta_i z_{ir'} \mid \G}  \notag \\
& = \sum_{r=0}^L 2^{2r} \expect{\bvtheta_i^2 z_{ir}\mid \G} - \sum_r 2^{2r} \left( \expect{ \bvtheta_i z_{ir} \mid \G} \right)^2 \le  \sum_{r=0}^L 2^{2r} \expect{\bvtheta_i^2 z_{ir}\mid \G}
\end{align}
The second-last step follows from the fact that, (i)  $z_{ir}z_{ir'} = 0$ whenever $r \ne r'$, since by construction, $i$ may lie in only one sampled group. We now simplify ~\eqref{eq:varYi}.
\begin{align} \label{eq:varYi1}
\expect{\bvtheta_i^2 z_{ir}\mid \G} & = \expect{ \bvtheta_i^2 \mid z_{ir}=1,\G} \prob{ z_{ir} = 1 \mid \G}
\end{align}
Assuming that $r$ is a level that is consistent with $i$ (otherwise $\prob{z_{ir} = 1\mid \G} = 0$), we have, by Lemma~\ref{lem:hssavtp}  that $\expect{\bvtheta_i \mid \G} \in \abs{f_i}^p(1 \pm O(n^{-80}))$, and $ \eta_{i, l_d(i)} \le \abs{f_i}/(13p)$. Using Lemma~\ref{lem:vbvtheta}, we obtain,
\begin{align}\label{eq:varYi2}
\expect{\bvtheta_i^2\mid z_{ir} = 1, \G} &  = \variance{\bvtheta_i \mid z_{ir} = 1, \G} + \left(\expect{ \bvtheta_i \mid z_{ir} =1, \G}  \right)^2  \le ( 3p^2 \abs{f_i}^{2p-2} \eta_{i,l_d(i)}^2)/(5k) + \abs{f_i}^{2p} (1 + O(n^{-10}))  \notag  \\
& \le \frac{ \abs{f_i}^{2p}}{280(16) (\log n)} + \abs{f_i}^{2p} (1 + O(n^{-10}))  \le  \abs{f_i}^{2p-2} (1.001) \enspace .
\end{align}
Substituting ~\eqref{eq:varYi2} and  ~\eqref{eq:varYi1} into ~\eqref{eq:varYi}, we have,
\begin{align} \label{eq:varYi3a}
\variance{ Y_i \mid \G}  & \le  \sum_{r=0}^L 2^{2r} \expect{\bvtheta_i^2 z_{ir}\mid \G} \le
(1.001) \abs{f_i}^{2p} \sum_{r=0}^L 2^{2r} \prob{i \in \bar{G}_r \mid \G} \notag  \\
 & \le 2^{l+1} (1.001) \abs{f_i}^{2p} \sum_{r=0}^L 2^r \prob{i \in \bar{G}_r \mid \G}  \le (1.001) 2^{l+1} \abs{f_i}^{2p} (1+ n^{-c})  \le (1.002) 2^{l+1} \abs{f_i}^{2p}\enspace .
\end{align}
Step 2 uses ~\eqref{eq:varYi2}, step 3 uses Lemma~\ref{lem:margin} to argue that if $i \in G_l$, then, $\prob{i \in \bar{G}_r \mid \G} = 0$ for all $ r > l+1$. Hence, the summation from $ r=0$ to $L$ is equivalent to $r$ ranging over  $ l-1, l$ and $l+1$. Hence the term $2^{2r} \le 2^{l+1} 2^r$.  The last step again uses Lemma~\ref{lem:margin} to note that $\sum_{r=0}^L 2^r \prob{i \in \bar{G}_r \mid \G} = 1 \pm O(2^ln^{-c})$.
 \hfill
\end{proof}

Lemma~\ref{lem:varcross} shows that the contribution of the cross terms of the form $\expect{Y_iY_j \mid \G}$, for $i \ne j$ is very small. This lemma essentially uses the results of Lemmas~\ref{lem:hssj} and ~\ref{lem:hsscond}. It is proved in the Appendix.
\begin{lemma} \label{lem:varcross} Let $i \ne j$. Then,
$ \card{\expect{ Y_i Y_{j} \mid \G} - \expect{Y_i \mid \G}\expect{Y_{j} \mid \G}} \le    o(n^{-2})\abs{f_i}^p \abs{f_j}^p$.
\end{lemma}
} 
\eat{
\begin{lemma}\label{lem:varFp}
$\variance{\hat{F}_p \mid \G} \le  \epsilon^2 F_p^2/{100} \enspace . $
\end{lemma}
\begin{proof}
We have,
\begin{align}\label{eq:varFpa}
\variance{\hat{F}_p} & = \variance{\sum_{i \in [n]} Y_i\mid \G}  = \sum_{i \in [n]} \variance{Y_i \mid \G} + \left(\sum_{i \ne  j}  \expect{Y_iY_j \mid \G} - \expect{ Y_i \mid \G} \expect{Y_j\mid \G}\right) \notag \\
& \le     \sum_{i \in [n]} \variance{Y_i \mid \G} + \left(\sum_{i \ne  j} \card{ \expect{Y_iY_j \mid \G} - \expect{ Y_i \mid \G} \expect{Y_j\mid \G}}\right) \notag \\
& \le \sum_{i \in \midreg(G_0)} \variance{Y_i \mid \G}+  \sum_{i \in [n], i \not\in  \midreg(G_0)} \variance{Y_i \mid \G} +  F_p^2 \cdot O(n^{-c+2}) \notag \\
& \le \sum_{i \in \midreg(G_0)} \cfrac{  \epsilon^2 \abs{f_i}^{2p-2}F_p^{2/p}}{(25)^3K} +
\left(\sum_{l=0}^{L} \sum_{i \in G_l, i \not\in \midreg(G_0)} 2^{l+1} (1.002) \abs{f_i}^{2p}\right)  + n^{-c+2}F_p^2
\end{align}

Step 4 follows from  Lemma~\ref{lem:varcross}, since,  $ \sum_{i \ne j}  \bigl\lvert \expect{Y_iY_j \mid \G} - \expect{ Y_i \mid \G} \expect{Y_j\mid \G}\bigr\rvert \le \sum_{i \ne j} O(n^{-c})\abs{f_i}^p \abs{f_j}^p  \le O(n^{-c} F_p^2)$. Step 5 uses Lemma~\ref{lem:varYi}. \\
Let $K' = 425 (2\alpha)^{p/2}$, so that $B = K' n^{1-2/p} \epsilon^{-2}/\min(\log (n),\epsilon^{4/p-2})$. Then,
\begin{equation} \label{eq:F2byBp} \left(\hat{F}_2/{B} \right)^{p/2}\le  \left(\frac{ (1 +0.001/(2p))
F_2}{K' n^{1-2/p} \epsilon^{-4/p}} \right)^{p/2}
\le  (1.001) \epsilon^2 F_p/{K'} \enspace .
\end{equation}
since, $\hat{F}_2 \le (1+0.001/p)$.

For any set $S \subset [n]$ and $q \ge 0$,  let $F_q(S)$ denote $\sum_{i \in S} \abs{f_i}^q$. For $i \in G_l$, for $l \ge 1$, $\abs{f_i} \le T_{l-1} = (F_2/((2\alpha)^{l-1})B)^{1/2}$. For $i \in \lmargin(G_0)$, $ \abs{f_i} \le T_0(1+\epsbar) = (F_2/B)^{1/2}(1+1/(25p))$, and $(1+\epsbar)^p \le e^{p/\epsbar} = e^{1/25}$.

 Continuing with ~\eqref{eq:varFpa}, except for  the small term $n^{-c+2}F_p^2$ which we add later, we obtain,
\begin{align} \label{eq:varFpb}
\varianceb{\hat{F}_p} & \le  \frac{ \epsilon^2 F_p^{2/p}}{(25)^3 K} F_{2p-2}(\midreg(G_0)) + \sum_{i \in \lmargin(G_0)} (2.004)(T_0(1+\epsbar))^p \abs{f_i}^p  + \sum_{l=1}^L \sum_{i \in G_l} 2^{l+1} T_{l-1}^p \abs{f_i}^p \notag \\
& \le \frac{ \epsilon^2 F_p^{2/p}}{(25)^3 K} F_p^{2-2/p} + (2.004)e^{1/25}\left( \frac{\hat{F}_2}{B} \right)^{p/2}F_p(\lmargin(G_0)) + (2.004) \sum_{l =1}^L  2^{l+1} \left( \frac{ \hat{F}_2}{ (2\alpha)^{l-1} B} \right)^{p/2} F_p(G_l) \notag \\
& = \frac{ \epsilon^2 F_p^2}{(25)^3 K} + \frac{2.12\epsilon^2 F_p}{K} \left( F_p(\lmargin(G_0)  + \sum_{l=1}^L 2^{l+1}
(2\alpha)^{(l-1)(-p/2)} F_p(G_l)\right)
\end{align}
Let $ \gamma = 1-\alpha = (1-2/p)\nu$, where, $\nu = 0.01$. Therefore, $ \ln \alpha \ge -2\gamma = -2(1-2/p)\nu$.
\begin{align} \label{eq:varFpc}
2^{l+1}
(2\alpha)^{(l-1)(-p/2)} =   2 (2\alpha)^{p/2} 2^{l -(lp/2)\log_2 (2\alpha)}
=  2 (2\alpha)^{p/2} 2^{l(1-(p/2)(1 + \ln (\alpha)/\ln (2)))}  \enspace .
\end{align}
Now,
\begin{align} \label{eq:varFpd}
1-\frac{p}{2}\left(1+\frac{\ln (\alpha)}{\ln(2)}\right) \le 1 - \frac{p}{2}\left (1 -\frac{2\nu(1-2/p)}{\ln 2} \right)
= 1 -\frac{p}{2} + \frac{2\nu}{\ln 2} (p/2-1) = (p/2-1)\left(-1 + \frac{2\nu}{\ln 2}) \right) = \beta \text{ (say)}
\end{align}
where $\beta$ is a negative constant, since, $p > 2$ and $\nu = 0.01$. Substituting in ~\eqref{eq:varFpc}, we obtain that  $2^{l+1}
(2\alpha)^{(l-1)(-p/2)} \le 2 (2\alpha)^{p/2} 2^{l\beta} \le 2 (2\alpha)^{p/2}$. Substituting this into ~\eqref{eq:varFpb}, we have,
\begin{multline} \label{eq:varFpe}
\varianceb{\hat{F}_p}  \le  \frac{ \epsilon^2 F_p^2}{(25)^3 K} + \frac{2.12 (2)(2\alpha)^{p/2}\epsilon^2 F_p}{K} \left( F_p(\lmargin(G_0)  + \sum_{l=1}^{L}  F_p(G_l)\right)  +n^{-c+2} F_p^2 \\ = \frac{ \epsilon^2 F_p^2}{(25)^3 K} + \frac{2.12 (2)(2\alpha)^{p/2}\epsilon^2 F_p^2}{K} + n^{-c+2} F_p^2  \le \frac{\epsilon^2F_p^2}{100} \enspace .
\end{multline}
by using the value of $K = 425 (2\alpha)^{p/2}$.  \hfill
}

\eat{

We will consider the summations in ~\eqref{eq:varFpa} separately. Since, $F_{2p-2} \le F_p^{1-2/p}$ for $p > 2$, and $F_2 \le F_p^{2/p}n^{1-2/p}$.

Recall that $C = (25p)^2 B = (25p)^2 (425 (2\alpha)^{p/2}) n^{1-2/p} \epsilon^{-2}/\min(\log (n), \epsilon^{4/p-2}) $. Hence, $C = \max(C^1, C^2)$ where, $C^1 = (25p)^2 (425 (2\alpha)^{p/2}) n^{1-2/p} \epsilon^{-2}/\log (n) $ and $C^2 = (25p)^2  (425 (2\alpha)^{p/2}) n^{1-2/p}\epsilon^{-4/p}$.

 Consider the first summation. Using $k = 10\log (n)$, we have,
\begin{align}\label{eq:varFpa1}
 \sum_{i \in \midreg(G_0)} \cfrac{ 3(1.01) p^2  F_2 \abs{f_i}^{2p-2}}{5kC} &  \le \sum_{i \in \midreg(G_0)} \cfrac{ 3 p^2 (1.01) F_2 \abs{f_i}^{2p-2}}{5(10 \log (n)) C^1}
= \sum_{i \in \midreg(G_0)} \cfrac{(3.03) (  F_p^{2/p} n^{1-2/p}) \abs{f_i}^{2p-2}}{(10 \log (n)) (5) (25)^2 (425 (2\alpha)^{p/2}) n^{1-2/p} \epsilon^{-2}/\log n} \notag  \\
 & \le \frac{(3.03)\epsilon^2 F_{2p-2} F_p^{2/p}}{(5) (10) (25)^2 (425 (2\alpha)^{p/2}) } \le  \frac{ \epsilon^2 F_p^{2-2/p} F_p^{2/p} n^{1-2/p}}{(5) (10) (25)^2 (425 (2\alpha)^{p/2}) } = \frac{\epsilon^2 F_p^2}{(10) (25)^2 (425 (2\alpha)^{p/2})} \enspace .
\end{align}
Step 1 follows since $C \ge C^1$, step 2 substitutes $C^1 = (25p)^2  (425 (2\alpha)^{p/2}) n^{1-2/p}\epsilon^{-4/p}$,  step 3 uses $F_2 \le F_p^{2/p} n^{1-2/p}$ and that $ \sum_{i \in \midreg(G_0)} \abs{f_i}^{2p-2} \le F_{2p-2} $, step 4 uses $F_{2p-2} \le F_p^{(2p-2)/p}$ since $p > 2$, and step 5 is a simplification step.

We now consider the second summation of ~\eqref{eq:varFpa}. For any set of items $S \subset [n]$, let $F_p(S)$ denote $\sum_{i \in S} \abs{f_i}^p$. As in the previous calculation, we can write $T_l = F_2/((2\alpha)^l B)$, where, $B = 425 (2\alpha)^{p/2} n^{1-2/p} \epsilon^{-2}/\min(\epsilon^{4/p-2}, \log n)$. Therefore, $B \ge 425 (2\alpha)^{p/2} n^{1-2/p} \epsilon^{-4/p}$. Also, we will use $F_2 \le F_p^{2/p} n^{1-2/p}$. Therefore,
\begin{align} \label{eq:f2byB} \left( \frac{F_2}{B} \right)^{p/2} \le \left( \frac{ F_p^{2/p} n^{1-2/p}}{425 (2\alpha)^{p/2} n^{1-2/p} \epsilon^{-4/p}}\right)^{p/2}
= \frac{ \epsilon^2 F_p}{425 (2\alpha)^{p/2}} \enspace .
\end{align}
Further,  for $l \ge 1$,
\begin{align} \label{eq:f2byB2}
2^{l+1} T_{l-1}^p & = 2^{l+1} \left( \frac{F_2}{(2\alpha)^{l-1} B} \right)^{p/2}  = 2^{l+1} (2\alpha)^{p/2}  \left( \frac{F_p^{2/p} n^{1-2/p}}{ (2\alpha)^{l} (425 (2\alpha)^{p/2}) n^{1-2/p} \epsilon^{-4/p}} \right)^{p/2}  =  \frac{ 2\epsilon^2 F_p }{425 } 2^{l (1 - (p/2)( \log (2\alpha)))}
\end{align}
Further, recall that $ \alpha = 1- \gamma$, where, $\gamma = (1-2/p)\nu$ and $\nu = 0.01$. Therefore, $\ln \alpha \ge -2\gamma$ and hence,
\begin{multline} \label{eq:beta}
1 - (p/2) \log (2\alpha) =  1 - (p/2)\left (1 + \frac{ \ln \alpha}{\ln 2} \right)  \le 1- (p/2) + \frac{ (p/2) (2\gamma)}{\ln 2} =   1- (p/2) +\frac{ (p/2)(2) (1-2/p)\nu}{\ln 2} \\
= 1-p/2 + (p/2-1) \frac{ 2\nu}{\ln 2} = (1-p/2)\left(1- \frac{2\nu}{\ln 2} \right)  = \beta {\text (say)}
\end{multline}
Clearly, $\beta < 0$, since, $p > 2$.  Thus, ~\eqref{eq:f2byB2} simplifies to
\begin{align} \label{eq:f2byB3}
2^{l+1} T_{l-1}^p & = \frac{ 2 \epsilon^2 F_p}{425} 2^{l\beta}, ~~~l \ge 1 \enspace .
\end{align}

 For $i \in G_l$ with $l \ge 1$,  $\abs{f_i} \le T_{l-1} = (\hat{F}_2/(2\alpha)^{l-1} B)^{1/2}$, and for $i \in \lmargin(G_0)$, $\abs{f_i} \le T_0(1+\epsbar) = (\hat{F}_2 (1+\epsbar)/B)^{1/2}$.

\begin{align} \label{eq:varFp4}
&  \sum_{l=0}^{L} \sum_{i \in G_l, i \not\in \midreg(G_0)}  (1.03)2^{l+1} \abs{f_i}^{2p}  =
\sum_{i \in \lmargin(G_0)}  2 (1.03)  \abs{f_i}^{p} (T_0(1+\epsbar))^p  + \sum_{l=1}^L \sum_{i \in G_l}   (1.03) 2^{l+1} \abs{f_i}^{p} T_{l-1}^p\notag \\
& =  (2.06)(1+\epsbar))^p \left( \frac{F_2}{B} \right)^{p/2} \sum_{i \in \lmargin(G_0)}  \abs{f_i}^p
+ \sum_{l=1}^L   (1.03) \sum_{i \in G_l}  2^{l\beta}\abs{f_i}^p \notag \\
&  \le    \left( \frac{ 2.15 \epsilon^2 F_p}{425 (2\alpha)^{p/2}} \right)  F_p(\lmargin(G_0)) + \frac{ (2.06) \epsilon^2 F_p}{425} \sum_{l=1}^L    2^{l \beta}  F_p(G_l) \notag \\
& \le \frac{2.15 \epsilon^2 F_p}{425 } \left( F_p(\lmargin(G_0)) + \sum_{l=1}^L   2^{l \beta} F_p(G_l) \right) \notag \\
& \le \frac{2.15 \epsilon^2 F_p^2 }{425}
\end{align}
Step 1 is obtained by separating the summation into two parts, (i) items in $\lmargin(G_0)$, and (ii) items in $G_l$, for $l \ge 1$. In step 2, we use the fact that for $i \in \lmargin(G_0)$, $\abs{f_i} \le T_0(1+ \epsbar) $ and for $i \in G_l$, $\abs{f_i} \le T_{l-1}$. In step 3, we use $(1+\epsbar)^p = (1+1/(25p))^p \le 1.041$ and use ~\eqref{eq:f2byB} to simplify. Also, we use ~\eqref{eq:f2byB3} to simplify $2^{l+1} T_{l-1}^p$. In step 4, we use that  $F_p(\lmargin(G_0)) + \sum_{l=1}^L   2^{l \beta} F_p(G_l) \le F_p(\lmargin(G_0)) + \sum_{l=1}^L   F_p(G_l) \le F_p$, since, $2^{l\beta} < 1$ as $\beta < 0$.

Substituting ~\eqref{eq:varFp4} and ~\eqref{eq:varFpa1} into ~\eqref{eq:varFpa} we have,
\begin{align*}
\variance{\hat{F}_p} \le \frac{\epsilon^2 F_p^2}{(10) (25)^2 (425 (2\alpha)^{p/2})} +  \frac{2.15 \epsilon^2 F_p^2 }{425} + n^{-c+2} F_p^2 \le \frac{\epsilon^2 F_p^2}{100}
\end{align*}

\end{proof}
}
\end{appendices}


\begin{thebibliography}{10}

\bibitem{ams:jcss98}
Noga Alon, Yossi Matias, and Mario Szegedy.
\newblock {``The space complexity of approximating frequency moments''}.
\newblock {\em Journal of Computer Systems and Sciences}, 58(1):137--147, 1998.
\newblock Preliminary version appeared in Proceedings of \textit{ACM} Symposium
  on Theory of Computing (\textit{STOC}) 1996, pp. 1-10.

\bibitem{ako:arxiv10}
Alexander Andoni, Robert Krauthgamer, and Krzysztof Onak.
\newblock {``Streaming Algorithms via Precision Sampling''}.
\newblock In {\em Proceedings of IEEE Foundations of Computer Science
  {(FOCS)}}, 2011.
\newblock A version appears in arXiv:1011.1263v1 [cs.DS] November 2010.

\bibitem{anpw:icalp13}
Alexandr Andoni, Huy~L. Nguyen, Yury Polyanskiy, and Yihong Wu.
\newblock {``Tight Lower Bound for Linear Sketches of Moments''}.
\newblock In {\em Proceedings of International Conference on Automata,
  Languages and Programming, (ICALP)}, July 2013.
\newblock Version published as arXiv:1306.6295, June 2013.

\bibitem{dipw:soda10}
Kanh~Do Ba, Piotr Indyk, Eric Price, and David Woodruff.
\newblock {``Lower bounds for sparse recovery''}.
\newblock In {\em Proceedings of ACM Symposium on Discrete Algorithms (SODA)},
  2008.

\bibitem{B-YJKS:stoc02}
Z.~Bar-Yossef, T.S. Jayram, R.~Kumar, and D.~Sivakumar.
\newblock {``An information statistics approach to data stream and
  communication complexity''}.
\newblock In {\em Proceedings of \textit{ACM} Symposium on Theory of Computing
  \textit{STOC}}, pages 209--218, 2002.

\bibitem{bgks:soda06}
L.~Bhuvanagiri, S.~Ganguly, D.~Kesh, and C.~Saha.
\newblock {``Simpler algorithm for estimating frequency moments of data
  streams''}.
\newblock In {\em Proceedings of ACM Symposium on Discrete Algorithms (SODA)},
  pages 708--713, 2006.

\bibitem{bksv:arxiv14}
Vladimir Braverman, Jonathan Katzman, Charles Seidell, and Gregory Vorsanger.
\newblock {``Approximating Large Frequency Moments with $O(n^{1-2/k})$ Bits''}.
\newblock In {\em Proceedings of International Workshop on Randomization and
  Computation (RANDOM)}, 2014.
\newblock Published earlier as {arXiv:1401.1763}, January 2014.

\bibitem{braost:arxiv10}
Vladimir Braverman and Rafail Ostrovsky.
\newblock {``Recursive Sketching For Frequency Moments''}.
\newblock arXiv:1011.2571v1 [cs.DS], November 2010.

\bibitem{crt:ieeetit06a}
Emmanuel Cand\`{e}s, Justin Romberg, and Terence Tao.
\newblock {``Robust uncertainty principles: Exact signal reconstruction from
  highly incomplete frequency information''}.
\newblock {\em IEEE Trans. Inf. Theory}, 52(2):489–--509, February 2006.

\bibitem{css:colt10}
Nicol\`{o} Cesa-Bianchi, Shai~Shalev Shwartz, and Ohad Shamir.
\newblock {``Online Learning of Noisy Data with Kernels''}.
\newblock In {\em Proceedings of ACM International Conference on Learning
  Theory (COLT)}, 2010.

\bibitem{cks:ccc03}
A.~Chakrabarti, S.~Khot, and X.~Sun.
\newblock {``Near-Optimal Lower Bounds on the Multi-Party Communication
  Complexity of Set Disjointness''}.
\newblock In {\em Proceedings of International Conference on Computational
  Complexity (CCC)}, 2003.

\bibitem{ccf:icalp02}
Moses Charikar, Kevin Chen, and Martin Farach-Colton.
\newblock {``Finding frequent items in data streams''}.
\newblock {\em Theoretical Computer Science}, 312(1):3--15, 2004.
\newblock Preliminary version appeared in Proceedings of ICALP 2002, pages
  693-703.

\bibitem{cm:sirocco06}
Graham Cormode and S.~Muthukrishnan.
\newblock {``Combinatorial Algorithms for Compressed Sensing''}.
\newblock In {\em Proceedings of International Colloquium on Structural
  Information \& Communication Complexity, (\textit{SIROCCO})}, 2006.

\bibitem{donoho:tit06}
David~L. Donoho.
\newblock {``Compressed Sensing''}.
\newblock {\em IEEE Trans. Inf. Theory}, 52(4):1289–--1306, April 2006.

\bibitem{gl:hssfull}
S.~Ganguly and L.~Bhuvanagiri.
\newblock {``Hierarchical Sampling from Sketches: Estimating Functions over
  Data Streams''}.
\newblock {\em Algorithmica}, 53:549--582, 2009.

\bibitem{gks:fsttcs05}
S.~Ganguly, D.~Kesh, and C.~Saha.
\newblock {``Practical Algorithms for Tracking Database Join Sizes''}.
\newblock In {\em Proceedings of Foundations of Software Technoogy and
  Theoretical Computer Science (FSTTCS)}, pages 294--305, Hyderabad, India,
  December 2005.

\bibitem{g:arxiv11b}
Sumit Ganguly.
\newblock {``A Lower Bound for Estimating High Moments of a Data Stream''}.
\newblock arXiv:1201.0253, December 2011.

\bibitem{g:isaac12}
Sumit Ganguly.
\newblock {``Precision vs. Confidence Tradeoffs for $\ell_2$-Based Frequency
  Estimation in Data Streams''}.
\newblock In {\em Proceedings of International Symposium on Algorithms,
  Automata and Computation (ISAAC), LNCS Vol. 7676}, pages 64--74, 2012.

\bibitem{indy:focs00}
Piotr Indyk.
\newblock Stable distributions, pseudorandom generators, embeddings, and data
  stream computation.
\newblock {\em J. ACM}, 53(3):307--323, 2006.
\newblock Preliminary Version appeared in Proceedings of IEEE FOCS 2000, pages
  189-197.

\bibitem{indwoo:stoc05}
Piotr Indyk and David Woodruff.
\newblock {``Optimal Approximations of the Frequency Moments''}.
\newblock In {\em Proceedings of \textit{ACM} Symposium on Theory of Computing
  \textit{STOC}}, pages 202--298, Baltimore, Maryland, USA, June 2005.

\bibitem{jw:soda11}
T.S. Jayram and David Woodruff.
\newblock {``Optimal Bounds for Johnson-Lindenstrauss Transforms and Streaming
  Problems with Low Error''}.
\newblock In {\em Proceedings of ACM Symposium on Discrete Algorithms (SODA)},
  2011.

\bibitem{jst:pods11}
Hossein Jowhari, Mert S\u{a}glam, and G\'{a}bor Tardos.
\newblock {``Tight Bounds for Lp Samplers, Finding Duplicates in Streams, and
  Related Problems''}.
\newblock In {\em Proceedings of ACM International Symposium on Principles of
  Database Systems (PODS)}, 2011.

\bibitem{liwood:random13}
Yi~Li and David Woodruff.
\newblock {``A Tight Lower Bound for High Frequency Moment Estimation with
  Small Error''}.
\newblock In {\em Proceedings of International Workshop on Randomization and
  Computation (RANDOM)}, 2013.

\bibitem{mw:soda10}
Morteza Monemizadeh and David Woodruff.
\newblock ``1-pass relative-error $l_p$-sampling with applications''.
\newblock In {\em Proceedings of ACM Symposium on Discrete Algorithms (SODA)},
  2010.

\bibitem{nisan:stoc90}
N.~Nisan.
\newblock {``Pseudo-Random Generators for Space Bounded Computation''}.
\newblock In {\em Proceedings of \textit{ACM} Symposium on Theory of Computing
  \textit{STOC}}, pages 204--212, May 1990.

\bibitem{pw:focs11}
Eric Price and David Woodruff.
\newblock {``$(1+\epsilon)$-approximate Sparse Recovery''}.
\newblock In {\em Proceedings of IEEE Foundations of Computer Science
  {(FOCS)}}, 2011.

\bibitem{ConcreteMath:book}
Oren~Patashnik Ronald L.~Graham, Donald E.~Knuth.
\newblock {\em {``Concrete Mathematics A Foundation for Computer Science''}}.
\newblock Addison-Wesley, 1994.

\bibitem{sss:soda93}
J.~Schmidt, A.~Siegel, and A.~Srinivasan.
\newblock {``Chernoff-Hoeffding Bounds with Applications for Limited
  Independence''}.
\newblock In {\em Proceedings of ACM Symposium on Discrete Algorithms (SODA)},
  pages 331--340, 1993.

\bibitem{singh:sankhya64}
R.~Singh.
\newblock {``Existence of unbiased estimates''}.
\newblock {\em Sankhya: The Indian Journal of Statistics}, 26(1):93--96, 1964.

\bibitem{tz:soda04}
M.~Thorup and Y.~Zhang.
\newblock {``Tabulation based 4-universal hashing with applications to second
  moment estimation''}.
\newblock In {\em Proceedings of ACM Symposium on Discrete Algorithms (SODA)},
  pages 615--624, New Orleans, Louisiana, USA, January 2004.

\bibitem{wood:soda04}
David~P. Woodruff.
\newblock {``Optimal space lower bounds for all frequency moments''}.
\newblock In {\em Proceedings of ACM Symposium on Discrete Algorithms (SODA)},
  pages 167--175, 2004.

\bibitem{wz:stoc12}
David~P. Woodruff and Qin Zhang.
\newblock {``Tight Bounds for Distributed Functional Monitoring''}.
\newblock In {\em Proceedings of \textit{ACM} Symposium on Theory of Computing
  \textit{STOC}}, 2012.

\end{thebibliography}
\end{document}